\documentclass[numberwithinsect,cleveref,thm-restate]{no-lipics-v2019}
\usepackage{bm}

\SetKwComment{Comment}{$\triangleright$\ \rm}{}

\usetikzlibrary{arrows,arrows.meta,shapes.misc, decorations.pathreplacing}
\tikzset{dotmark/.style={circle,fill,inner sep=1.5pt}}
\tikzset{emptymark/.style={circle,draw,fill=white,inner sep=1.5pt}}
\tikzset{crossmark/.style={thick,inner sep=1.5pt}}

\newcommand{\Oh}{\mathcal{O}}
\newcommand{\tOh}{{\tilde{\Oh}}}

\newcommand{\eps}{\varepsilon}
\newcommand{\lcp}{\mathsf{lcp}}

\newcommand{\cO}{\mathcal{O}}

\newcommand{\ceil}[1]{\lceil #1 \rceil}
\newcommand{\floor}[1]{\lfloor #1 \rfloor}
\newcommand{\per}{\operatorname{per}}
\newcommand{\rot}{\operatorname{rot}}

\newcommand{\hd}{\delta_H}
\newcommand{\OccEx}{\mathrm{Occ}}
\newcommand{\Occ}{\mathrm{Occ}^H}
\newcommand{\OccE}{\mathrm{Occ}^E}
\newcommand{\MIS}{\mathrm{Mis}}
\newcommand{\ed}{\delta_E}

\newcommand{\G}{\mathcal{G}}

\newcommand{\X}{\mathcal{X}}

\newcommand{\gen}{\textsf{gen}}
\newcommand{\val}{\textsf{val}}
\newcommand{\Tr}{\mathsf{PT}}

\newcommand{\sub}{\subseteq}

\newcommand{\edl}[2]{{\delta_E}(#1,{}^*\!#2^*)}
\newcommand{\eds}[2]{{\delta_E}(#1,{}^*\!#2)}

\def\modelname{{\tt PILLAR}\xspace}
\def\lceOp#1#2{{\tt LCP}(#1, #2)}
\def\lcbOp#1#2{{\tt LCP}^R(#1, #2)}
\def\ipmOp#1#2{{\tt IPM}(#1, #2)}
\def\perOp#1{{\tt Period}(#1)}

\def\accOp#1#2{#1\position{#2}}
\def\ipmOpName{{\tt IPM}\xspace}
\def\misOpName{{\tt Mismatches}\xspace}

\def\misOpLName{{\tt MismGenerator}\xspace}
\def\mibOpLName{{\tt MismGenerator$^R$}\xspace}
\def\accOpName{{\tt Access}\xspace}
\def\extractOpName{{\tt Extract}\xspace}
\def\lenOpName{{\tt Length}\xspace}
\def\perOpName{{\tt Period}\xspace}
\def\lceOpName{{\tt LCP}\xspace}
\def\lcbOpName{{\tt LCP$^R$}\xspace}
\def\cycEqOpName{{\tt Rotations}\xspace}

\def\misOp#1#2{{\tt Mismatches}(#1,#2)}

\def\misOpL#1#2{{\tt MismGenerator}(#1,{#2}^*)}
\def\mibOpL#1#2{{\tt MismGenerator}^R(#1,{}^*\!{#2})}
\def\cycEqOp#1#2{{\tt Rotations}(#1,#2)}

\newcommand{\concat}{\ensuremath{\mathtt{concat}}}
\newcommand{\makestring}{\ensuremath{\mathtt{makestring}}}
\newcommand{\splitOp}{\ensuremath{\mathtt{split}}}

\DeclareMathOperator*{\polylog}{polylog}

\def\substr{\ensuremath \preccurlyeq}
\def\fragmentco#1#2{\bm{[}\,#1\,\bm{.\,.}\,#2\,\bm{)}}
\def\fragmentoc#1#2{\bm{(}\,#1\,\bm{.\,.}\,#2\,\bm{]}}
\def\fragmentoo#1#2{\bm{(}\,#1\,\bm{.\,.}\,#2\,\bm{)}}
\def\fragment#1#2{\bm{[}\,#1\,\bm{.\,.}\,#2\,\bm{]}}
\def\position#1{\bm{[}\,#1\,\bm{]}}

\makeatother

\title{Faster Approximate Pattern Matching:\texorpdfstring{\\}{}
A Unified Approach}
\titlerunning{Faster Approximate Pattern Matching: A Unified Approach}

\author{Panagiotis Charalampopoulos}{Department of~Informatics, King's College London,
United Kingdom \and Institute of~Informatics, University of~Warsaw, Poland}{panagiotis.charalampopoulos@kcl.ac.uk}{https://orcid.org/0000-0002-6024-1557}{Partially supported by ERC grant TOTAL  under the
European Union’s Horizon 2020 Research and Innovation Programme (agreement no. 677651).}
\author{Tomasz Kociumaka}{Department of Computer Science, Bar-Ilan University, Ramat Gan, Israel
\and University of California, Berkeley, U.S.}{kociumaka@berkeley.edu}{https://orcid.org/0000-0002-2477-1702}{Supported by ISF grants no. 1278/16 and 1926/19, by a BSF grant no. 2018364, and by an ERC grant MPM under the EU's Horizon 2020 Research and Innovation Programme (agreement no. 683064).}
\author{Philip Wellnitz}{Max Planck Institute for Informatics,
    Saarland Informatics Campus (SIC),
Saarbrücken, Germany}{wellnitz@mpi-inf.mpg.de}{https://orcid.org/0000-0002-6482-8478}{}

\authorrunning{P. Charalampopoulos, T. Kociumaka, and P. Wellnitz}

\Copyright{Panagiotis Charalampopoulos, Tomasz Kociumaka, and Philip Wellnitz}

\nolinenumbers

\usepackage{pgf}
\usepgflibrary{fpu}

\makeatletter
\newcommand{\eviltwo}[1]{
    \pgfkeys{/pgf/fpu}
    \pgfkeys{/pgf/fpu/output format=float}
    \pgfmathsetmacro{\@Res}{#1}%
    \pgfmathfloattoint{\@Res}
    \pgfkeys{/pgf/fpu=false}%
}
\makeatother

\def\alphav{128}
\def\betav{8}
\def\deltavN{3}
\def\deltavD{8}
\pgfmathsetmacro{\ubjNu}{int(\betav*\deltavN+2*\deltavD)}
\pgfmathsetmacro{\ubjDu}{int(\betav*\deltavD)}
\pgfmathsetmacro{\ubjNr}{int(\ubjNu/gcd(\ubjNu,\ubjDu))}
\pgfmathsetmacro{\ubjDr}{int(\ubjDu/gcd(\ubjNu,\ubjDu))}
\def\ubjv{\ubjNr/\ubjDr\cdot}
\pgfmathsetmacro{\alphavd}{int(2*\alphav)}
\def\gammav{\alphavd}
\pgfmathsetmacro{\betavh}{int(\betav/2)}

\pgfmathsetmacro{\gammap}{int(12*\betav)}
\pgfmathsetmacro{\tbetav}{int(2*\betav)}

\pgfmathsetmacro{\ialphavbh}{int(2*\alphav/\betav)}

\pgfmathsetmacro{\alphavq}{int(\alphav/4)}
\pgfmathsetmacro{\alphavh}{int(\alphav/2)}
\pgfmathsetmacro{\alphavb}{int(\alphav/\betav)}
\def\threehalfs{{}^3{\mskip -4mu/\mskip -3.5mu}_2\,}
\pgfmathsetmacro{\gbetavh}{int(12 * \tbetav)}

\pgfmathsetmacro{\gammabv}{int(\gbetavh * \betavh*\deltavN / (\betavh*\deltavN - \deltavD))}

\pgfmathsetmacro{\thmbound}{int(max(\gammap, max(\gammabv, \gammav)))}
\def\thmboundt{\alphav}

\def\qvarphiv{152}

\pgfmathsetmacro{\pvarphiv}{int(2 * \qvarphiv)}
\pgfmathsetmacro{\varphiv}{int(4 * \pvarphiv)}

\pgfmathsetmacro{\alphavt}{int(4 * \alphav)}
\pgfmathsetmacro{\alphavdt}{int(4 * \alphavd)}
\pgfmathsetmacro{\varphibq}{int(8 * \varphiv)}
\eviltwo{\tbetav * \varphiv}
\let\varphivb\pgfmathresult

\pgfmathsetmacro{\betavpf}{int(3+ \betav)}

\eviltwo{\tbetav * \tbetav * \varphiv}

\eviltwo{\betavpf * \varphivb }
\let\varphivbq\pgfmathresult

\eviltwo{\varphivbq * \betavh * (\deltavN/(\deltavN*\betavh-\deltavD))}
\let\varphig\pgfmathresult

\eviltwo{\tbetav * \varphiv}

\eviltwo{9 * \varphibq}
\let\gammaEp\pgfmathresult

\eviltwo{max(\alphavdt, max(\gammaEp, \varphig))}
\let\Ethmboundt\pgfmathresult

\begin{document}
\maketitle

\begin{abstract}
    In the approximate pattern matching problem, given a
    text $T$, a pattern $P$, and a threshold $k$, the task is to find (the starting positions of) all
    substrings of~$T$ that are at distance at most $k$ from $P$.
    We consider the two most fundamental string metrics:
    Under the \emph{Hamming distance}, we search for substrings of~$T$ that have at most $k$
    \emph{mismatches} with~$P$, while under the \emph{edit distance}, we search for substrings of~$T$
    that can be transformed to $P$ with at most $k$ \emph{edits}.

    Exact occurrences of~$P$ in $T$ have a very simple structure: If we assume for simplicity
     that $|P|<|T|
    \le \threehalfs |P|$ and that $P$ occurs both as a prefix and as a suffix of $T$, then
    both $P$ and $T$ are periodic with a common period.
    However, an analogous characterization for occurrences
    with up to $k$ mismatches was proved only recently by Bringmann~et~al.\ [SODA'19]:
    Either there are $\Oh(k^2)$ $k$-mismatch occurrences of~$P$ in~$T$,
    or both $P$ and $T$ are at Hamming distance
    $\Oh(k)$ from strings with a common string period of length $\Oh(m/k)$.
    We tighten this characterization by showing that there are $\Oh(k)$ $k$-mismatch
    occurrences in the non-periodic case,
    and we lift it to the edit distance setting,
    where we tightly bound the number of~$k$-error occurrences by $\Oh(k^2)$ in the
    non-periodic case.
    Our proofs are constructive and let us obtain a
    unified framework for approximate pattern matching with respect to both considered distances.
    In particular, we provide meta-algorithms that only rely on a small set of
    primitive operations.
    We showcase the generality of~our meta-algorithms with results for the
    following settings:
    \begin{itemize}
        \item The \emph{fully compressed} setting, where
            both $T$ and $P$ are given as straight-line programs
            of~sizes $n$ and $m$, respectively.
            Here, we obtain an $\tOh((n + m) k^2)$-time
            and an $\tOh((n + m) k^4)$-time
            algorithm for pattern matching with mismatches and edits, respectively.
            Note that while our algorithms are the first
            to work in the fully compressed setting (that is, without first decompressing the input),
            they also improve the state of the art for the setting where only
            the text is compressed: For pattern matching with mismatches, we improve
            the dependency on $k$ from $\tOh((n + |P|) k^4)$ [Bringmann~et~al.
            SODA'19]; for pattern matching with edits, we improve the overall running
            time from $\tOh(n\sqrt{|P|}\,k^3)$ [Gawrychowski, Straszak, ISAAC'13].
        \item The \emph{dynamic setting}, where we maintain a collection of~strings
            $\mathcal{X}$ of~total length~$N$ using the data structure of~Gawrychowski
            et~al.\ [SODA'18], which supports each of~the operations ``split'',
            ``concatenate'' and ``insert a length-$1$ string'' in
            $\Oh(\log N)$ time with high probability.
            Here, for any two strings $T, P \in \mathcal{X}$, we can compute all
            occurrences of~$P$ in $T$ with up to $k$ mismatches in time $\tOh(|T|/|P|\cdot k^2)$
            or up to $k$ edits in time $\tOh(|T|/|P|\cdot k^4)$.
        \item The \emph{standard setting}, where $T$ and $P$ are given explicitly.
            Here, we obtain an $\Oh(|T| + |T|/|P|\cdot k^2\log\log k)$-time algorithm for the Hamming
            distance case (improving $\polylog |T|$ factors compared to the deterministic algorithm by Clifford et~al.\ [SODA'18] and matching, up to the $\log\log k$ factor, the randomized algorithm by Chan et al.\ [STOC'20], the state of the art for $k\le \sqrt{|P|}$),
            and an $\Oh(|T| + |T|/|P|\cdot k^4)$-time algorithm for the edit
            distance case (matching the algorithm by Cole and Hariharan [J.~Comput.'02],
            the state of the art for $k\le \sqrt[3]{|P|}$).
    \end{itemize}
\end{abstract}

\clearpage
\section{Introduction}

The \emph{pattern matching} problem is perhaps the most fundamental problem on strings:
Given a pattern $P$ and a text $T$, the task is to find all occurrences of $P$ in $T$.
However, in most applications, finding
all \emph{exact} occurrences of~a pattern is not enough: Think of~human spelling
mistakes or DNA sequencing errors, for example. In this work, we focus on
\emph{approximate} pattern matching, where we are interested in finding substrings of~the
text that are ``similar'' to the pattern.
While various similarity measures are imaginable, we study the two
most commonly encountered metrics in this context: the \emph{Hamming distance}
and the \emph{edit distance}.

\paragraph*{Hamming Distance}
Recall that the Hamming distance of~two
(equal-length) strings is the number of~positions where the strings differ.
Now, given a text $T$ of~length $n$, a~pattern $P$ of~length $m$, and an integer threshold $k>0$,
we want to compute the \emph{$k$-mismatch occurrences} of~$P$ in $T$,
that is, all length-$m$ substrings of~$T$ that are at Hamming distance at most
$k$ from $P$.
This \emph{pattern matching with mismatches} problem has been extensively studied.
In the late 1980s, Abrahamson~\cite{Abrahamson} and Kosaraju~\cite{Kosaraju} independently
proposed an FFT-based $\cO(n\sqrt{m \log m})$-time algorithm for computing the Hamming
distance of~$P$ and all the length-$m$ fragments of~$T$.
While their algorithms can be used to solve the pattern matching with mismatches problem,
the first algorithm to benefit from the threshold $k$ was given by
Landau and Vishkin~\cite{LandauV86} and slightly improved by Galil and Giancarlo~\cite{GG86}: Based on so-called ``kangaroo
jumping'', they obtained an $\Oh(nk)$-time algorithm, which is faster than $\cO(n\sqrt{m
\log m})$ even for moderately large $k$.
Amir et al.~\cite{AmirLP04} developed two algorithms with running time
$\cO(n\sqrt{k \log k})$ and $\tOh(n+k^3 n/m)$, respectively;
the latter algorithm was then improved upon by Clifford et
al.~\cite{CliffordFPSS16}, who presented an $\tOh(n+k^2 n/m)$-time solution.
Subsequently, Gawrychowski and Uznański~\cite{GawrychowskiU18} provided a smooth trade-off between
the running times $\tOh(n\sqrt{k})$ and $\tOh(n+k^2 n/m)$ by designing an
$\tOh(n+ kn/\!\sqrt{m})$-time algorithm.
Very recently, Chan et al.~\cite{cgkkp20} removed most of the $\polylog n$ factors in the
latter solution at the cost of (Monte-Carlo) randomization.
Furthermore, Gawrychowski and Uznański~\cite{GawrychowskiU18} showed that a significantly
faster ``combinatorial'' algorithm would
have (unexpected) consequences for the complexity of~Boolean matrix multiplication.
Pattern matching with mismatches on strings is thus well understood in the standard setting.
Nevertheless, in the settings where the strings are not given explicitly,
a similar understanding is yet to be obtained.
One of~the main contributions of~this work is to improve the upper bounds for two
such settings, obtaining algorithms with running times analogous to the algorithm
of Clifford et al.~\cite{CliffordFPSS16}.

\paragraph*{Edit Distance}
Recall that the edit distance (also known as the Levenshtein distance) of two strings $S$ and $T$ is the minimum number of~edits required to transform $S$ into $T$.
Here, an edit
is an insertion, a substitution, or a deletion of a single character.
In the \emph{pattern matching with edits} problem, we are given a text $T$, a pattern $P$,
and an integer threshold $k>0$, and the task is to find the starting positions of all the \emph{$k$-edit} (or \emph{$k$-error}) \emph{occurrences}
of~$P$ in $T$. Formally, we are to find all positions $i$ in $T$ such that the edit distance between ${T\fragment{i}{j}}$ and $P$ is at most $k$ for some position $j$.
Again, a classic algorithm by Landau and Vishkin~\cite{LandauV89} runs in $\Oh(nk)$ time.
Subsequent research~\cite{SV96,ColeH98} resulted in an $\cO(n+k^4 n/m)$-time algorithm (which is faster for $k\le \sqrt[3]{m}$).
From a lower-bound perspective, we can benefit from the discovery that the classic quadratic-time
algorithm for computing the edit distance of two strings is essentially optimal:
Backurs and Indyk~\cite{bi18} recently proved that a significantly faster algorithm would
yield a major breakthrough for the satisfiability problem.
For pattern matching with edits, this means that there is no hope for an
algorithm that is significantly faster than $\Oh(n + k^2n/m)$; however,
apart from that ``trivial'' lower bound and the 20-year-old conjecture of~Cole and
Hariharan~\cite{ColeH98} that an $\Oh(n + k^3n/m)$-time algorithm \emph{should be possible},
nothing is known that would close this gap.
While we do not manage to tighten this gap, we do believe that
the structural insights we obtain may be useful for
doing so. What we do manage, however, is to significantly improve the running time of~the
known algorithms in two settings where $T$ and $P$ are not given explicitly,
thereby obtaining running times that can be seen as analogous to the running time of~Cole
and Hariharan's algorithm~\cite{ColeH98}.

\paragraph*{Grammar Compression}

One of~the settings that we consider in this paper is the
\emph{fully compressed} setting, where both the text~$T$ and the pattern $P$
are given as straight-line programs. Compressing the text and the pattern is, in general, a natural
thing to do---think of~huge natural-language texts or genomic databases, which are easily compressible.
While one approach to solve pattern matching in the fully compressed setting is to first
decompress the strings and then run an algorithm for the standard setting, this voids most benefits
of compression in the first place. Hence, there has been a long line of~research
with the goal of designing text algorithms directly operating on compressed strings.
Naturally, such algorithms highly depend on the chosen compression method.
In this work, we consider \emph{grammar compression},
where a string $T$ is represented using a context-free grammar that generates the singleton language $\{T\}$;
without loss of generality, such a grammar is a \emph{straight-line program} (SLP).

Straight-line programs are popular due to mathematical
elegance and equivalence~\cite{r03,KP18,KK20} (up to logarithmic factors and moderate constants) to
widely-used dictionary compression schemes,
including the LZ77 parsing~\cite{lz77} and the run-length-encoded Burrows--Wheeler transform~\cite{BWT}.
Many more schemes, such as
byte-pair encoding~\cite{shishia99}, Re-Pair~\cite{LM00}, Sequitur~\cite{nw97}, and further members of~the Lempel--Ziv family~%
\cite{lz78, w84}, to name but a few, can be expressed as straight-line programs.
We refer an interested reader to~\cite{r04,ab10, l12, sa14}
to learn more about grammar compression.

Working directly with a compressed representation of~a text, intuitively at
least, seems to be hard in general---in fact, Abboud et al.~\cite{abbk17} showed that, for some problems,
decompress-and-solve is the best we can hope for, under some reasonable assumptions from
fine-grained complexity theory. Nevertheless, Jeż~\cite{talg/Jez15} managed to prove
that exact pattern matching can be solved on grammar-compressed strings in near-linear
time: Given an SLP of~size $n$ representing a string $T$ and an SLP of~size
$m$ representing a string~$P$, we can find all exact occurrences of~$P$ in $T$
in $\Oh((n + m)\log |P|)$ time.
For fully compressed \emph{approximate} pattern matching, no such near-linear time
algorithm is known, though. While the $\tOh((n + |P|) k^4)$-time algorithm by Bringmann et al.~\cite{bkw19}
for pattern matching with mismatches comes close, it works in an easier setting where only the text is compressed.
We fill this void by providing the first algorithm for fully compressed pattern matching
with mismatches that runs in near-linear time.
Denote by $\Occ_k(P, T)$ the set of (starting positions of) $k$-mismatch occurrences
of~$P$ in $T$; then, our result reads as follows.

\begin{restatable}{mtheorem}{gchdalgmain}\label{gc_hd_alg_intro}
    Let $\G_T$ denote an SLP of~size~$n$ generating a text~$T$,
    let~$\G_P$ denote an SLP of~size~$m$ generating a pattern~$P$,
    let $k>0$ denote an integer threshold,
    and set $N := |T| + |P|$.

    Then, we can compute $|\Occ_k(P, T)|$ in time
    $\Oh(m\log N + n\, k^2 \log^2 N \log\log N)$.
    The elements of $\Occ_k(P, T)$ can be reported within $\Oh(|\Occ_k(P, T)|)$ extra
    time.\lipicsEnd
\end{restatable}

\noindent For pattern matching with edits, near-linear time
algorithms are not known even in the case that the pattern is given explicitly.
Currently, the best pattern matching algorithms on an SLP-compressed text run in
time $\Oh(n|P|\log
|P|)$~\cite{t14} and $\Oh(n(\min\{|P|k, k^4 + |P|\} + \log |T|))$~\cite{BilleLRSSW15}. Moreover, an $\tOh(n\sqrt{|P|} k^3)$-time solution~\cite{DBLP:conf/isaac/GawrychowskiS13}
is known for (weaker) LZW compression~\cite{w84}.
 Again, we obtain a near-linear time algorithm for fully compressed
pattern matching with edits.
Denote by $\OccE_k(P, T)$ the set of~all
starting positions of~$k$-error occurrences of~$P$ in $T$; then, our result reads as
follows.

\begin{restatable}{mtheorem}{gcedalgmain}\label{gc_ed_alg_intro}
    Let $\G_T$ denote an SLP of~size~$n$ generating a string~$T$,
    let~$\G_P$ denote an SLP of~size~$m$ generating a string~$P$,
    let $k>0$ denote an integer threshold,
    and set $N := |T| + |P|$.

    Then, we can compute $|\OccE_k(P, T)|$ in time
    $\Oh(m\log N + n\, k^4 \log^2 N \log\log N)$.
    The elements of $\OccE_k(P, T)$ can be reported within $\Oh(|\OccE_k(P, T)|)$ extra time.\lipicsEnd
\end{restatable}
Note that our algorithms also improve the state of the art when the pattern is given in an
uncompressed form; this is because any string $P$ admits a trivial SLP of size $\Oh(|P|)$.

\paragraph*{Dynamic Strings}
While compression handles large static data sets,
a different approach is available if the data changes
frequently.
Several works on pattern matching in dynamic strings considered the indexing problem,
assuming that the text is maintained subject to updates and the pattern is given explicitly at query time;
we refer an interested reader to~\cite{Gu94,FG98,SV96,abr00,nii20} and references therein.

Recently, Clifford et al.~\cite{CGLS18} considered the problem of maintaining a data
structure for a text $T$ and a pattern~$P$,
both of which undergo character substitutions, in order to be able to efficiently compute
the Hamming distance between $P$ and any given fragment of $T$.
Among other results, for the case where $|T|\leq 2|P|$ and constant alphabet size, they
presented a data structure with $\cO(\sqrt{m \log m})$
time per operation, and they proved that, conditioned on the Online Boolean Matrix-Vector
Multiplication (OMv) Conjecture~\cite{HKNS15}, one cannot simultaneously achieve
$\cO(m^{1/2-\varepsilon})$ for the query and update time for any constant
$\varepsilon>0$.

We consider the following more general setting: We maintain an initially empty
collection of strings $\X$ that can be modified via the following
``update'' operations:
\begin{itemize}
    \item $\makestring(U)$: Insert a string $U$ to $\X$.
    \item $\concat(U,V)$: Insert $UV$ to $\X$, for $U,V \in \X$.
    \item $\splitOp(U,i)$: Insert $U\fragmentco{0}{i}$ and $U\fragmentco{i}{|U|}$ in $\X$,
        for $U \in \X$ and $i \in \fragmentco{1}{|U|}$.
\end{itemize}
The strings in $\X$ are \emph{persistent}, meaning that \concat{} and \splitOp{} do not destroy their
arguments.

The main goal in this model (and for the dynamic setting in general) is to provide
algorithms that are faster than recomputing the answer from scratch after every update.
Specifically for dynamic strings, already the task of testing equality of strings in $\X$ is challenging.
After a long line of research~\cite{st94,ksu97,abr00}, Gawrychowski et al.~\cite{ods} proved the following:
There is a data structure that supports equality queries in $\Oh(1)$ time, while each of~the update operations takes
$\Oh(\log N)$ time, where $N$ is an upper bound on the total length of~all strings in
$\X$.\footnote{Strictly speaking, $\makestring(U)$ costs
$\Oh(|U| + \log N)$ time.}
In fact, as shown in~\cite{ods},
one can also support $\Oh(1)$-time queries for the longest common prefix of two strings in $\X$ with no increase in the update times.
The data structure of~\cite{ods} is Las-Vegas randomized:
the answers are correct, but the update times are guaranteed only with high probability.
Randomization can be avoided at the cost of extra logarithmic factors in the running times (see~\cite{ksu97,nii20}),
and the same is true for our results.

We extend the data structure of~Gawrychowski et al.~\cite{ods} with approximate pattern
matching queries:
\begin{restatable}{mtheorem}{dynalgmain}\label{thm:dynalgmain}
    A collection $\X$ of non-empty persistent strings of~total length $N$ can be
    maintained subject to
    $\makestring(U)$, $\concat(U,V)$, and $\splitOp(U,i)$ operations requiring $\cO(\log N +|U|)$,
    $\cO(\log N)$, and $\cO(\log N)$ time, respectively, so that given two strings $P,T \in \X$
    with $|P|=m$ and $|T|=n$ and an integer threshold $k>0$, we can compute $|\Occ_k(P, T)|$ in
    time $\Oh(n/m \cdot k^2 \log^2 N)$ and $|\OccE_k(P, T)|$ in time $\Oh(n/m \cdot k^4 \log^2 N)$.\footnote{All running time bounds hold with high probability (i.e., $1-N^{\Omega(1)}$).}
    The elements of $\Occ_k(P, T)$ and $\OccE_k(P, T)$ can be reported in $\Oh(|\Occ_k(P, T)|)$ and $\Oh(|\OccE_k(P, T)|)$ extra time, respectively.
    \lipicsEnd
\end{restatable}

In the Hamming distance case, for $k < \sqrt{m}/\log N$, the data structure of \cref{thm:dynalgmain}
is faster than recomputing the occurrences from scratch after each update:
Recall that, in the standard setting, the fastest known algorithm for pattern matching with mismatches
costs $\tOh(n + kn/\sqrt{m}) = \tOh(n + k\sqrt{m}\cdot n/m)$ time;
in particular, the additive $\tOh(n)$ term dominates the time complexity for the considered parameter range.
Observe further that, for any~$k$, \cref{thm:dynalgmain} is not slower (ignoring ${\polylog}$ factors) than running the $\tOh(n+k^2 n/m)$-time algorithm by Clifford~et~al.~\cite{CliffordFPSS16} after every update.

In the edit distance case, for $k < (m/\log^2N)^{1/3}$, the data structure of \cref{thm:dynalgmain}
is faster than running the $\Oh(nk)$-time Landau--Vishkin algorithm for the standard
setting. Note further that, for any $k$,
\cref{thm:dynalgmain} is not slower (ignoring $ \polylog N$ factors) than running  after every update the $\Oh(n+k^4 n/m)$-time
Cole--Hariharan algorithm, whose bottleneck for $k<m^{1/4}$ is the additive $\Oh(n)$ term.

\paragraph*{Structure of~Pattern Matching with Mismatches}

As in~\cite{bkw19}, we obtain our algorithms by exploiting new structural insights for
approximate pattern matching.
Our contribution in this area is two-fold:
We strengthen the structural result of~\cite{bkw19} for pattern matching with mismatches,
and we prove a similar result for pattern matching with edits.

Before we describe our new structural insights, let us recall the structure of
exact pattern matching. Let $P$ denote a pattern of~length $m$ and let $T$ denote a text
of~length $n \le \threehalfs m$. Assume that $T$ is trimmed so that $P$ occurs both at the beginning and
at the end of~$T$, that is, $P=T{\fragmentco{0}{m}} = T\fragmentco{n - m}{n}$.
By the length constraints, we have $P{\fragmentco{n - m}
{m}} = P\fragmentco{0}{2m - n}$. Repeating this argument for the overlapping parts of~$P$,
we obtain the following well-known characterization (where $X^\infty$ denotes the concatenation
of infinitely many copies of~a string $X$):

\begin{fact}[folklore~\cite{BG95}]\label{ft:per}
    Let $P$ denote a pattern of~length $m$ and let $T$ denote a text of~length $n \le
    \threehalfs m$.
    If $T{\fragmentco{0}{m}} = T\fragmentco{n-m}{m} = P$, then there is a
    string $Q$ such that $P = Q^\infty\fragmentco{0}{m}$ and $T =
    Q^\infty\fragmentco{0}{n}$, that is, both the text and the pattern are periodic with a common string period~$Q$,
    and the starting positions of~all exact occurrences of~$P$ in $T$ form an arithmetic
    progression with difference $|Q|$.
    \lipicsEnd
\end{fact}

Hence, it is justified to say that the structure of~exact pattern matching is
fully understood.
Surprisingly, a similar characterization for approximate pattern matching was missing for
a long time. Only recently, Bringmann et al.~\cite{bkw19} proved a similar result for
pattern matching with mismatches (we write $\hd(S, T)$ for the Hamming distance of~$S$
and~$T$):

\begin{theorem}[{\cite[Theorem~1.2]{bkw19}}, simplified]\label{thm:old}
    Given a pattern $P$ of~length $m$, a text $T$ of~length $n \le\threehalfs m$, and a
    positive integer threshold $k\le m$, at least one of~the following holds:
    \begin{itemize}
        \item The number of~$k$-mismatch occurrences of~$P$ in $T$ is bounded by
            $\Oh(k^2)$.
        \item There is a primitive string $Q$ of~length $\Oh(m/k)$ such that
            $\hd(P, Q^{\infty}\fragmentco{0}{m}) \le 6k$.\lipicsEnd
    \end{itemize}
\end{theorem}

\noindent Motivated by the absence of~examples proving the tightness of~their result,
Bringmann et al.~\cite{bkw19} conjectured that the bound on the number of $k$-mismatch occurrences in
\cref{thm:old} can be improved to $\Oh(k)$.
We resolve their conjecture positively by proving the following stronger variant of
\cref{thm:old}.

\begin{mtheorem}[Compare~\cref{thm:old}]\label{hd:mthm_intro}
    Given a pattern $P$ of~length $m$, a text $T$ of~length $n \le\threehalfs m$, and a
    positive integer threshold $k\le m$, at least one of~the following holds:
    \begin{itemize}
        \item The number of~$k$-mismatch occurrences of~$P$ in $T$ is bounded by
            $\Oh(k)$.
        \item There is a primitive string $Q$ of~length $\Oh(m/k)$ that satisfies
            $\hd(P, Q^{\infty}\fragmentco{0}{m}) < 2k$.
            \lipicsEnd
    \end{itemize}
\end{mtheorem}

\begin{figure}[ht]
    \begin{subfigure}[b]{.48\textwidth}
        \centering
    \centering
\begin{tikzpicture}
    \node at (-.5, .25) {$T$};
    \node at (-.5, -.5) {$P$};
    \foreach\c[count=\x from 1] in {a,\phantom{a},a,{\phantom{a}},a}{
        \node(a\x) [inner sep=.3em] at (\x/2.5,.25) {\tt \c};
    }
    \foreach\c[count=\x from 6] in {c, {\phantom{c}},c,\phantom{c},c}{
        \node(a\x) [inner sep=.3em] at (\x/2.5,.25) {\tt \c};
    }
    \draw (a1.north west) rectangle (a5.south east);
    \draw[fill=lipicsYellow!80] (a6.north west) rectangle (a10.south east);
    \foreach\c[count=\x from 6] in {c, {\phantom{c}},c,\phantom{c},c}{
        \node(a\x) [inner sep=.3em] at (\x/2.5,.25) {\tt \c};
    }
    \foreach\c[count=\x from 3] in {a,\phantom{a}, a}{
        \node(b\x)[inner sep=.3em] at (\x/2.5,-.5) {\tt \c};
    }
    \foreach\c[count=\x from 6] in {c,\phantom{c}, c}{
        \node(b\x)[inner sep=.3em] at (\x/2.5,-.5) {\tt \c};
    }
    \draw (b3.north west) rectangle (b5.south east);
    \draw[fill=lipicsYellow!80] (b6.north west) rectangle (b8.south east);
    \foreach\c[count=\x from 6] in {c,\phantom{c}, c}{
        \node(b\x)[inner sep=.3em] at (\x/2.5,-.5) {\tt \c};
    }
    \foreach\c in {2,4,7,9}{
        \node at (\c/2.5+.03,.25) {\tt$\scriptsize\cdots$};
    }
    \foreach\c in {4,7}{
        \node at (\c/2.5+.03,-.5) {\tt$\scriptsize\cdots$};
    }
    \node at (3/2.5, .9) {${\tt a}^{3m/4}$};
    \node at (8/2.5, .9) {${\tt c}^{3m/4}$};
    \node at (4/2.5, -1.1) {${\tt a}^{m/2}$};
    \node at (7/2.5, -1.1) {${\tt c}^{m/2}$};
    \draw [decorate,decoration={brace,amplitude=4pt}] (5.5/2.5,-.8) -- (2.5/2.5,-.8);
    \draw [decorate,decoration={brace,amplitude=4pt}] (8.5/2.5,-.8) -- (5.5/2.5,-.8);
    \draw [decorate,decoration={brace,amplitude=4pt}] (.5/2.5,.5) -- (5.5/2.5,.5);
    \draw [decorate,decoration={brace,amplitude=4pt}] (5.5/2.5,.5) -- (10.5/2.5,.5);
\end{tikzpicture}
\caption{Consider a text $T := \texttt{a}^{3m/4}\texttt{c}^{3m/4}$ and a pattern $P := \texttt{a}^{m/2}\texttt{c}^{m/2}$, neither of which is approximately periodic.
Then, shifting the exact occurrence of~$P$ in $T$
by up to $k$ positions in either direction still yields a $k$-mismatch
occurrence. Hence, we need $\Omega(k)$ distinct $k$-mismatch occurrences
to derive approximate periodicity~of~$P$.}\label{nfig:exa}
\end{subfigure}%
~~~%
\begin{subfigure}[b]{.48\textwidth}
\centering
    \centering
\begin{tikzpicture}
    \node at (-.5, .25) {$T$};
    \node at (-.5, -.75) {$P$};
    \foreach\c[count=\x from 1] in {a,\phantom{a},\phantom{a},a,\phantom{a},a,\phantom{a},a,\phantom{a},a,\phantom{a}}{
        \node(a\x) [inner sep=.3em] at (\x/2.5,.25) {\tt \c};
    }
    \foreach\c[count=\x from 3] in {\phantom{a}, a, \phantom{a}, a, \phantom{a}, a}{
        \node(b\x)[inner sep=.3em] at (\x/2.5,-.75) {\tt \c};
    }

    \draw (a1.north west) rectangle (a4.south east);
    \draw (a6.north west) rectangle (a6.south east);
    \draw (a8.north west) rectangle (a10.south east);
    \draw[fill=lipicsYellow!80] (a5.north west) rectangle (a5.south east);
    \draw[fill=lipicsYellow!80] (a7.north west) rectangle (a7.south east);
    \draw[fill=lipicsYellow!80] (a11.north west) rectangle (a11.south east);

    \draw (b4.north west) rectangle (b6.south east);
    \draw (b8.north west) rectangle (b8.south east);
    \draw[fill=lipicsYellow!80] (b3.north west) rectangle (b3.south east);
    \draw[fill=lipicsYellow!80] (b7.north west) rectangle (b7.south east);

    \node(c) at (5.5/2.5, -.25){\scriptsize${\tt c}$ at $k/2$ random positions in each string};
    \foreach\x in {5, 7, 11}{
        \node(e\x) [inner sep=.3em] at (\x/2.5,.25)
        {\tt c};
        \draw {(c.north)++(0,-.1)} -- (e\x.south);
    }
    \foreach\x in {3, 7}{
        \node(d\x)[inner sep=.3em] at (\x/2.5,-.75)
        {\tt c};
        \draw {(c.south)++(0,.1)}-- (d\x.north);
    }
    \foreach\c in {2.5,9}{
        \node at (\c/2.5+.03,.25) {\tt$\scriptsize\cdots$};
    }
    \foreach\c in {5}{
        \node at (\c/2.5+.03,-.75) {\tt$\scriptsize\cdots$};
    }

    \node at (5.5/2.5, .95) {${\tt a}^{3m/2}$};
    \node at (5.5/2.5, -1.45) {${\tt a}^{m}$};
    \draw [decorate,decoration={brace,amplitude=4pt}] (8.5/2.5,-1.07)--(2.5/2.5,-1.07);
    \draw [decorate,decoration={brace,amplitude=4pt}] (.5/2.5,.57) -- (10.5/2.5,.57);
\end{tikzpicture}
    \caption{Consider a text $T$
    and a pattern $P$ obtained from $\texttt{a}^{3m/2}$ and $\texttt{a}^m$,
    respectively, by substituting $\texttt{a}$ to $\texttt{c}$ at $k/2$
    random positions. Then, all length-$m$ fragments of $T$
    are $k$-mismatch occurrences of~$P$, but, with high probability, neither $T$ nor
    $P$ is perfectly periodic. Hence, we need a relaxed periodicity notion  allowing for $\Omega(k)$ mismatches.}\label{nfig:exb}
\end{subfigure}
\caption{Examples (1) and (2) from~\cite{bkw19}.}\label{nfig:ex}
\end{figure}

Examples from~\cite{bkw19}, illustrated in \cref{nfig:ex},
prove the asymptotic tightness of \cref{hd:mthm_intro}.

As in the exact pattern matching case, we can also characterize the (approximately) periodic
case in more detail.
\begin{mtheorem}[{Compare~\cite[Claim 3.1]{bkw19}}]\label{lem:aux_intro}
    Let $P$ denote a pattern of~length $m$, let $T$ denote a text of~length $n\le
    \threehalfs m$, and let $0 \le k\le m$ denote an integer threshold.
    Suppose that both $T\fragmentco{0}{m}$ and $T\fragmentco{n-m}{n}$ are
    $k$-mismatch occurrences of~$P$.
    If there is a positive integer $d\ge 2k$
    and a primitive string $Q$ with $|Q|\le m/8d$ and $\hd(P,Q^\infty\fragmentco{0}{m})
    \le d$, then each of~following holds:
    \begin{enumerate}[(a)]
        \item Every $k$-mismatch occurrence of~$P$ in $T$ starts at a
            position that is a multiple of~$|Q|$.
        \item The string $T$ satisfies $\hd(T,Q^\infty\fragmentco{0}{n})\le 3d$.
        \item The set $\Occ_k(P,T)$ can be decomposed into $\Oh(d^2)$ arithmetic
            progressions with difference $|Q|$.\label{hdinnew}
        \lipicsEnd
    \end{enumerate}
\end{mtheorem}

Note that~\cref{thm:old}, as originally formulated in~\cite{bkw19}, includes a weaker version of~\cref{lem:aux_intro}.
We also observe that part~\eqref{hdinnew} of the new characterization is asymptotically tight, as justified by
modifying the example of \cref{nfig:exb}: Let $P$ be obtained from $\texttt{a}^m$
by placing $\texttt{c}$ at $(k+1)/2$ random positions,
and let~$T$ be obtained from $\texttt{a}^{3m/2}$ by placing $\texttt{c}$ at $(k+1)/2$ random positions within the middle third of $\texttt{a}^{3m/2}$.
Then, each $k$-mismatch occurrence must align at least one $\texttt{c}$ from $P$ with one
$\texttt{c}$ from $T$ and, conversely, each such alignment results in a $k$-mismatch occurrence.
Hence, the number of $k$-mismatch occurrences is $\Theta(k^2)$.
Furthermore, for every $q$, with high probability,
$\Occ_k(P,T)$ can only be decomposed into $\Theta(k^2)$ progressions with difference~$q$.

\paragraph*{Structure of~Pattern Matching with Edits}

Having understood the structure of~pattern matching with mismatches, we turn to the more
complicated situation of pattern matching with edits.
First, observe that the examples of \cref{nfig:exa,nfig:exb} are still valid: Any
$k$-mismatch occurrence is also a $k$-error occurrence. However, as the edit distance
allows insertions and deletions of~characters, we can construct an example where neither $P$
nor $T$ is approximately periodic, yet the number of~$k$-error occurrences is $\Omega(k^2)$;
see \cref{fig:intro_ex2}.
In the example of \cref{fig:intro_ex2}, there are still only $\Oh(k)$ regions of~size $\Oh(k)$ each
where $k$-error occurrences start. In fact, we can show that this is the worst that can
happen
(we write $\ed(S, T)$ for the edit distance of~$S$
and~$T$):

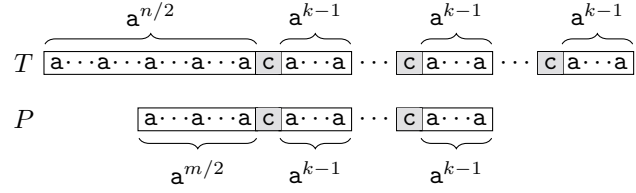
\begin{figure}[ht]
    \begin{center}
    \begin{tikzpicture}
        \node at (.5, .25) {$T$};
        \node at (.5, -.5) {$P$};
        \foreach\c[count=\x from 3] in
        {a,\phantom{a},a,\phantom{a},a,\phantom{a},a,\phantom{a},a,{\phantom{a}},a,
            {\phantom{a}},a,\phantom{a},\phantom{a},{\phantom{a}},a,{\phantom{a}},a
            ,\phantom{a},\phantom{a},{\phantom{a}},a,{\phantom{a}},a}{
            \node(a\x) [inner sep=.2em] at (\x/3.2,.25) {\tt \c};
        }
        \draw (a3.north west) rectangle (a11.south east);
        \draw[fill = lipicsYellow!80] (a12.north west) rectangle (a12.south east);
        \draw (a13.north west) rectangle (a15.south east);
        \draw[fill = lipicsYellow!80] (a18.north west) rectangle (a18.south east);
        \draw (a19.north west) rectangle (a21.south east);
        \draw[fill = lipicsYellow!80] (a24.north west) rectangle (a24.south east);
        \draw (a25.north west) rectangle (a27.south east);
        \foreach\c in {12,18,24}{
            \node(a\c) [fill=lipicsYellow!80, inner sep=.2em] at
            (\c/3.2,.25) {\tt c};
        }
        \foreach\c[count=\x from 7] in
        {a,\phantom{a},a,\phantom{a},a,{\phantom{a}},a,{\phantom{a}},a,\phantom{a},\phantom{a},
            {\phantom{a}},a,{\phantom{a}},a}{
            \node(b\x)[inner sep=.2em] at (\x/3.2,-.5) {\tt \c};
        }

        \draw (b7.north west) rectangle (b11.south east);
        \draw[fill = lipicsYellow!80] (b12.north west) rectangle (b12.south east);
        \draw (b13.north west) rectangle (b15.south east);
        \draw[fill = lipicsYellow!80] (b18.north west) rectangle (b18.south east);
        \draw (b19.north west) rectangle (b21.south east);

        \foreach\c in {12,18}{
            \node(b\c)[fill=lipicsYellow!80, inner sep=.2em] at
            (\c/3.2,-.5) {{\tt c}};
        }
        \foreach\c in {4,6,8,10,14,16.5,20,22.5,26}{
            \node at (\c/3.2+.03,.25) {\tt$\scriptsize\cdots$};
        }
        \foreach\c in {8,10,14,16.5,20}{
            \node at (\c/3.2+.03,-.5) {\tt$\scriptsize\cdots$};
        }
        \node at (7/3.2, .9) {${\tt a}^{n/2}$};
        \node at (14/3.2, .9) {${\tt a}^{k-1}$};
        \node at (20/3.2, .9) {${\tt a}^{k- 1}$};
        \node at (26/3.2, .9) {${\tt a}^{k- 1}$};
        \node at (9/3.2, -1.2) {${\tt a}^{m/2}$};
        \node at (14/3.2, -1.2) {${\tt a}^{k-1}$};
        \node at (20/3.2, -1.2) {${\tt a}^{k-1}$};
        \draw [decorate,decoration={brace,amplitude=4pt}] (11.5/3.2,-.75) -- (6.5/3.2,-.75);
        \draw [decorate,decoration={brace,amplitude=4pt}] (15.5/3.2,-.75) -- (12.5/3.2,-.75);
        \draw [decorate,decoration={brace,amplitude=4pt}] (21.5/3.2,-.75) -- (18.5/3.2,-.75);
        \draw [decorate,decoration={brace,amplitude=4pt}] (2.5/3.2,.5) -- (11.5/3.2,.5);
        \draw [decorate,decoration={brace,amplitude=4pt}] (12.5/3.2,.5) -- (15.5/3.2,.5);
        \draw [decorate,decoration={brace,amplitude=4pt}] (18.5/3.2,.5) -- (21.5/3.2,.5);
        \draw [decorate,decoration={brace,amplitude=4pt}] (24.5/3.2,.5) -- (27.5/3.2,.5);
    \end{tikzpicture}
\end{center}
    \caption{Consider a text
    $T := \texttt{a}^{n/2} \cdot (\texttt{c}\cdot \texttt{a}^{k - 1})^{n/2k}$ and a
    pattern $P := \texttt{a}^{m/2}\cdot (\texttt{c}\cdot \texttt{a}^{k - 1})^{m/2k}$
    for $n := m + 2k^2$.
    Now, for every $i \in {\fragment{-k}{k}}$,
    an $|i|$-mismatch occurrence of $P$ starts at position $n/2 - m/2 + i\cdot k$ in $T$.
    The remaining budget on the number of errors can be spent on
    shifting the starting positions, so for every $j\in \fragment{|i|-k}{k-|i|}$,
    there is a $k$-error occurrence starting at position $n/2 - m/2 + i\cdot k + j$ in $T$.
    Overall, the number of $k$-error occurrences of $P$ in $T$ is $\Omega(k^2)$,
    but neither $P$ nor $T$ is approximately periodic.}\label{fig:intro_ex2}
\end{figure}

\begin{mtheorem}\label{ed:mthm_intro}
    Given a pattern $P$ of~length $m$, a text $T$ of~length $n\le\threehalfs m$,
    and a positive integer threshold $k\le m$, at least one of~the following holds:
    \begin{itemize}
        \item The starting positions of~all $k$-error occurrences of~$P$ in $T$ lie
            in $\Oh(k)$ intervals of~length $\Oh(k)$ each.
        \item There is a primitive string $Q$ of~length $\Oh(m/k)$ and integers
            $i, j$ such that $\ed(P, Q^{\infty}\fragment{i}{j}) < 2k$.\lipicsEnd
    \end{itemize}
\end{mtheorem}

Again, we treat the (approximately) periodic case separately, thereby obtaining a result similar to
\cref{lem:aux_intro}:

\begin{mtheorem}\label{lem:Eaux_intro}
    Let $P$ denote a pattern of~length $m$, let~$T$ denote a text of~length $n$,
    and let $0 \le k\le m$ denote an integer threshold such that $n < \threehalfs m+k$.
    Suppose that the $k$-error occurrences of~$P$ in $T$ include a prefix of~$T$ and a
    suffix of~$T$.
    If there is a positive integer $d\ge 2k$ and a primitive string~$Q$
    satisfying $|Q|\le m/8d$ and $\ed(P, Q^{\infty}{\fragment{i}{j}}) \le d$ for some integers $i$, $j$, then each of~following holds:
    \begin{enumerate}[(a)]
        \item For every $p\in \OccE_k(P,T)$, we have $p\bmod |Q|\le 3d$ or $p\bmod |Q|\ge
            |Q|-3d$.
        \item The string $T$ satisfies $\ed(T, Q^{\infty}\fragment{i'}{j'}) \le 3d$
            for some integers $i'$ and $j'$.
        \item The set $\OccE_k(P,T)$ can be decomposed into $\Oh(d^3)$ arithmetic
            progressions with difference $|Q|$.\label{edintroc}
        \ifx\edauxt\undefined\lipicsEnd\fi
    \end{enumerate}
\end{mtheorem}


\subsection*{Technical Overview}
\paragraph*{Gaining Structural Insights}
To highlight the novelty of~our approach, let us first outline the proof
\cref{thm:old} by Bringmann et al.~\cite{bkw19}.
Consider a pattern $P$
of~length~$m$ and a text $T$ of~length $n \le \threehalfs m$.
Split the pattern into
$\Theta(k)$ blocks of length $\Theta(m/k)$ each and process each such block $P_i$ as follows:
Compute the shortest string period $Q_i$ of $P_i$
and align $P$ with a substring of
$Q_i^\infty$, starting from $P_i=Q_i^\infty\fragmentco{0}{|P_i|}$ and extending to both
directions, with mismatches allowed.
If there are $\Oh(k)$ mismatches for any block $P_i$,
then $P$ is approximately periodic;
otherwise, there are many mismatches for every block $P_i$.
In particular, in every $k$-mismatch occurrence where a block $P_i$ is matched exactly,
all but at most $k$ of~these mismatches between $P$ and $Q_i^\infty$ must be aligned to
the corresponding mismatches
between~$T$ and $Q_i^\infty$. Observing that, in any $k$-mismatch occurrence, all but at most $k$
of~the blocks must be matched exactly, this yields an $\Oh(k^2)$ bound on the number of~$k$-mismatch
occurrences of~$P$ in $T$.

The main shortcoming of~this approach is the initial treatment of~the pattern:
Since the pattern $P$ is independently aligned with $Q_i^\infty$ for every block $P_i$,
the same position in $P$ may be accounted for as a mismatch for multiple blocks $P_i$.
In particular, this happens if several adjacent blocks share the same period.
This leads to an overcounting of the $k$-mismatch occurrences that is hard to control.

What we do instead is a more careful analysis of~the pattern. Instead of~creating all
blocks $P_i$ at once, we process $P$ from left to right, as described below.
Suppose that $P\fragmentco{j}{m}$
is the unprocessed suffix of $P$.
We first consider the length-$m/8k$ prefix $P'$ of $P\fragmentco{j}{m}$
and compute its shortest string period $Q$.
If $|Q|$ exceeds a certain constant fraction of $|P'|$,
we set $P'$ aside as a \emph{break} and continue processing $P\fragmentco{j+|P'|}{m}$.
Now, if $P'$ is the $2k$-th break that we set aside, our process stops,
and we continue to work only with the breaks.
If $P'$ does not form a break, we try extending $P'$ to a prefix $R$ of $P\fragmentco{j}{m}$
that satisfies $\hd(R, Q^\infty\fragmentco{0}{|R|}) = \Theta(k\cdot |R|/m)$.
If such a prefix $R$ exists, we set it
aside as a \emph{repetitive region} and continue processing $P\fragmentco{j+|R|}{m}$.
Now, if all the repetitive regions collected so far have a total length of~at least $3/8 \cdot m$, we stop our process
and continue to work only with the repetitive regions computed so far.
A repetitive region $R$ does not exist only if $P\fragmentco{j}{m}$ has too few
mismatches with $Q^\infty$.
In this case, we try extending $P\fragmentco{j}{m}$ to a
suffix~$R'$ of~$P$ that satisfies $\hd(R', \overline{Q}^\infty\fragmentco{0}{|R'|}) =
\Theta(k\cdot |R'|/m)$; where $\overline{Q}$ is a suitable rotation of~$Q$.
If we fail again, we report that $P$ is approximately periodic; otherwise, we continue to work
with the single repetitive region~$R'$, disregarding the previously generated repetitive regions.
For this, we note that $|R'| \ge 3/8\cdot m$ because all
breaks and repetitive regions found beforehand have a total length of~at most $5/8\cdot
m$.

Overall, for every pattern $P$, we obtain either $2k$ disjoint breaks, or disjoint repetitive regions of total~length at least $3/8 \cdot m$, or a string with period $\Oh(m/k)$ at Hamming distance $\Oh(k)$ from $P$ (see \cref{prp:I}).

If the analysis results in breaks, we observe
that at least $k$ breaks need to be matched
exactly in every $k$-mismatch occurrence of~$P$ in $T$.
As both the length and the shortest period of each break are $\Theta(n/k)$, there are at most $\Oh(k)$ exact
matches of each break in the text. Now, a simple marking argument shows that
the number of $k$-mismatch occurrences of~$P$ in $T$ is $\Oh(k)$ (see \cref{lm:hdC}).

If the analysis results in repetitive regions,
for each region $R_i$, we consider its $k_i$-mismatch occurrences in $T$ with $k_i := \Theta(k \cdot |R_i|/m)$.
Intuitively, this distributes the available
budget of~$k$ mismatches among the repetitive regions according to their lengths.
Next, we try extending each $k_i$-mismatch occurrence of each $R_i$
to an approximate occurrence of $P$, and we assign $|R_i|$ marks to this extension.
Using insights gained in the periodic case, we bound the total number of~marks
by $\Oh(k \cdot \sum_i |R_i|)$.
Independently, we show that each $k$-mismatch occurrence of~$P$ has at least $\sum_i |R_i| - m/4$ marks. Using $\sum_i |R_i| \ge 3/8\cdot m$, we finally obtain
a bound of~$\Oh(k)$ on the number of $k$-mismatch occurrences of~$P$ in~$T$ (see \cref{lm:hdB}).

In total, this proves \cref{hd:mthm_intro}. For the characterization of~the periodic
case (\cref{lem:aux_intro}), we use a reasoning similar to that in~\cite{bkw19}.
As in the theorem, assume that $P$ has $k$-mismatch occurrences
both as a prefix and as a suffix of $T$.
Further, fix a threshold $d \ge 2k$ and a primitive string $Q$ such that $\hd(P,
Q^{\infty}\fragmentco{0}{m}) \le d$.
First, we show that every $k$-mismatch occurrence
of~$P$ in $T$ starts at a multiple of~$|Q|$. In particular, $|Q|$ divides $n-m$ and, using this observation,
we bound $\hd(T,Q^\infty\fragmentco{0}{n})$.
Finally, to decompose $\Occ_k(P, T)$ into
$\Oh(k^2)$ arithmetic progressions, we analyze the sequence of~Hamming distances between
$P$ and the length-$m$ fragments of $T$ starting at the multiples of~$|Q|$:
we observe that the number of~changes in this
sequence is bounded by $\Oh(d^2)$, which then yields the claim.

For pattern matching with edits, surprisingly few modifications in our arguments are
necessary. In fact, the analysis of~the pattern stays essentially the same. The main
difference in the subsequent arguments is that we need to account for shifts of~up to $\Oh(k)$ positions; this causes
the increase in the bound on the number of~occurrences.
Unfortunately, for the periodic case of~pattern matching with edits, the situation is
messier.
The key difficulty that we overcome is that an alignment corresponding to a specific edit
distance may not be unique.
In particular, due to insertions and deletions, combining
(the arguments for) two disjoint substrings is not as easy as in the Hamming distance case.
We address these issues by enclosing individual errors between a string and its approximate period with so-called \emph{locked fragments}, which admit a unique canonical alignment.
(A similar idea was used by Cole and Hariharan~\cite{ColeH98}.)
Combining this with a more
involved marking scheme, we then obtain \cref{lem:Eaux_intro}.

\paragraph*{A Unified Approach to Approximate Pattern Matching}
The proofs of our new structural insights are already essentially algorithmic.
To obtain algorithms for all the considered settings at once,
we proceed in two steps.
In the first step, we devise meta-algorithms that only rely on a core set of abstract operations;
in the second step, we implement these operations in various settings.
Specifically, we introduce the \modelname model---a novel abstract
interface to handle strings represented in a setting-specific manner.
For two strings $S$ and $T$, the following operations are supported:
\begin{itemize}
    \item ${\tt Extract}(S, \ell, r)$: Retrieve a string $S\fragment{\ell}{r}$.
    \item $\lceOp{S}{T}$: Compute the length of~the longest common prefix of~$S$ and $T$.
    \item $\lcbOp{S}{T}$: Compute the length of~the longest common suffix of~$S$ and $T$.
    \item $\ipmOp{S}{T}$: Assuming that $|T|\le 2|S|$, compute the starting positions of
        all exact occurrences of~$S$ in~$T$.
    \item $\accOpName(S,i)$: Retrieve the character $\accOp{S}{i}$.
    \item $\lenOpName(S)$: Compute the length $|S|$ of~the string $S$.
\end{itemize}

Using the \modelname-model operations, the meta-algorithms for both
pattern matching with mismatches and with errors follow the same overall
structure:
\begin{itemize}
    \item First, we implement the analysis of~the pattern.
        Here, the key difficulty is to detect repetitive regions. Our algorithm
        finds the \emph{shortest} repetitive region: Starting from the prefix $P'$
        of the unprocessed suffix $P\fragmentco{j}{m}$, we enumerate the mismatches (or errors) between
        $P\fragmentco{j}{m}$ and $Q^\infty$. We stop
        when the number of~mismatches (or errors) within the constructed region $R$ reaches $\Theta(k/m \cdot |R|)$. Intuitively, this is correct because the number
        of~mismatches (or errors) increases at most as fast as the length of~$|R|$.
        We treat the special case when we reach the end of~the pattern symmetrically.

        Note that computing the next mismatch between two strings is a prime application
        of~the \lceOpName operation. For finding a next edit, we adapt the Landau--Vishkin algorithm~\cite{LandauV89},
        based on \lceOpName operations as well.
    \item Next, we deal with the periodic case. This turns out to be the main difficulty.
        For the Hamming distance, implementing the proof~of
        \cref{lem:aux_intro} is rather straightforward. However, for the edit distance
        case, the more complicated proof~of~\cref{lem:Eaux_intro} gets complemented with
        even more sophisticated algorithms. Hence, we do not discuss them in this outline.
    \item Finding the occurrences in the presence of~$2k$ breaks is
        easy: We first use \ipmOpName operations to locate exact occurrences of
        the breaks in the text and then perform a straightforward marking step;
        for the Hamming distance, we lose an $\Oh(\log\log k)$ factor for
        sorting marks.
    \item Finding the occurrences in the presence of~repetitive regions
        is implemented similarly; the key difference is that we use our algorithm for the
        periodic case to find approximate occurrences of repetitive regions.
\end{itemize}
Overall, this approach then yields the main technical results of~this work (stated below for
strings of arbitrary lengths):
\begin{restatable}{mtheorem}{hdalg}\label{thm:hdalg}
    Given a pattern $P$ of~length $m$, a text $T$ of~length $n$, and a positive integer
    $k\le m$, we can compute (a representation of) the set $\Occ_k(P,T)$
    using $\Oh(n/m \cdot k^2 \log\log k)$ time plus
    $\Oh(n/m \cdot k^2)$ \modelname operations.\ifx\hdalgt\undefined\lipicsEnd\fi
\end{restatable}

For pattern matching with edits, the number of \modelname-model operations matches
the time cost of non-\modelname-model operations; hence the simplified theorem statement.%
\begin{restatable}{mtheorem}{edalgI}\label{thm:edalgI}
    Given a pattern $P$ of~length $m$, a text $T$ of~length $n$, and a positive integer
    $k\le m$, we can compute (a representation of) the set $\OccE_k(P,T)$
    using $\Oh(n/m \cdot k^{4})$ time in the \modelname model.
    \ifx\edalgIt\undefined\lipicsEnd\fi
\end{restatable}

Finally, we show how to implement the \modelname model in the
settings that we consider:
\begin{itemize}
    \item As a toy example, we start with the standard setting. Here, implementing the
        \modelname-model operations boils down to collecting known tools on strings.
    \item For the fully compressed setting, we heavily rely on the recompression
        technique by Jeż~\cite{talg/Jez15,jacm/Jez16} (especially for internal
        pattern matching queries) and on other works on
        straight-line programs~\cite{BilleLRSSW15,I17}.
    \item Finally, for the dynamic setting, we use the data structure by
        Gawrychowski et al.~\cite{ods} (for \lceOpName and \lcbOpName operations).
        Furthermore, we reuse some tools from the fully compressed setting, because the data
        structure of~\cite{ods} actually works with (a form of) straight-line programs.
\end{itemize}

As the primitive operations of the \modelname model are rather simple, we believe that
they can be efficiently implemented in further settings not considered here.

\clearpage
{\tableofcontents}
\clearpage
\section{Preliminaries}
\paragraph*{Sets and Arithmetic Progressions}
For $n\in \mathbb{Z}_{\ge 0}$, we write $\position{n}$ to denote the set $\{0, \dots, n-1\}$.
Further, for $i,j\in \mathbb{Z}$,
we write $\fragment{i}{j}$ to denote $\{i, \dots, j\}$ and
$\fragmentco{i}{j}$ to denote $\{i ,\dots, j - 1\}$;
the sets $\fragmentoc{i}{j}$ and $\fragmentoo{i}{j}$ are defined similarly.

For $a,d,\ell\in \mathbb{Z}$ with $\ell > 0$,
the set $\{ a + j \cdot d \mid j \in \fragmentco{0}{\ell}\}$
is an \emph{arithmetic progression} with starting value $a$, difference $d$, and
length~$\ell$.
Whenever we use arithmetic progressions in an algorithm, we store them
as a triple $(a,d,\ell)$ consisting of the first value, the difference, and the length.

For a set $X\sub \mathbb{Z}$, we write $kX$ to denote the set containing all elements of~$X$ multiplied
by $k$, that is, $kX := \{ k\cdot x \mid x \in X\}$. Similarly, we define
$\floor{X/k} := \{ \floor{x/k} \mid x \in X \}$ and $k\floor{X/k} := \{k \cdot \floor{x/k}
\mid x \in X \}$.

\paragraph*{Strings}

We write $T=T\position{0}\, T\position{1}\cdots T\position{n-1}$ to denote a \textit{string} of
length $|T|=n$ over an alphabet $\Sigma$. The elements of~$\Sigma$ are called \textit{characters}.
We write $\varepsilon$ to denote the \emph{empty string}.

For a string $T$, we denote the \emph{reverse string} of~$T$ by $T^R$, that is,
$T^R :=T\position{n-1}T\position{n-2}\cdots T\position{0}$.
For two positions $i\le j$ in $T$, we write
$T\fragmentco{i}{j + 1} := T\fragment{i}{j} := T\position{i}\cdots
T\position{j}$ for the \textit{fragment} of~$T$ that starts at position $i$
and ends at position $j$.
We set $T\fragment{i}{j} := \varepsilon$ whenever $j < i$.

A \emph{prefix} of~a string $T$ is a fragment that starts at position~$0$ (that is, a
prefix is a fragment of~the form $T\fragmentco{0}{j}$ for some $j \ge 0$).
A \emph{suffix} of~a string $T$ is a fragment that ends at position ${|T|-1}$ (that is,
a suffix is a fragment of~the form $T\fragmentco{i}{|T|}$ for some $i \le |T|$).
We denote the length of the \emph{longest common prefix} (\emph{longest common suffix}) of two strings $U$ and $V$, defined as~the longest string that occurs as a prefix (suffix) of~both $U$ and~$V$, by $\lcp(U, V)$ (respectively, $\lcp^R(U, V)$).

A string $P$ of~length $m\in \fragment{0}{|T|}$ is a \emph{substring} of a string~$T$ (denoted $P \substr T$)
 if there is a fragment $T\fragmentco{i}{i + m}$ matching $P$.
In this case, we say that there is an \emph{exact occurrence} of~$P$ at position $i$
in~$T$, or, more simply, that $P$ \emph{exactly occurs in} $T$.

For two strings $U$ and $V$, we write $UV$ or $U\cdot V$ to denote their concatenation.
We also write $U^k := U\cdots U$ to denote the concatenation of~$k$ copies of~the
string $U$. Furthermore, $U^\infty$ denotes an infinite string obtained by concatenating
infinitely many copies of~$U$.
A string $T$ is called \emph{primitive} if it cannot be expressed as $T=U^k$
for a string~$U$ and an integer~$k > 1$.

A positive integer $p$ is called a \emph{period} of~a string $T$ if $T[i] = T[i + p]$ for
all $i \in \fragmentco{0}{|T|-p}$. We refer to the smallest
period as \emph{the period} $\per(T)$ of~the string.
The string $T\fragmentco{0}{\per(T)}$ is called the \emph{string period} of $T$.
We call a string \emph{periodic} if its period is at most half of~its length.

For a string $T$, we define the following \emph{rotation} operations:
The operation $\rot(\cdot)$ takes as input a string, and moves its last character to the
front; that is,~$\rot(T) := T\position{n-1}T\fragment{0}{n-2}$.
The inverse operation $\rot^{-1}(\cdot)$ takes as input a string and
moves its initial character to the end; that is,~$\rot^{-1}(T) := T\fragment{1}{n-1}T\position{0}$.
Note that a primitive string $T$ does not match any of~its non-trivial rotations,
that is, we have $T=\rot^j(T)$ if and only if $j \equiv 0 \pmod{|T|}$.

Finally, the \emph{run-length encoding} (RLE) of~a string $T$ is a decomposition of~$T$ into
maximal blocks such that each block is a power of~a single character.
(For instance, the RLE of~the string \texttt{aaabbabbbb} is~\texttt{a$^3$b$^2$ab$^4$}.)
Note that each block of~the RLE can be represented in $\cO(1)$ space.

\paragraph*{Hamming Distance and Pattern Matching with Mismatches}
For two strings $S$ and $T$ of~the same length $n$,
we define the set of~\emph{mismatches} between $S$ and $T$  as $\MIS(S,T) :=\{i\in
    \position{n} \mid S\position{i}\ne T\position{i}\}$.
Now, the \emph{Hamming distance} of~$S$ and $T$ is
defined as the number of~mismatches between $S$ and $T$, that is, $\hd(S,
T) := |\MIS(S,T)|$.

It is easy to verify that the Hamming distance satisfies the triangle inequality:

\begin{fact}[Triangle inequality for Hamming distance]\label{tria}
    Any strings $A$, $B$, and $C$ of~the same length satisfy
        $\hd(A, C) + \hd(C, B) \ge \hd(A, B) \ge |\hd(A, C) - \hd(C, B)|$.
    \lipicsEnd
\end{fact}

As we are often concerned with the Hamming distance of~a string $S$ and a prefix
of~$T^{\infty}$ for a string $T$, we write $\MIS(S,T^*) := \MIS(S, T^{\infty}\fragmentco{0}{|S|})$
and $\hd(S, T^*) := |\MIS(S,T^*)|$.

Now, for a string $P$ (also called a \emph{pattern}), a string $T$ (also called a \emph{text}),
and an integer $k\ge 0$ (also called a \emph{threshold}), we say that there is a
\emph{$k$-mismatch occurrence} of~$P$ in $T$ at position $i$ if $\hd(P, T\fragmentco{i}{i+|P|})\leq k$.
We write $\Occ_k(P,T)$ to denote the set of~all positions of~$k$-mismatch occurrences of
$P$ in $T$, that is, $\Occ_k (P,T):= \{i \mid \hd(P,T\fragmentco{i}{i+|P|}\leq k)\}$.
Lastly, we define the \emph{pattern matching with mismatches} problem.
\begin{problem}[Pattern matching with mismatches]
    Given a pattern $P$, a text $T$, and a threshold $k$, compute the set $\Occ_k(P,
    T)$.\lipicsEnd
\end{problem}
Note that, depending on the use case (especially if the set $\Occ_k(P, T)$
is relatively large), we may only want to compute the size $|\Occ_k(P,T)|$
or a space-efficient representation of the set $\Occ_k(P, T)$,
e.g., as the union of disjoint arithmetic progressions.

\paragraph*{Edit Distance and Pattern Matching with Edits}

The \emph{edit distance} (also known as \emph{Levenshtein distance}) between two
strings $S$ and $T$, denoted $\ed(S,T)$, is the minimum number of
character insertions, deletions, and substitutions required to transform $S$ into $T$.

Again, it is easy to verify that the edit distance satisfies the triangle inequality:
\begin{fact}[Triangle inequality for edit distance]\label{Etria}
    Any strings $A$, $B$, and $C$ of~the same length satisfy
        $\ed(A, C) + \ed(C, B) \ge \ed(A, B) \ge |\ed(A, C) - \ed(C, B)|$.
    \lipicsEnd
\end{fact}

Further, similarly to the Hamming distance, we write $\ed(S, T^*) := \min
\{\ed(S,T^\infty\fragmentco{0}{j}) \mid j \in \mathbb{Z}_{\ge 0}\}$  to denote the minimum edit
distance between a string $S$ and any prefix of~a string $T^\infty$.
Further, we write $\edl{S}{T}:= \min\{\ed(S,T^\infty\fragmentco{i}{j}) \mid i, j \in \mathbb{Z}_{\ge 0}, i \le j\}$
to denote the minimum edit distance between $S$ and any substring of~$T^\infty$,
and we set $\eds{S}{T} :=  \min\{\ed(S,T^\infty\fragmentco{i}{j|T|}) \mid i, j \in \mathbb{Z}_{\ge 0}, i \le j|T|\}$.

Now, for a string $P$ (also called a \emph{pattern}), a string $T$ (also called a \emph{text}),
and an integer $k\ge 0$ (also called a \emph{threshold}), we say that there is a
\emph{$k$-error occurrence} of~$P$ in $T$ at position $i$ if
$\ed(P, T\fragmentco{i}{j})\leq k$ for some $j \ge i$.
We write $\OccE_k(P,T)$ to denote the set of~all positions of~$k$-error occurrences of
$P$ in $T$, that is, $\OccE_k (P,T):= \{i \mid \exists_{j \ge i} \ed(P,T\fragmentco{i}{j}\leq k)\}$.
Lastly, we define the \emph{pattern matching with edits} problem.
\begin{problem}[Pattern matching with edits]
    Given a pattern $P$, a text $T$, and a threshold $k$, compute the set $\OccE_k(P,
    T)$.\lipicsEnd
\end{problem}
Again, we may only want to compute the size $|\OccE_k(P,T)|$
or a space-efficient representation of $\OccE_k(P, T)$.

\subsection{The \modelname Model}\label{sec:pillar}

In order to unify the implementations of~our approach for approximate pattern matching
in different settings, we introduce the \modelname model.
The \modelname model captures certain primitive operations which can be
implemented efficiently in all considered settings (see~\cref{sec:model} for the
actual implementations).
Thus, in the algorithmic sections of~this work, \cref{sec:pmm,sec:pme}, we bound the
running times in terms of the number of~\modelname operations;
if this value is asymptotically smaller than the time complexity of the remaining computations,
we also specify the extra running time.

In the \modelname model, we are given a family of~strings $\X$ for preprocessing.
The elementary objects are fragments $X\fragmentco{\ell}{r}$ of~strings $X\in \X$.
Each such fragment $S$ is represented via a \emph{handle}, which is how $S$ can be passed as input to \modelname operations.
Initially, the model provides a handle to each $X\in \X$, interpreted as $X\fragmentco{0}{|X|}$.
Handles to other fragments can be obtained through an \extractOpName operation:
\begin{itemize}
    \item $\extractOpName(S,\ell,r)$: Given a fragment $S$ and positions $0 \le \ell \le r
        \le |S|$, extract the (sub)fragment $S\fragmentco{\ell}{r}$. If
        $S=X\fragmentco{\ell'}{r'}$ for $X\in \X$, then $S\fragmentco{\ell}{r}$ is defined
        as $X\fragmentco{\ell'+\ell}{\ell'+r}$.
\end{itemize}
Furthermore, the following primitive operations are supported in the \modelname model:
\begin{itemize}
    \item $\lceOp{S}{T}$: Compute the length of~the longest common prefix of~$S$ and $T$.
    \item $\lcbOp{S}{T}$: Compute the length of~the longest common suffix of~$S$ and $T$.
    \item $\ipmOp{P}{T}$: Assuming that $|T|\le 2|P|$, compute $\OccEx(P,T)$ (represented
        as an arithmetic progression with difference $\per(P)$).
    \item $\accOpName(S,i)$: Assuming $i\in \position{|S|}$, retrieve the character $\accOp{S}{i}$.
    \item $\lenOpName(S)$: Retrieve the length $|S|$ of~the string $S$.
\end{itemize}

We now provide a small toolbox that is to be used in~\cref{sec:pmm,sec:pme}.
First, we note that $\lceOp{S}{T}$, $\lcbOp{S}{T}$, or $\ipmOp{S}{T}$ can be used to check whether two strings $S$ and $T$ are equal.
\begin{lemma}[Equality,~{\cite[Fact 2.5.2]{thesis}}]\label{lm:streq}
    Given strings $S$ and $T$, we can  check whether $S$ and $T$ are equal
    in $\Oh(1)$ time in the \modelname model.\lipicsEnd
\end{lemma}
A more involved operation allows checking if a given string is periodic and, if so, computing the period.
\begin{lemma}[$\perOp{S}$,~\cite{IPM,thesis}]
    Given a string $S$, we can compute $\per(S)$ or declare that $\per(S) > |S| / 2$
    in $\Oh(1)$ time in the \modelname model.\lipicsEnd
\end{lemma}
Next, we introduce an operation that checks cyclic equivalence and retrieves the witness shifts.
\begin{lemma}[$\cycEqOp{S}{T}$,~\cite{IPM,thesis}]
    Given strings $S$ and $T$, we can find all integers $j$ such that $T = \rot^j(S)$
    in $\Oh(1)$ time in the \modelname model. The output is represented as an arithmetic
    progression.\lipicsEnd
\end{lemma}

Our subsequent goal is to generalize the primitive \lceOpName operation to also support
fragments of~infinite powers of strings. We start with a special case,
and then we cover the general case.
\begin{lemma}[$\lceOp{S}{Q^{\infty}}$,~{\cite[Fact 2.5.2]{thesis}; see also~\cite{BabenkoGKKS16}}]\label{lm:inflcpold}
    Given strings $S$ and $Q$, we can compute
    $\lceOp{S}{Q^{\infty}}$ in $\Oh(1)$ time in the \modelname model.\lipicsEnd
\end{lemma}
\begin{corollary}[$\lceOp{S}{Q^{\infty}\fragmentco{\ell}{r}}$]\label{lm:inflcp}
    Given strings $S$ and $Q$ and integers $0 \le \ell \le r$, we can compute
    $\lceOp{S}{Q^{\infty}\fragmentco{\ell}{r}}$ in $\Oh(1)$ time in the \modelname model.
\end{corollary}
\begin{proof}
    We first compute $\lceOp{S}{Q\fragmentco{\ell\bmod{|Q|}}{|Q|}}$ using a primitive
    operation. If we reach the end of~the string $Q$, we continue with an $\lceOp{S\fragmentco{|Q|-\ell\bmod |Q|}{|S|}}{Q^\infty}$ query, implemented using~\cref{lm:inflcpold}.
    This yields  $\lceOp{S}{Q^{\infty}\fragmentco{\ell}{}}$,
    so we cap the obtained value with $r-\ell$ to retrieve   $\lceOp{S}{Q^{\infty}\fragmentco{\ell}{r}}$.
\end{proof}
A similar procedure yields an analogous generalization of the primitive \lcbOpName operation.
\begin{corollary}[$\lcbOp{S}{Q^{\infty}\fragmentco{\ell}{r}}$]\label{lm:inflcpR}
    Given strings $S$ and $Q$ and integers $0 \le \ell \le r$, we can compute
    $\lcbOp{S}{Q^{\infty}\fragmentco{\ell}{r}}$ in $\Oh(1)$ time in the \modelname model.\lipicsEnd
\end{corollary}

Lastly, we discuss an operation finding all exact occurrences of~a given string $P$ in
a given string~$T$.
\begin{lemma}[{\tt ExactMatches}$(P, T)$]\label{lm:emath}
    Let $T$ denote a string of~length $n$ and let $P$ denote a string of~length~$m$.
    We can compute the set $\OccEx(P, T)$ using $\Oh(n/\per(P))$ time and $\Oh(n/m)$
    \modelname operations.
\end{lemma}
\begin{proof}
    We perform an $\ipmOp{P}{T_i}$ query with $T_i:= T\fragmentco{i{m}}{\min\{n, (i+2){m}-1\}}$
    for each $i\in \position{\floor{n/m}}$; that is a total of~$\Oh(n/m)$ \modelname operations.
    Each occurrence $j\in \OccEx(P,T)$ corresponds to a single occurrence of~$P$ in a single $T_i$,
    namely, $j \bmod m \in \OccEx(P, T_{\floor{j/m}})$, and vice versa.
    Furthermore, each $\ipmOp{P}{T_i}$ query returns an arithmetic progression with
    difference $\per(P)$, which thus consists of~$\Oh(m/\per(P))$ elements.
    Hence, the total number of~elements of~all arithmetic progressions is $\Oh(n/\per(P))$.
\end{proof}

We conclude this section with introducing the concept of~a generator.
\begin{definition}[Generator of~a set]
    For an (ordered) set $S$, an {\em $(\Oh(P), \Oh(Q))$-time generator of~$S$}
    is a data structure that after $\Oh(P)$-time initialization in the \modelname model,
    supports the following operation:
    \begin{itemize}
        \item {\tt Next}: In the $i$th call of~{\tt Next}, return the $i$th smallest
            element of~the set $S$ (or $\bot$ if $i > |S|$), using $\Oh(Q(i))$ time in the \modelname model.\lipicsEnd
    \end{itemize}
\end{definition}

Note that a generator is not specific to the \modelname model. We use generators
to obtain positions where two strings differ---either by a mismatch or by an edit. Consult
\cref{sc:auxphd,sc:auxped} for the details, as well as for other \modelname operations
specific to pattern matching with mismatches or edits, respectively.

\section{Improved Structural Insights into Pattern Matching with Mismatches}\label{sec:km}
In this section, we provide insight into the structure of~$k$-mismatch occurrences of~a pattern $P$
in a text $T$.
In particular, we improve the result of~\cite{bkw19} and show the following asymptotically tight
characterization (which is \cref{hd:mthm_intro} with explicit constants and without the
restriction on the length of~$T$).

\begin{restatable}{theorem}{hdmain}\label{thm:hdmain}
    Given a pattern $P$ of~length $m$, a text $T$ of~length $n$, and a threshold $k \in \fragment{1}{m}$,
    at least one of~the following holds:
    \begin{itemize}
        \item The number of~$k$-mismatch occurrences of~$P$ in $T$ is bounded by
            $|\Occ_k(P,T)|\le \thmbound \cdot n/m \cdot k$.
        \item There is a primitive string $Q$ of~length $|Q| \le m/\thmboundt k$ that
            satisfies $\hd(P, Q^*) < 2k$.\ifx\hdmaint\undefined\lipicsEnd\fi
    \end{itemize}
\end{restatable}
\def\hdmaint{1}

\subsection{Characterization of~the Periodic Case}
In order to prove \cref{thm:hdmain}, we first discuss
the (approximately) periodic case, that is, the case when we have $\hd(P, Q^*) < 2k$.
In particular, we prove the following theorem, which strengthens~\cite[Claim~3.1]{bkw19}.

\begin{restatable}[{Compare \cref{lem:aux_intro} and~\cite[Claim 3.1]{bkw19}}]{theorem}{lemaux}\label{lem:aux}
    Let $P$ denote a pattern of~length $m$, let $T$ denote a text of~length $n\le \threehalfs m$,
    and let $k\in \fragment{0}{m}$ denote a threshold.
    Suppose that both $T\fragmentco{0}{m}$ and $T\fragmentco{n-m}{n}$ are $k$-mismatch occurrences
    of~$P$ (that is, $\{0,n-m\}\subseteq \Occ_k(P,T)$).
    If there is a positive integer $d\ge 2k$ and a primitive string $Q$
    with $|Q|\le m/8d$ and $\hd(P,Q^*)\le d$, then each of~following holds:
     \begin{enumerate}[(a)]
        \item Every position in $\Occ_k(P,T)$ is a multiple of~$|Q|$.\label{it:mult}
        \item The string $T$ satisfies $\hd(T,Q^*)\le 3d$.\label{it:text}
        \item The set $\Occ_k(P,T)$ can be decomposed into $3d(d+1)$ arithmetic
            progressions with difference $|Q|$.\label{it:seq}
        \item If $\hd(P,Q^*)=d$, then $|\Occ_k(P,T)|\le
        6d$.\label{it:few}\ifx\lemauxt\undefined\lipicsEnd\fi
     \end{enumerate}
\end{restatable}
\def\lemauxt{1}

Before proving \cref{lem:aux}, we characterize the values $\hd(T\fragmentco{j|Q|}{j|Q|+m},P)$
under an extra assumption that $\hd(T,Q^*)$ is small as well; this assumption is dropped in~\cref{lem:aux}.
\begin{lemma}\label{lem:rle}
    Let $P$ denote a pattern of~length $m$ and let $T$ denote a text of~length $n\le
    \threehalfs m$.
    Further, let~$Q$ denote a string of~length $q$ and set $d:=\hd(P,Q^*)$ and $d':=\hd(T,Q^*)$.
    Then, the sequence of~values $h_j := \hd(T\fragmentco{jq}{jq+m},P)$
    for $0\le j \le (n-m)/q$ contains at most $d'(2d+1)$ entries $h_j$ with $h_j\ne h_{j+1}$
    and, unless $d=0$, at most $2d'$ entries $h_j$ with $h_j \le d/2$.
\end{lemma}
\begin{proof}
    For every $\tau \in \MIS(T,Q^*)$ and $\pi \in \MIS(P,Q^*)$,
    let us put $(2 - \hd(P\position{\pi}, T\position{\tau}))$ marks at position $\tau - \pi$
    in~$T$, if it exists.
    For each $0\le j\le (n-m)/q$, let $\mu_j(\tau, \pi)$ denote the number of~marks placed
    at position $jq$ due to the mismatches $\tau$ in~$T$ and $\pi$ in $P$,
    that is,
    \[
        \mu_j(\tau, \pi) := \left\{
            \begin{array}{c l}
                2 - \hd(P\position{\pi}, T\position{\tau}) &
                \text{if $\pi \in \MIS(P,Q^*)$ and $\tau = jq + \pi \in \MIS(T,Q^*)$,}\\
                0 & \text{otherwise.}
        \end{array}\right.
    \]
    Further, define $\mu_j := \sum_{\tau, \pi} \mu_j(\tau,\pi)$ as the total number of~marks at
    position~$jq$.

    Next, for every $0\le j\le (n-m)/q$, we relate the Hamming distance
    $h_j := \hd(T\fragmentco{jq}{jq+m},P)$ to the
    number of~marks $\mu_j$ at position $jq$ and the Hamming distances
    $\hd(T\fragmentco{jq}{jq+m}, Q^*)$ and $\hd(P, Q^*)$;
    consult \cref{fig:lem_aux} for an illustration.

\begin{figure}[t]
    \centering
    \begin{tikzpicture}
                \begin{scope}[yshift=-1cm]
                    \draw[{Latex[length=1.5mm, width=1mm]}-{Latex[length=1.5mm, width=1mm]}] (0,0.55) -- (2,0.55);
                    \node[label = {above: $q$}]  at (1,0.4) {};
                	\node[label = {left: $T$}]  at (0,0.2) {};
                    \draw (0,0) rectangle (11,0.4);
                    \foreach \x in {2,4,6,8,10}{
                    	\draw (\x,0) -- (\x,0.4);
                    }
                    \node[label = {above: $\texttt{b}$}]  at (1.5,-0.15) {};
                    \node[label = {above: $\texttt{b}$}]  at (3.5,-0.15) {};
                    \node[label = {above: $\texttt{c}$}]  at (5,-0.15) {};
                    \node[label = {above: $\texttt{a}$}]  at (5.5,-0.15) {};
					\node[label = {above: $\texttt{c}$}]  at (7.5,-0.15) {};
					\node[label = {above: $\texttt{a}$}]  at (10.5,-0.15) {};
					\node[label = {below: $jq$}]  at (2.25,0.2) {};
					\foreach \x/\y in {1.5/1, 3.5/2, 5/3, 5.5/4, 7.5/5, 10.5/6}{
                    	\node[label = {below: $\tau_{\y}$},xshift=0.05 cm]  at (\x,0.15) {};
                	}
                \end{scope}

                \node[label = {left: $P$}]  at (2,0.2) {};
				\draw (2,0) rectangle (9.5,0.4);
				\foreach \x in {4,6,8}{
                    	\draw (\x,0) -- (\x,0.4);
                }
				\node[label = {above: $\texttt{a}$}]  at (3.5,-0.15) {};
				\node[label = {above: $\texttt{a}$}]  at (5.5,-0.15) {};
				\node[label = {above: $\texttt{b}$}]  at (7.5,-0.15) {};
				\foreach \x/\y in {3.5/1, 5.5/2, 7.5/3}{
                    	\node[label = {above: $\pi_{\y}$},xshift=0.1 cm]  at (\x,0.2) {};
                }
                \foreach \x in {3.5, 5, 7.5}{
                    	\node[label = {below: $*$}]  at (\x,0) {};
                }
\end{tikzpicture}
    \caption{In both strings, all blocks apart from the last ones are of~length $q$.
        For each string $X\in \{P,T\}$, we show the characters at positions in $\MIS(X,Q^*)$ only.
        At position $jq$ in $T$, we place
        $\mu_j(\tau_2,\pi_1)+\mu_j(\tau_4,\pi_2)+\mu_j(\tau_5,\pi_3)=1+2+1=4$ marks.
        We have $\hd(P,Q^*)=|\{\pi_1,\pi_2,\pi_3\}|=3$ and $\hd(T\fragmentco{jq}{jq+m}, Q^*)=|\{\tau_2,\tau_3,\tau_4,\tau_5\}|=4$.
        Using~\cref{{clm:hj}}, we obtain that $h_j=3$;
        the three corresponding mismatches are indicated by asterisks.
    }\label{fig:lem_aux}
\end{figure}
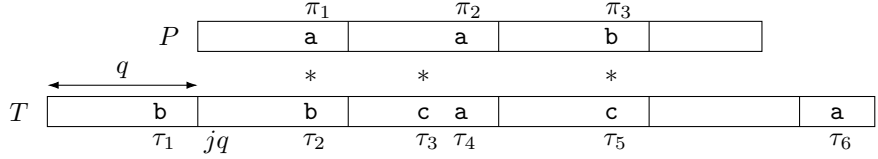

    \begin{claim}\label{clm:hj}
        For each $0\le j\le (n-m)/q$, we have $h_j =
        \hd(P,Q^*)+\hd(T\fragmentco{jq}{jq+m}, Q^*)-\mu_j$.
    \end{claim}
    \begin{claimproof}
        We show the following equivalent statement:
        \begin{equation}
            |\MIS(T\fragmentco{jq}{jq+m},P)|
            = |\MIS(P, Q^*)| + |\MIS(T\fragmentco{jq}{jq+m},Q^*)| -
            \sum_{\tau, \pi} \mu_j(\tau,\pi).\label{eq:hj}
        \end{equation}
        By construction, $\mu_j(\tau, \pi) = 0$ whenever $\tau \ne \pi + jq$.
        Hence, we can prove~\eqref{eq:hj} by showing that for every position $\pi \in \fragmentco{0}{m}$ in $P$
        and every position $\tau := jq + \pi$ in $T$, the following equation holds:
\[
            \hd(T\position{\tau}, P\position{\pi})
            =
            \hd(P\position{\pi}, Q^\infty\position{\pi})
            + \hd(T\position{\tau}, Q^\infty\position{\tau})
            - \mu_j(\tau, \pi).
        \]
        We proceed by case distinction on whether $\pi \in  \MIS(P,Q^*)$ and
        whether $\tau \in \MIS(T,Q^*)$.
        \begin{itemize}
            \item If $\pi \notin  \MIS(P,Q^*)$ and $\tau \notin \MIS(T,Q^*)$, then we have
                $P\position{\pi}=Q^\infty\position{\pi}=Q^\infty\position{\tau}
                = T\position{\tau}$ and thus \[
                    \hd(T\position{\tau}, P\position{\pi})=
                    0
                    =0+0-0 =
                    \hd(P\position{\pi}, Q^\infty\position{\pi})
                    + \hd(T\position{\tau}, Q^\infty\position{\tau})
                - \mu_j(\tau, \pi).\]
            \item If $\pi \in \MIS(P,Q^*)$ and $\tau \notin  \MIS(T,Q^*)$, then we have
                $P\position{\pi}\ne Q^\infty\position{\pi}=Q^\infty\position{\tau}=
                T\position{\tau}$ and thus
                \[
                    \hd(T\position{\tau}, P\position{\pi})=
                    1=1+0-0
                    =\hd(P\position{\pi}, Q^\infty\position{\pi})
                    + \hd(T\position{\tau}, Q^\infty\position{\tau})
                - \mu_j(\tau, \pi).
                \]
            \item If $\pi \notin \MIS(P,Q^*)$ and $\tau \in  \MIS(T,Q^*)$, then we have
                $P\position{\pi}=Q^\infty\position{\pi}=Q^\infty\position{\tau}
                \ne T\position{\tau}$ and thus
                \[
                    \hd(T\position{\tau}, P\position{\pi})=
                    1=0 + 1 -0
                    =\hd(P\position{\pi}, Q^\infty\position{\pi})
                    + \hd(T\position{\tau}, Q^\infty\position{\tau})
                - \mu_j(\tau, \pi).
                \]
            \item If $\pi \in \MIS(P,Q^*)$ and $\tau \in  \MIS(T,Q^*)$, then we have
            $P\position{\pi}\ne Q^\infty\position{\pi}=Q^\infty\position{\tau}
            \ne T\position{\tau}$ and thus
                \[
                    \hd(T\position{\tau}, P\position{\pi})
                    =1+1-(2-\hd(T\position{\tau}, P\position{\pi}))
                    =\hd(P\position{\pi}, Q^\infty\position{\pi})
                    + \hd(T\position{\tau}, Q^\infty\position{\tau})
                - \mu_j(\tau, \pi).\]
        \end{itemize}
        Combining the equations obtained for every pair of~positions $\pi$ and $\tau$,
        we derive~\eqref{eq:hj}.
    \end{claimproof}

    In particular, \cref{clm:hj} yields
    \[h_{j+1} - h_j =
        |\MIS(T,Q^*)\cap \fragmentco{jq+m}{(j+1)q+m}|
        -|\MIS(T,Q^*)\cap \fragmentco{jq}{(j+1)q}|
    -\mu_{j+1}+\mu_j.
\]
    Hence, in order for $h_{j+1}$ not to equal $h_j$, at least one of~the four terms on
    the right-hand side of~the equation above must be non-zero.
    Let us analyze when this is possible.
    To that end, we first observe that the set $\MIS(T,Q^*)\cap
    \fragmentco{jq+m}{(j+1)q+m}$ contains only elements $\tau \in \MIS(T,Q^*)$ with
    $\tau \ge m$, and that the set
    $\MIS(T,Q^*)\cap \fragmentco{jq}{(j+1)q}$ only contains elements $\tau \in \MIS(T,Q^*)$
    with $\tau < n - m$.
    Using $n\le \threehalfs m$, we observe that $h_{j+1}$ can be different from
    $h_j$ due to the first or second term at most $d'$ times.
    Further, each non-zero value in one of~the terms $\mu_{j+1}$ and $\mu_j$ can be
    attributed to a marked position ($jq$ or $(j+1)q$, respectively).
    The total number of~marked positions is at most
    $dd'$, so $h_{j+1}$ can be different from $h_j$ due one of~the terms
    $\mu_{j+1}$ or $\mu_j$ at most $2dd'$ times.
    In total, we conclude that the number of~entries~$h_j$ with $h_{j}\ne h_{j+1}$ is at
    most $d'(2d+1)$.

    Next, observe that $\mu_j \le 2|\MIS(T,Q^*)\cap
    \fragmentco{jq}{jq+m}|=2\hd(T\fragmentco{jq}{jq+m},Q^*)$,
    and therefore \[
        h_j = \hd(P,Q^*)+\hd(T\fragmentco{jq}{jq+m},Q^*)-\mu_j \ge d-\mu_j/2.
    \] Consequently, $h_j \le  d/2$ yields $\mu_j\ge d$,
    that is, that there are at least $d$ marks at position $jq$.
    Given that the total number of~marks is at most $2dd'$, the number of
    entries $h_j$ with $h_j \le d/2$ is at most $2d'$, assuming that $d>0$.
\end{proof}

Now, we drop the assumption that $\hd(T,Q^*)$ is small and prove \cref{lem:aux}.

\lemaux*
\begin{proof}
    Consider any position $\ell \in \Occ_k(P,T)$.
    By the definition of~a $k$-mismatch occurrence,
    we have $\hd(T\fragmentco{\ell}{\ell + m}, P) \le k \le d/2$.
    Combining this inequality with $\hd(P,Q^*)\le d$ via the triangle inequality
    yields $\hd(T\fragmentco{\ell}{\ell + m}, Q^*) \le \threehalfs d$.
    Similarly, for the position $0\in \Occ_k(P,T)$,
    we obtain $\hd(T\fragmentco{0}{m},\allowbreak Q^*) \le \threehalfs d$,
    which lets us compare the overlapping parts of~$Q^{\infty}$.
    Replacing strings by superstrings and applying the triangle inequality yields
    \begin{align*}
        \hd(Q^{\infty}\fragmentco{\ell}{m}, Q^{\infty}\fragmentco{0}{m-\ell})
            &\le \hd(T\fragmentco{\ell}{m}, Q^{\infty}\fragmentco{\ell}{m})+
            \hd(T\fragmentco{\ell}{m}, Q^{\infty}\fragmentco{0}{m-\ell})\\
            &\le \hd(T\fragmentco{0}{m}, Q^{\infty}\fragmentco{0}{m})+
            \hd(T\fragmentco{\ell}{\ell + m}, Q^{\infty}\fragmentco{0}{m})\\
            & = \hd(T\fragmentco{0}{m}, Q^*) + \hd(T\fragmentco{\ell}{\ell+m}, Q^*)\\
            &\le 3d.
    \end{align*}
    Towards a proof~by contradiction, suppose that $\ell$ is not a multiple of~$|Q|$.
    As $Q$ is primitive, we have\[
        3d \ge \hd(Q^{\infty}\fragmentco{\ell}{m}, Q^{\infty}\fragmentco{0}{m-\ell})
        \ge \left\lfloor{\frac{m - \ell}{|Q|}}\right\rfloor
        \ge \left\lfloor{\frac{m/2}{m/8d}}\right\rfloor
        = 4d,
    \]
    where the second bound follows from $\ell \leq m/2$ and $|Q|\le m/8d$.
    This contradiction yields Claim~\eqref{it:mult}.

    In order to prove Claim~\eqref{it:text},
    we observe that $n-m\in \Occ_k(P,T)$ is a multiple of~$|Q|$. Consequently,
    \[
        \hd(T, Q^*) = \hd(T\fragmentco{0}{n-m}, Q^*)+\hd(T\fragmentco{n-m}{n}, Q^*)
        \le \hd(T\fragmentco{0}{m}, Q^*)+ \threehalfs d \le 3d,
    \]
    which concludes the proof~of~Claim~\eqref{it:text}.

    For a proof~of~Claims~\eqref{it:seq} and~\eqref{it:few}, we apply \cref{lem:rle}.
    Due to Claim~\eqref{it:mult}, each position in $\Occ_k(P,T)$ corresponds to an entry
    $h_j$ with $h_j \le k$.
    In particular, each block of~consecutive entries $h_j,\ldots,h_{j'}$ not exceeding $k$
    yields an arithmetic progression (with difference $|Q|$) in $\Occ_k(P,T)$.
    The number of~entries~$h_j$ with $h_j\le k < h_{j+1}$ or $h_j > k \ge h_{j+1}$
    is in total at most $3d(2d+1)$, so the number of~arithmetic progressions is at most
    $1+1/2\cdot 3d(2d+1) \le 3d(d+1)$, which proves Claim~\eqref{it:seq}.

    For Claim~\eqref{it:few}, we observe that if $d=\hd(P,Q^*)$,
    then each position in $\Occ_k(P,T)$ corresponds to an entry $h_j$ with $h_j \le k \le
    d/2$; thus $|\Occ_k(P,T)|\le 2\cdot 3d = 6d$.
\end{proof}
\begin{corollary}\label{cor:aux}
    Let $P$ denote a pattern of~length $m$, let~$T$ denote a text of~length $n$,
    and let $k\in \fragment{0}{m}$ denote a threshold.
    If there is a positive integer $d\ge 2k$ and a primitive string $Q$
    with $|Q|\le m/8d$ and $\hd(P,Q^*)\le d$,
    then the set $\Occ_k(P,T)$ can be decomposed into $6\cdot n/m \cdot d(d + 1)$
    arithmetic progressions with difference $|Q|$.
    Moreover, if $\hd(P,Q^*)= d$, then $|\Occ_k(P,T)|\le 12\cdot n/m \cdot d$.
\end{corollary}
\begin{proof}
    Partition the string $T$ into $\floor{2n/m}$ blocks $T_0, \dots, T_{\floor{2n/m}-1}$ of~length
    less than $\threehalfs m$ each, where the $i$th block starts at position
    $\floor{i\cdot m/2}$; formally, we set $T_i := T\fragmentco{{\floor{i\cdot {m}/2}}}
        {\min\{n, \floor{(i+3)\cdot {m}/2}-1\}}$.
    If $\Occ_k(P,T_i)\ne \emptyset$, we define $T'_i$ to be the shortest fragment
    of~$T_i$ containing all $k$-mismatch occurrences of~$P$ in $T_i$.
    As a result, $T'_i$ satisfies the assumptions of~\cref{lem:aux}.
    Hence, $\Occ_k(P,T_i')$ can be decomposed into $3d(d+1)$ arithmetic
    progressions with difference $|Q|$, and $|\OccEx(P,T_i')|\le 6d$ if $\hd(P,Q^*)= d$.

    We conclude that $\Occ_k(P,T_i')$ decomposes into $6\cdot n/m \cdot
    d(d+1)$ arithmetic progressions with difference~$|Q|$; further, $|\OccEx(P,T_i)|\le
    12\cdot n/m \cdot d$  if $\hd(P,Q^*)= d$.
\end{proof}

\subsection{The Non-Periodic Case}

Having dealt with the (approximately) periodic case, we now turn to the general case.
In particular, we show that whenever the string $P$ is sufficiently far from being periodic,
the number of $k$-mismatch occurrences of~$P$ in any string~$T$ of~length $n\le \threehalfs m$ is $\Oh(k)$.

Intuitively, we proceed (and thereby prove \cref{thm:hdmain}) as follows:
We first analyze the string $P$ for useful structure that can help in bounding the number of
occurrences of~$P$ in any string $T$.
If we fail to find any special structure in $P$, then we conclude that the string $P$ is
close to a periodic string with a small period (compared to $|P|$)---a case that
we already understand thanks to the previous subsection.

\begin{lemma}\label{prp:I}
    Given a string $P$ of~length $m$ and and a threshold $k\in \fragment{1}{m}$,
    at least one of~the following holds:
    \begin{enumerate}[(a)]
        \item The string $P$ contains $2k$ disjoint \emph{breaks} $B_1,\ldots, B_{2k}$
            each having period $\per(B_i)> m/\alphav k$ and length $|B_i| = \lfloor
            m/\betav k\rfloor$.
        \item The string $P$ contains $r$ disjoint \emph{repetitive regions} $R_1,\ldots, R_{r}$
        of~total length $\sum_{i=1}^r |R_i| \ge \deltavN/\deltavD \cdot m$ such
        that each region $R_i$ satisfies
        $|R_i| \ge m/\betav k$ and has a primitive \emph{approximate period} $Q_i$
        with $|Q_i| \le m/\alphav k$ and $\hd(R_i,Q_i^*) = \ceil{\betav k/m\cdot |R_i|}$.
        \item The string $P$ has a primitive \emph{approximate period} $Q$ with $|Q|\le
            m/\alphav k$ and $\hd(P,Q^*) < \betav k$.
    \end{enumerate}
\end{lemma}

\clearpage
\SetInd{0.6em}{0.6em}
\begin{algorithm}[t]
    $\mathcal{B} \gets \{\}; \mathcal{R} \gets \{\}$\;
    \While{\bf true}{
        Consider the fragment $P' = P\fragmentco{j}{j+\floor{m/\betav k}}$
        of~the next $\floor{m/\betav k}$ unprocessed characters of~$P$\;
        \If{$\per(P') > m/\alphav k$}{
            $\mathcal{B} \gets \mathcal{B} \cup \{P'\}$\;
            \lIf{$|\mathcal{B}| = 2k$}{\Return{breaks $\mathcal{B}$}}
            }\Else{
            $Q \gets P\fragmentco{j}{j+\per(P')}$\;
            Search for a prefix $R$ of~$P\fragmentco{j}{m}$
            with $|R|>|P'|$ and~$\hd(R, Q^*) =\lceil\betav k/m\cdot |R|\rceil$\;
            \If{such $R$ exists}{
                $\mathcal{R} \gets \mathcal{R} \cup \{(R, Q)\}$\;
                \If{$\sum_{(R, Q) \in \mathcal{R}} |R|\ge \deltavN/\deltavD \cdot
                    m$}{
                    \Return{repetitive regions (and their corresponding
                    periods) $\mathcal{R}$}\;
                }
                }\Else{
                Search for a suffix $R'$ of~$P$
                with $|R'|\ge m-j$ and~$\hd(R', \rot^{|R'|-m+j}(Q)^*) = \lceil\betav
                k/m\cdot |R'|\rceil$\;
                \lIf{such $R'$ exists}{%
                    \Return{repetitive region $(R', \rot^{|R'|-m+j}(Q))$}
                }\lElse{%
                    \Return{approximate period $\rot^{j}(Q)$}
                }
            }
        }
    }
\caption{A constructive proof~of~\cref{prp:I}.}\label{alg:P1}
\end{algorithm}
\begin{proof}
    We prove the claim constructively, that is, we construct either a set $\mathcal{B}$
    of~$2k$ breaks, or a set $\mathcal{R}$ of~repetitive regions, or, if we fail to
    construct either, we derive an approximate string period $Q$ of the string $P$ with the
    desired properties.

    We process the string $P$ from left to right as follows:
    If the fragment $P'$ of~the next $\floor{m/\betav k}$ (unprocessed) characters of~$P$
    has a long period, we have found a new break and continue (or return the found set of~$2k$
    breaks). Otherwise, if $P'$ has a short string period $Q$,
    we try to extend the fragment $P'$ (to the right) into a repetitive region.
    If we succeed, we have found a new repetitive region and continue (or return the found set
    of~repetitive regions if the total length of~all repetitive regions found so far is
    at least $\deltavN/\deltavD \cdot m$).
    If we fail to construct a new repetitive region, then we conclude that the suffix of
    $P$ starting with $P'$ has an approximate period~$Q$.
    We try to construct a repetitive region by extending this suffix to the left,
    dropping all the repetitive regions computed beforehand.
    If we fail again, we declare that an appropriate rotation of $Q$ is an approximate period of~the string $P$.
    Consider~\cref{alg:P1} for a detailed description.

    By construction, all breaks in the set $\mathcal{B}$ and repetitive regions
    in the set $\mathcal{R}$ returned by the
    algorithm are disjoint and satisfy the claimed properties.
    To prove that the algorithm is also correct when it fails to find a
    new repetitive region,
    we start by bounding from above the length of~the processed prefix of~$P$.
    \begin{claim}\label{clm:j}
        Whenever we consider a new fragment $P'=P\fragmentco{j}{j + \floor{m/\betav k}}$ of
        the next $\floor{m/\betav k}$ unprocessed characters of~$P$, such a fragment
        starts at a position $j<\ubjv m$.
    \end{claim}
    \begin{claimproof}
        Observe that whenever we consider a new fragment
        $P\fragmentco{j}{j+\floor{m/\betav k}}$, the string $P\fragmentco{0}{j}$ has been
        partitioned into breaks and repetitive regions.
        The total length of~breaks is less than $2k\floor{m/\betav k}\le 2/\betav \cdot m$,
        and the total length of~repetitive regions is less than $\deltavN/\deltavD \cdot m$.
        Hence, $j<\ubjv m$, yielding the claim.
    \end{claimproof}
    Note that \cref{clm:j} also shows that whenever we consider a new fragment $P'$
    of~$\floor{m/\betav k}$  characters, there is indeed such a fragment, that is, $P'$
    is well-defined.

    Now, consider the case when, for a fragment $P' = P\fragmentco{j}{j + \floor{m/\betav k}}$
    (that is not a break) and its string period $Q = P\fragmentco{j}{j + \per(P')}$,
    we fail to obtain a new repetitive region~$R$.
    In this case, we search for a repetitive region
    $R'$ of~length $|R'|\ge m-j$ that is a suffix of~$P$ and has an approximate period $Q'
    := \rot^{|R'|-m+j}(Q)$.
    If we indeed find such a region $R'$, then $|R'|\ge m-j \ge m-\ubjv m =
    \deltavN/\deltavD \cdot m$ by \cref{clm:j}, so $R'$ is long enough to be reported on
    its own.
    However, if we fail to find such $R'$, we need to show that $\rot^{j}(Q)$ can be
    reported as an approximate period of~$P$, that is, $\hd(P,\rot^{j}(Q)^*)< \betav k$.

    We first derive  $\hd(P\fragmentco{j}{m},Q^*)< \ceil{\betav k/m \cdot (m-j)}$.
    For this, we inductively prove that the values $\Delta_{\rho}:=\ceil{\betav k/m \cdot
    \rho}-\hd(P\fragmentco{j}{j+\rho},Q^*)$ for $\rho \in \fragment{|P'|}{m-j}$ are all at least 1.
    In the base case of~$\rho=|P'|$, we have $\Delta_{\rho}=1-0$ because $Q$ is the string
    period of~$P'$.
    To carry out an inductive step, suppose that $\Delta_{\rho-1}\ge 1$ for some $\rho \in \fragment{|P'|}{m-j}$.
    Notice that $\Delta_{\rho}\ge \Delta_{\rho-1}-1\ge 0$: The first term in the
    definition of~$\Delta_{\rho}$ has not decreased  compared to $\Delta_{\rho-1}$, and the second term
    $\hd(P\fragmentco{j}{j+\rho},Q^*)$ may have
    increased by at most one.
    Moreover, $\Delta_{\rho} \ne 0$ because $R=P\fragmentco{j}{j+\rho}$ could not be
    reported as a repetitive region. Since $\Delta_{\rho}\in \mathbb{Z}$, we conclude that
    $\Delta_{\rho}\ge 1$.
    This inductive reasoning ultimately shows that $\Delta_{m-j}>0$, that is,
    $\hd(P\fragmentco{j}{m},Q^*)< \ceil{\betav k/m \cdot (m-j)}$.

    A symmetric argument holds for the values $\Delta'_{\rho}:= {\ceil{\betav k/m \cdot
    \rho}}-\hd(P\fragmentco{m-\rho}{m},\rot^{\rho-m+j}(Q)^*)$ for $\rho \in \fragment{m-j}{m}$
    because no repetitive region $R'$ was found as an extension of~$P\fragmentco{j}{m}$ to
    the left. Thus, $\hd(P,\rot^{j}(Q)^*)< \betav k$,
    that is,  $\rot^{j}(Q)$ is an approximate period of~$P$.
\end{proof}

In the next steps, we discuss how to exploit the structure obtained by \cref{prp:I}.
First, we discuss the case that a string $P$ contains $2k$ disjoint breaks.

\begin{lemma}\label{lm:hdC}
    Let $P$ denote a pattern of~length $m$, let~$T$
    denote a text of~length $n$, and let $k\in \fragment{1}{m}$ denote a threshold.
    Suppose that $P$ contains $2k$
    disjoint breaks $B_1,\dots,B_{2k}$ 
    each satisfying $\per(B_i) \ge m / \alphav k$.
    Then, $|\Occ_k(P,T)|\le \gammav \cdot n/m \cdot k$.
\end{lemma}
\begin{proof}
    For every break $B_i=P\fragmentco{b_i}{b_i + |B_i|}$
    we mark a position $j$ in $T$ if $j+b_i \in \OccEx(B_i,T)$.
    \begin{claim}\label{cl:c1}
        We place at most $\alphavd \cdot n/m \cdot k^2$ marks in total.
    \end{claim}
    \begin{claimproof}
        Fix a break $B_i$ and notice that the positions in $\OccEx(B_i,T)$ are at distance at least
        $\per(B_i)$ from each other. Hence,
        for the break $B_i$, we place at most $\alphav \cdot n/m \cdot k$ marks in $T$.
        In total, we therefore place at most
        $2k\cdot \alphav n/m \cdot k = \alphavd \cdot n/m \cdot k^2$ marks in $T$.
    \end{claimproof}
    Next, we show that every $k$-mismatch occurrence of~$P$ in~$T$
    starts at a position with at least $k$ marks.
    \begin{claim}\label{cl:c2}
        Each position $\ell\in \Occ_k(P,T)$ has at least $k$ marks in $T$.
    \end{claim}
    \begin{claimproof}
        Fix $\ell \in \Occ_k(P,T)$.
        Out of~the $2k$ breaks, at least $k$ breaks are matched exactly, as not matching
        a break exactly incurs at least one mismatch.
        If a break $B_i$ is matched exactly, then we have $\ell+b_i\in \OccEx(B_i,T)$.
        Hence, we have placed a mark at position $\ell$.
        Thus, there is a mark at position $\ell$ for every break $B_i$ matched
        exactly in the corresponding occurrence of~$P$ in~$T$. In total, there are at
        least~$k$ marks at position $\ell$ in $T$.
    \end{claimproof}
    By \cref{cl:c1,cl:c2}, we have  $|\Occ_k(P,T)|\le (\alphavd \cdot n/m \cdot k^2)/k
    = \alphavd \cdot n/m \cdot k$.
\end{proof}

Secondly, we discuss how to use repetitive regions in the string $P$ to bound
$|\Occ_k(P,T)|$.
\begin{lemma}\label{lm:hdB}
    Let $P$ denote a pattern of~length $m$, let~$T$ denote a text of~length $n$,
    and let $k\in \fragment{1}{m}$ denote a threshold.
    If $P$ contains disjoint repetitive regions $R_1,\ldots, R_{r}$
    of~total length at least $\sum_{i=1}^r |R_i| \ge \deltavN/\deltavD\cdot m$
    such that each region $R_i$ satisfies $|R_i| \ge m/\betav k$ and has a
    primitive approximate period~$Q_i$
    with $|Q_i| \le m/\alphav k$ and $\hd(R_i,Q_i^*) = \ceil{\betav k/m\cdot |R_i|}$,
    then $|\Occ_k(P,T)|\le \gammabv \cdot n/m \cdot k$.
\end{lemma}
\begin{proof}
    Set  $m_R := \sum_{i=1}^r |R_i|$.
    For each  repetitive
     region
    $R_i=P\fragmentco{r_i}{r_i+|R_i|}$,
    set $k_i := \floor{\betavh k/m\cdot |R_i|}$,
    and place $|R_i|$ marks at each position $j$ with $j+r_i \in \Occ_{k_i}(R_i,T)$.
    \begin{claim}\label{cl:b1}
        The total number of~marks placed is at most $\gbetavh \cdot n/m \cdot k\cdot m_R$.
    \end{claim}
    \begin{claimproof}
        We use \cref{cor:aux} to bound $|\Occ_{k_i}(R_i,T)|$.
        For this, we set $d_i := \hd(R_i,Q_i^*)$ and notice that $d_i = \ceil{\betav
        k/m\cdot |R_i|} \le \tbetav \cdot k/m\cdot |R_i|$
        since $|R_i| \ge m/\betav k$.
        Moreover, $d_i \ge 2k_i$ and $|Q_i| \le m/\alphav k \le |R_i|/8d_i$ due to $d_i
        \le  \tbetav\cdot k/m\cdot |R_i|$. Hence, the assumptions of~\cref{cor:aux} are
        satisfied.
        Consequently, $|\Occ_{k_i}(R_i,T)| \le 12\cdot n/|R_i| \cdot d_i\le \gbetavh \cdot
        n/m \cdot k$;
        the last inequality holds as $d_i \le  \tbetav \cdot k/m\cdot |R_i|$.

        The total number of~marks placed due to $R_i$ is therefore
        bounded by $\gbetavh \cdot n/m \cdot k\cdot |R_i|$.
        Across all repetitive regions, this sums up to $\gbetavh \cdot n/m \cdot k\cdot
        m_R$, yielding the claim.
    \end{claimproof}
    Next, we show that every $k$-mismatch occurrence of~$P$ in~$T$,
    starts at a position with many marks.
    \begin{claim}\label{cl:b2}
        Each $\ell\in\Occ_k(P,T)$ has at least $m_R-m/\betavh$ marks.
    \end{claim}
    \begin{claimproof}
        Let us fix $\ell\in \Occ_k(P,T)$ and denote
        $k'_i := \hd(R_i, T\fragmentco{\ell+r_i}{\ell+r_i+|R_i|})$ to be the number of~mismatches incurred by repetitive region $R_i$.
        Further, let
        $I := \{i\in\fragment{1}{r} \mid k'_i \le k_i\}=\{i\in\fragment{1}{r} \mid k'_i \le \betavh k/m\cdot |R_i|\}$
        denote the set of~indices of all repetitive regions that have $k_i$-mismatch occurrences at the corresponding
        positions in~$T$. By construction, for each $i\in I$,
        we have placed $|R_i|$ marks at position $\ell$.
        Hence, the total number of~marks at position $\ell$ is at least
        $\sum_{i\in I}|R_i|= m_R -\sum_{i\notin I}|R_i|$.
        It remains to bound the term $\sum_{i\notin I}|R_i|$.
        Using the definition of~$I$,
        we obtain
        \[\sum_{i\notin I} |R_i| = \sum_{i\notin I} \frac{\betavh m k}{\betavh m k}
            \cdot|R_i| = \frac{m}{\betavh k} \cdot  \sum_{i\notin I} \left(\betavh \cdot
            |R_i|/m \cdot k\right)
            < \frac{m}{\betavh k} \cdot  \sum_{i\notin I} k'_i \le
                \frac{m}{\betavh k} \cdot  \sum_{i=1}^r k'_i \le \frac{m}{\betavh},\]
        where the last bound holds as, in total, all repetitive regions incur at most
        $\sum_{i=1}^r k'_i \le k$ mismatches (since all repetitive regions are pairwise disjoint).
        Hence, the number of~marks placed is at least $m_R-m/\betavh$,
        completing the proof~of~the claim.
    \end{claimproof}
    In total, by \cref{cl:b1,cl:b2}, the number of~$k$-mismatch occurrences of~$P$ in $T$
    is at most \[
        |\Occ_k(P, T)| \le  \tfrac{\gbetavh \cdot n/m \cdot k \cdot m_R}{m_R - m/\betavh}
        = \tfrac{\gbetavh \cdot n/m \cdot k}{1 - m/(\betavh m_R)}.
    \]
    As this bound is a decreasing function in $m_R$, the assumption $m_R \ge
    \deltavN/\deltavD\cdot m$ yields the upper bound
    \[|\Occ_k(P, T)| \le \tfrac{\gbetavh \cdot n/m \cdot k \cdot \deltavN/\deltavD\cdot m}
        {\deltavN/\deltavD\cdot m - m/\betavh} = \gammabv \cdot n/m \cdot k,
    \]completing the proof.
\end{proof}

Finally, we consider the case that $P$ is approximately periodic, but not too close to the periodic string in scope.%
\begin{lemma}\label{lm:hdA}
    Let $P$ denote a string of~length $m$, let $T$ denote a string of~length $n$,
    and let $k\in \fragment{1}{m}$ denote threshold.
    If there is a primitive string $Q$
    of~length at most $|Q| \le m/\alphav k$ that satisfies $2k\le\hd(P, Q^*)\le\betav k$,
    then $|\Occ_k(P,T)|\le \gammap \cdot n/m \cdot k$.
\end{lemma}
\begin{proof}
    We apply \cref{cor:aux} with $d =\hd(P,Q^*)$.
    As $2k \le d\le \betav k$ yields $|Q|\le m/\alphav k \le m/8d$,
    the assumptions of~\cref{cor:aux} are met.
    Consequently, $|\Occ_k(P,T)|\le 12 \cdot n/m \cdot d
    \le \gammap \cdot n/m \cdot k$.
\end{proof}

Gathering \cref{prp:I,lm:hdC,lm:hdB,lm:hdA},
we are now ready to prove \cref{thm:hdmain}, which we repeat here for convenience.
\hdmain*
\begin{proof}
    We apply \cref{prp:I} on the string $P$ and proceed
    depending on the structure found in~$P$.

    If the string $P$ contains $2k$ disjoint breaks $B_1,\dots,B_{2k}$
    (in the sense of~\cref{prp:I}), we apply \cref{lm:hdC}
    and obtain that $|\Occ_k(P,T)|\le \gammav \cdot n/m \cdot k$.

    If the string $P$ contains $r$ disjoint repetitive regions $R_1,\dots,R_r$
    (again, in the sense of~\cref{prp:I}), we apply \cref{lm:hdB} and obtain that
    $|\Occ_k(P,T)|\le \gammabv \cdot n/m \cdot k$.

    Otherwise, \cref{prp:I} guarantees that there is a primitive string $Q$ of~length
    at most $|Q| \le m/\alphav k$ that satisfies $\hd(P, Q^*) < \betav k$.
    If $\hd(P, Q^*)\ge 2k$, then \cref{lm:hdA} yields $|\Occ_k(P,T)|\le \gammap \cdot n/m \cdot k$.
    If, however, $\hd(P, Q^*)< 2k$, then we are in the second alternative of~the theorem statement.
\end{proof}

\section{Algorithm: Pattern Matching with Mismatches in the \modelname Model}\label{sec:pmm}

In this section, we implement the improved structural result (\cref{thm:hdmain})
from the previous section to obtain the following result.

\hdalg*
\def\hdalgt{1}
In general, the algorithm follows the outline given by the proof~of~\cref{thm:hdmain}:
We first show how to implement \cref{prp:I} to preprocess the given pattern $P$.
Then, depending on the structure of~$P$, we (construct and) use algorithms implementing the
insights from the corresponding lemmas from the previous section.

\subsection{Auxiliary \modelname Model Operations for Pattern Matching with Mismatches}
\label{sc:auxphd}

We start by introducing some commonly used operations for Pattern Matching with Mismatches
and show how to implement them efficiently in the \modelname model.

\begin{algorithm}[t]
    \SetKwBlock{Begin}{}{end}
    \SetKwFunction{nxt}{Next}

    ${\misOpL{S}{Q}}$\Begin{
        \Return{$\mathbf{G}\gets \{ S\gets S;\, Q\gets Q;\, i \gets 0 \}$}\;
    }
    \BlankLine
    \nxt{$\mathbf{G}=\{S;\,Q;\,i\}$}\Begin{
        \lIf{$i \ge |S|$}{\Return{$\bot$}}
        $\pi \gets \lceOp{S\fragmentco{i}{|S|}}
        {Q^\infty\fragmentco{i}{}}$\;
        $i \gets  i + \pi + 1$\;
        \lIf{$i > |S|$}{\Return{$\bot$}}
        \lElse{\Return{$i - 1$}}
    }
    \caption{A generator for the set $\MIS(S, Q^*)$.}\label{alg:gen}
\end{algorithm}
\begin{lemma}[$\misOpL{S}{Q}$, $\mibOpL{S}{Q}$]\label{lm:misoph}
    For every pair of~strings $S$ and $Q$, the sets $\MIS(S, Q^*)$ and $\MIS(S^R, (Q^R)^*)$
    admit $(\Oh(1), \Oh(1))$-time generators.
\end{lemma}
\begin{proof}
    We only develop the \misOpLName generator; \mibOpLName can be obtained similarly.

    Given strings $S$ and $Q$, the generator itself just stores $S$, $Q$, and an index
    $i$ of~the position \emph{after} the last returned value by {\tt Next};
    initially, we set $i$ to $0$.

    We implement the {\tt Next} operation by using \cref{lm:inflcp} to compute
    $\pi=\lceOp{S\fragmentco{i}{|S|}}{Q^{\infty}\fragmentco{i}{}}$.
    If we observe that $i + \pi = |S|$, that is, if we reached the end of~the string
    $S$, then we return $\bot$. Otherwise, we report $i + \pi$ and update the
    index $i$ to $i + \pi + 1$. See \cref{alg:gen} for a pseudo-code.

    For the correctness, we observe that due to storing the index $i$, we are able to retrieve
    the suffixes of~$S$ and $Q$ to be compared, so the correctness follows.

    For the running time, we observe that the creation of~a generator
    is only bookkeeping, which takes constant time.
    Further, the {\tt Next} operation uses one call to the primitive \lceOpName operation
    and a single call to the \lceOpName operation from \cref{lm:inflcp}, which uses
    $\Oh(1)$ \modelname operations. Thus in total, the {\tt Next} operation also uses $\Oh(1)$
    \modelname operations, completing the proof.
\end{proof}
\begin{corollary}[\misOpName{\tt ($S$, $Q^*$)}]\label{lm:misoph2}
    Given strings $S$ and $Q$, we can compute the set $\MIS(S, Q^*)$,
    using $\Oh(\hd(S, Q^*)+1)$ primitive operations in the \modelname model.
\end{corollary}
\begin{proof}
    We use a \misOpLName from \cref{lm:misoph} and call its {\tt Next} operation until
    the {\tt Next} operation returns~$\bot$. The claim follows.
\end{proof}

\begin{lemma}[{\tt Verify($S$, $T$, $k$)}]\label{lm:verifyh}
    Let $S$ and $T$ denote strings of~length $m$ each, and let $k \le m$ denote a positive integer.
    Using $\Oh(k)$ \modelname operations, we can check whether $\hd(S, T) \le k$.
\end{lemma}
\begin{proof}
    We use a \misOpLName from \cref{lm:misoph} and call its {\tt Next} operation until
    either the {\tt Next} operation returns~$\bot$ (in which case we return {\tt true})
    or until we obtain the $(k+1)$st 
    mismatch between $S$ and $T$ (in which case we return {\tt false}). The claim follows.
\end{proof}

\subsection{Computing Structure in the Pattern}

In this section, we show how to implement \cref{prp:I}.
While the proof~of~\cref{prp:I} is already constructive, we still need to fill in some
implementation details.

\begin{lemma}[{\tt Analyze($P$, $k$)}: Implementation of~\cref{prp:I}]\label{prp:Ialg}
    Let $P$ denote a string of~length $m$ and let $k \le m$ denote a positive integer.
    Then, there is an algorithm that computes one of~the following:
    \begin{enumerate}[(a)]
        \item $2k$ disjoint breaks $B_1,\ldots, B_{2k} \substr P$
            each having period $\per(B_i)> m/\alphav k$ and length $|B_i| = \lfloor
            m/\betav k\rfloor$;
        \item disjoint repetitive regions $R_1,\ldots, R_{r} \substr P$
        of~total length $\sum_{i=1}^r |R_i| \ge \deltavN/\deltavD \cdot m$ such
        that each region $R_i$ satisfies
        $|R_i| \ge m/\betav k$ and is constructed along with a primitive approximate period $Q_i$
        such that $|Q_i| \le m/\alphav k$ and $\hd(R_i,Q_i^*) = \ceil{\betav k/m\cdot
        |R_i|}$; or
        \item a primitive approximate period $Q$ of~$P$
            with $|Q|\le m/\alphav k$ and $\hd(P,Q^*) < \betav k$.
    \end{enumerate}
    The algorithm uses $\Oh(k)$ time plus $\Oh(k)$ \modelname operations.
\end{lemma}
\begin{algorithm}[t]
    \SetKwBlock{Begin}{}{end}
    \SetKwFunction{anly}{Analyze}
    \anly{$P$, $k$}\Begin{
        $j\gets 0$; $r \gets 1$; $b \gets 1$\;
        \While{\bf true}{
            $j' \gets j+\floor{m/\betav k}$\; 
            \If{$\perOp{P\fragmentco{j}{j'}} > m/\alphav k$}{
                $B_b \gets P\fragmentco{j}{j'}$\;
                \lIf{$b = 2k$}{\Return{breaks $B_1,\ldots, B_{2k}$}}
                $b \gets b+1$; $j \gets j'$\;
                }\Else{
                $q \gets \perOp{P\fragmentco{j}{j'}}$\;
                $Q_r \gets P\fragmentco{j}{j+q}$; $\delta \gets 0$\;
                generator $\mathbf{G} \gets \misOpL{P\fragmentco{j}{m}}{Q_r}$\;
                \While{$\delta < \betav k/m\cdot (j'-j)$ {\bf and}
                    $(\pi \gets \nxt{$\mathbf{G}$}) \ne \bot$}{
                    $j' \gets j + \pi+1$; $\delta \gets \delta + 1$\;
                }
                \If{$\delta \ge \betav k/{m}\cdot (j'-j)$}{
                    $R_r \gets P\fragmentco{j}{j'}$\;
                    \If{$\sum_{i=1}^r |R_i|\ge \deltavN/\deltavD \cdot  m$}{
                        \Return{repetitive regions $R_1,\ldots,R_r$ with periods
                        $Q_1,\ldots, Q_r$}\;
                    }
                    $r \gets r+1$;  $j \gets j'$\;
                    }\Else{
                    $Q\gets Q_r$; $j'' \gets j$\;
                    generator $\mathbf{G}' \gets \mibOpL{P\fragmentco{0}{j}}{Q}$\;
                    \While{$\delta < \betav k/m\cdot (m-j'')$ {\bf and}
                        $(\pi \gets \nxt{$\mathbf{G}'$}) \ne \bot$}{
                        $j'' \gets \pi$; $\delta \gets \delta + 1$\;
                    }
                    $Q \gets P\fragmentco{j+(j''-j)\!\bmod{q} }
                    {j+(j''-j)\!\bmod{q} + q}$\tcp*{$Q \gets \rot^{j-j''}(Q)$}
                    \label{ln:22}
                    \If{$\delta \ge \betav k/m\cdot (m-j'')$}{
                        \Return{repetitive region $P\fragmentco{j''}{m}$
                        with period $Q$}
                    }
                    \lElse{\Return{approximate period $Q$}}
                }
            }
        }
    }
    \caption{A \modelname model implementation of~\cref{alg:P1}.}\label{alg:iP1}
\end{algorithm}
\begin{proof}
    Our implementation follows \cref{alg:P1} from the proof~of~\cref{prp:I}:
    Recall that $P$ is processed from left to right and split into breaks and repetitive regions.
    In each iteration, the algorithm first considers a fragment of~length $\floor{m/\betav k}$.
    This fragment either becomes the next break (if its shortest period is long enough)
    or is extended to the right to a repetitive region (otherwise).
    Having constructed sufficiently many breaks or repetitive regions of~sufficiently large
    total length, the algorithm stops. Processing the string $P$ in this manner
    guarantees disjointness of~breaks and repetitive regions.
    As in the proof~of~\cref{prp:I}, a slightly different approach is needed if the
    algorithm encounters the end of~$P$ while growing a repetitive region. If this
    happens, the region is also extended to the left.
    This way, the algorithm either obtains a single repetitive region (which is not
    necessarily disjoint with the previously created ones, so it is returned on its own)
    or learns that the whole string $P$ is close to being periodic.

    Next, we fill in missing details of~the implementation of~the previous steps in the
    \modelname model.
    To that end, first note that the \modelname model includes a \perOpName operation of
    checking if
    the period of~a string~$S$ satisfies $\per(S)\le |S|/2$ and computing $\per(S)$
    in case of~a positive answer. Since our threshold $m/\alphav k$ satisfies $\floor{m/\alphav k}
    \le \floor{m/\betav k}/2$, no specific work is required to obtain the period of~an
    unprocessed fragment of~$\floor{m/\betav k}$ characters of~$P$.

    To compute a repetitive region starting from a fragment $P' = P\fragmentco{j}{j +
    \floor{m/\betav k}}$ with string period $Q = P'\fragmentco{0}{\per(P')}$, we use a
    $\misOpL{P\fragmentco{j}{m}}{Q}$ generator from \cref{lm:misoph}:
    We extend $P'$ up to the next mismatch between $P'$ and $Q^{\infty}$ until
    we either reach the end of~$P$ or the number $\delta = \hd(P', Q^*)$ reaches the bound
    $\betav k/m\cdot |P'|$.
    If we reach the end of~$P$, we similarly extend $P' = P\fragmentco{j}{m}$
    to the left using a $\mibOpL{P\fragmentco{0}{j}}{Q}$ generator from \cref{lm:misoph}:
    Again, we always extend $P'$ up to the next mismatch
    until we reach the start of~$P$ or the number $\delta = \hd(P', \overline{Q}^*)$
    reaches the bound $\betav k/m\cdot |P'|$ (where $\overline{Q}=\rot^{|P'|-m+j}(Q)$ is
    the corresponding cyclic rotation of~$Q$).
    If we reach the start of~the string, we return a suitable cyclic rotation of~$Q$;
    otherwise we found a long repetitive region, which we then return.
    Consider \cref{alg:iP1} for a detailed pseudo-code of~the implementation.

    For the correctness, since our algorithm follows the proof~of~\cref{prp:I}, we only
    need to show that our implementation of~finding repetitive regions correctly
    implements the corresponding step in
    \cref{alg:P1}. However, this is easy, as with each extension of~$P'$, the number $\delta$
    may increase by at most $1$. As we start with $\delta =
    \hd(P\fragmentco{j}{j+\floor{m/\betav k}}, Q^*) = 0$, we thus never
    skip over a repetitive region.
    Further, the fragment $P'=P\fragmentco{j}{j + \floor{m/\betav k}}$ by construction
    contains at least two repetitions of~the period $Q$, so we can obtain
    each cyclic rotation of~$Q$ as a fragment of~$P$. In particular we indeed
    compute a cyclic rotation of~$Q$ in \cref{ln:22} of~\cref{alg:iP1}.
    Consequently, \cref{alg:iP1} indeed correctly implements \cref{alg:P1}.

    For the running time analysis, observe that each iteration of~the outer while loop
    processes at least $\floor{m/\betav k}$ characters of~$P$,
    so there are at most $\Oh(k)$ iterations of~the outer while loop.
    In each iteration, we perform one \perOpName operation, a constant number of
    \accOpName operations, and at most $\betav k/m \cdot (j'-j)$ calls to the generator
    $\misOpLName$. Each of~these calls uses $\Oh(1)$ \modelname operations,
    which is $\Oh(\betav k/m \cdot m) = \Oh(k)$ in total across all iterations.
    Similarly, we bound the running time of~the calls to the generator \mibOpLName:
    As we find at most $\betav k/m \cdot m = \betav k$ mismatches,
    \mibOpLName uses at most $\Oh(k)$ operations.
    Overall, \cref{alg:iP1} thus uses $\Oh(k)$ \modelname operations.

    The remaining running time is bounded by $\Oh(k)$ in the same way, completing the proof.
\end{proof}

\subsection{Computing Occurrences in the Periodic Case}

\begin{lemma}[\texttt{FindRotation($k$, $Q$, $S$)}]\label{lem:findAPeriod}
    Let $k$ denote a positive integer, let $Q$ denote a primitive string, and
    let~$S$ denote a string with $|S|\ge (2k+1)|Q|$.
    Then, we can compute a unique integer $j\in \fragmentco{0}{|Q|}$
    such that $\hd(S,\rot^j(Q)^*)\le k$,
    or report $\bot$ if no such integer exists,
    using $\Oh(k)$ time plus $\Oh(k)$ \modelname operations.
\end{lemma}
\begin{proof}
    For every $0 \le i \le 2k$, define $S_i := S\fragmentco{i\,|Q|}{(i+1)\,|Q|}$.
    We compute the majority of~$S_{0},\ldots,S_{2k}$
    (using \cref{lm:streq} for checking equality of~fragments).
    If no majority exists, then we return $\bot$.
    Otherwise, we set $\bar{Q}$ to be the majority string of~$S_0,\dots,S_{2k}$
    and check if $\hd(S,\bar{Q}^*)\le k$ using a \misOpLName from \cref{lm:misoph}.
    If this test succeeds, we use a \cycEqOpName operation
    to retrieve all $j\in \fragmentco{0}{|Q|}$ with $\bar{Q}=\rot^j(Q)$
    and return any such $j$.
    If the test fails or if no such $j$ is found, then we return $\bot$.

    For the correctness, observe that if $\hd(S,\bar{Q}^*)\le k$,
    then at least $k+1$ fragments $S_i$ match $\bar{Q}$ exactly, so~$\bar{Q}$ must be the
    majority of~$S_{0},\ldots,S_{2k}$.
    Moreover, since $Q$ is primitive, there is at most one $j\in \fragmentco{0}{|Q|}$
    with $\bar{Q}=\rot^j(Q)$.

    For the running time, note that we can compute the majority of~$\Oh(k)$ elements
    with a classic linear-time algorithm by Boyer and Moore \cite{Moore91} using $\Oh(k)$
    equality tests; as (by \cref{lm:streq}) each equality test takes $\Oh(1)$ time in the
    \modelname model, we obtain the claimed running time and hence the claim.
\end{proof}

\begin{algorithm}
    \SetKwBlock{Begin}{}{end}
    \SetKwFunction{RFR}{FindRelevantFragment}
    \SetKwFunction{Rots}{FindRotation}
    \RFR{$P$, $T$, $d$, $Q$}\Begin{
        $j \gets \Rots{$\floor{\threehalfs d}$, $Q$, $T\fragmentco{n-m}{m}$}$\;
        \lIf{$j=\bot$}{\Return{$\eps$}}
        $\delta \gets 0$; $r \gets n-m+j$\;
        generator $\mathbf{G} \gets \misOpL{T\fragmentco{n-m+j}{n}}{Q}$\;
        \While{$\delta \le \threehalfs d$ \KwSty{and} $(\pi \gets \nxt{$\mathbf{G}$}) \ne \bot$}{
            $r \gets n-m+j+\pi$\;
            $\delta \gets \delta + 1$\;
        }
        \lIf{$\delta \le \threehalfs d$}{$r \gets m$}
        $\delta' \gets 0$; $\ell \gets n-m+j$; $\ell' \gets (n-m+j)\bmod {|Q|}$\;
        generator $\mathbf{G'} \gets \mibOpL{T\fragmentco{\ell'}{n-m+j}}{Q}$\;
        \While{$\delta' \le \threehalfs d$ \KwSty{and} $(\pi \gets \nxt{$\mathbf{G'}$})\ne \bot$}{
            $\ell \gets \ell'+|Q|\cdot \ceil{(\pi+1)/|Q|}$\;
            $\delta' \gets \delta'+1$\;
        }
        \lIf{$\delta' \le \threehalfs d$}{$\ell \gets \ell'$}
        \Return{$T\fragmentco{\ell}{r}$}\;
    }
    \caption{A \modelname algorithm computing a {\em relevant} fragment $T$:
    a fragment $T'$ such that all $k$-mismatch  occurrences (for any $k \le d/2$) of~$P$ in $T$
        start at a position in $T'$ which is a multiple of~$|Q|$.}\label{alg:relevant}
\end{algorithm}

\begin{lemma}[\texttt{FindRelevantFragment($P$, $T$, $d$, $Q$)}]\label{lem:relevant}
    Let $P$ denote a pattern of~length $m$ and let $T$ denote a text of~length $n\le
    \threehalfs m$. Further, let $d$ denote a positive integer and let
    $Q$ denote a primitive string that satisfies $|Q|\le m/8d$ and $\hd(P,Q^*)\le d$.

    Then, using $\Oh(d)$ time plus $\Oh(d)$ \modelname operations,
    we can report a fragment $T'=T\fragmentco{\ell}{r}$ such that $\hd(T',Q^*)\le 3d$
    and, for every $k\le d/2$, the set $\Occ_k(P,T')=\{p-\ell \mid p \in \Occ_k(P,T)\}$
    contains only multiples of~$|Q|$.
\end{lemma}
\begin{proof}
    We start by using a call to \texttt{FindRotation} from \cref{lem:findAPeriod}
    to find the unique integer $j$ such that $\hd(T\fragmentco{n-m}{m},\rot^j(Q)^*)\le
    \threehalfs d$. If no such $j$ exists, then we return the empty string $\varepsilon$.
    Otherwise, we proceed by computing the rightmost position $r$
    such that $\hd(T\fragmentco{n-m+j}{r},Q^*)\le \threehalfs d$
    and the leftmost position $\ell$ (with $\ell \equiv (n-m+j) \pmod{|Q|}$)
    such that $\hd(T\fragmentco{\ell}{n-m+j},Q^*)\le \threehalfs d$;
    afterwards, we return the fragment $T\fragmentco{\ell}{r}$.
    Consider \cref{alg:relevant} for implementation details.

    For the correctness, first observe that $\hd(T\fragmentco{n-m}{m},
    \rot^{p-n+m}(Q))\le \threehalfs d$ for each  $p\in \Occ_k(P,T)$:
    By triangle inequality (\cref{tria}), we have\[
        \hd(T\fragmentco{p}{p+m},Q^*)\le k+\hd(P,Q^*)\le \threehalfs d.
    \] Since $p\le n-m$ and $p+m\ge m$,
    this yields $\hd(T\fragmentco{n-m}{m}, \rot^{p-n+m}(Q))\le \threehalfs d$.
    Moreover, $|T\fragmentco{n-m}{m}|=2m-n \ge m/2 \ge 4d|Q| \ge (2\cdot
    \floor{\threehalfs d} + 1)|Q|$,
    so the call to \texttt{FindRotation} is valid.

    Hence, if the call to {\tt FindRotation} returns $\bot$, then $\Occ_k(P,T)=\emptyset$
    (for each $k \le d/2$).
    Otherwise, each position $p\in \Occ_k(P,T)$ satisfies $p \equiv  n-m+j \equiv \ell
    \pmod{|Q|}$. Moreover, we have \begin{align*}
        \hd(T\fragmentco{n-m+j}{p+m},Q^*) &\le \hd(T\fragmentco{p}{p+m},Q^*)
    \le \threehalfs d, \quad\text{and}\\
        \hd(T\fragmentco{p}{n-m+j},Q^*)&\le \hd(T\fragmentco{p}{p+m},Q^*) \le
        \threehalfs d.
        \end{align*}
    Hence, the fragment $T'=T\fragmentco{\ell}{r}$ contains all $k$-mismatch occurrences of~$P$
    in $T$ (for any $k \le d/2$), and all these occurrences start at multiples of~$|Q|$ in $T'$.
    Moreover, $\hd(T',Q^*)=
    \hd(T\fragmentco{\ell}{n-m+j},Q^*)+\hd(T\fragmentco{n-m+j}{r},Q^*)\le 3d$.

    For the running time (and the number of~\modelname operations used), the call to
    {\tt FindRotation} uses $\Oh(d)$ time plus $\Oh(d)$ \modelname operations;
    the same is true for the usage of~\misOpLName and \mibOpLName.
    Thus, the algorithm uses $\Oh(d)$ time plus $\Oh(d)$ \modelname operations in total,
    completing the proof.
\end{proof}

\begin{lemma}[\texttt{DistancesRLE($P$, $T$, $Q$)}: Implementation of
    \cref{lem:rle}]\label{lem:rle_alg}
    Let $P$ denote a pattern of~length $m$ and let $T$ denote a text of~length $n \le
    \threehalfs m$. Further, let $d$ denote a positive integer and
    let $Q$ denote a string that satisfies $\hd(P,Q^*)= \Oh(d)$ and $\hd(T,Q^*)=\Oh(d)$.

    Then, using $\Oh(d^2\log \log d)$ time plus $\Oh(d)$ \modelname operations,
    we can compute a run-length encoded sequence of
    $h_j := \hd(T\fragmentco{j|Q|}{j|Q|+m},P)$ for $0\le j \le (n-m)/|Q|$.
\end{lemma}
\begin{algorithm}[t]
    \SetKwBlock{Begin}{}{end}
    \SetKwFunction{RLE}{DistancesRLE}
    \RLE{$P$, $T$, $Q$}\Begin{
    \tcp{Marking phase}
    $M \gets \{\}$\;
    \ForEach{$\tau \in \misOp{T}{Q^*}$}{
        $M \gets M \cup \{(\tau-m,1),(\tau,-1)\}$\;
        \ForEach{$\pi \in \misOp{P}{Q^*}$}{
            $M \gets M \cup \{(\tau-\pi-1,\hd(P[\pi],T[\tau])-2),(\tau-\pi,2-\hd(P[\pi],T[\tau]))\}$\;
        }
    }
    \BlankLine
    \tcp{Sliding-window phase}
    sort $M$\;
    $h \gets |\misOp{P}{Q^*}|$\;
    \lForEach{$(i',w)\in M$ with $i' < 0$}{$h \gets h+w$}
    $i \gets 0$\;
    \ForEach{$(i',w)\in M$ with $0\le i' < n-m$ {\rm sorted by $i'$}}{
        Output a block of~$\ceil{(i'+1)/q}-\ceil{i/q}$ values $h$\;
        $i \gets i'+1$\;
        $h \gets h+w$\;
    }
    Output a block of~$\ceil{(n-m+1)/q}-\ceil{i/q}$ values $h$\;
    }
    \caption{A \modelname algorithm for \cref{lem:rle}}\label{alg:rle}
\end{algorithm}
\begin{proof}
    Observe that \cref{clm:hj} already gives rise to an algorithm:
    Starting with $h_0$, we can obtain the value $h_{j+1}$ from $h_j$ by adding the value\[
        h_{j+1} - h_j = |\MIS(T,Q^*)\cap \fragmentco{jq+m}{(j+1)q+m}|
        -|\MIS(T,Q^*)\cap \fragmentco{jq}{(j+1)q}|
        -\mu_{j+1}+\mu_j,
    \] where $\mu_{j+1}$ and $\mu_{j}$ are defined as in \cref{lem:rle}.

    We implement this idea in two steps: In the first step, we compute
    the values $\mu_j$ (using marking) and the positions of~mismatches in $\MIS(T, Q^*)$
    (using \misOpName from \cref{lm:misoph2}).
    In the second step, we use a sliding-window approach (with the positions computed
    in the first step interpreted as events) to output the sequence of~values of~$h_j$.
    Consider the pseudo-code (\cref{alg:rle}) for implementation details.

    For the correctness, in the marking phase,
    the algorithm constructs a multiset $M$ of~pairs $(i,w)$ (where~$i$ can be interpreted
    as a position in $T$ and $w$ as the weight) so that $h_j - \hd(P,Q^*)=w(M,j|Q|)$,
    where $w(M,i)=\sum_{\{(i',w)\in M \mid i' < i \}} w$ denotes the sum of~weights of~pairs
    $(i',w)$ with $i'<i$.

    Specifically, for each $\tau\in \MIS(T,Q^*)$, the algorithm first inserts to $M$ pairs
    $(\tau-m,1)$ and $(\tau,-1)$.
    As a result, for each position $i$ with $0\le i \le n-m$, we have
    $w(M,i)=|\MIS(T,Q^*)\cap \fragmentco{i}{i+m}|$. In particular, if $i=j|Q|$, then
    $w(M,i)=\hd(T\fragmentco{j|Q|}{j|Q|+m},Q^*)$.
    Next, for each $\tau \in \MIS(T,Q^*)$ and each $\pi \in \MIS(P,Q^*)$,
    the algorithm inserts to $M$ pairs $(\tau-\pi-1,\hd(P[\pi],T[\tau])-2)$ and
    $(\tau-\pi,2-\hd(P[\pi],T[\tau]))$.
    As a result, the values
    $w(M,i)$ with $i\ne \tau-\pi$ are not altered, whereas $w(M,i)$ for $i=\tau-\pi$ is
    decreased by the number of~marks placed in the proof~of~\cref{lem:rle} at position
    $i=\tau-\pi$ of~$T$ due to positions $\tau$ in $T$ and $\pi$ in $P$.
    Consequently, we have
    $w(M,j|Q|)=\hd(T\fragmentco{j|Q|}{j|Q|+m},Q^*)-\mu_j$, which yields $h_j =
    \hd(P,Q^*)+w(M,j|Q|)$ due to \cref{clm:hj}.

    Hence, in order to construct the sequence $h_j$, the algorithm sorts the pairs in $M$
    and determines the partial sums $w(M,i)$. In each block of~$\fragment{i}{i'}$ of~equal
    partial sums, the algorithm reports a block with all $\ceil{(i'+1)/q}-\ceil{i/q}$
    entries $h_j$ for $j|Q|\in \fragmentco{i}{i'}$, which is indeed correct.

    The running time is $\Oh(d^2 \log \log d)$ (dominated by sorting $M$, which consists
    of~$\Oh(d^2)$ integer pairs) plus $\Oh(d)$ \modelname operations (for the calls to
    $\misOp{P}{Q^*}$ and $\misOp{T}{Q^*}$
    and for accessing the mismatching positions of~$P$ and $T$), thus completing the
    proof.
\end{proof}

\begin{lemma}[\texttt{PeriodicMatches($P$, $T$, $k$, $d$, $Q$)}: Implementation of~\cref{cor:aux}]\label{lm:permat}
    Let $P$ denote a pattern of~length $m$ and let $T$ denote a text of~length $n$.
    Further, let $k \le m$ denote a non-negative integer, let $d \ge 2k$
    denote a positive integer, and let $Q$ denote a primitive string $Q$ that satisfies
    $|Q| \le n/8d$ and $\hd(P, Q^*) \le d$.

    There is an algorithm that computes the set $\Occ_k(P,T)$,
    represented as $\Oh(n/m\cdot d^2)$ arithmetic progressions
    with difference $|Q|$ (or as $\Oh(n/m \cdot d)$ individual positions if $\hd(P,Q^*)= d$).
    The algorithm uses $\Oh(n/m\cdot d^2\log \log d)$ time plus $\Oh(n/m\cdot d)$
    \modelname operations.
\end{lemma}
\begin{proof}
    First, we split the string $T$ into $\floor{2n/m}$ blocks
    $T_i := T\fragmentco{\floor{i\cdot {m}/2}}{\min\{n, \floor{(i+3)\cdot {m}/2}-1\}}$
    for $0\le i < \floor{2n/m}$.
    For each block $T_i$, we call \texttt{FindRelevantFragment($P$, $T_i$, $d$, $Q$)}
    from \cref{lem:relevant} to obtain a fragment $T'_i = T\fragmentco{\ell_i}{r_i}$ containing
    all $k$-mismatch occurrences of~$P$ in $T_i$.
    Next, we call \texttt{DistancesRLE($P$, $T'_i$, $Q$)} from
    \cref{lem:rle_alg}, yielding a run-length encoded sequence of~values
    $h_t := \hd(T'_i\fragmentco{t|Q|}{t|Q|+m},P)$ for $0\le t \le (|T'_i|-m)/|Q|$.
    For each run $h_{t}=\cdots = h_{t'} \le k$, we add the arithmetic progression
    $\{ \ell_i + j \cdot |Q|: j\in \fragment{t}{t'}\}$ to $\Occ_k(P, T)$.
    In the end, we return the set $\Occ_k(P, T)$.

    For the correctness, note that we essentially follow the proof~of~\cref{cor:aux}.
    In particular, each $k$-mismatch occurrence of~$P$ in $T$ is contained in exactly one
    of~the fragments $T_i$.
    By \cref{lem:relevant}, we see that $T'_i$ contains all the $k$-mismatch occurrences
    of~$P$ in $T_i$. Moreover, as $\Occ_k(P,T'_i)$ only contains  multiples of~$|Q|$,
    each $p\in \Occ_k(P,T'_i)$ corresponds to an entry $h_j \le k$.
    Consequently, all the $k$-mismatch occurrences of~$P$ in $T$ are found.
    Furthermore, since $h_j=\hd(T\fragmentco{\ell_i+j|Q|}{\ell_i+j|Q|+m},P)\le k$ holds
    whenever $\ell_i+j|Q|$ is reported, there are no false positives.

    If $\hd(P,Q^*)= d$, then for each $i$, the number of~entries $h_j$ with $h_j\le k$ is
    $\Oh(d)$ by
    \cref{lem:rle}, so the corresponding positions $\ell_i+j|Q|$ can be added to
    $\Occ_k(P,T)$ individually.

    The bounds on the overall running time follow from \cref{lem:relevant,lem:rle_alg}
    due to $\hd(P,Q^*)\le d$ and since $\hd(T_i,Q^*)\le 3d$ holds for each $i$ by
    \cref{lem:rle_alg}.
\end{proof}

\subsection{Computing Occurrences in the Non-Periodic Case}

\begin{algorithm}[t]
    \SetKwBlock{Begin}{}{end}
    \SetKwFunction{verify}{Verify}
    \SetKwFunction{brmtch}{BreakMatches}
    \SetKwFunction{exmtch}{ExactMatches}
    \brmtch{$P$, $T$, $\{ B_1 = P\fragmentco{b_1}{b_1 + |B_1|}, \dots, B_{2k} =
        P\fragmentco{b_{2k}}{b_{2k} + |B_{2k}|} \}$, $k$}\Begin{
        multi-set $M \gets \{\}$; $\Occ_k(P, T) \gets \{\}$\;
        \For{$i \gets 1$ \KwSty{to} $2k$}{
            \ForEach{$\tau \in \exmtch{$B_i$, $T$}$}{
                $M \gets M\cup\{ \tau - b_i\}$\tcp*{Mark position $\tau - b_i$ in $T$}
            }
        }
        sort $M$\;
        \ForEach{$\pi\in \fragment{0}{n-m}$ that appears at least $k$ times in $M$}{
            \lIf{\verify{$P$, $T\fragmentco{\pi}{\pi+m}$, $k$}}{$\Occ_k(P, T) \gets
            \Occ_k(P, T) \cup \{ \pi \}$}
        }
        \Return{$\Occ_k(P, T)$}\;
    }
    \caption{A \modelname model algorithm for \cref{lm:hdC}}\label{alg:hdC}
\end{algorithm}
\begin{lemma}[{\tt BreakMatches($P$, $T$, $\{B_1,\dots,B_{2k}\}$, $k$)}:
    Implementation of~\cref{lm:hdC}]\label{lm:imphdA}
    Let $P$ denote a string of~length $m$ having $2k$ disjoint breaks $B_1,\dots,B_{2k}
    \substr P$ each satisfying $\per(B_i) \ge m / \alphav k$.
    Further, let $T$ denote a string of~length $n \le \threehalfs m$.

    Then, we can compute the set $\Occ_k(P, T)$ using $\Oh(k^2 \log \log k)$
    time plus $\Oh(k^2)$ \modelname operations.
\end{lemma}
\begin{proof}
    The implementation of~(the marking in the proof~of) \cref{lm:hdC} is straightforward:
    For each break $B_i = P\fragmentco{b_i}{b_i + |B_i|}$,
    we use a call to {\tt ExactMatches}$(B_i, T)$ from \cref{lm:emath} to find all
    exact occurrences $\OccEx(B_i, T)$.
    For each occurrence $\pi \in \OccEx(B_i, T)$, we mark position $\pi - b_i$ in $T$.
    Having placed all marks, we run {\tt Verify} from \cref{lm:verifyh}
    for every position $\pi\in \fragment{0}{n-m}$ in $T$ that has at least $k$ marks.
    In the end, we return all positions where {\tt Verify} confirmed a $k$-mismatch
    occurrence. See \cref{alg:hdC} for a pseudo-code.

    For the correctness, note that we placed the marks as in the proof~of~\cref{lm:hdC};
    in particular, by \cref{cl:c2}, any $\pi\in \Occ_k(P,T)$  has at least $k$ marks.
    As we verify every possible candidate using {\tt Verify},
    we report no false positives, and thus the algorithm is correct.

    We continue with analyzing the number of~\modelname
    operations used. As every break~$B_i$ has period $\per(B_i)>m/\alphav k$,
    every call to {\tt ExactMatches} uses $\Oh(k)$ basic \modelname operations;
    thus, all calls to {\tt ExactMatches} use $\Oh(k^2)$ basic operations in total.
    As there are at most $\Oh(k^2 / k) = O(k)$ positions that we verify, and every
    call to {\tt Verify} uses $\Oh(k)$ \modelname operations,
    the verifications use $\Oh(k^2)$ \modelname operations in total.

    Finally, for the running time, by \cref{cl:c1}, we place at most $\Oh(k^2)$ marks in $T$, so
    the marking step uses $\Oh(k^2)$ operations in total.
    Further, finding all positions in $T$ with at least $k$ marks can be done via a linear scan
    over the multiset $M$ of~all marks after sorting $M$, which can be done in time
    $\Oh(k^2 \log \log k)$. Overall, \cref{alg:hdC} runs in time $\Oh(k^2 \log\log
    k)$ plus $\Oh(k^2)$ \modelname operations.
\end{proof}

\begin{algorithm}[t]
    \SetKwBlock{Begin}{}{end}
    \SetKwFunction{appm}{PeriodicMatches}
    \SetKwFunction{rpmtch}{RepetitiveMatches}
    \rpmtch{$P$, $T$, $\{ (R_1 = P\fragmentco{r_1}{r_1 + |R_1|}, Q_1) \dots, (R_{r} =
        P\fragmentco{r_{r}}{r_{r} + |R_{r}|}, Q_r) \}$, $k$}\Begin{
        multi-set $M \gets \{\}$; $\Occ_k(P, T) \gets \{\}$\;
        \For{$i \gets 1$ \KwSty{to} $r$}{
            \ForEach{$\tau \in \appm{$R_i$, $T$, $\floor{\betavh \cdot k/m \cdot |R_i|}$,
                $\ceil{\betav\cdot k/m \cdot |R_i|}$, $Q_i$}$}{
                $M \gets M\cup\{ (\tau - r_i, |R_i|) \}$\tcp*{Place $|R_i|$ marks at
                position $\tau - r_i$ in $T$}
            }
        }
        sort $M$ by positions\;
        \ForEach{$\pi\in \fragment{0}{n-m}$ appearing at least
            $\sum_{(\pi, v) \in M} v \ge \sum_{i=1}^r |R_i| - m/\betavh$ times in $M$}{
            \lIf{\verify{$P$, $T\fragmentco{\pi}{\pi+m}$, $k$}}{$\Occ_k(P, T) \gets
            \Occ_k(P, T) \cup \{ \pi \}$}
        }
        \Return{$\Occ_k(P, T)$}\;
    }
    \caption{A \modelname model algorithm for \cref{lm:hdB}}\label{alg:hdB}
\end{algorithm}
\begin{lemma}[{\tt RepetitiveMatches($P$,$T$,$\{ (R_1, Q_1) \dots, (R_{r},Q_r)\}$,$k$)}:
    Implementation of~\cref{lm:hdB}]\label{lm:imphdB}
    Let $P$ denote a string of~length~$m$,
    let $T$ denote a string of~length~$n \le \threehalfs m$,
    and let $k\le m$ denote a positive integer.
    Suppose that $P$ contains disjoint repetitive regions $R_1,\ldots, R_{r}$
    of~total length at least $\sum_{i=1}^r |R_i| \ge \deltavN/\deltavD\cdot m$
    such that each region $R_i$ satisfies $|R_i| \ge m/\betav k$ and has a
    primitive approximate period~$Q_i$
    with $|Q_i| \le m/\alphav k$ and $\hd(R_i,Q_i^*) = \ceil{\betav k/m\cdot |R_i|}$.

    Then, we can compute the set $\Occ_k(P,T)$ using $\Oh(k^2 \log \log k)$ time plus
    $\Oh(k^2)$ \modelname operations.
\end{lemma}
\begin{proof}
    As in the proof~of~\cref{lm:hdB},
    set $m_R := \sum_{i=1}^r |R_i| \ge \deltavN/\deltavD\cdot m$ and
    define for every $1 \le i \le r$ the values
    $k_i := \floor{\betavh \cdot k/m \cdot |R_i|}$ and  $d_i := \ceil{\betav \cdot k/m
    \cdot |R_i|}=|\MIS(R_i, Q_i^*)|$.
    Further, write $R_i = P\fragmentco{r_i}{r_i + |R_i|}$.

    We implement the marking of~the proof~of~\cref{lm:hdB}:
    for every repetitive region $R_i$, we call $\appm{$R_i$, $T$, $k_i$, $d_i$, $Q_i$}$
    from \cref{lm:permat} to obtain the set $\Occ_{k_i}(R_i, T)$.
    Next, for each position $\tau \in \Occ_{k_i}(R_i, T)$,
    we place $|R_i|$ marks at position $\tau - r_i$.
    Note that for performance reasons,
    instead of~placing $|R_i|$ unweighted marks, we place a single mark of
    weight $|R_i|$ at position $\tau - r_i$.

    Having placed all marks, we run \verify from \cref{lm:verifyh}
    for every position $\pi\in \fragment{0}{n-m}$ in $T$ that has marks of~total weight at
    least $m_r - m/\betavh$.
    In the end, we return all positions where \verify confirmed a $k$-mismatch
    occurrence. See \cref{alg:hdB} for a pseudo-code.

    For the correctness, first note that in every call to \appm from \cref{lm:permat},
    we have $\tbetav k/m \cdot |R_i| \ge d_i =  \ceil{\betav k/m\cdot |R_i|}  = \hd(R_i,
    Q_i^*) \ge 2k_i$,
    so $|Q_i| \le m/\alphav k \le |R_i|/8d_i$; hence, we can indeed call \appm in this case.
    Further, note that we placed the marks as in the proof~of~\cref{lm:hdB};
    in particular, by \cref{cl:b2}, any $\pi \in \Occ_k(P,T)$
    has at least $m_R - m/\betavh$ marks.
    As we verify every possible candidate using \verify,
    we report no false positives, and thus the algorithm is correct.

    For the number of~\modelname operations,
    observe that during the marking step, for every repetitive region~$R_i$,
    we call \appm once, and the call uses $\Oh(n/|R_i| \cdot d_i)
    = \Oh(m/|R_i| \cdot k/m\cdot |R_i|) = \Oh(k)$ \modelname
    operations. Hence, the marking step uses
    $\Oh(r\cdot k) = \Oh(k^2)$ \modelname operations in total.
    Next, during the verification step, by \cref{cl:b1,cl:b2}, we call \verify
    at most $\Oh(k)$ times. As each call to \verify uses $\Oh(k)$ \modelname
    operations, the verification step in total uses $\Oh(k^2)$ \modelname
    operations.
    Overall, \cref{alg:hdB} uses $\Oh(k^2)$ \modelname operations.

    Finally, for the running time, with similar calculations as for the number of
    \modelname operations, we see that the marking step, including calls to \appm, takes time
    $\sum_i \Oh(n/|R_i|\cdot d_i^2 \log \log d_i)=\sum_i \Oh(|R_i|/m \cdot k^2 \log\log k)
    = \Oh(k^2 \log\log k)$.
    Further, for every~$R_i$, we place at most
    $|\Occ_{k_i}(R_i, T)|$ (weighted) marks, which can be bounded using \cref{cor:aux} by
    $|\Occ_{k_i}(R_i,T)| = \Oh(n/|R_i| \cdot d_i) = \Oh(k).$
    Thus, we place $|M| = \Oh(k^2)$ (weighted) marks in total.
    Therefore, we can sort $M$ (by positions) in time $\Oh(k^2 \log\log k)$;
    afterwards, we can find the elements with total weight at least $m_R - m/\betavh$
    via a linear scan over $M$ in time $\Oh(k^2)$.
    Hence, \cref{alg:hdB} runs in $\Oh(k^2 \log \log k)$ overall time, completing the
    proof.
\end{proof}

\subsection{A \modelname Model Algorithm for Pattern Matches with Mismatches}

\begin{algorithm}[t]
    \SetKwBlock{Begin}{}{end}
    \SetKwFunction{mism}{MismatchOccurrences}
    \mism{$P$, $T$, $k$}\Begin{
        {\bf (} $B_1,\dots,B_{2k}$ {\bf or} $(R_1, Q_1),\dots,(R_r, Q_r)$ {\bf or}
        $Q$ {\bf )} $\gets \anly{$P$, $k$}$\;
        $\Occ_k(P, T) \gets \{\}$\;
        \For{$i \gets 0$ \KwSty{to} $\floor{2n/m}$}{
            $T_i \gets T\fragmentco{\floor{i\cdot {m}/2}}   {\min\{n, \floor{(i+3)\cdot {m}/2}-1\}}$\;
        \If{breaks $B_1,\dots,B_{2k}$ exist}{
             $\Occ_k(P, T_i) \gets \brmtch{$P$, $T_i$, $\{ B_1, \dots, B_{2k}\}$, $k$}$\;
        }\ElseIf{repetitive regions $(R_1, Q_1),\dots,(R_r, Q_r)$ exist}{
             $\Occ_k(P, T_i) \gets \rpmtch{$P$, $T_i$, $\{(R_1, Q_1),\dots,(R_r, Q_r)\}$, $k$}$\;
        }\lElse{%
            $\Occ_k(P, T_i) \gets \appm{$P$, $T_i$, $k$, $\betav k$, $Q$}$%
        }
        $\Occ_k(P, T)\gets \Occ_k(P, T) \cup\{\ell + im/2 : \ell \in \Occ_k(P,T_i)\}$\;
    }
    \Return{$\Occ_k(P, T)$}\;
    }
    \caption{A \modelname model algorithm for \cref{thm:hdmain}}\label{alg:hdmain}
\end{algorithm}
\hdalg*
\begin{proof}
    First, we split $T$ into overlapping parts $T_1,\dots,T_{\floor{2n/m}}$ of
    length less than $\threehalfs m$ each. In order to compute $\Occ_k(P,T_i)$ for each $i$,
    we follow the structure of~the proof~of~\cref{thm:hdmain}:
    We first call $\anly{$P$,$k$}$ from \cref{prp:Ialg}.
    If the call to \anly{$P$,$k$} yields $2k$ disjoint
    breaks $B_1,\dots,B_{2k}$ in~$P$,
    then we call \brmtch{$P$, $T_i$,$\{B_1,\dots,B_{2k}\}$, $k$}
    from \cref{lm:imphdA}.
    If the call to \anly{$P$,$k$} yields disjoint repetitive regions $R_1,\dots,R_r$
    (and corresponding approximate periods $Q_1, \dots, Q_r$), then we call
    \rpmtch{$P$, $T_i$, $\{(R_1, Q_1),\dots,(R_r, Q_r)\}$, $k$} from \cref{lm:imphdB}.
    Finally, if the call to \anly{$P$,$k$} yields an approximate period $Q$,
    then we call \appm{$P$, $T_i$, $k$, $\betav k$, $Q$} from \cref{lm:permat}.
    The resulting set $\Occ_k(P,T)$ is obtained by combining the sets $\Occ_k(P,T_i)$.
    Consider \cref{alg:hdmain} for a visualization as pseudo-code.

    \noindent For the correctness, first observe that we do not lose any occurrences
    by splitting the string $T$, since every length-$m$ fragment of
    $T$ is contained in one of~the fragments $T_i$.
    Second, by \cref{prp:Ialg} and due to $|T_i|\le \threehalfs m$,
    the parameters in the calls to \brmtch and \rpmtch each satisfy the
    requirements.
    Lastly, if we use \appm, notice that again by \cref{prp:Ialg} the string $Q$ satisfies
    $\hd(P, Q^*) \le \betav k$ and $|Q|\le m/\alphav k \le m/(8\cdot \betav k)$;
    hence, we can indeed call \appm in this case.

    For the number of~\modelname operations used, the call to \anly
    uses $\Oh(k)$ \modelname operations, each call to \brmtch and \rpmtch
    uses $\Oh(k^2)$ \modelname operations, and each call to \appm
    uses $\Oh(k)$ \modelname operations.
    As there are at most $\Oh(n/m)$ calls to \brmtch, \rpmtch, and \appm,
    we can bound the total number of~\modelname operations used by $\Oh(n/m
    \cdot k^2)$.

    Similarly, for the running time, the call to \anly takes $\Oh(k)$ time,
    whereas each call to \brmtch, \rpmtch, and \appm takes $\Oh(k^2 \log\log k)$ time.
    Again, since there are at most $\Oh(n/m)$ calls to \brmtch, \rpmtch, and \appm each,
    and combining  the sets $\Occ_k(P,T_i)$ to  $\Occ_k(P,T)$ can be implemented
    in total time $\Oh(n/m \cdot k^2)$,
    we can bound the total running time by $\Oh(n/m \cdot k^2 \log\log k)$, thus
    completing the proof.
\end{proof}

\section{Structural Insights into Pattern Matching with Edits}

In this section, we develop insight into the structure of~$k$-error occurrences of~a pattern $P$
in a text $T$.
We prove the following result, which is analogous to~\cref{thm:hdmain} and
is~\cref{ed:mthm_intro} with explicit constants.

\begin{restatable}[Compare~\cref{thm:hdmain}]{theorem}{edmain}\label{thm:edmain}
    Given a pattern $P$ of~length $m$, a text $T$ of~length $n$, and a positive integer $k\le m$,
    then at least one of the following holds.
    \begin{itemize}
        \item The $k$-error occurrences of~$P$ in $T$ satisfy
            $|\floor{\OccE_k(P,T)/k}|\le \Ethmboundt \cdot n/m \cdot k$.
        \item There is a primitive string $Q$ of~length $|Q| \le m/\thmboundt k$ that
            satisfies $\edl{P}{Q} < 2k$.\ifx\edmaint\undefined\lipicsEnd\fi
    \end{itemize}
\end{restatable}
\def\edmaint{1}

Similarly to~\cref{sec:km}, we start with an analysis of~the (approximately) periodic case.

\subsection{Characterization of~the Periodic Case}

\begin{restatable}[{Compare~\cref{lem:aux}}]{theorem}{edaux}\label{lem:Eaux}
    Let $P$ denote a pattern of~length $m$, let $T$ denote a text of~length~$n$,
    and let $k\le m$ denote a non-negative integer such that $n < \threehalfs m +k$.
    Suppose that the $k$-error occurrences of~$P$ in $T$ include a prefix of~$T$ and a
    suffix of~$T$.
    If there are a positive integer $d\ge 2k$ and a primitive string $Q$
    with $|Q|\le m/8d$ and $\ed(P,Q^*)=\edl{P}{Q}\le d$, then each of~following holds:
     \begin{enumerate}[(a)]
        \item For every $p\in \OccE_k(P,T)$, we have $p\bmod |Q|\le 3d$ or $p\bmod |Q|\ge
            |Q|-3d$.\label{it:Emult}
        \item The string $T$ satisfies $\edl{T}{Q}\le 3d$.\label{it:Etext}
        \item If $\edl{P}{Q}=d$, then $|\floor{\OccE_k(P,T)/d}|\le \pvarphiv d$.\label{it:Efew}
        \item The set $\OccE_k(P,T)$ can be decomposed into $617d^3$ arithmetic progressions
        with difference $|Q|$. \label{it:Eprog}
        \ifx\edauxt\undefined\lipicsEnd\fi
     \end{enumerate}
\end{restatable}
\def\edauxt{1}

\begin{lemma}\label{fct:split}
    Let $k$ denote a positive integer, let $Q$ denote a primitive string, and let $S$ denote a
    string of~length $|S|\ge (2k+1)|Q|$.
    If there are integers $\ell \le r$ and $\ell' \le r'$ such that
    $\ed(S,Q^\infty\fragmentco{\ell}{r})\le k$ and $\ed(S,Q^\infty\fragmentco{\ell'}{r'})\le k$,
    then there are integers $j, j'$ and a decomposition $S=S_L\cdot S_R$
    that satisfy
    \begin{align*}
        \ed(S,Q^\infty\fragmentco{\ell\phantom{'}}{\phantom{'}r})
        & =\ed(S_L,Q^\infty\fragmentco{\ell\phantom{'}}{\phantom{'}j|Q|})+\ed(S_R,
        Q^\infty\fragmentco{j|Q|\phantom{'}}{\phantom{'}r})\quad\text{and}\\
        \ed(S,Q^\infty\fragmentco{\ell'}{r'})
        &=\ed(S_L,Q^\infty\fragmentco{\ell'}{j'|Q|})+\ed(S_R,Q^\infty\fragmentco{j'|Q|}{r'}).
    \end{align*}
    Furthermore, if $|Q|=1$, then the assumption $|S|\ge (2k+1)|Q|$ is not required.
\end{lemma}
\begin{proof}
    If $|Q|=1$, we can set $S_L := S$, $S_R :=\eps$, $j :=r$, and $j' :=r'$.

    Now assume that we have $|S|\ge (2k+1)|Q|$
    and define $S_i := S\fragmentco{i|Q|}{(i+1)|Q|}$ for $0\le i \le 2k$.
    Further, fix optimal alignments between $S$ and $Q^\infty\fragmentco{\ell}{r}$
    and between $S$ and $Q^\infty\fragmentco{\ell'}{r'}$.

    Observe that at least one of~the fragments $S_i$ is aligned without errors in
    \emph{both} alignments.
    Let us fix such a fragment $S_i$ and observe that $S_i$ is a length-$|Q|$ substring of
    $Q^\infty$, so
    $S_i = \rot^p(Q)$ for some $p\in \fragmentco{0}{|Q|}$. An illustration is provided in
    \cref{fig:fct_split}.
    Based on this value, we set $S_L := S\fragmentco{0}{i|Q|+p}$ and $S_R :=
    S\fragmentco{i|Q|+p}{|S|}$.

    Next, consider the fragment $Q'$ of~$Q^\infty\fragmentco{\ell}{r}$ that is aligned to $S_i$
    in the alignment fixed earlier.
    The fragment $Q'$ matches $\rot^p(Q)$. As $Q$ is primitive, $Q'$ is thus of~the form
    $Q' = Q^{\infty}\fragmentco{j|Q|-p}{(j+1)|Q|-p}$ for some integer~$j$.
    Consequently, \[
        \ed(S,Q^\infty\fragmentco{\ell}{r})
        = \ed(S_L,Q^\infty\fragmentco{\ell}{j|Q|})+\ed(S_R,Q^\infty\fragmentco{j|Q|}{r}).
        \] A similar argument shows that for some integer $j'$, we also have\[
        \ed(S,Q^\infty\fragmentco{\ell'}{r'})
        =\ed(S_L,Q^\infty\fragmentco{\ell'}{j'|Q|})+\ed(S_R,Q^\infty\fragmentco{j'|Q|}{r'}).
    \] This completes the proof.
\end{proof}

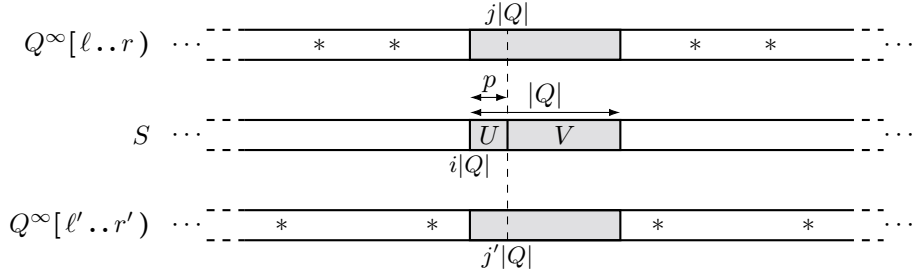
\begin{figure}[t]
    \centering
                   \begin{tikzpicture}
                	\draw[fill=lipicsYellow!80,thick] (4,0) rectangle (6,0.4);

                    \draw[thick] (1,0) -- (9,0);
                    \draw[thick] (1,0.4) -- (9,0.4);
                    \draw[dashed,thick] (0.5,0) -- (1,0);
                    \draw[dashed,thick] (0.5,0.4) -- (1,0.4);
                    \draw[dashed,thick] (9,0) -- (9.5,0);
                    \draw[dashed,thick] (9,0.4) -- (9.5,0.4);
                    \node[label = {left: $\cdots$}]  at (0.75,0.2) {};
                    \node[label = {right: $\cdots$}]  at (9.25,0.2) {};

                    \node[label = {left: $S$}]  at (0,0.2) {};
                    \node[label = {below: \small{$i|Q|$}}]  at (4,0.2) {};

                    \draw[thick] (4.5,0) -- (4.5,0.4);
                    \node[label = {left: $U$}]  at (4.65,0.2) {};
                    \node[label = {left: $V$}]  at (5.65,0.2) {};
                    \draw[{Latex[length=1.5mm, width=1mm]}-{Latex[length=1.5mm,
                    width=1mm]}] (4,0.7) -- (4.5,0.7);
                    \node[label = {above: $p$}]  at (4.26,0.52) {};
                    \draw[{Latex[length=1.5mm, width=1mm]}-{Latex[length=1.5mm,
                    width=1mm]}] (4,0.5) -- (6,0.5);
                    \node[label = {above: $|Q|$}]  at (5,0.3) {};

                \begin{scope}[yshift=1.2cm]
                 \draw[fill=lipicsYellow!80,thick] (4,0) rectangle (6,0.4);

                    \draw[thick] (1,0) -- (9,0);
                    \draw[thick] (1,0.4) -- (9,0.4);
                    \draw[dashed,thick] (0.5,0) -- (1,0);
                    \draw[dashed,thick] (0.5,0.4) -- (1,0.4);
                    \draw[dashed,thick] (9,0) -- (9.5,0);
                    \draw[dashed,thick] (9,0.4) -- (9.5,0.4);
                    \node[label = {left: $\cdots$}]  at (0.75,0.2) {};
                    \node[label = {right: $\cdots$}]  at (9.25,0.2) {};

                 \node[label = {left: $Q^\infty \fragmentco{\ell}{r}$}]  at (0,0.2) {};
                 \node[crossmark] at (2,0.2){$\mathbf{*}$};
                 \node[crossmark] at (3,0.2){$\mathbf{*}$};
                 \node[crossmark] at (7,0.2){$\mathbf{*}$};
                 \node[crossmark] at (8,0.2){$\mathbf{*}$};
                 \node[label = {above: \small{$j|Q|$}}]  at (4.5,0.2) {};
                \end{scope}

                \begin{scope}[yshift=-1.2cm]
                 \draw[fill=lipicsYellow!80,thick] (4,0) rectangle (6,0.4);

                    \draw[thick] (1,0) -- (9,0);
                    \draw[thick] (1,0.4) -- (9,0.4);
                    \draw[dashed,thick] (0.5,0) -- (1,0);
                    \draw[dashed,thick] (0.5,0.4) -- (1,0.4);
                    \draw[dashed,thick] (9,0) -- (9.5,0);
                    \draw[dashed,thick] (9,0.4) -- (9.5,0.4);
                    \node[label = {left: $\cdots$}]  at (0.75,0.2) {};
                    \node[label = {right: $\cdots$}]  at (9.25,0.2) {};

                 \node[label = {left: $Q^\infty \fragmentco{\ell'}{r'}$}]  at (0,0.2) {};
                 \node[crossmark] at (1.5,0.2){$\mathbf{*}$};
                 \node[crossmark] at (3.5,0.2){$\mathbf{*}$};
                 \node[crossmark] at (6.5,0.2){$\mathbf{*}$};
                 \node[crossmark] at (8.5,0.2){$\mathbf{*}$};
                 \node[label = {below: \small{$j'|Q|$}}]  at (4.5,0.2) {};
                \end{scope}

                \draw[dashed] (4.5,-1.2) -- (4.5,1.6);
                \end{tikzpicture}
    \caption{The setting in~\cref{fct:split}. Asterisks denote edit operations on the
        respective strings in their optimal alignment with $S$. $S_i$ is denoted by a shaded
        rectangle, $U=Q\fragmentco{|Q|-p}{|Q|}$ and
    $V=Q\fragmentco{0}{|Q|-p}$.}\label{fig:fct_split}
\end{figure}

\begin{lemma}\label{lem:synchr}
    Let $T$ denote a text of~length $n$, let $k$ denote a positive integer,
    and let $Q$ denote a primitive string.
    Suppose that $\ed(T\fragmentco{0}{q},Q^\infty\fragmentco{x}{y})\le k$ and
    $\ed(T\fragmentco{p}{n},Q^\infty\fragmentco{x'}{y'})\le k$ holds for some integers $p\le q$, $x\le y$,
    and $x'\le y'$.
    If $|Q|=1$ or $q-p \ge (2k+1)|Q|$,
    then $\ed(T,Q^\infty\fragmentco{x''}{y})=\ed(T,Q^\infty\fragmentco{x}{y''})\le 2k$ for some $x''\equiv x'\pmod{|Q|}$ and $y''\equiv y' \pmod {|Q|}$, and $(p+x-x'+2k) \bmod |Q| \le 4k$.
\end{lemma}
\begin{proof}
    Observe that, for some integer $z\in \fragment{x}{y}$, we have
    \[\ed(T\fragmentco{0}{q},Q^\infty\fragmentco{x}{y}) = \ed(T\fragmentco{0}{p},Q^\infty\fragmentco{x}{z}) +
        \ed(T\fragmentco{p}{q},Q^\infty\fragmentco{z}{y}).\]
    Similarly,  for some integer $z'\in \fragment{x'}{y'}$, we have
    \[\ed(T\fragmentco{p}{n},Q^\infty\fragmentco{x'}{y'}) = \ed(T\fragmentco{p}{q},Q^\infty\fragmentco{x'}{z'}) +
        \ed(\fragmentco{q}{n},Q^\infty\fragmentco{z'}{y'}).\]
    Now, \cref{fct:split} applied to $S:=T\fragmentco{p}{q}$ yields an integer $r\in \fragment{p}{q}$
    and integers $j,j'$ such that (see also \cref{fig:lem_synchr})
    \begin{align*}
        \ed(T\fragmentco{p}{q},Q^\infty\fragmentco{z}{y})
            &=\ed(T\fragmentco{p}{r},Q^\infty\fragmentco{z}{j|Q|})
            + \ed(T\fragmentco{r}{q},Q^\infty\fragmentco{j|Q|}{y}), \;\text{and}\\
        \ed(T\fragmentco{p}{q},Q^\infty\fragmentco{x'}{z'})
            &=\ed(T\fragmentco{p}{r},Q^\infty\fragmentco{x'}{j'|Q|})
            + \ed(T\fragmentco{r}{q},Q^\infty\fragmentco{j'|Q|}{z'}).
    \end{align*}
    This implies that
    \begin{align*}
        \ed(T\fragmentco{0}{r},Q^\infty\fragmentco{x}{j|Q|})
            &=\ed(T\fragmentco{0}{p},Q^\infty\fragmentco{x}{z})
            + \ed(T\fragmentco{p}{r},Q^\infty\fragmentco{z}{j|Q|})\le k,
            \;\text{and}\\
        \ed(T\fragmentco{r}{n},Q^\infty\fragmentco{j'|Q|}{y'})
            &= \ed(T\fragmentco{r}{q},Q^\infty\fragmentco{j'|Q|}{z'})+\ed(\fragmentco{q}{n},Q^\infty\fragmentco{z'}{y'})\le k.
    \end{align*}
    Combining the equations yields
    \begin{multline*}\ed(T,Q^\infty\fragmentco{x+(j'-j)|Q|}{y'})=\ed(T,Q^\infty\fragmentco{x}{y'+(j-j')|Q|})\\
    \le \ed(T\fragmentco{0}{r},Q^\infty\fragmentco{x}{j|Q|}) +\ed(T\fragmentco{r}{n},Q^\infty\fragmentco{j'|Q|}{y'})\le 2k.\end{multline*}
    Moreover, we deduce $|j|Q|-x-r|\le k$ and $|j'|Q|-x' - r + p|\le k$,
    which yields $|p+x-x'-(j-j')|Q||\le 2k$, and therefore $(p+x-x'+2k)\bmod |Q|\le 4k$.
\end{proof}

\begin{figure}[t]
    \centering
                    \begin{tikzpicture}

                \filldraw[lipicsYellow!80] (0,1.2) rectangle (5.5,1.6);
                \filldraw[lipicsYellow!80] (5.5,-1.2) rectangle (10,-0.8);

                \draw (0,0) rectangle (10,0.4);
                \node[label = {left: $T$}]  at (0,0.2) {};
                \node[label = {below: \small{$p$}}]  at (3,0.15) {};
                \node[label = {below: \small{$r$}}]  at (8,0.15) {};
                \node[label = {below: \small{$q$}}]  at (5.5,0.15) {};

			\begin{scope}[yshift=1.2cm]
				\draw (0,0) rectangle (8,0.4);
				\draw (3,0) -- (3,0.4);
				\node[label = {right: $Q^\infty \fragmentco{x}{y}$}]  at (8,0.2) {};
				\node[label = {below: \small{$z$}}]  at (3,0.15) {};
				\node[label = {below: \small{$y$}}]  at (8,0.15) {};
				\node[label = {below: \small{$x$}}]  at (0,0.15) {};
				\node[label = {below: \small{$j|Q|$}}]  at (5.5,0.2) {};
				 \node[crossmark] at (1,0.2){$\mathbf{*}$};
                 \node[crossmark] at (4,0.2){$\mathbf{*}$};
                 \node[crossmark] at (6,0.2){$\mathbf{*}$};
                 \node[crossmark] at (7,0.2){$\mathbf{*}$};
            \end{scope}

			\begin{scope}[yshift=-1.2cm]
				\draw (3,0) rectangle (10,0.4);
				\draw (8,0) -- (8,0.4);
				\node[label = {left: $Q^\infty \fragmentco{x'}{y'}$}]  at (3,0.2) {};
				\node[label = {below: \small{$x'$}}]  at (3,0.2) {};
				\node[label = {below: \small{$z'$}}]  at (8,0.2) {};
				\node[label = {below: \small{$y'$}}]  at (10,0.2) {};
				\node[label = {below: \small{$j'|Q|$}}]  at (5.5,0.2) {};
                 \node[crossmark] at (3.5,0.2){$\mathbf{*}$};
                 \node[crossmark] at (6.5,0.2){$\mathbf{*}$};
                 \node[crossmark] at (9,0.2){$\mathbf{*}$};
            \end{scope}

			\draw[dashed] (3,-0.4) -- (3,-1.2);
			\draw[dashed] (3,0) -- (3,0.8);

			\draw[dashed] (5.5,-0.4) -- (5.5,-1.2);
			\draw[dashed] (5.5,0) -- (5.5,0.8);
			\draw[dashed] (5.5,1.2) -- (5.5,1.6);

			\draw[dashed] (8,-0.4) -- (8,-1.2);
			\draw[dashed] (8,0) -- (8,0.8);
                \end{tikzpicture}
    \caption{The setting in~\cref{lem:synchr}. $T$ is at edit distance at most $2k$ from
    $Q^\infty\fragmentco{x}{j|Q|}Q^\infty\fragmentco{j'|Q|}{y'} = Q^\infty\fragmentco{x+(j'-j)|Q|}{y'}
    =Q^\infty\fragmentco{x}{y'+(j-j')|Q|}$.}\label{fig:lem_synchr}
\end{figure}
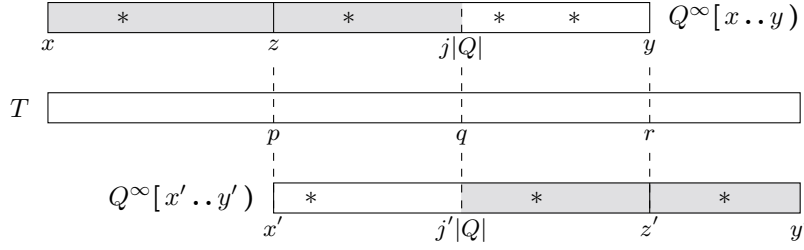

\begin{definition}
    Let $S$ denote a string and let $Q$ denote a primitive string.
    We say that a fragment $L$ of~$S$ is \emph{locked} (with respect to $Q$)
    if at least one of~the following holds:
    \begin{itemize}
        \item For some integer $\alpha$, we have $\edl{L}{Q}=\ed(L,Q^\alpha)$.
        \item The fragment $L$ is a suffix of~$S$ and $\edl{L}{Q}=\ed(L,Q^*)$.
        \item The fragment $L$ is a prefix of~$S$ and $\edl{L}{Q}=\eds{L}{Q}$.
        \item We have $L = S$.\lipicsEnd
    \end{itemize}
\end{definition}

The notion of~locked fragments was also used in~\cite{ColeH98}.
In order to develop some intuition, let us consider the following example: A string
$U=Q^{k+1} S Q^{k+1}$ such that $\edl{U}{Q} \leq k$ and $Q$ is primitive.
Then, in any optimal alignment of~$U$ with a substring of~$Q^\infty$ at least one of~the
leading (or trailing) $k+1$ occurrences of~$Q$ in $U$ is matched exactly and hence also
all occurrences preceding it (or succeeding it).
Thus $U$ is locked with respect to $Q$.

Next, we show that we can identify short locked fragments
covering all errors with respect to $\!{}^*Q^*$. Intuitively, our strategy is to start
with at most $k$ $|Q|$-length fragments of~$S$ that contain all the errors and extend
or/and merge them
(in a sense similar to that of~the intuitive example provided above),
so that the resulting fragments contain sufficiently many copies of~$Q$
\begin{lemma}\label{lem:locked}
    Let $S$ denote a string and let $Q$ denote a primitive string.
    There are disjoint locked fragments $L_1,\ldots,L_{\ell} \preceq S$
    with $\edl{L_i}{Q} > 0$ such that \[
        \edl{S}{Q}=\sum_{i=1}^{\ell} \edl{L_i}{Q}\quad
        \text{and}\quad\sum_{i=1}^{\ell}|L_i| \le (5|Q|+1)\edl{S}{Q}.\]
\end{lemma}
\begin{proof}
    Let us choose integers $x\le y$ so that $\edl{S}{Q}=\ed(S,Q^\infty\fragmentco{x}{y})$.
    Without loss of~generality, we may assume that $x\in \fragmentco{0}{|Q|}$.
    If $y \le |Q|$, then $|S|\le |Q|+\edl{S}{Q}$; thus setting the whole string~$S$
    as the only locked fragment satisfies the claimed conditions.
    Hence, we may assume that $y > |Q|$.

    An arbitrary optimum alignment of~$S$ and $Q^\infty\fragmentco{x}{y}$ yields
    a partition  $S=S_0^{(0)}\cdots S_{\!s^{(0)}}^{(0)}$
    with $s^{(0)}=\floor{(y-1)/|Q|}$
    such that  $\ed(S,Q^\infty\fragmentco{x}{y}) = \sum_{i=0}^{s^{(0)}} \Delta^{(0)}_i$,
    where
    \[\Delta^{(0)}_i = \begin{cases}
        \ed(S^{(0)}_0,Q\fragmentco{x}{|Q|}) &\text{if }i=0,\\
        \ed(S^{(0)}_i,Q) &\text{if }0 < i < s^{(0)},\\
        \ed(S^{(0)}_{s^{(0)}}, Q\fragmentco{0}{y-s^{(0)}|Q|}) & \text{if }i = s^{(0)}.
    \end{cases}\]
    See \cref{fig:lem_locked} for an illustration.

    We start with this partition and then coarsen it by exhaustively applying
    the merging rules specified below, where each rule is applied only if the previous rules cannot
    be applied.
    In each case, we re-index the unchanged
    fragments $S^{(t)}_i$ to obtain a new partition $S = S^{(t+1)}_0\cdots
    S^{(t+1)}_{\!s^{(t+1)}}$ and re-index the corresponding values $\Delta^{(t)}_i$ accordingly.
    We say that a fragment $S^{(t)}_i$ is \emph{interesting}
    if $i=0$, $i=s^{(t)}$, $S^{(t)}_i\ne Q$, or $\Delta_i^{(t)}>0$.
    \begin{enumerate}
        \item\label{it:type1} If subsequent fragments $S^{(t)}_i$ and $S^{(t)}_{i+1}$ are
            both interesting, then merge $S^{(t)}_i$ and $S^{(t)}_{i+1}$, obtaining
            $S^{(t+1)}_i := S^{(t)}_i S^{(t)}_{i+1}$ and $\Delta^{(t+1)}_i :=
            \Delta^{(t)}_i + \Delta^{(t)}_{i+1}$.
        \item\label{it:type2} If $0 < i < s^{(t)}$ and $\Delta^{(t)}_i>0$,
            then merge the subsequent fragments $S^{(t)}_{i-1}=Q$, $S^{(t)}_i$, and
            $S^{(t)}_{i+1}=Q$, obtaining $S^{(t+1)}_{i-1} := S^{(t)}_{i-1} S^{(t)}_i
            S^{(t)}_{i+1}$, and set $\Delta^{(t+1)}_{i-1}:=\Delta^{(t)}_{i}-1$.
        \item\label{it:type3} If $0 < i = s^{(t)}$ and $\Delta^{(t)}_{i}>0$,
            then merge the subsequent fragments $S^{(t)}_{i-1}=Q$ and
            $S^{(t)}_{i}$, obtaining $S^{(t+1)}_{i-1} := S^{(t)}_{i-1}S^{(t)}_{i}$, and
            set $\Delta^{(t+1)}_{i-1} := \Delta^{(t)}_{i}-1$.
        \item\label{it:type4} If $0 = i < s^{(t)}$ and $\Delta^{(t)}_i>0$,
            then merge the subsequent fragments $S^{(t)}_i$ and $S^{(t)}_{i+1}=Q$,
            obtaining $S^{(t+1)}_i := S^{(t)}_i S^{(t)}_{i+1}$,
            and set $\Delta^{(t+1)}_{i}:=\Delta^{(t)}_{i}-1$.
    \end{enumerate}

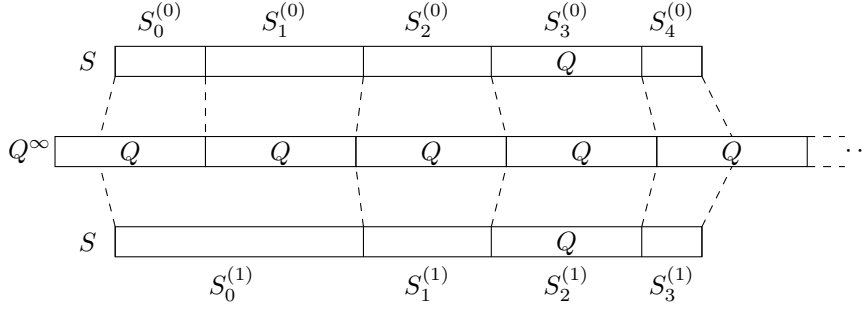
\begin{figure}[t]
    \centering
                   \begin{tikzpicture}

                \begin{scope}[yshift=-1.2cm]
                	\node[label = {left: $Q^\infty$}]  at (0.2,0.2) {};
  					\foreach \x in {0,2,4,6,8}{
                 		\draw[xshift=\x cm] (0,0) rectangle (2,0.4);
                 		\node[xshift=\x cm,label = {above: $Q$}]  at (1,-0.225) {};
                 	}
                    \draw[dashed] (10,0) -- (10.5,0);
                    \draw[dashed] (10,0.4) -- (10.5,0.4);
                    \node[label = {right: $\cdots$}]  at (10.25,0.2) {};
                \end{scope}

					\node[label = {left: $S$}]  at (0.8,0.2) {};
					\draw (0.8,0) rectangle (8.6,0.4);
                 	\node[label = {above: $Q$}]  at (6.8,-0.225) {};
					\foreach \i/\x/\z/\y in {0.6/0.8/1.4/0, 2/2/3.05/1, 4/4.1/4.95/2, 6/5.8/6.8/3, 8/7.8/8.2/4, 9/8.6/8.2/5}{
                 		\draw (\x,0) -- (\x,0.4);
						\draw[dashed] (\x,0) -- (\i,-0.8);
						\ifthenelse{\y=0 \OR \y=1 \OR \y=2 \OR \y=3 \OR \y=4}{\draw (\z,0.4) node[above] {$S^{(0)}_\y$};}{}
					}
				\begin{scope}[yshift=-2.4cm]
					\node[label = {left: $S$}]  at (0.8,0.2) {};
					\draw (0.8,0) rectangle (8.6,0.4);
                 	\node[label = {above: $Q$}]  at (6.8,-0.225) {};
					\foreach \i/\x/\z/\y in {0.6/0.8/2.35/0, 4/4.1/4.95/1, 6/5.8/6.8/2, 8/7.8/8.2/3, 9/8.6/8.2/4}{
                 		\draw (\x,0) -- (\x,0.4);
						\draw[dashed] (\x,0.4) -- (\i,1.2);
						\ifthenelse{\y=0 \OR \y=1 \OR \y=2 \OR \y=3}{\draw (\z,0) node[below] {$S^{(1)}_\y$};}{}
					}
                \end{scope}

                \end{tikzpicture}
    \caption{A partition $S_0^{(0)}\cdots S_4^{(0)}$ of~a string $S$ is shown
        ($s^{(0)}=4$), in which all
        fragments apart from $S_0^{(3)}$ are interesting.
        A merge of~the fragments $S_0^{(0)}$ and $S_1^{(0)}$ yields the shown partition
    $S_0^{(1)}\cdots S_3^{(1)}$ of~$S$.\vspace{-2ex}}\label{fig:lem_locked}
\end{figure}

    Let $S=S_0\cdots S_{s}$ denote the obtained final partition.
    We select as locked fragments all the fragments~$S_i$ with $\edl{S_i}{Q}> 0$.
    Below,  we show that this selection satisfies the desired properties.
    We start by proving that we indeed picked locked fragments.
    \begin{claim}\label{clm:locked}
        Each fragment $S_i^{(t)}$ of~each partition $S=S_0^{(t)}\cdots S_{s^{(t)}}^{(t)}$
        satisfies at least one of~the following:
        \begin{itemize}
            \item $\ed(S^{(t)}_i,Q^\alpha)\le \edl{S^{(t)}_i}{Q}+\Delta^{(t)}_i$ for some
                integer $\alpha$;
            \item $i=s^{(t)}$ and $\ed(S^{(t)}_i,Q^*)\le \edl{S^{(t)}_i}{Q}+\Delta^{(t)}_i$;
            \item $i=0$ and $\eds{S^{(t)}_i}{Q}\le \edl{S^{(t)}_i}{Q}+\Delta^{(t)}_i$;
            \item $i=0=s^{(t)}$.
        \end{itemize}
    \end{claim}
    \begin{claimproof}
        We proceed by induction on $t$. The base case follows from the definition of~the
        values $\Delta^{(t)}_i$.

        As for the inductive step, we assume that the claim holds for all fragments $S^{(t)}_i$
        and we prove that it holds for all fragments $S^{(t+1)}_i$. We consider several cases based on the merge rule applied.
        \begin{enumerate}
            \item For a type-\ref{it:type1} merge of~interesting fragments $S^{(t)}_i$ and $S^{(t)}_{i+1}$ into $S^{(t+1)}_i$, it suffices to prove that $S^{(t+1)}_i$ satisfies the claim.
            \begin{itemize}
                \item If $0 < i < s^{(t+1)}$, then $\ed(S^{(t)}_i,Q^\alpha)\le \edl{S^{(t)}_i}{Q}+\Delta^{(t)}_i$ and $\ed(S^{(t)}_{i+1},Q^{\alpha'})\le \edl{S^{(t)}_{i+1}}{Q}+\Delta^{(t)}_{i+1}$ hold by the inductive assumption for some integers $\alpha,\alpha'$. Consequently,
                \begin{multline*}\ed(S^{(t+1)}_i,Q^{\alpha+\alpha'})=\ed(S^{(t)}_i S^{(t)}_{i+1} ,Q^{\alpha} Q^{\alpha'})\le \ed(S^{(t)}_i,Q^\alpha)+\ed(S^{(t)}_{i+1},Q^{\alpha'})\\
                \le  \edl{S^{(t)}_i}{Q}+\Delta^{(t)}_i+\edl{S^{(t)}_{i+1}}{Q}+\Delta^{(t)}_{i+1}\le  \edl{S^{(t+1)}_i}{Q}+\Delta^{(t+1)}_i.
                \end{multline*}
                \item If $0 < i = s^{(t+1)}$, then $\ed(S^{(t)}_i,Q^\alpha)\le \edl{S^{(t)}_i}{Q}+\Delta^{(t)}_i$ and $\ed(S^{(t)}_{i+1},Q^*)\le \edl{S^{(t)}_{i+1}}{Q}+\Delta^{(t)}_{i+1}$ hold by the inductive assumption for some integer $\alpha$. Consequently,
                \begin{multline*}\ed(S^{(t+1)}_i,Q^*)=\ed(S^{(t)}_i S^{(t)}_{i+1},Q^*)\le \ed(S^{(t)}_i,Q^\alpha)+\ed(S^{(t)}_{i+1},Q^*)\\
                \le  \edl{S^{(t)}_i}{Q}+\Delta^{(t)}_i+\edl{S^{(t)}_{i+1}}{Q}+\Delta^{(t)}_{i+1}\le  \edl{S^{(t+1)}_i}{Q}+\Delta^{(t+1)}_i.
                \end{multline*}
                \item The analysis of~the case that $0 = i < s^{(t+1)}$ is symmetric to that of~the above case -- this can be seen by reversing all strings in scope.
                \item If $0 = i = s^{(t+1)}$, then the claim holds trivially.
            \end{itemize}
            \item For a type-\ref{it:type2} merge of~$S^{(t)}_{i-1}$, $S^{(t)}_i$, and $S^{(t)}_{i+1}$ into $S^{(t+1)}_{i-1}$, it suffices to prove that $S^{(t+1)}_{i-1}$ satisfies the claim.
            \begin{itemize} \item If $\edl{S^{(t+1)}_{i-1}}{Q}>\edl{S^{(t)}_i}{Q}$, we observe that  $\ed(S^{(t)}_i,Q^\alpha)\le \edl{S^{(t)}_i}{Q}+\Delta^{(t)}_i$ holds by the inductive assumption for some integer $\alpha$. Consequently,
            \begin{multline*}
                \ed(S^{(t+1)}_{i-1},Q^{\alpha+2})=\ed(QS^{(t)}_iQ,QQ^{\alpha}Q)\le \ed(S^{(t)}_i,Q^\alpha)\\
                \le \edl{S^{(t)}_i}{Q}+\Delta^{(t)}_i \le \edl{S^{(t+1)}_{i-1}}{Q} - 1 + \Delta^{(t)}_i=\edl{S^{(t+1)}_{i-1}}{Q} + \Delta^{(t+1)}_{i-1}.
            \end{multline*}
                \item If $\edl{S^{(t+1)}_{i-1}}{Q}=\edl{S^{(t)}_i}{Q}$, then let $x' \le
                    y'$ denote integers that satisfy
                $\edl{S^{(t+1)}_{i-1}}{Q}=\ed(S^{(t+1)}_{i-1},Q^\infty\fragmentco{x'}{y'})$.
                This also yields integers $x'',y''$ with
            $x'\le x'' \le y'' \le y'$ such that\[
                \ed(S^{(t+1)}_{i-1},Q^\infty\fragmentco{x'}{y'}) =
                    \ed(Q,Q^\infty\fragmentco{x'}{x''})
                    +\ed(S^{(t)}_i,Q^\infty\fragmentco{x''}{y''})
                    +\ed(Q,Q^\infty\fragmentco{y''}{y'}).
            \]
            Due to\[
                \edl{S^{(t)}_i}{Q}
                \le \ed(S^{(t)}_i,Q^\infty\fragmentco{x''}{y''})
                \le \ed(S^{(t+1)}_{i-1},Q^\infty\fragmentco{x'}{y'})
                = \edl{S^{(t+1)}_{i-1}}{Q}=\edl{S^{(t)}_i}{Q},
            \]
            we have
            $\ed(Q,Q^\infty\fragmentco{x'}{x''})=0=\ed(Q,Q^\infty\fragmentco{y''}{y'})$.
            As the string $Q$ is primitive, this means that $x',x'',y'',y'$ are all multiples
            of~$|Q|$.
            Consequently,
            \[\ed(S^{(t+1)}_{i-1},Q^{(y'-x')/|Q|}) =\ed(S^{(t+1)}_{i-1},Q^\infty\fragmentco{x'}{y'}) = \edl{S^{(t+1)}_{i-1}}{Q} \le \edl{S^{(t+1)}_{i-1}}{Q} + \Delta^{(t-1)}_{i-1}.\]
            \end{itemize}
            \item For a type-\ref{it:type3} merge of~$S^{(t)}_{i-1}$ and $S^{(t)}_i$ into $S^{(t+1)}_{i-1}$, it suffices to prove that $S^{(t+1)}_{i-1}$ satisfies the claim.
            \begin{itemize} \item If $\edl{S^{(t+1)}_{i-1}}{Q}>\edl{S^{(t)}_i}{Q}$, we observe that  $\ed(S^{(t)}_{i},Q^*)\le \edl{S^{(t)}_{i}}{Q}+\Delta^{(t)}_{i}$ holds by the inductive assumption. Consequently,
                \begin{multline*}
                    \ed(S^{(t+1)}_{i-1},Q^*)=\ed(QS^{(t)}_i,Q^*)\le \ed(S^{(t)}_i,Q^*)\\
                    \le \edl{S^{(t)}_{i}}{Q}+\Delta^{(t)}_{i} \le \edl{S^{(t+1)}_{i-1}}{Q} - 1 + \Delta^{(t)}_i=\edl{S^{(t+1)}_{i-1}}{Q} + \Delta^{(t+1)}_{i-1}.
                \end{multline*}
                \item If $\edl{S^{(t+1)}_{i-1}}{Q}=\edl{S^{(t)}_i}{Q}$, then let $x' \le
                    y'$ denote integers that satisfy
                $\edl{S^{(t+1)}_{i-1}}{Q}=\ed(S^{(t+1)}_{i-1},Q^\infty\fragmentco{x'}{y'})$.
            This also yields an integer $x''$ with
            $x'\le x'' \le y'$ such that\[
                \ed(S^{(t+1)}_{i-1},Q^\infty\fragmentco{x'}{y'}) =
                    \ed(Q,Q^\infty\fragmentco{x'}{x''})
                    +\ed(S^{(t)}_i,Q^\infty\fragmentco{x''}{y'}).
            \]
            Due to\[
                \edl{S^{(t)}_i}{Q}
                \le \ed(S^{(t)}_i,Q^\infty\fragmentco{x''}{y'})
                \le \ed(S^{(t+1)}_{i-1},Q^\infty\fragmentco{x'}{y'})
                = \edl{S^{(t+1)}_{i-1}}{Q}=\edl{S^{(t)}_i}{Q},
            \]
            we have
            $\ed(Q,Q^\infty\fragmentco{x'}{x''})=0$.
            As the string $Q$ is primitive, this means that $x',x''$ are both multiples
            of~$|Q|$.
            Consequently,
            \[\ed(S^{(t+1)}_{i-1},Q^*) =\ed(S^{(t+1)}_{i-1},Q^\infty\fragmentco{x'}{y'})= \edl{S^{(t+1)}_{i-1}}{Q} \le \edl{S^{(t+1)}_{i-1}}{Q} + \Delta^{(t-1)}_{i-1}.\]
            \end{itemize}
        \item The analysis of~type~\ref{it:type4} merges is symmetrical to that of~~\ref{it:type3} merges -- this can be seen by reversing all strings in scope.
        \end{enumerate}

        This completes the proof~of~the inductive step.
    \end{claimproof}

    Observe that if no merge rule can be applied to a partition $S=S_{0}^{(t)}\cdots S_{s^{(t)}}^{(t)}$,
    then $s^{(t)}=0$ or $\Delta_{0}^{(t)}= \cdots = \Delta_{s^{(t)}}^{(t)}=0$.
    Consequently, \cref{clm:locked} implies that all fragments $S_i$ in the final partition $S=S_0\cdots S_s$
    are locked.

    \begin{claim}\label{clm:locked_short}
        For each partition $S=S_0^{(t)}\cdots S_{\!s^{(t)}}^{(t)}$,
        the total length $\lambda^{(t)}$ of~interesting fragments satisfies\[
            \lambda^{(t)} + 2|Q|\sum_{i=0}^{s^{(t)}} \Delta_i^{(t)} \le
            (5|Q|+1)\edl{S}{Q}.
        \]
    \end{claim}
    \begin{claimproof}
        We proceed by induction on $t$.
        In the base case of~$t=0$, each interesting fragment other than $S_{0}^{(0)}$ and $S_{s^{(0)}}^{(0)}$
        satisfies $\Delta_i^{(0)}>0$. Hence, the number of~interesting fragments is at most $2+\sum_{i=0}^{s^{(0)}} \Delta_i^{(0)}= 2+\edl{S}{Q}$.
        Moreover, the length of~each fragment $S^{(0)}_i$ does not exceed $|Q|+\Delta_i^{(0)}$.
        Consequently,
        \[
            \lambda^{(0)} + 2|Q|\sum_{i=0}^{s^{(0)}} \Delta_i^{(0)} \le
            (2+\edl{S}{Q})|Q|+(2|Q|+1)\sum_{i=0}^{s^{(0)}}\Delta_i^{(0)}\le(5|Q|+1)\,\edl{S}{Q}.
        \]This completes the proof~in the base case.

        As for the inductive step, it suffices to prove that $\lambda^{(t+1)}+2|Q|\sum_{i=0}^{s^{(t+1)}} \Delta_i^{(t+1)} \le \lambda^{(t)}+2|Q|\sum_{i=0}^{s^{(t)}} \Delta_i^{(t)}$:
        \begin{itemize}
            \item For a type-\ref{it:type1} merge
                (where we merge two interesting fragments),
                we have
                \[\lambda^{(t+1)} +2|Q|\sum_{i=0}^{s^{(t+1)}} \Delta_i^{(t+1)}
                =  \lambda^{(t)} + 2|Q|\sum_{i=0}^{s^{(t)}} \Delta_i^{(t)}.\]
            \item For a type-\ref{it:type2}, type-\ref{it:type3}, or type-\ref{it:type4} merge
                (where we merge a fragment with its one or two non-interesting neighbors), we have
                \[\lambda^{(t+1)} +2|Q|\sum_{i=0}^{s^{(t+1)}} \Delta_i^{(t+1)}
                \le \lambda^{(t)} + 2|Q|+2|Q|\sum_{i=0}^{s^{(t+1)}} \Delta_i^{(t+1)}
                =\lambda^{(t)} + 2|Q|\sum_{i=0}^{s^{(t)}} \Delta_i^{(t)}.\]
        \end{itemize}
        Overall, we obtain the claimed bound.
    \end{claimproof}
    We conclude that the total length of~interesting fragments $S_i$ does not exceed $(5|Q|+1)\edl{S}{Q}$.

    \begin{claim}\label{clm:locked_whole}
        We have $\edl{S}{Q}=\sum_{i=0}^{s} \edl{S_i}{Q}$.
    \end{claim}
    \begin{claimproof}
        The claim is immediate if $s = 0$; hence, assume that $s \ge 1$.
        Observe that the inequality $\sum_{i=0}^{s} \edl{S_i}{Q}\le \edl{S}{Q}$
        easily follows from disjointness of~fragments $S_i$;
        thus, we focus on proving $\edl{S}{Q}\le \sum_{i=0}^{s} \edl{S_i}{Q}$.

        For $0\le i \le s$, let $Q_i$ denote a substring of~$Q^\infty$ that satisfies
        $\edl{S_i}{Q}=\ed(S_i,Q_i)$.
        Since each $S_i$ is locked (by \cref{clm:locked}),
        we may assume that for $0 < i < s$ the substring $Q_i$ is a power of~$Q$,
        the substring $Q_s$ is a prefix of~a power of~$Q$,
        and the substring $Q_0$ is a suffix of~a power of~$Q$.
        Consequently, $Q_0\cdots Q_s$ is a substring of~$Q^\infty$,
        and we have\vspace{-1.5ex}\[
            \edl{S}{Q} \le \ed(S_0\cdots S_s,Q_0\cdots Q_s)
                       \le \sum_{i=0}^s \ed(S_i,Q_i) = \sum_{i=0}^s \edl{S_i}{Q},
                   \]\vspace{-1.5ex} thus completing the proof.
    \end{claimproof}

    The locked fragments created satisfy $\edl{S}{Q}=\sum_{i=1}^\ell \edl{L_i}{Q}$
    due to \cref{clm:locked_whole}.
    Moreover, since $\edl{S_i}{Q}>0$ holds only for interesting fragments, \cref{clm:locked_short}
    yields $\sum_{i=1}^{\ell} |L_i| \le (5|Q|+1)\,\edl{S}{Q}$, completing the proof.
\end{proof}

The definition and lemma that follow, as well as~\cref{lem:Eaux}~\eqref{it:Eprog},
are not needed for our proof~of~the main result of~this section, \cref{thm:edmain} --
 a reader interested only in that result can safely skip them.
They, however, provide additional structural insights that we exploit in~\cref{sec:pme}.

\begin{definition}\label{def:klocked}
    Let $S$ denote a string, let $Q$ denote a primitive string,
    and let $k\ge 0$ denote an integer.
    We say that a prefix $L$ of~$S$ is \emph{$k$-locked} (with respect to $Q$)
    if at least one of~the following holds:
    \begin{itemize}
        \item For every $p\in \fragmentco{0}{|Q|}$, if $\ed(L,\rot^p(Q)^*)\le k$, then
            $\ed(L,\rot^p(Q)^*)=\ed(L,Q^\infty\fragmentco{|Q|-p}{j|Q|})$ holds for some
            integer $j$.
        \item We have $L=S$.\lipicsEnd
    \end{itemize}
\end{definition}

\begin{lemma}\label{lem:klocked}
    Let $S$ denote a string, let $Q$ denote a primitive string, and let $k\ge 0$ be an integer.
    There are disjoint locked fragments $L_1,\ldots,L_{\ell} \preceq S$,
    such that $L_1$ is a $k$-locked prefix of~$S$, $L_{\ell}$ is a suffix of~$S$,
    $\edl{L_i}{Q} > 0$ for $1 < i < \ell$,\vspace{-1.5ex} \[
        \edl{S}{Q}=\sum_{i=1}^{\ell} \edl{L_i}{Q},\quad
        \text{and}\quad\sum_{i=1}^{\ell}|L_i| \le (5|Q|+1)\edl{S}{Q}+2(k+1)|Q|.\]
\end{lemma}
\begin{proof}
    We proceed as in the proof~of~\cref{lem:locked} except that $\Delta_{0}^{(0)}$
    is artificially increased by $k+1$, the prefix $S_0$ in the final partition
    is included as $L_1$ among the locked fragments even if $\edl{S_0}{Q}=0$,
    and the suffix $S_s$ is included as $L_{\ell}$ among the locked fragments even if $\edl{S_s}{Q}=0$.

    It is easy to see that \cref{clm:locked,clm:locked_whole} remain satisfied,
    whereas the upper bound in \cref{clm:locked_short} is increased by $2(k+1)|Q|$.
    We only need to prove that $S_0$ is a $k$-locked prefix of~$S$.
    For this, we prove the following claim using induction.

    \begin{claim}\label{clm:klocked}
        For each partition $S=S_0^{(t)}\cdots S_{s^{(t)}}^{(t)}$, at least one of~the following holds:
        \begin{itemize}
                \item For every $p\in \fragmentco{0}{|Q|}$, if $\ed(S_0^{(t)},\rot^p(Q)^*)\le k-\Delta_0^{(t)}$, then $\ed(S_0^{(t)},Q^\infty\fragmentco{|Q|-p}{j|Q|}) \le \ed(S_0^{(t)},\rot^p(Q)^*)+\Delta_0^{(t)}$ holds for some integer $j$.
                \item We have $S_0^{(t)}=S$.
        \end{itemize}
    \end{claim}
    \begin{claimproof}
        We proceed by induction on $t$. In the base case of~$t=0$, the claim holds trivially since
        $\ed(S_0^{(0)},\rot^p(Q)^*) \ge 0 > k-\Delta_0^{(0)}$ holds for every $p$ due to $\Delta_0^{(0)}\ge k+1$.

        As for the induction step, we assume that the claim holds for $S_0^{(t)}$ and we prove that it holds
        for $S_0^{(t+1)}$. The claim holds trivially if the merge rule applied did not affect $S_0^{(t)}$.
        Given that  $S_0^{(t)}$ is interesting by definition, the merges that might affect $S_0^{(t)}$ are of~type~\ref{it:type1} (if $S_1^{(t)}$ is interesting) or~\ref{it:type4} (otherwise).

        \begin{enumerate}
            \item Consider a type-\ref{it:type1} merge of~$S_0^{(t)}$ and $S_1^{(t)}$. If $s^{(t)}=1$, then $S_0^{(t+1)}=S$ satisfies the claim trivially.
            Hence, we may assume that $1 < s^{(t)}$ so that \cref{clm:locked} yields
            $\ed(S_1^{(t)},Q^{\alpha})\le \edl{S_1^{(t)}}{Q} + \Delta_{1}^{(t)}$ for some integer $\alpha$.
            Let us fix $p\in \fragmentco{0}{|Q|}$ with $\ed(S_0^{(t+1)},\rot^p(Q)^*)\le k-\Delta_0^{(t+1)}$.
            Due to $\Delta_0^{(t+1)}\ge \Delta_0^{(t)}$, this yields
            $\ed(S_0^{(t)},\rot^p(Q)^*)\le k-\Delta_0^{(t)}$, so the inductive assumption
            implies $\ed(S_0^{(t)},Q^\infty\fragmentco{|Q|-p}{j|Q|}) \le
            \ed(S_0^{(t)},\rot^p(Q)^*)+\Delta_0^{(t)}$ for some integer $j$.
            Consequently,\vspace{-1ex}
            \begin{multline*}
                \ed(S_0^{(t+1)},Q^\infty\fragmentco{|Q|-p}{(j+\alpha)|Q|})
                = \ed(S_0^{(t)}S_1^{(t)},Q^\infty\fragmentco{|Q|-p}{j|Q|}Q^{\alpha}) \\
                \le \ed(S_0^{(t)},Q^\infty\fragmentco{|Q|-p}{j|Q|}) + \ed(S_1^{(t)},Q^{\alpha})
                \le \ed(S_0^{(t)},\rot^p(Q)^*)+\Delta_0^{(t)}+\edl{S_1^{(t)}}{Q} + \Delta_{1}^{(t)}\\
                \le \ed(S_0^{(t+1)},\rot^p(Q)^*) + \Delta_0^{(t+1)}.
            \end{multline*}
            \item Consider a type-\ref{it:type4} merge of~$S_0^{(t)}$ and $S_1^{(t)}$.
            Let us fix $p\in \fragmentco{0}{|Q|}$ with $\ed(S_0^{(t+1)},\rot^p(Q)^*)\le k-\Delta_0^{(t+1)}$.
            \begin{itemize}
                \item If $\ed(S_0^{(t+1)},\rot^p(Q)^*)>\ed(S_0^{(t)},\rot^p(Q)^*)$,
                then \[\ed(S_0^{(t)},\rot^p(Q)^*) \le\ed(S_0^{(t+1)},\rot^p(Q)^*) -1 \le k-\Delta_0^{(t+1)}-1=
                k-\Delta_0^{(t)}.\]
                Hence, the inductive assumption implies $\ed(S_0^{(t)},Q^\infty\fragmentco{|Q|-p}{j|Q|}) \le \ed(S_0^{(t)},\rot^p(Q)^*)+\Delta_0^{(t)}$ for some integer $j$.
                Consequently,\vspace{-1ex}
                \begin{multline*}
                    \ed(S_0^{(t+1)},Q^\infty\fragmentco{|Q|-p}{(j+1)|Q|})
                    = \ed(S_0^{(t)}Q,Q^\infty\fragmentco{|Q|-p}{j|Q|}Q) \\
                    \le \ed(S_0^{(t)},Q^\infty\fragmentco{|Q|-p}{j|Q|})
                    \le \ed(S_0^{(t)},\rot^p(Q)^*)+\Delta_0^{(t)}
                    \le \ed(S_0^{(t+1)},\rot^p(Q)^*)-1 + \Delta_0^{(t)} \\
                    = \ed(S_0^{(t+1)},\rot^p(Q)^*) + \Delta_0^{(t+1)}.
                \end{multline*}
                \item If $\ed(S_0^{(t+1)},\rot^p(Q)^*)=\ed(S_0^{(t)},\rot^p(Q)^*)$,
                then let $y'$ denote an arbitrary integer that satisfies
                $\ed(S_0^{(t+1)},\rot^p(Q)^*)=\ed(S_0^{(t+1)},Q^\infty\fragmentco{|Q|-p}{y'})$.
                This also yields an integer $y''$ with
                $|Q|-p\le y'' \le y'$ such that\vspace{-1ex}\[
                    \ed(S^{(t+1)}_{0},Q^\infty\fragmentco{|Q|-p}{y'}) =
                        \ed(S^{(t)}_0,Q^\infty\fragmentco{|Q|-p}{y''}) +
                        \ed(Q,Q^\infty\fragmentco{y''}{y'}).
                \]
                Due to\vspace{-2ex}\begin{multline*}
                    \ed(S_0^{(t)},\rot^p(Q)^*)
                    \le  \ed(S^{(t)}_0,Q^\infty\fragmentco{|Q|-p}{y''})
                    \le \ed(S^{(t+1)}_{0},Q^\infty\fragmentco{|Q|-p}{y'})\\
                    = \ed(S_0^{(t+1)},\rot^p(Q)^*)=\ed(S_0^{(t)},\rot^p(Q)^*),
                \end{multline*}
                we have
                $\ed(Q,Q^\infty\fragmentco{y''}{y''})$.
                As the string $Q$ is primitive, this means that $y'',y'$ are both multiples
                of~$|Q|$.
                Consequently,
                \[\ed(S^{(t+1)}_{0},Q^\infty\fragmentco{|Q|-p}{y'}) = \ed(S_0^{(t+1)},\rot^p(Q)^*) \le \ed(S_0^{(t+1)},\rot^p(Q)^*) + \Delta^{(t+1)}_{0}.\]
            \end{itemize}
        \end{enumerate}
        This completes the proof~of~the inductive step.
    \end{claimproof}
    Given that the final partition $S=S_{0}^{(t)}\cdots S_{s^{(t)}}^{(t)}$
    satisfies $s^{(t)}=0$ or $\Delta_{0}^{(t)}= 0$, we conclude that $S_0$ is indeed $k$-locked.
\end{proof}

We are now ready to prove \cref{lem:Eaux}, which we restate here for convenience.
\edaux*
\begin{proof}

    Consider any $k$-error occurrence $T\fragmentco{\ell}{r}$ of~$P$.
    By definition, $\ed(T\fragmentco{\ell}{r}, P) \le k \le d/2$.
    Combining this inequality with $\ed(P,Q^*) \le d$ via the triangle inequality
    (\cref{Etria}), we obtain the bound $\ed(T\fragmentco{\ell}{r}, Q^*) \le \threehalfs d$.
    In particular, this inequality is true for the $k$-error occurrence of~$P$
    as a prefix~of~$T$. Hence, for some integer $m' \in
    \fragment{m-k}{m+k}$, we have $\ed(T\fragmentco{0}{m'}, Q^*)\le \threehalfs d$,
    and thus also $\ed(T\fragmentco{0}{\min(r,m')}, Q^*)\le \threehalfs d$.

    Next, we apply \cref{lem:synchr} on the fragment $T\fragmentco{0}{r}$,
    whose prefix $T\fragmentco{0}{\min(r,m')}$ satisfies $\ed(T\fragmentco{0}{\min(r,m')}, Q^*)\le \threehalfs d$
    and whose suffix $T\fragmentco{\ell}{r}$ satisfies  $\ed(T\fragmentco{\ell}{r}, Q^*) \le \threehalfs d$.
    Further, if $|Q| > 1$, we also have $|T\fragmentco{\ell}{\min(r,m')}| \ge (3d+1)|Q|$:
    Due to $r-\ell \ge m-k$ and $m'\ge m-k$,
    we have $\min(r,m')-\ell\ge 2(m-k)-n > m/2 - 3k \ge 4d|Q|-3k$.
    Hence, it suffices to prove that $(d-1)|Q|\ge 3k-1$.
    This equality holds trivially if $k=0$.
    For $k\ge 1$, we have $(d-1)|Q|\ge (2k-1)\cdot 2 = 4k-2 \ge 3k-1$ due to $d\ge 2k$ and $|Q|\ge 2$.
    Thus, we can indeed use \cref{lem:synchr}.

    In particular, \cref{lem:synchr} implies $(\ell+3d)\bmod |Q|\le 6d$.
    Since  $T\fragmentco{\ell}{r}$ was an arbitrary $k$-error occurrence of~$P$ in $T$,
    we conclude that Claim~\eqref{it:Emult} holds.

    Moreover, \cref{lem:synchr} implies $\ed(T\fragmentco{0}{r},Q^*)\le 3d$.
    If we choose $T\fragmentco{\ell}{r}$ to be a $k$-error occurrence of~$P$ that is a suffix of~$T$,
    we have $r=n$ and therefore $\ed(T,Q^*) \le 3d$, which proves Claim~\eqref{it:Etext}.

    We proceed to the proof~of~Claim~\eqref{it:Efew}.
    Let $L_1,\ldots,L_{\ell}$ denote locked fragments of~$P$ obtained from \cref{lem:locked}.
    Note that we thus have $\sum_{i=1}^{\ell} \edl{L_i}{Q}=\edl{P}{Q}=d$, $\ell\le d$,
    and $\sum_{i=1}^{\ell}|L_i|\le (5|Q|+1)d$.

    Moreover, let us fix an optimal alignment between $T$ and a substring $Q'$ of~$Q^*$,
    and define $d':= \ed(T,Q^*)$. This yields partitions $T=T_0\cdots T_{2d'}$
    and $Q'=Q'_0\cdots Q'_{2d'}$ such that:
    \begin{itemize}
        \item $T_i = Q'_i$ for even $i$,
        \item $T_i \ne Q'_i$ and $|T_i|,|Q'_i|\le 1$ for odd $i$.
    \end{itemize}
    We create a multi-set $E = \{\sum_{i'<i}|T_{i'}| : i\text{ is odd}\}$ of~size $d'$.
    Its elements can be interpreted as positions in $T$ which incur errors in
    an optimal alignment with $Q'$.
    In particular, we show the following:
    \begin{claim}\label{clm:e}
        For every fragment $T\fragmentco{x}{y}$, we have $\edl{T\fragmentco{x}{y}}{Q}\le
        |\{e \in E \mid x \le e < y\}|$.
    \end{claim}
    \begin{claimproof}
        It suffices to observe that the alignment between $T$ and $Q'$
        yields an alignment between  $T\fragmentco{x}{y}$ and a fragment $Q'\fragmentco{x'}{y'}$
        with $\ed(T\fragmentco{x}{y},Q'\fragmentco{x'}{y'})\le |\{e \in E \mid x \le e < y\}|$ edits.
    \end{claimproof}

    We split $\mathbb{Z}$ into disjoint blocks of~the form $\fragmentco{jd}{(j+1)d}$ for $j\in\mathbb{Z}$.
    We say that a block $\fragmentco{jd}{(j+1)d}$ is \emph{synchronized}
    if it contains a position $p$ such that $(p+3d)\bmod |Q| \le 6d$.
    For every locked fragment
    $L_{i}=P\fragmentco{\ell_i}{r_i}$ and every $e\in E$, we
    mark a synchronized block $\fragmentco{jd}{(j+1)d}$ if
    \[e\in \fragmentco{jd+\ell_i-k}{(j+1)d-1+r_i+k}.\]

    \begin{claim}
        If $\fragmentco{jd}{(j+1)d} \cap \OccE_k(P,T) \ne \emptyset$,
        then $\fragmentco{jd}{(j+1)d}$ has at least $d-k$ marks.
    \end{claim}
    \begin{claimproof}
        Consider a $k$-error occurrence of~$P$ in $T$ starting at position $p\in \fragmentco{jd}{(j+1)d}$
        and fix its arbitrary optimal alignment with $P$.
        For each locked fragment $L_i$, let $L'_i$ be the fragment of~$T$ aligned with $L_i$ in this alignment.
        Moreover, $L'_i = T\fragmentco{\ell'_i}{r'_i}$ for
        \[\fragmentco{\ell'_i}{r'_i}\subseteq \fragmentco{p+\ell_i-k}{p+r_i+k}\subseteq \fragmentco{jd+\ell_i-k}{(j+1)d-1+r_i+k}.\]
        Also, by \cref{clm:e}, we have \[
            \edl{L'_i}{Q} \le |\{e \in E \mid \ell'_i \le e < r'_i\}|
        \le  |\{e \in E \mid jd+\ell_i-k\le e < (j+1)d-1+r_i+k\}|.
        \] Hence, the number $\mu$ of~marks at $\fragmentco{jd}{(j+1)d}$ is at least
        $\mu \ge \sum_{i=1}^{\ell}\edl{L'_i}{Q}$.
        On the other hand, by disjointness of~regions $L_i$,
        we have $\sum_{i=1}^{\ell}\ed(L'_i,L_i)\le k$.
        By the triangle inequality (\cref{Etria}), this yields
        $\sum_{i=1}^{\ell}\edl{L_i}{Q}\le k+\mu$.
        Since $\sum_{i=1}^{\ell}\edl{L_i}{Q}=\edl{P}{Q}=d$,
        we conclude that $\mu \ge d-k$.
    \end{claimproof}

    \begin{claim}
        The total number of~marks placed is at most $\qvarphiv d^2$.
    \end{claim}
    \begin{claimproof}
        Let us fix $e\in E$ and a locked fragment $L_i=P\fragmentco{\ell_i}{r_i}$.
        Recall that a mark is placed at a synchronized block $\fragmentco{jd}{(j+1)d}$
        if $e\in  \fragmentco{jd+\ell_i-k}{(j+1)d-1+r_i+k}$,
        or, in other words,
        \[\fragmentco{jd}{(j+1)d} \cap \fragmentco{e-r_i-k}{e-\ell_i+k} \ne \emptyset.\]
        The length of~the interval  $I_{i,e}:=\fragmentco{e-r_i-k}{e-\ell_i+k}$ satisfies $|I_{i,e}|= 2k+|L_i| \le d+|L_i|$.

        We now consider two cases depending on whether or not the inequality
        $|Q|< 9d$ is satisfied.
        If $|Q|< 9d$, it suffices to observe that any interval $I$
        overlaps with at most $2+|I|/d$ blocks.
        Hence, $I_{i,e}$ overlaps with at most $3+|L_i|/d$ blocks,
        and thus the number of~marks we have placed due to $e$ and $L_i$ is bounded by $3+|L_i|/d$.
        The overall number of~marks is therefore at most
        \[|E|\cdot \sum_{i=1}^{\ell}(3+|L_i|/d) \le 9d^2+3\sum_{i=1}^{\ell}|L_i|
        \le 9d^2 + 3(5|Q|+1)d \le 9d^2+135d^2=144d^2.\]

        On the other hand, if $|Q| \ge 9d$, we utilize the fact that only synchronized blocks are marked.
        For this, observe that any interval $I$ overlaps at most $2+|I|/|Q|$ intervals
        of~the form $\fragmentco{(j'-1)|Q|-4d}{j'|Q|+4d}$ each of~which overlaps with at most $7$ synchronized blocks
        (covering $\fragment{j'|Q|-3d}{j'|Q|+3d}$, which is of~length $6d+1$).
        Hence, the total number of~blocks marked due to $e$ and $L_i$
        is bounded by $7(2+(d+|L_i|)/|Q|)$.
        The overall number of~marks is therefore at most
        \begin{multline*}|E|\cdot \sum_{i=1}^{\ell}7(2+(d+|L_i|)/|Q|) \le
            42d^2+21d^2/|Q|+21d/|Q|\cdot \sum_{i=1}^{\ell}|L_i| \\
        \le42d^2+21d^2/|Q|+21d^2\cdot (5|Q|+1)/|Q|
        \le42d^2+21d^2/|Q|+105d^2 + 21d^2/|Q|< 152d^2.
    \end{multline*}
    This completes the proof~of~the claim.
    \end{claimproof}
    Hence, the number of~blocks with at least $d/2$ marks is at most $\pvarphiv d$,
    completing the proof~of~Claim~\eqref{it:Efew}.

    We proceed to the proof~of~Claim~\eqref{it:Eprog}.
    Let $L^{P}_1,\ldots,L^P_{\ell^P}$ denote locked fragments of~$P$ obtained from \cref{lem:klocked}
    (so that $L_1^P$ is a $k$-locked prefix of~$P$),
    and let $L^T_1,\ldots,L^T_{\ell^T}$ denote locked fragments of~$T$ obtained from \cref{lem:locked}.
    Denote $L^P_i = P\fragmentco{\ell^P_i}{r^P_i}$ for $i\in \fragment{1}{\ell^P}$ and $L^T_j=T\fragmentco{\ell^T_j}{r^T_j}$ for $j\in \fragment{1}{\ell^T}$.
    We say that a position $p\in \fragmentco{0}{n}$ is \emph{marked} if $p \in \fragment{n-m-k}{n-m+k}$ or
    $\fragmentco{p+\ell^P_i-k}{p+r^P_i+k}\cap \fragmentco{\ell^T_j}{r^T_j} \neq \emptyset$
    holds for some $i\in \fragment{1}{\ell^P}$ and $j\in \fragment{1}{\ell^T}$.
    (The positions in $\fragment{n-m-k}{n-m+k}$ are marked for a technical reason that will become clear in the proof~of~\cref{cl:progr_un}.)
    Furthermore, we say that $p$ is \emph{synchronized} if $p \bmod |Q| \le 3d$ or $p\bmod |Q| \ge |Q|-3d$.

Let us provide some intuition.
Informally, if there is a $k$-occurrence in an unmarked position $p$, then no locked region of~$P$ can overlap a locked region of~$T$ in any corresponding optimal alignment and we can exploit this structure.
On the other hand, we now show that there are only a few synchronized marked positions.

    \begin{claim}\label{cl:sync_marked}
        Marked positions can be decomposed into at most $10d^2$ integer intervals.
        Moreover, the number of~synchronized marked positions is at most $547d^3$.
    \end{claim}
    \begin{claimproof}
        First, we have the interval $\fragment{n-m-k}{n-m+k}$.
        Observe that each pair $i\in \fragment{1}{\ell^P}$ and $j\in \fragment{1}{\ell^T}$
        yields to marking positions $p\in \fragmentoo{\ell_j^T-r_i^P-k}{r_j^T-\ell_i^P+k}$.
        Consequently, marked positions can be decomposed into $1+\ell^P\cdot \ell^T$ integer intervals.
        Due to $\ell^P \le \edl{P}{Q}+2\le d+2 \le 3d$ and $\ell^T \le \edl{T}{Q}\le 3d$,
        the number of~intervals is at most $10d^2$.

        The interval of~positions marked due to $i$ and $j$ is of~length $2k+|L_i^P|+|L_j^T|-1\le d+|L_i^P|+|L_j^T|-1$.
        Out of~any $|Q|$ consecutive positions at most $6d+1\le 7d$ are synchronized.
        Hence, the number of~synchronized positions in any such interval $I$ is at most $7d(|I|-7d+|Q|)/|Q|$.
        Consequently, the total number of~synchronized marked positions does not exceed $2k+1\leq d^3$ (for $\fragment{n-m-k}{n-m+k}$) plus
        \begin{align*}
         \sum_{i=1}^{\ell^P}\sum_{j=1}^{\ell^T}7d\frac{|L_i^P|+|L_j^T|-7d+|Q|}{|Q|}
         &
        \le \frac{7d\ell^T }{|Q|}\sum_{i=1}^{\ell^P}|L_i^P|+\frac{7d\ell^P}{|Q|}\sum_{i=1}^{\ell^T}|L_i^T|+
        \frac{7d\ell^P\ell^T(|Q|-7d)}{|Q|} \\
        &\le \frac{21d^2}{|Q|}\left(\sum_{i=1}^{\ell^P}|L_i^P|+\sum_{i=1}^{\ell^T}|L_i^T|+3d(|Q|-7d)\right)  \\
        &\le \frac{21d^2}{|Q|}\left((5|Q|+1)d+2(k+1)|Q|+(5|Q|+1)3d + 3d(|Q|-7d)\right)  \\
        &\le \frac{21d^2}{|Q|}\left((5|Q|+1)d+3d|Q|+(5|Q|+1)3d + 3d(|Q|-7d)\right)  \\
        &\le \frac{21d^2}{|Q|}\left(26d|Q|+4d-21d^2\right)  \\
        & \le 21\cdot 26 d^3 = 546d^3.
        \end{align*}
        This completes the proof.
    \end{claimproof}

Next, we characterize unmarked positions $p\in \OccE_k(P,T)$.
They can be decomposed into at most $10d^2$ integer intervals by~\cref{cl:sync_marked} and by the fact that $p<n-m-k$.
Consider any such interval $I$.

\begin{claim}\label{cl:progr_un}
For any $p,p' \in I$ such that $p \equiv p' \pmod{|Q|}$ we have $p \in \OccE_k(P,T)$ if and only if $p' \in \OccE_k(P,T)$.
In particular, $I \cap \OccE_k(P,T)$ can be decomposed into at most $6d+1$ arithmetic progressions with difference $|Q|$.
\end{claim}
\begin{claimproof}
By our marking scheme, for any $i,j$, for any pair $x \in \fragmentco{\ell_i^P}{r_i^P}$
and $y \in \fragmentco{\ell_j^P}{r_j^P}$, we see that $|(p+x)-y|>k$.
Consider an unmarked position $p \in \OccE_k(P,T) \cap I$ and fix any alignment of~some prefix $T\fragmentco{p}{t}$ of~$T\fragmentco{p}{n}$ with $P$ with at most $k$ errors.
Then, for all $i$, the fragment $T_i$ of~$T$ aligned with $L^P_i$ is disjoint from all locked fragments of~$T$ and hence is a substring of~$Q^\infty$.
Now, recall that $L_1^P$ is a prefix of~$P$ and $L_{\ell^P}^P$ is a suffix of~$P$.
Hence, the locked fragments of~$T$ that are considered are exactly those $L^T_j$s such that $p+k < \ell_j < r_j < p+m-k$; say that this holds for $j \in \fragment{j_1}{j_2}$.
For all $j \in \fragment{j_1}{j_2}$, the fragment $P_j$ of~$P$ aligned with $L^T_j$ is disjoint from all locked fragments of~$P$ and is a substring of~$Q^\infty$.
In addition, since $I=\fragment{i_1}{i_2}$ is an interval of~unmarked positions, $T\fragmentco{i_1}{i_2+|L_1^P|+k}$ is disjoint from all locked fragments of~$T$ and hence is equal to a substring of~$Q^\infty$.
Thus, for some $r$ we see that:
    \begin{equation}\label{ineq:progr}
    \begin{split}
\ed(P,T\fragmentco{p}{t}) & \geq \ed(L_1^P,\rot^r(Q)^*)+\sum_{i=2}^{\ell^P}\ed(L^P_i,Q_i)+\sum_{j=j_1}^{j_2}\ed(\ell^T_j,P_j)\\
& \geq \ed(L_1^P,\rot^r(Q)^*)+\sum_{i=2}^{\ell^P}\edl{L^P_i}{Q}+\sum_{j=j_1}^{j_2}\edl{\ell^T_j}{Q}.
    \end{split}
    \end{equation}
Note that since $L_1^P$ is $k$-locked and $\ed(L_1^P,\rot^r(Q)^*)\leq k$ there exists a $j$ such that $\ed(L_1^P,\rot^r(Q)^*)=\ed(L_1^P,Q^\infty\fragmentco{|Q|-r}{j|Q|})$.
Let $b=|Q^\infty\fragmentco{|Q|-r}{j|Q|}|$.

Consider any position $p' \in I$ with $p' \equiv p \pmod{|Q|}$.
$T\fragmentco{p'}{p'+b}=T\fragmentco{p}{p+b}$ since both fragments lie in $T\fragment{i_1}{i_2+|L_1^P|+k}$ and start a multiple of~$|Q|$ positions apart.
In addition, we have $P\fragmentco{|L_1^P|}{m}=Q^{\alpha_1}L_2^P Q^{\alpha_2} \cdots Q^{\alpha_{\ell^P-1}} L^P_{\ell^P}$ for some non-negative integers $\alpha_i$ and that
$T\fragmentco{p'+b}{p'+m+k}=Q^{\beta_{j_1-1}}L_{j_1}^T Q^{j_1} \cdots L^T_{j_2} Q^{\beta{j_2}}Q'$ for some non-negative integers $\beta_j$ and a prefix $Q'$ of~$Q$.
(Note that $p'+m+k<n$ since $p'<n-m-k$.)
We claim that there exists a $t'$ such that:
\begin{align*}
 \ed(P,T\fragmentco{p'}{t'}) & \leq
\ed(L_1^P,T\fragmentco{p'}{p'+b})+\ed(P\fragmentco{|L_1^P|}{m},T\fragmentco{p'+b}{t'}) \\
& = \ed(L_1^P,\rot^r(Q)^*)+ \sum_{i=2}^{\ell^P}\edl{L^P_i}{Q}+\sum_{j=j_1}^{j_2}\edl{\ell^T_j}{Q} \leq \ed(P,T\fragmentco{p}{t}).
\end{align*}

In order to prove this, let us consider the following greedy alignment of~$P\fragmentco{|L_1^P|}{m}$ and $T\fragmentco{p'+b}{t'}$.
We start at the leftmost position in both strings.
We will maintain the invariant that the remainder of~each string starts with either $Q$ or a locked fragment,
except possibly for the case that the remainder of~$P$ is~$L_{\ell^P}^P$, in which case the remainder of~$T$ can be $Q'$.
While we have not reached the end of~$P\fragmentco{|L_1^P|}{m}$ we repeat the following procedure.
If both strings have a prefix equal to $Q$, we align those prefixes exactly.
Else, the prefix of~one of~the strings is a locked fragment $L$.
Let us first assume that $L \neq L_{\ell^P}^P$.
Then, since~$p'$ is unmarked, $Q^\infty\fragmentco{0}{|L|+k-k'}$ is a prefix of~the other
string, where $k'$ is the number of~edits already performed by our greedy alignment.
Since $L$ is locked, $\edl{L}{Q}=\ed(L,Q{^\alpha})$ for some integer $\alpha$.
We have $|\alpha|Q|- |L|| \leq k-k'$ due to the fact that otherwise~\eqref{ineq:progr} would imply that $\ed(P,T\fragmentco{p}{t})>k$, a contradiction.
Hence, $Q^{\alpha}$ is a prefix of~$Q^\infty\fragmentco{0}{|L|+k-k'}$; we optimally align those two fragments.
If $L = L_{\ell^P}^P$, an analogous argument shows that there exists a prefix $Q''$ of~the remainder of~$T\fragmentco{p'+b}{t'}$ such that $\edl{L}{Q}=\ed(L,Q'')$.
Upon termination of~this greedy alignment, the equality in the above equation holds.

We have thus proved that if $p \in I \cap \OccE_k(P,T)$ then all $p' \in I$ such that $p' \equiv p \pmod{|Q|}$ are also in $\OccE_k(P,T)$.
Thus, for any fixed $0 \leq j <|Q|$, for $U=\{i\cdot |Q|+j \in I \mid i \in \mathbb{Z}\}$ either $U \subseteq \OccE_k(P,T)$ or $U \cap \OccE_k(P,T) = \emptyset$.
By Claim~\eqref{it:Emult}, we can restrict our attention to synchronized positions.
We can thus decompose $I \cap \OccE_k(P,T)$ into at most $6d+1$ arithmetic progressions.
\end{claimproof}

Combining~\cref{cl:sync_marked,cl:progr_un} we conclude that $\OccE_k(P,T)$ can be decomposed into at most $547d^3+10d^2 (6d+1)\leq 617d^3$ arithmetic progressions with difference $|Q|$, thus completing the proof.
\end{proof}

\begin{corollary}[Compare~\cref{cor:aux}]\label{cor:Eaux}
    Let $P$ denote a pattern of~length $m$, let $T$ denote a string of~length $n$,
    and let $k\le m$ denote a non-negative integer.
    Suppose that there are a positive integer $d\ge 2k$ and a primitive string $Q$
    with $|Q|\le m/8d$ and $\edl{P}{Q}=d$.
    Then $|\floor{\OccE_k(P,T)/d}|\le \varphiv \cdot n/m \cdot d$.
\end{corollary}
\begin{proof}
    Partition the string $T$ into $\floor{2n/m}$ blocks $T_0, \dots, T_{\floor{2n/m}-1}$ of~length
    at most $\threehalfs m+k-1$ each, where the $i$th block starts at position
    $i\cdot m/2$, that is, $T_i := T\fragmentco{\floor{i\cdot {m}/2}}
        {\min\{n, \floor{(i+3)\cdot {m}/2} + k - 1\}}$.
    Observe that each $k$-error occurrence of~$P$ in $T$ is contained in at least one of~the
    fragments $T_i$:
    Specifically, $T_i$ covers all  the occurrences starting in $\fragmentco{\floor{i\cdot
    m/2}}{\floor{(i+1)\cdot m/2}}$.
    If $\OccE_k(P,T_i)\ne \emptyset$, we define $T'_i:=T\fragmentco{t'_i}{t'_i+|T'_i|}$ to
    be the shortest fragment
    of~$T_i$ containing all $k$-error occurrences of~$P$ in~$T_i$.
    As a result, $T'_i$ satisfies the assumptions of~\cref{lem:Eaux}, so
    $|\floor{\OccE_k(P,T_i')}/k|\le \pvarphiv d$.
    Each block $\fragmentco{j'k}{(j'+1)k}$ of~positions in $T_i$ corresponds to a
    block $\fragmentco{t'_i+j'k}{t'_i+(j'+1)k}$ of~positions in $T$,
    which intersects at most two blocks of~the form $\fragmentco{jk}{(j+1)k}$.
    In total, we conclude that $|\floor{\OccEx(P,T_i)}/k|\le
     \floor{2n/m} \cdot 2 \cdot \pvarphiv d \le \varphiv \cdot n/m \cdot d$.
\end{proof}

\subsection{Bounding the Number of~Occurrences in the Non-Periodic Case}

\begin{lemma}[Compare~\cref{prp:I}]\label{prp:EI}
    Let $P$ denote a string of~length $m$ and let $k \le m$ denote a positive integer.
    Then, at least one of~the following holds:
    \begin{enumerate}[(a)]
        \item The string $P$ contains $2k$ disjoint \emph{breaks} $B_1,\ldots, B_{2k}$
            each having period $\per(B_i)> m/\alphav k$ and length $|B_i| = \lfloor
            m/\betav k\rfloor$.
        \item The string $P$ contains disjoint \emph{repetitive regions} $R_1,\ldots, R_{r}$
            of~total length $\sum_{i=1}^r |R_i| \ge \deltavN/\deltavD \cdot m$ such
            that each region $R_i$ satisfies
            $|R_i| \ge m/\betav k$ and has a primitive \emph{approximate period} $Q_i$
            with $|Q_i| \le m/\alphav k$ and $\edl{R_i}{Q_i} = \ceil{\betav k/m\cdot |R_i|}$.
        \item The string $P$ has a primitive \emph{approximate period} $Q$ with $|Q|\le
            m/\alphav k$ and
            $\edl{P}{Q} < \betav k$.
    \end{enumerate}
\end{lemma}
\begin{algorithm}[t]
    $\mathcal{B} \gets \{\}; \mathcal{R} \gets \{\}$\;
    \While{\bf true}{
        Consider the fragment $P' = P\fragmentco{j}{j+\floor{m/\betav k}}$
        of~the next $\floor{m/\betav k}$ unprocessed characters of~$P$\;
        \If{$\per(P') > m/\alphav k$}{
            $\mathcal{B} \gets \mathcal{B} \cup \{P'\}$\;
            \lIf{$|\mathcal{B}| = 2k$}{\Return{breaks $\mathcal{B}$}}
            }\Else{
            $Q \gets P\fragmentco{j}{j+\per(P')}$\;
            Search for prefix $R$ of~$P\fragmentco{j}{m}$
            with $\colorbox{lipicsYellow!80}{$\edl{R}{Q}$} =
            \lceil\betav k/m\cdot |R|\rceil$ and $|R|>|P'|$\;\label{ln:efwd}
            \If{such $R$ exists}{
                $\mathcal{R} \gets \mathcal{R} \cup \{(R, Q)\}$\;
                \If{$\sum_{(R, Q) \in \mathcal{R}} |R|\ge \deltavN/\deltavD \cdot
                    m$}{
                    \Return{repetitive regions (and their corresponding
                    periods) $\mathcal{R}$}\;
                }
                }\Else{
                Search for suffix $R'$ of~$P$
                with $\colorbox{lipicsYellow!80}{$\edl{R'}{Q}$} = \lceil\betav
                k/m\cdot |R'|\rceil$ and $|R'|\ge m-j$\;\label{ln:ebcw}
                \lIf{such $R'$ exists}{%
                    \Return{repetitive region $(R',\colorbox{lipicsYellow!80}{$\!Q\!$})$}%
                    }\lElse{%
                    \Return{approximate period \colorbox{lipicsYellow!80}{$\!Q\!$}}%
                }
            }
        }
    }
    \caption{A constructive proof~of~\cref{prp:EI}. Changes to \cref{alg:P1}
    are highlighted.}\label{alg:E1}
\end{algorithm}

\begin{proof}
    We use essentially the same algorithm as in the proof~of~\cref{prp:I}: We replace all checks
    for a specific Hamming distance with the corresponding counterpart for the edit
    distance. Further, as we are only interested in (approximate) periods under an
    arbitrary rotation, we do not need to explicitly rotate the string $Q$ in the
    algorithm anymore. Consider \cref{alg:E1} for a visualization; the changes to
    \cref{alg:P1} are highlighted.

    In particular, we directly get an analogue of~\cref{clm:j}:
    \begin{claim}[See \cref{clm:j}]\label{clm:Ej}
        Whenever we consider a new fragment $P\fragmentco{j}{j + \floor{m/\betav k}}$ of
        $\floor{m/\betav k}$ unprocessed characters of~$P$, such a fragment
        starts at a position $j<\ubjv m$.\lipicsClaimEnd
    \end{claim}
    Again, note that \cref{clm:Ej} also shows that whenever we consider a new fragment $P'$
    of~$\floor{m/\betav k}$  characters, there is indeed such a fragment, that is, $P'$
    is well-defined.

    Now consider the case when, for a fragment $P' = P\fragmentco{j}{j + \floor{m/\betav k}}$
    (that is not a break) and its corresponding period $Q = \fragmentco{j}{j + \per(P')}$,
    we fail to obtain a new repetitive region $R$.
    Recall that in this case, we search for a repetitive region
    $R'$ of~length $|R'|\ge m-j$ that is a suffix of~$P$ and has an approximate period $Q$.
    If we indeed find such a region $R'$, then $|R'|\ge m-j \ge m-\ubjv m =
    \deltavN/\deltavD \cdot m$ by \cref{clm:Ej}, so $R'$ is long enough to be reported on
    its own.
    However, if we fail to find such $R'$, we need to show that $Q$ can be
    reported as an approximate period of~$P$, that is, $\edl{P}{Q} < \betav k$.

    Similar to \cref{prp:I}, we first show that
    $\edl{P\fragmentco{j}{m}}{Q} < \ceil{\betav k/m \cdot (m-j)}$.
    For this, we inductively prove that the values $\Delta_{\rho}:=\ceil{\betav k/m \cdot
    \rho}-\edl{P\fragmentco{j}{j+\rho}}{Q}$ for $|P'|\le \rho \le m-j$ are all at least 1.
    In the base case of~$\rho=|P'|$, we have $\Delta_{\rho}=1-0$ because $Q$ is the string
    period of~$P'$.
    To carry out an inductive step, suppose that $\Delta_{\rho-1}\ge 1$ for some $|P'|<
    \rho \le m-j$.
    Notice that $\Delta_{\rho}\ge \Delta_{\rho-1}-1\ge 0$: the first term in the
    definition of~$\Delta_{\rho}$ has not decreased, and the second term
    $\edl{P\fragmentco{j}{j+\rho}}{Q}$ may have
    increased by at most $1$ compared to $\Delta_{\rho-1}$.
    Moreover, $\Delta_{\rho} \ne 0$ because $R=P\fragmentco{j}{j+\rho}$ could not be
    reported as a repetitive region. Since $\Delta_{\rho}$ is an integer, we conclude that
    $\Delta_{\rho}\ge 1$.
    This inductive reasoning ultimately shows that $\Delta_{m-j}>0$, that is,
    $\edl{P\fragmentco{j}{m}}{Q} < \ceil{\betav k/m \cdot (m-j)}$.

    A symmetric argument holds for values $\Delta'_{\rho}:= \ceil{\betav k/m \cdot
    \rho}-\edl{P\fragmentco{m-\rho}{m}}{Q}$ for $m-j\le \rho \le m$
    because no repetitive region $R'$ was found as an extension of~$P\fragmentco{j}{m}$ to
    the left. Note that in contrast to the proof~of~\cref{prp:I}, the rotation of~$Q$ is implicit.
    This completes the proof~that $\edl{P}{Q}< \betav k$,
    that is, $Q$ is an approximate period of~$P$.
\end{proof}

\begin{lemma}[Compare~\cref{lm:hdC}]\label{lm:EdC}
    Let $P$ denote a pattern of~length $m$, let
    $T$ denote a text of~length $n$, and let $k\le m$ denote a positive integer.
    Suppose that $P$ that contains $2k$
    disjoint breaks $B_1,\dots,B_{2k} \substr P$
    each satisfying $\per(B_i) \ge m / \alphav k$.
    Then, $|\floor{\OccE_k(P,T)/k}|\le \alphavdt \cdot n/m \cdot k$.
\end{lemma}
\begin{proof}
    The proof~proceeds similarly to the proof~of~\cref{lm:hdC}. The only major difference
    is that we obtain length-$k$ \emph{blocks} of~possible starting positions instead
    of~single starting positions. This is because the edit distance allows for deletions
    and insertions of~characters.

    Hence, we split $\mathbb{Z}$ into disjoint blocks
    of~the form $\fragmentco{jk}{(j+1)k}$ for $j\in \mathbb{Z}$.
    Now for every break $B_i = P\fragmentco{b_i}{b_i + |B_i|}$, we mark
    a block $\fragmentco{jk}{(j+1)k}$
    if \[\fragmentco{(j - 1)\cdot k + b_i}{(j + 2)\cdot k + b_i}
        \cap \OccEx(B_i, T) \ne \emptyset.\]

    Similarly to the proof~of~\cref{lm:hdC}, we proceed to show that we place at most $\Oh(n/m \cdot k^2)$
    marks and that every $k$-error occurrence starts in a block with at least $k$
    marks.
    \begin{claim}\label{cl:Ec1}
        The total number of~marks placed at blocks is at most $\alphavdt \cdot n/m \cdot k^2$.
    \end{claim}
    \begin{claimproof}
        Fix a break $B_i$. Notice that positions in $\OccEx(B_i,T)$ are at distance at least
        $\per(B_i)$ from each other. Furthermore, note that for every occurrence
        in $\OccEx(B_i, T)$ we mark at most $4$ blocks. Hence,
        for the break $B_i$, we place at most $\alphavt \cdot n/m \cdot k$ marks.
        In total, the number of~marks placed is thus at most
        $2k\cdot \alphavt n/m \cdot k = \alphavdt \cdot n/m \cdot k^2$.
    \end{claimproof}
    Next, we show that every $k$-error occurrence of~$P$ in $T$
    starts in a block  with at least $k$ marks.
    \begin{claim}\label{cl:Ec2}
        If $\fragmentco{jk}{(j+1)k}\cap \OccE_k(P,T) \ne \emptyset$, then $\fragmentco{jk}{(j+1)k}$ has at least $k$ marks.
    \end{claim}
    \begin{claimproof}
        Consider a $k$-error occurrence of~$P$ in $T$ starting at position $\ell\in \fragmentco{jk}{(j+1)k}$
        and fix an arbitrary optimal alignment of~it with $P$.
        Out of~the $2k$ breaks, at least $k$ breaks are matched exactly, as not matching
        a break exactly incurs at least one error.
        If a break $B_i$ is matched exactly, then for at least one $s \in
        \fragment{-k}{k}$, we have  $\ell + b_i + s \in \OccEx(B_i,T)$.
        Since $jk \le \ell < (j+1)k$, we conclude that
        $\fragmentco{(j - 1)\cdot k + b_i}{(j + 2)\cdot k + b_i}
        \cap \OccEx(B_i, T) \ne \emptyset$, that is, that the block
        $\fragmentco{jk}{(j+1)k}$ has been marked for $B_i$.
        In total, there are at least $k$ marks for the at least $k$ breaks matched exactly.
    \end{claimproof}
    By \cref{cl:Ec1,cl:Ec2}, the number of~blocks where $k$-error occurrences of~$P$ in $T$ may start
    is $|\floor{\OccE_k(P,T)/k}|
    \le (\alphavdt \cdot n/m \cdot k^2)/k = \alphavdt \cdot n/m \cdot k$.
\end{proof}

\begin{lemma}[Compare~\cref{lm:hdB}]\label{lm:EdB}
    Let $P$ denote a pattern of~length $m$, let $T$ denote a text of~length $n$,
    and let $k\le m$ denote a positive integer.
    Suppose that $P$ contains disjoint repetitive regions $R_1,\ldots, R_{r}$
    of~total length at least $\sum_{i=1}^r |R_i| \ge \deltavN/\deltavD\cdot m$
    such that each region $R_i$ satisfies $|R_i| \ge m/\betav k$ and has a
    primitive approximate period~$Q_i$
    with $|Q_i| \le m/\alphav k$ and $\edl{R_i}{Q_i} = \ceil{\betav k/m\cdot |R_i|}$.
    Then, $|\floor{\OccE_k(P,T)/k}|\le \varphig\cdot n/m \cdot k$.
\end{lemma}
\begin{proof}
    Again, the proof~is similar to its Hamming distance counterpart; as before, a major
    difference is that we only obtain length-$k$ \emph{blocks} of~possible starting positions
    instead of~single starting positions.

    As in the proof~of~\cref{lm:EdC}, we split $\mathbb{Z}$ into disjoint
    blocks of~the form $\fragmentco{jk}{(j+1)k}$ for $j\in \mathbb{Z}$.
    Further, we set $m_R := \sum_r |R_i|$ and define $k_i := \floor{\betavh \cdot |R_i|/m \cdot k}$
    for every $1 \le i \le r$.

    For every repetitive region $R_i = P\fragmentco{r_i}{r_i + |R_i|}$,
    we place $|R_i|$ marks on block $\fragmentco{jk}{(j+1)k}$ if
    \[\fragmentco{(j-1)\cdot k + r_i}{(j + 2)\cdot k + r_i}
    \cap \OccE_{k_i}(R_i, T) \ne \emptyset.\]

    Similarly to the proof~\cref{lm:hdB}, we proceed to show that we placed at most $\Oh(n/m\cdot k
    \cdot m_R)$ marks and that every $k$-error occurrence of~$P$ in $T$ starts in a block
    with at last $m_R - m/\betavh$ marks.
    \begin{claim}\label{cl:Eb1}
        The total number of~marks placed is at most $\varphivbq \cdot n/m \cdot k\cdot m_R$,
    \end{claim}
    \begin{claimproof}
        We use \cref{cor:Eaux} to analyze $\OccE_{k_i}(R_i,T)$.
        For this, we set $d_i := \edl{R_i}{Q_i}=\ceil{\betav
        k/m\cdot |R_i|}$ and notice that $d_i  \le \tbetav k/m\cdot |R_i|$
        since $|R_i| \ge m/\betav k$.
        Moreover, $d_i \ge 2k_i$ and $|Q_i| \le m/\alphav k \le |R_i|/ \betav d_i$. Hence, the assumptions of~\cref{cor:Eaux} are
        satisfied, so $\floor{\OccE_{k_i}(R_i,T)/d_i} \le \varphiv \cdot n/|R_i| \cdot d_i
        \le \varphivb \cdot k \cdot n/m$.

        For a block $\fragmentco{j'd_i}{(j'+1)d_i}$ intersecting $\OccE_{k_i}(R_i,T)$,
        we mark a block $\fragmentco{j k}{(j+1)k}$ only if \[\fragmentco{(j-1)\cdot k + r_i}{(j + 2)\cdot k + r_i}
        \cap \fragmentco{j'd_i}{(j'+1)d_i}\ne \emptyset,\]
        which holds only if
        $jk \in \fragmentco{j'd_i - r_i - 2k}{(j'+1)d_i -r_i + k}$.
        The length of~the latter interval is $d_i+3k=\ceil{\betav
        k/m\cdot |R_i|}+3k \le  \betavpf k$,
        so the interval contains at most $\betavpf$ multiples of~$k$.
        Hence, the total number of~marks placed due to $R_i$ is bounded by
        $\betavpf \cdot \varphivb \cdot n/m \cdot k\cdot |R_i|$.
        Across all repetitive regions, this sums up to no more than $\varphivbq \cdot n/m \cdot k\cdot
        m_R$, yielding the claim.
    \end{claimproof}
    Next, we show that every $k$-error occurrence of~$P$ in $T$ starts in a block with many marks.
    \begin{claim}\label{cl:Eb2}
        If $\fragmentco{jk}{(j+1)k}\cap \OccE_k(P,T) \ne \emptyset$, then $\fragmentco{jk}{(j+1)k}$ has at least $m_R-m/\betavh$ marks.
    \end{claim}
    \begin{claimproof}
        Consider a $k$-error occurrence of~$P$ in $T$ starting at position $\ell\in \fragmentco{jk}{(j+1)k}$
        and fix an arbitrary optimal alignment of~it with $P$.
        For each repetitive region $R_i$, let $R'_i$ be the fragment of~$T$ aligned with $R_i$ in this alignment. Define $k'_i = \ed(R_i,R'_i)$ and observe that $R'_i=T\fragmentco{r'_i}{r'_i+|R'_i|}$ for some $r'_i \in \fragment{\ell+r_i-k}{\ell+r_i+k}\subseteq \fragmentco{(j-1)\cdot k + r_i}{(j + 2)\cdot k + r_i}$. Consequently,
        \[\fragmentco{(j-1)\cdot k + r_i}{(j + 2)\cdot k + r_i} \cap \OccE_{k'_i}(R_i,T)\ne \emptyset.\]
        Further, let $I := \{i \mid k'_i \le k_i\}=\{i \mid k'_i \le \betavh \cdot |R_i|/m \cdot k\}$ denote the set of~indices $i$ for which $R'_i$ is a $k_i$-error occurrence of~$R_i$.
        By construction, for each $i\in I$,
        we have placed $|R_i|$ marks at the block $\fragmentco{jk}{(j+1)k}$.

        Hence, the total number of~marks at the block $\fragmentco{jk}{(j+1)k}$ is at least
        $\sum_{i\in I}|R_i|= m_R -\sum_{i\notin I}|R_i|$.
        It remains to bound the term $\sum_{i\notin I}|R_i|$. Using the definition of~$I$,
        we obtain
        \[\sum_{i\notin I} |R_i| = \frac{m}{\betavh k} \cdot  \sum_{i\notin I} \left(\betavh \cdot |R_i|/m \cdot k\right)
            < \frac{m}{\betavh k} \cdot  \sum_{i\notin I} k'_i \le
                \frac{m}{\betavh k} \cdot  \sum_{i=1}^r k'_i \le \frac{m}{\betavh},\]
        where the last bound holds because, in total, all repetitive regions incur at most
        $\sum_{i=1}^r k'_i \le k$ errors (since the repetitive regions are disjoint).
        Hence, the number of~marks placed is at least $m_r-m/\betavh k$,
        completing the proof~of~the claim.
    \end{claimproof}
    In total, by \cref{cl:Eb1,cl:Eb2}, the number of~$k$-error occurrences of~$P$ in $T$
    is at most \[
        \OccE_k(P, T) \le  \frac{\varphivbq \cdot n/m \cdot k \cdot m_R}{m_R - m/\betavh}.
    \]
    As this bound is a decreasing function in $m_R$, the assumption $m_R \ge
    \deltavN/\deltavD\cdot m$ yields
    \[\OccE_k(P, T) \le \frac{\varphivbq \cdot n/m \cdot k \cdot \deltavN/\deltavD\cdot m}
        {\deltavN/\deltavD\cdot m - m/\betavh} = \varphig \cdot n/m \cdot k,
    \]completing the proof.
\end{proof}

\begin{lemma}[Compare~\cref{lm:hdA}]\label{lm:EdA}
    Let $P$ denote a pattern of~length $m$, let $T$ denote a text of~length~$n$,
    and let $k\le m$ denote a positive integer.
    If there is a primitive string $Q$
    of~length at most $|Q| \le m/\alphav k$ that satisfies $2k\le\edl{P}{Q}\le\betav k$,
    then $|\floor{\OccE_k(P,T)/k}|\le \gammaEp \cdot n/m \cdot k$.
\end{lemma}
\begin{proof}
    We apply \cref{cor:Eaux} with $d =\edl{P}{Q}$.
    Observe that $d\ge 2k$ and that $|Q|\le m/\alphav k \le m/8d$ due to $d\le \betav k$.
    Hence, the assumptions of~\cref{cor:Eaux} are met.

    Consequently, $|\floor{\OccE_k(P,T)/d}|
        \le \varphiv \cdot n/m \cdot d
    \le \varphibq \cdot n/m \cdot k.$
    Every block $\fragmentco{j'd}{(j'+1)d}$
    is of~length at most $8k$, and thus may intersect at most 9
    blocks of~the form $\fragmentco{jk}{(j+1)k}$.
    Consequently, $\floor{\OccE_k(P,T)/k}\le 9 \cdot \varphibq \cdot n/m \cdot k$,
    completing the proof.
\end{proof}

\edmain*
\begin{proof}
    The proof~proceeds similarly to the proof~of~\cref{thm:hdmain}:
    We apply \cref{prp:EI} on the string $P$ and proceed
    depending on the structure found in $P$.

    If the string $P$ contains $2k$ disjoint breaks $B_1,\dots,B_{2k}$
    (in the sense of~\cref{prp:EI}), we apply \cref{lm:EdC}
    and obtain that $|\floor{\OccE_k(P,T)/k}|\le \alphavdt \cdot n/m \cdot k$.

    If the string $P$ contains disjoint repetitive regions $R_1,\dots,R_r$
    (again, in the sense of~\cref{prp:EI}), we apply \cref{lm:EdB} and obtain that
    $|\floor{\OccE_k(P,T)/k}|\le \varphig \cdot n/m \cdot k$.

    Otherwise, \cref{prp:EI} guarantees that there is a primitive string $Q$ of~length
    at most $|Q| \le m/\alphav k$ that satisfies $\edl{P}{Q} < \betav k$.
    If $\edl{P}{Q}\ge 2k$, then \cref{lm:EdA} yields
    $|\floor{\OccE_k(P,T)/k}|\le \gammaEp \cdot n/m \cdot k$.
    If, however, $\edl{P}{Q}< 2k$, then we are in the second alternative of~the theorem statement.
\end{proof}

\section{Algorithm: Pattern Matching with Edits in the \modelname Model}\label{sec:pme}

In this section, we discuss how to solve pattern matching with edits in the \modelname
model. Specifically, we prove the following.
\edalgI*
\def\edalgIt{1}

The overall structure of the algorithm is similar to the Hamming distance case: We first
introduce useful tools for the algorithms later. Then, we implement {\tt Analyze}, which
is then followed by a discussion of the case when the pattern is periodic. Finally, we
discuss the easier non-periodic case and conclude with combining the various auxiliary
algorithms.

\subsection{Auxiliary \modelname Model Operations for Pattern Matching with Edits}
\label{sc:auxped}

As in the Hamming distance setting, we start the discussion of~the algorithms with
general tools that we use as auxiliary procedures in the remaining algorithms.
Specifically, we discuss a generator that computes the ``next'' error between two
strings. Further, we discuss a procedure to verify whether there is an occurrence of~the pattern
at a given (interval of) position(s) in the text.

\begin{lemma}[{\tt EditGenerator($S$, $Q$)}, {\tt EditGenerator$^R$($S$, $Q$)}]\label{lm:misope}
    Let $S$ denote a string and let $Q$ denote a string (that is possibly given as
    a cyclic rotation $Q' = \rot^j(Q)$).
    Then, there is an $(\Oh(1), \Oh(k))$-time generator that in the $k$-th call to
    \nxt returns the length of~the longest prefix (suffix) $S'$ of~$S$ and the length of
    the corresponding prefix (suffix) $Q'$ of~$Q^\infty$ such that $\ed(S', Q') \le k$.
    (Note that $k \ge 0$, that is, the initial call to \nxt is the zeroth call.)

    Further, both generators support an additional operation {\tt Alignment},
    that outputs a witness for the result returned by $k$-th call to \nxt
    that is, {\tt Alignment} outputs a sequence of~edits ($(i, j)$ for a replacement,
    $(i, \bot)$ for an insertion in $S$, and $(\bot, i)$ for an insertion in $Q$).
    The operation {\tt Alignment} takes $\Oh(k)$ time in the \modelname model.
\end{lemma}
\begin{algorithm}[p]
    \SetKwBlock{Begin}{}{end}
    \SetKwFunction{errgen}{EditGenerator}
    \SetKwFunction{errgenR}{EditGenerator$^R$}
    \SetKwFunction{nxt}{Next}
    \SetKwFunction{align}{Alignment}

    \errgen{$S$, $Q' = \rot^j(Q)$}\Begin{
        \Return{{\tt\{}\Begin{
                    $S \gets S$; $Q' \gets Q'$; $k \gets 0$; $j \gets j$; $end \gets {\tt
                    false}$\;
                $(\ell^{(-1)}_{-1},\dots,\ell^{(-1)}_1)\gets (-\infty, -\infty, -\infty)$\;
                    $A\fragment{-2}{2} \gets \position{(), \dots, ()}$;
     }{\tt\}}}\;
    }
    \BlankLine
    \nxt{$\mathbf{G}=\{S;\,Q';\,k;\,j;\,(\ell^{(k-1)}_{-k-1},\dots,\ell^{(k-1)}_{k+1})$;
        $A$; $end\}$}\Begin{
        \If{$k = 0$}{
            $(\ell^{(0)}_{-2},\dots,\ell^{(0)}_2)\gets (-\infty, -\infty, \lceOp{S}{Q^{\infty}},
                -\infty, -\infty)$\;
            replace $\ell^{(k-1)}$ with  $\ell^{(k)}$; $k \gets k + 1$\;
            \Return{$\ell^{(0)}_0$}\;
        }

        \BlankLine
        $(\ell^{(k)}_{-k - 2},\dots,\ell^{(k)}_{k + 2}) \gets (-\infty,
        \dots, -\infty)$\;
        $A'\fragment{-k-2}{k+2} \gets \position{(),\dots, ()}$;
        $r \gets -\infty; a_r \gets -1$\;
        \For{$i \gets -k - 1$ \KwSty{to} $k + 1$}{
            \tcp{Compute new prefix lengths as long as we did not reach the end of~$S$}
            \If{{\tt not} $end$}{
            $\ell_{insert} \gets \ell^{(k-1)}_{i-1} + 1 +
            \lceOp{S\fragmentco{\ell^{(k-1)}_{i-1} - i}{|S|}}
            {Q'^{\infty}\fragmentco{j + \ell^{(k-1)}_i + 1}{}}$\;
            $\ell_{replace} \gets \ell^{(k-1)}_i + 1 +
            \lceOp{S\fragmentco{\ell^{(k-1)}_{i} - i + 1}{|S|}}
                {Q'^{\infty}\fragmentco{j + \ell^{(k-1)}_i + 1}{}}$\;
            $\ell_{delete} \gets \ell^{(k-1)}_{i+1} + \lceOp{S\fragmentco{\ell^{(k-1)}_{i} - i +
                1}{|S|}}{Q'^{\infty}\fragmentco{j + \ell^{(k-1)}_i}{}}$\;
                $\ell^{(k)}_i \gets \max(\ell_{insert}, \ell_{replace}, \ell_{delete})$\;
            }\lElse{$\ell^{(k)}_i \gets \ell^{(k-1)}_i$}
            $r \gets \max(r, \ell^{(k)}_i)$\;
            \lIf{$r = \ell^{(k)}_i - i$}{$a_r \gets \ell^{(k)}_i + i$}
            \lIf{$end$}{{\tt continue}}
            \BlankLine
            \tcp{Store witness for \align}
            \If{$\ell^{(k)}_i = \ell_{insert}$}{$A'\position{i} \gets (A\position{i-1},
            (\bot,j + \ell^{(k-1)}_i))$\;}
            \If{$\ell^{(k)}_i = \ell_{replace}$}{$A'\position{i} \gets (A\position{i},
            (\ell_i - i, j + \ell^{(k-1)}_i))$\;}
            \If{$\ell^{(k)}_i = \ell_{delete}$}{$A'\position{i} \gets (A\position{i+1},
            (\ell^{(k-1)}_i - i, \bot))$\;}
        }
        replace $\ell^{(k-1)}$ with  $\ell^{(k)}$;
        $k \gets k + 1$;
        $A \gets A'$\;
        \lIf{$r \ge |S|$}{$end \gets {\tt true}$}
        \Return{$(r, a_r)$}\;
    }
    \BlankLine
    \align{$\mathbf{G}=\{S;\,Q';\,k;\,j;\,(\ell^{(k-1)}_{-k-1},\dots,\ell^{(k-1)}_{k+1})$;
        $A$; $end\}$}\Begin{
        $r \gets -\infty$; $a_r \gets -1$\;
        \For{$i \gets -k - 1$ \KwSty{to} $k + 1$}{
            $r \gets \max(r, \ell^{(k-1)}_i)$\;
            \lIf{$r = \ell^{(k-1)}_i$}{$a_r \gets i$}
        }
        \Return{$A\position{a_r}$}\;
    }

    \caption{An analogue of~\cref{alg:gen} for edit distance: We adapt Landau and
    Vishkin's algorithm to make its execution ``resumable''.}\label{alg:errgen}
\end{algorithm}
\begin{proof}
    We focus on the generator {\tt EditGenerator($S$, $Q$)};
    the generator {\tt EditGenerator$^R$($S$, $Q$)} can be obtained in a symmetric
    manner.

    We construct the generator as follows:
    We run the dynamic programming algorithm by Landau and Vishkin \cite{LandauV89}
    for one additional error (per call to \nxt) at a time, storing the dynamic
    programming table as a state in the generator.
    In particular, we maintain a sequence $\ell$ that after $k$ calls to \nxt
    stores at a position $i \in \fragment{-k}{k}$ the length of~the longest
    prefix $S'$ of~$S$ that satisfies $\ed(S', Q^\infty\fragment{0}{|S'| + i}) \le k$
    (that is, intuitively, we store how far each of~the ``diagonals'' of~the dynamic
    programming table extends).

    \noindent
    In each call to \nxt, we update the values stored in the sequence $\ell$ by using three
    calls to \lceOpName from \cref{lm:inflcp} (one call for each of~the {\tt insert}, {\tt
    replace}, {\tt delete} cases of~the
    edit distance) to compute each new entry. We then obtain the result as the maximum
    value of~the newly computed sequence.

    In order to support the \align operation, we additionally store every diagonal
    represented as list of~pairs of~{\tt insert}, {\tt replace}, and {\tt delete}
    operations and the position(s) in the strings $S$ and $Q^\infty$ where the edits happened.
    In each call to \nxt, we also update the representations of~the diagonals.
    (Note that, for performance reasons, we need to avoid copying whole diagonals;
    this can be done by storing all diagonals together as a (directed) graph.)

    Consider \cref{alg:errgen} for a visualization of~the generator as pseudo-code;
    note that we simplified how we store the diagonals for the \align operation
    to improve readability.

    For the correctness, we show by induction that the values computed in the array $\ell$
    are indeed correct, that is, after $k$ calls to the \nxt operation,
    for all $i \in \fragment{-k}{k}$ we have
    \[
        \ell^{(k)}_i = \max_{r} \{ r \mid \ed(S\fragment{0}{r}, Q^{\infty}\fragment{0}{r +
        i}) \le k \},
    \] or we reached the end of~the string $S$ (in which case the output does not change
    anymore after calling \nxt).
    For the zeroth call to next, we explicitly return
    \[  \ell^{(0)}_0 = \lceOp{S}{Q^{\infty}} = \max_{r} \{ r \mid \ed(S\fragment{0}{r},
        Q^{\infty}\fragment{0}{r+0}) = 0 \},
    \] which is thus correct. Now consider the $k$-th call to \nxt and assume that
    the values computed so far are correct. In particular, for all
    $i \in \fragment{-k+1}{k-1}$, we have
    \[
        \ell^{(k-1)}_i = \max_{r} \{ r \mid \ed(S\fragment{0}{r},
        Q^{\infty}\fragment{0}{r+i}) \le k - 1 \}.
    \] Now, fix a $j \in \fragment{k}{-k}$ and consider the longest prefix
    $S' = S\fragment{0}{r}$ with $\ed(S', Q^{\infty}\fragment{0}{r + j}) = k'$,
    for some $0 < k' \le k$. By definition of~the edit distance, there is an integer $r'$
    that satisfies $S'\fragmentoc{r'}{r} = Q^{\infty}\fragmentoc{r' + j}{r + j}$ and
    at least one of~the following
    \begin{itemize}
        \item $k' = \ed(S'\fragmentco{0}{r'}, Q^{\infty}\fragmentco{0}{r' + j}) + 1$
            and $S\position{r'} \ne Q^{\infty}\position{r'}$
            (when changing the character $S\position{r'}$ to the character
            $Q^{\infty}\position{r' + j}$);
        \item $k' = \ed(S'\fragmentco{0}{r'}, Q^{\infty}\fragmentco{0}{r' + j - 1}) + 1$
            (when inserting the character $Q^{\infty}\position{r' + j}$); or
        \item $k' = \ed(S'\fragment{0}{r'}, Q^{\infty}\fragmentco{0}{r' + j}) + 1$
            (when deleting the character $Q^{\infty}\position{r' + j}$).
    \end{itemize}
    Note that (as $S'\fragmentoc{r'}{r} = Q^{\infty}\fragmentoc{r' + j}{r + j}$) we may
    assume that $r'$ is maximal (that is there is no larger integer with the same
    properties as $r'$).
    In particular, we have \[
        r' = \max(\ell^{(k-1)}_{j} + 1, \ell^{(k-1)}_{j-1} + 1, \ell^{(k-1)}_{j+1}),
    \] and hence \nxt computes $\ell^{(k)}_j$ indeed correctly.

    Using the computed values $\ell^{(k)}$, we can easily compute
    the length $|S'|$ of~the longest prefix $S'$ of~$S$ and the length $|Q'|$ of
    the corresponding prefix $Q'$ of~$Q^\infty$ such that $\ed(S', Q') \le k$:
    For $|S'|$, we have\[
        |S'| = \max_{j, r} \{ r \mid \ed(S\fragment{0}{r}, Q^{\infty}\fragment{0}{r +
        y}) \le k \} = \max_j \ell^{(k)}_j,
    \] as $k$ edits only allow for up to $k$ insertions or deletions (that is, operations
    that can change the shift between $Q^{\infty}$ and $S$). Hence, the computation
    of~$|S'|$ in \nxt is correct. For $|Q'|$, observe that if $|S'| = \ell^{(k)}_j$
    for some $j \in \fragment{-k}{k}$, by construction, we have $|Q'| = \ell^{(k)}_j + j$.
    Hence, also the computed value for $|Q'|$ is indeed correct.

    For the correctness of~\align, observe that we store information computed in
    \nxt (which is correct); further we always synchronize the information stored, hence
    also \align is correct.

    For the running time, observe that in the $k$-th call to \nxt we call \lceOpName
    three times for each of~the $2k + 3$ values~$\ell^{(k)}_i$. Further, all other
    operations are essentially book-keeping that can be implemented in $\Oh(1)$ time.
    Hence in total, the $k$-th call to \nxt takes $\Oh(k)$ time in the \modelname model.

    For the {\tt Alignment} operation, observe that we traverse the sequence $\ell^{(k)}$
    exactly once, hence  {\tt Alignment} uses $\Oh(k)$ time in the \modelname model,
    completing the proof.
\end{proof}

\begin{lemma}[{\tt Verify($P$, $T$, $k$, $I$)}, {\cite[Section 5]{ColeH98}}]\label{lm:verifye}
    Let $P$ denote a string of~length $m$, let $T$ denote a string,
    and let $k \le m$ and denote a positive integer.
    Further, let $I$ denote an interval of~positive integers. Using $\Oh(k(k + |I|))$
    \modelname operations, we can compute $\{(\ell,\min_r \ed(P,T\fragmentco{\ell}{r}))
    \mid \ell \in \OccE_k(P, T)\cap I\}$.
\end{lemma}
\begin{proof}
    Observe that the algorithm in \cite[Section 5]{ColeH98} only uses \lceOpName
    operations, as it mainly implements~\cite{LandauV89}.
    In particular, the algorithm in \cite[Section 5]{ColeH98} uses
    $\Oh(k(k + |I|))$ \lceOpName operations. The claim follows.
\end{proof}
Note that we can also call {\tt Verify($P$, $T$, $k$, $I$)} for strings $T = Q^{\infty}$
(for some primitive $Q$), as, by \cref{lm:inflcp}, we can also efficiently compute
$\lceOp{P}{Q^{\infty}}$.

\subsection{Computing Structure in the Pattern}

We proceed to discuss the implementation of~\cref{prp:EI}, that is, the analysation of the
pattern. While the algorithm itself is similar to the Hamming distance case, the analysis
requires more involved arguments. We start with an auxiliary combinatorial lemma:

\begin{lemma}\label{lem:fixed}
    Let $S$ denote a string, let $k$ denote a positive integer,
    and let $Q$ denote a primitive string such that $|Q|=1$ or $|S|\ge (2k+1)|Q|$.
    Suppose that $\edl{S}{Q}=\ed(S,Q^\infty\fragmentco{x}{y})\le k$ for integers $x\le y$.
    Then, for any string $S'$,
    if $\edl{SS'}{Q}\le k$, then $\edl{SS'}{Q}=\ed(SS',\rot^{-x}(Q)^*)$,
    and if $\edl{S'S}{Q}\le k$, then $\edl{S'S}{Q}=\eds{S'S}{\rot^{-y}(Q)}$.
\end{lemma}
\begin{proof}
    Suppose that $\edl{SS'}{Q}=\ed(SS',Q^\infty\fragmentco{x'}{z'})\le k$ for some
    integers $x'\le z'$.
    Then, there is a position $y'\in \fragment{x'}{z'}$ such that
    \[\edl{SS'}{Q} = \ed(S,Q^\infty\fragmentco{x'}{y'}) +
    \ed(S',Q^\infty\fragmentco{y'}{z'})\le k.\]
    Due to $\ed(S,Q^\infty\fragmentco{x}{y})\le k$ and $\ed(S,Q^\infty\fragmentco{x'}{y'})\le k$,
    we may apply \cref{fct:split},
    which yields a decomposition $S=S_L\cdot S_R$ and integers $j,j'$ such that
    \begin{align*}
        \ed(S,Q^\infty\fragmentco{x\phantom{'}}{\phantom{'}y})
        & =\ed(S_L,Q^\infty\fragmentco{x\phantom{'}}{\phantom{'}j|Q|})+\ed(S_R,
        Q^\infty\fragmentco{j|Q|\phantom{'}}{\phantom{'}y})\quad\text{and}\\
        \ed(S,Q^\infty\fragmentco{x'}{y'})
        &=\ed(S_L,Q^\infty\fragmentco{x'}{j'|Q|})+\ed(S_R,Q^\infty\fragmentco{j'|Q|}{y'}).
    \end{align*}
    In particular, we have
    \begin{multline*}
        \ed(S_L,Q^\infty\fragmentco{x}{j|Q|})+\ed(S_R,
        Q^\infty\fragmentco{j|Q|}{y}) = \ed(S,Q^\infty\fragmentco{x}{y}) = \edl{S}{Q}
        \\ \le \ed(S,Q^\infty\fragmentco{x'}{y+(j'-j)|Q|})\le
        \ed(S_L,Q^\infty\fragmentco{x'}{j'|Q|})+\ed(S_R,
        Q^\infty\fragmentco{j|Q|}{y}),
    \end{multline*}
    which implies $\ed(S_L,Q^\infty\fragmentco{x}{j|Q|})\le
    \ed(S_L,Q^\infty\fragmentco{x'}{j'|Q|})$.
    Consequently,
    \begin{align*}
        \ed(SS',\rot^{-x}(Q)^*)&\le\ed(SS',Q^\infty\fragmentco{x}{z'+(j-j')|Q|})\\
        &\le \ed(S_L,Q^\infty\fragmentco{x\phantom{'}}{\phantom{'}j|Q|})+
        \ed(S_R,Q^\infty\fragmentco{j'|Q|}{y'})+ \ed(S',Q^\infty\fragmentco{y'}{z'})\\
        &\le \ed(S_L,Q^\infty\fragmentco{x'}{j'|Q|})+
        \ed(S_R,Q^\infty\fragmentco{j'|Q|}{y'})+ \ed(S',Q^\infty\fragmentco{y'}{z'})\\
        &=\edl{SS'}{Q},\end{align*}
        which implies $\edl{SS'}{Q}=\ed(SS',\rot^{-x}(Q)^*)$.

    The claim regarding $\edl{S'S}{Q}$ and $\eds{S'S}{\rot^{-y}(Q)}$ is symmetric.
\end{proof}

\begin{algorithm}[t!]
    \SetKwBlock{Begin}{}{end}
    \SetKwFunction{anly}{Analyze}
    \SetKwRepeat{Do}{do}{while}
    \anly{$P$, $k$}\Begin{
        $j\gets 0$; $r \gets 1$; $b \gets 1$\;
        \While{\bf true}{
            $j' \gets j+\floor{m/\betav k}$\; 
            \If{$\perOp{P\fragmentco{j}{j'}} > m/\alphav k$}{
                $B_b \gets P\fragmentco{j}{j'}$\;
                \lIf{$b = 2k$}{\Return{breaks $B_1,\ldots, B_{2k}$}}
                $b \gets b+1$; $j \gets j'$\;
                }\Else{
                $q \gets \perOp{P\fragmentco{j}{j'}}$;
                $Q_r \gets P\fragmentco{j}{j+q}$\;
                generator $\mathbf{G} \gets \errgen{$P\fragmentco{j}{m}$, $Q_r$}$\;
                $\delta \gets 0$\;
                \While{$\delta < \betav k/m\cdot (j'-j)$ \KwSty{and} $j' \le m$}{
                    $(\pi, \pi') \gets \nxt{$\mathbf{G}$}$\;
                    $j'\gets j+\pi+1$; $\delta \gets \delta + 1$\;
                }
                \If{$j' \le m$}{
                    $R_r \gets P\fragmentco{j}{j'}$\;
                    \If{$\sum_{i=1}^r |R_i|\ge \deltavN/\deltavD \cdot  m$}{
                        \Return{repetitive regions $R_1,\ldots,R_r$ with periods
                        $Q_1,\ldots, Q_r$}\;
                    }
                    $r \gets r+1$;  $j \gets j'$\;
                    }\Else{
                    $Q \gets Q_r$\;
                    generator ${\mathbf{G}'} \gets \errgenR{$P$, $\rot^{-\pi'}(Q)$}$\;
                    $j''=m$; $\delta \gets 0$\;
                    \While{$(j''\ge j$ \KwSty{or} $\delta < \betav k/m\cdot (m-j''))$ \KwSty{and}
                        $j'' \ge 0$}{
                        $(\pi,\_) \gets \nxt{${\mathbf{G}'}$}$\;
                        $j''\gets m-\pi - 1$; $\delta \gets \delta + 1$\;
                    }
                    \lIf{$j'' \ge 0$}{%
                        \Return{repetitive region $P\fragmentco{j''}{m}$
                        with period $Q$}
                    }
                    \lElse{\Return{approximate period $Q$}}
                }
            }
        }
    }
    \caption{Analyzing the pattern: A \modelname model implementation
    of~\cref{alg:E1}.}\label{alg:EiP1}
\end{algorithm}
\begin{lemma}[{\tt Analyze($P$, $k$)}: Implementation of~\cref{prp:EI}]\label{prp:EIalg}
    Let $P$ denote a string of~length $m$ and let $k \le m$ denote a positive integer.
    Then, there is an algorithm that computes one of~the following:
    \begin{enumerate}[(a)]
        \item $2k$ disjoint breaks $B_1,\ldots, B_{2k} \substr P$,
            each having period $\per(B_i)> m/\alphav k$ and length $|B_i| = \lfloor
            m/\betav k\rfloor$.
        \item Disjoint repetitive regions $R_1,\ldots, R_{r} \substr P$
            of~total length $\sum_{i=1}^r |R_i| \ge \deltavN/\deltavD \cdot m$ such
            that each region~$R_i$ satisfies
            $|R_i| \ge m/\betav k$ and is constructed along with a primitive approximate
            period $Q_i$
            such that $|Q_i| \le m/\alphav k$ and $\edl{R_i}{Q_i} = \ceil{\betav k/m\cdot
            |R_i|}$.
        \item A primitive approximate period $Q$ of~$P$
            with $|Q|\le m/\alphav k$ and $\edl{P}{Q} < \betav k$.
    \end{enumerate}
    \noindent The algorithm uses $\Oh(k^2)$ time plus $\Oh(k^2)$ \modelname operations.
\end{lemma}
\begin{proof}
    In a similar manner to the Hamming distance case,
    our implementation follows \cref{alg:E1} from the proof~of~\cref{prp:EI}.

    Recall that $P$ is processed from left to right and split into breaks and repetitive regions.
    In each iteration, the algorithm first considers a fragment of~length $\floor{m/\betav k}$.
    This fragment either becomes the next break (if its shortest period is long enough)
    or is extended to the right to a repetitive region (otherwise).
    Having constructed sufficiently many breaks or repetitive regions of~sufficiently large
    total length, the algorithm stops. Processing the string $P$ in this manner
    guarantees disjointness of~breaks and repetitive regions.
    As in the proof~of~\cref{prp:EI}, a slightly different approach is needed if the
    algorithm encounters the end of~$P$ while growing a repetitive region. If this
    happens, the region is also extended to the left.
    This way, the algorithm either obtains a single repetitive region (which is not
    necessarily disjoint with the previously created ones, so it is returned on its own)
    or learns that the whole string $P$ is approximately periodic.

    Next, we fill in missing details of~the implementation of~the previous steps in the
    \modelname model.
    To that end, first note that the \modelname model includes a \perOpName operation of
    checking if the period of~a string~$S$ satisfies $\per(S)\le |S|/2$; computing $\per(S)$
    in case of~a positive answer. Since our threshold $m/\alphav k$ satisfies $\floor{m/\alphav k}
    \le \floor{m/\betav k}/2$, no specific work is required to obtain the period of~an
    unprocessed fragment of~$\floor{m/\betav k}$ characters of~$P$.

    To compute a repetitive region starting from a fragment $P\fragmentco{j}{j'}$ with
    string period $Q$, we use a generator $\mathbf{G}=\errgen{$P\fragmentco{j}{m}$, $Q$}$
    from \cref{lm:misope}: for subsequent values $\delta \ge 1$, we
    find the shortest prefix $P'_\delta$ of~$P\fragmentco{j}{m}$ such that
    $\ed(P'_\delta,Q^*)=\delta$,
    until no such prefix exists or $\delta\ge \betav k/m\cdot |P'_\delta|$.
    This is possible because the $(\delta-1)$-st call to \nxt{$\mathbf{G}$}
    returns the length $\pi$ of~the longest prefix of~$P\fragmentco{j}{m}$ with
    $\ed(P\fragmentco{j}{j+\pi},Q^*)<\delta$.
    If $\delta\ge \betav k/m\cdot |P'_\delta|$, then we have identified a repetitive
    region $P'_\delta$.
    Otherwise, we reach  $\pi=m-j$ and retrieve $\pi'$ such that
    $\ed(P\fragmentco{j}{m},Q^*)=\ed(P\fragmentco{j}{m},Q^\infty\fragmentco{0}{\pi'})$
    from the last call to  \nxt{$\mathbf{G}$}.
    In this case, we similarly use a generator $\mathbf{G}'=\errgenR{$P$,
    $\rot^{-\pi'}(Q)$}$ from \cref{lm:misope}:
    For subsequent values $\delta \ge 1$,
    we find the shortest suffix $P''_\delta$ of~$P$ such that
    $\eds{P''_\delta}{\rot^{-\pi'}(Q)}=\delta$,
    until no such suffix exists or $|P''_{\delta}|\ge |P\fragmentco{j}{m}|$ and $\delta\ge
    \betav k/m\cdot |P''_\delta|$.
    Again, this is possible because the $(\delta-1)$-st call to \nxt{$\mathbf{G'}$}
    returns the length $\pi$ of~the longest suffix of~$P$ with
    $\eds{P\fragmentco{m-\pi}{m}}{\rot^{-\pi'}(Q)}<\delta$.
    If we reach $\pi=m$, then we return $Q$ as an approximate period of~$P$;
    otherwise, we return the final suffix $P''_\delta$ as a long repetitive region.
    Consider \cref{alg:EiP1} for implementation details.

    For the correctness, since our algorithm follows the proof~of~\cref{prp:EI}, we only
    need to show that our implementation of~finding repetitive regions correctly
    implements the corresponding steps in \cref{alg:E1}.

    First, we inductively prove that each considered prefix $P'_\delta$
    of~$P\fragmentco{j}{m}$ satisfies $\edl{P'_\delta}{Q}=\delta \le \lceil \betav k/m
    \cdot |P'_{\delta}| \rceil$.
    The case of~$\delta=1$ is easy since $\edl{P'_1}{Q}\le\ed(P'_1,Q^*) = 1$,
    since $\edl{P'_1}{Q}=0$ would imply $\ed(P'_1,Q^*)=0$
    because $Q$ is a prefix of~$P'_1$,
    and since $\betav k/m \cdot |P'_{1}| > 0$ due to $|P'_{1}|>0$.
    Next, we prove that the claim holds for $\delta+1$ assuming that it holds for $\delta$.
    The inductive assumption guarantees $\edl{P'_\delta}{Q}=\ed(P'_\delta,Q^*)=\delta$.
    Since the algorithm proceeded to the next step, we have
    $\delta < \betav k/m\cdot |P'_{\delta}|$.
    In particular, $|P'_{\delta}| \ge (2\delta+1)\cdot m/\alphav k \ge (2\delta+1)|Q|$,
    so we can apply \cref{lem:fixed} to $P'_{\delta}$.
    If $\edl{P'_{\delta+1}}{Q}\le\delta$, then \cref{lem:fixed} yields
    $\ed(P'_{\delta+1},Q^*)=\edl{P'_{\delta+1}}{Q}\le \delta$,
    which contradicts the definition of~$P'_{\delta+1}$.
    Hence, $\edl{P'_{\delta+1}}{Q}\ge \delta+1$.
    Due to $\ed(P'_{\delta+1},Q^*)=\delta+1$, we have $\edl{P'_{\delta+1}}{Q}=\delta+1$.
    Moreover, $\lceil  \betav k/m\cdot |P'_{\delta+1}|\rceil \ge \lceil  \betav k/m\cdot
    |P'_{\delta}|\rceil > \delta$ guarantees $ \lceil  \betav k/m\cdot
    |P'_{\delta+1}|\rceil\ge \delta+1$,
    which completes the inductive proof.

    In particular, if we encounter a prefix $P'_{\delta}$ that satisfies $\delta \ge \betav
    k/m\cdot |P'_{\delta}|$, then $\edl{P'_\delta}{Q}=\lceil \betav k/m \cdot
    |P'_{\delta}| \rceil$.
    However, if no such prefix $P'_{\delta}$ exists,
    then $\edl{R}{Q} < \betav k/m\cdot |R|$ holds for each non-empty prefix
    of~$P\fragmentco{j}{m}$
    (because $R=P'_{\delta}$ is the shortest prefix $R$ of~$P\fragmentco{j}{m}$ with
    $\edl{R}{Q}=\delta$).
    Thus, \cref{ln:efwd} of~\cref{alg:E1} is implemented correctly.

    In the following, we assume that no such prefix $R$ exists.
    In particular, we have $\edl{P\fragmentco{j}{m}}{Q}< \betav k/m\cdot |P\fragmentco{j}{m}|$.
    Then, the last call to \nxt{$\mathbf{G}$} resulted in $(m-j,\pi')$
    such that $\ed(P\fragmentco{j}{m},Q^*)=\ed(P\fragmentco{j}{m},Q^\infty\fragmentco{0}{\pi'})$.
    Moreover, since $P'_{\delta}$ with $\delta=\ed(P\fragmentco{j}{m},Q^*)$
    satisfies $\edl{P'_{\delta}}{Q}=\ed(P'_{\delta},Q^*)=\delta$,
    we have
    \[\edl{P\fragmentco{j}{m}}{Q}=\ed(P\fragment{j}{m},Q^*)=\eds{P\fragmentco{j}{m}}{\rot^{-\pi'}(Q)}=\delta.\]

    We inductively prove that each considered suffix of~$P''_{\delta}$ of~$P$
    with $|P''_{\delta}|> j-m$ satisfies
    $\edl{P''_\delta}{Q}=\delta \le \lceil \betav k/m \cdot |P''_{\delta}| \rceil$.
    Let us prove that this claim is true for $\delta+1$ assuming that is it true for
    $\delta$ (unless $\delta=0$, when there is no assumption).
    If $\delta < \edl{P\fragmentco{j}{m}}{Q}$, then $|P''_{\delta+1}|\le j-m$ and the
    claim is void, so we only consider $\delta \ge \edl{P\fragmentco{j}{m}}{Q}$.
    If $\delta > \edl{P\fragmentco{j}{m}}{Q}=\eds{P\fragmentco{j}{m}}{\rot^{-\pi'}(Q)}$,
    then $|P''_{\delta}|>j-m$ and the inductive assumption guarantees
    $\edl{R}{Q}=\eds{R}{\rot^{-\pi'}(Q)}=\delta$ for $R=P''_{\delta}$.
    Otherwise, we have
    \[\delta=\edl{P\fragmentco{j}{m}}{Q}=\eds{P\fragmentco{j}{m}}{\rot^{-\pi'}(Q)},\]
    in which case $\edl{R}{Q}=\eds{R}{\rot^{-\pi'}(Q)}=\delta$ holds for $R=P\fragmentco{j}{m}$.
    In either case, we also have $\delta < \betav k/m\cdot |R|$,
    and hence $|R| \ge (2\delta+1)\cdot m/\alphav k \ge (2\delta+1)|Q|$.
    Therefore, we can apply \cref{lem:fixed} to $R$.
    If $\edl{P''_{\delta+1}}{Q}\le \delta$, then \cref{lem:fixed} yields
    $\eds{P''_{\delta+1}}{\rot^{-\pi'}(Q)}=\edl{P''_{\delta+1}}{Q}\le \delta$,
    which contradicts the definition of~$P''_{\delta+1}$.
    Hence, $\edl{P''_{\delta+1}}{Q}\ge \delta+1$.
    Due to $\eds{P''_{\delta+1}}{\rot^{-\pi'}(Q)}=\delta+1$, we have
    $\edl{P''_{\delta+1}}{Q}=\delta+1$.
    Moreover, $\lceil  \betav k/m\cdot |P''_{\delta+1}|\rceil \ge \lceil  \betav k/m\cdot
    |R|\rceil > \delta$ guarantees $ \lceil  \betav k/m\cdot |P''_{\delta+1}|\rceil\ge
    \delta+1$,
    which completes the inductive proof.

    In particular, if we encounter a suffix $P''_{\delta}$ that satisfies $|P''_{\delta}|>m-j$
    and $\delta \ge \betav k/m\cdot |P''_{\delta}|$,
    then $\edl{P''_\delta}{Q}=\lceil \betav k/m \cdot |P''_{\delta}| \rceil$.
    On the other hand, if no such suffix $P''_{\delta}$ exists,
    then $\edl{R}{Q} < \betav k/m\cdot |R|$ holds for each suffix $R$ of~$P$ of~length $|R|>m-j$
    (because $R=P'_{\delta}$ is the shortest suffix $R$ of~$P$ with $\edl{R}{Q}=\delta$
    assuming $\delta > \edl{P\fragmentco{j}{m}}{Q}$). Thus, \cref{ln:ebcw}
    of~\cref{alg:E1} is also implemented correctly.

    For the running time analysis, observe that each iteration of~the outer while loop
    processes at least $\floor{m/\betav k}$ characters of~$P$,
    so there are at most $\Oh(k)$ iterations of~the outer while loop.
    In each iteration, we perform one \perOpName operation, a constant number of
    \accOpName operations, and at most $\betav k/m \cdot (j'-j)$ calls to the generator
    \errgen. These calls use $\Oh((\betav k/m \cdot (j'-j))^2)$ time in the \modelname model,
    which is $\Oh(k^2)$ in total across all iterations (since the function $x\mapsto x^2$
    is convex).
    Similarly, we bound the running time of~the calls to the generator \errgenR:
    As we find at most $\betav k/m \cdot m = \betav k$ errors,
    \errgenR uses at most $\Oh(k^2)$ time.
    Overall, \cref{alg:EiP1} thus uses $\Oh(k^2)$ time in the \modelname model.
\end{proof}

\subsection{Computing Occurrences in the Periodic Case}

We start this subsection with a subroutine to compute a \emph{witness} that
a string $S$ has a small edit distance to a string $Q^{\infty}\!$.

\begin{lemma}[\texttt{FindAWitness($k$, $Q$, $S$)}]\label{lm:witness}
    Let $k$ denote a positive integer,
    let $S$ denote a string,
    and let $Q$ denote a primitive string that satisfies $|S|\ge (2k+1)|Q|$ or $|Q|\le 3k+1$.

    Then, we can be compute a \emph{witness} $Q^\infty\fragmentco{x}{y}$
    such that $\ed(S,Q^\infty\fragmentco{x}{y})=\edl{S}{Q}\le k$,
    or report that $\edl{S}{Q}>k$.
    The algorithm takes $\Oh(k^2)$ time in the \modelname model.
\end{lemma}
\begin{algorithm}[t]
    \SetKwBlock{Begin}{}{end}
    \SetKwFunction{verify}{Verify}
    \SetKwFunction{witness}{FindAWitness}
    \witness{$k$, $Q$, $S$}\Begin{
        \tcp{Compute ``correct'' rotation(s) of~$Q$}
        \lIf{$|Q|\le 3k+1$}{$J \gets \fragmentco{0}{|Q|}$}
        \Else{
        multi-set $R \gets \{\}$\;
        \For{$i\gets 0$ \KwSty{to} $2k$}{
            $R \gets R \cup \cycEqOp{S\fragmentco{i|Q|}{(i+1)|Q|}}{Q}$\;
        }
        $I \gets \bigcup\{\fragment{p}{p+k} : p\in \mathbb{Z}\text{ and }\fragment{p}{p+k}
        \text{ contains at least }k+1\text{ elements of~}R\}$\;
        Let $J\subseteq \fragmentco{0}{2|Q|}$ denote a shortest interval that satisfies $I \bmod
        |Q|\sub J\bmod |Q|$\;
        }
        \BlankLine
        \tcp{Compute the start position of~the witness}
        $Occ \gets \verify{$S$, $Q^\infty$, $k$, $J$}$\;
        \lIf{$Occ = \emptyset$}{\Return$\bot$}
        Let $(x,k') \in Occ$ be an arbitrary element minimizing $k'$\;
        \BlankLine
        \tcp{Compute the end position of~the witness}
        generator $\mathbf{G} \gets \errgen{$S$, $\rot^{-x}(Q)$}$\;
        \lFor{$i \gets 1$ \KwSty{to} $k'$}{\nxt{$\mathbf{G}$}}
        $(\lambda,\lambda')\gets \nxt{$\mathbf{G}$}$\;
        \Return{$Q^\infty\fragmentco{x}{x+\lambda'}$}\;
    }
    \caption{Finding a witness $Q^{\infty}\fragmentco{x}{y}$ for $\edl{S}{Q} \le k$.}\label{alg:witness}
\end{algorithm}
\begin{proof}
    For a set $A\subseteq \mathbb{Z}$ and an integer $p>0$, we define
    $A\bmod p := \{a \bmod p \mid a\in A\}$.
    We first compute a (short) interval $J$ such that $\OccE_k(S,Q^\infty)\bmod |Q| \sub J
    \bmod |Q|$.
    If $|Q|\le 3k+1$, then we simply set $J=\fragmentco{0}{|Q|}$.
    Otherwise, we proceed similarly as in the Hamming distance setting (\cref{lem:findAPeriod}),
    where we computed a majority string of~the first $2k + 1$ length-$|Q|$ subsequent
    fragments $S_1,\ldots,S_{2k}$ of~$S$.
    However, now we need to accommodate for insertions and deletions of~a $k$-error occurrence.
    Hence, we first compute an auxiliary
    set $I$ defined as the union of~intervals $\fragment{p}{p+k}$
    such that for at least $k + 1$ fragments $S_i$ of~$S$,
    we have $Q=\rot^{j}(S_i)$ for $j \in \fragment{p}{p+k}$.
    Finally, we set $J\sub \fragmentco{0}{2|Q|}$ to be a shortest interval satisfying
    $I\bmod |Q|\sub J\bmod |Q|$.
    Here, $J\bmod |Q|$ can be interpreted as a shortest cyclic interval (modulo $|Q|$)
    containing $I\bmod |Q|$.

    Having computed the set $J$, we use \verify
    from \cref{lm:verifye} to determine at which starting position~$x$ in $J$ we have an occurrence
    with the fewest number of~errors (or to report that the number of~errors is greater
    than $k$ everywhere).
    Finally, we use an \errgen
    from \cref{lm:misope} to compute the ending position $y$ of~the occurrence of
    $S$ as a prefix of~$Q^{\infty}\fragmentco{x}{}$.
    Consider \cref{alg:witness} for a pseudo-code implementation.

    The correctness is based on the aforementioned characterization of~$J$:
    \begin{claim}\label{cl:wt0}
        The interval $J$ satisfies
        $\OccE_k(S,Q^\infty)\bmod |Q| \subseteq J \bmod |Q|$.
    \end{claim}
    \begin{claimproof}
        The claim trivially holds if $|Q|\le 3k+1$, so we assume that $|Q|> 3k+1$.
        For every $i \in \fragment{0}{2k}$, define $S_i := S\fragmentco{i\,|Q|}{(i+1)\,|Q|}$.
        Consider an optimum alignment between $S$ and its $k$-error occurrence
        $Q^\infty\fragmentco{x}{y}$,
        and let $Q_i=Q^\infty\fragmentco{x_i}{x_{i+1}}$ denote the fragment aligned to $S_i$.
        Consider the multi-set $R:= \bigcup_i \cycEqOp{S_i}{Q}$.
        Next, consider the values $\delta_i := x_i-i|Q|$ for $i \in \fragment{0}{2k}$.
        We have $\delta_0=x$, and $\delta_{i+1}=\delta_i+|Q_i|-|S_i|$ for $i>0$.
        Since \[
            \sum_{i=0}^{2k}\big||Q_i|-|S_i|\big|\le \sum_{i=0}^{2k}\ed(Q_i,S_i)\le k,
        \]
        all values $\delta_i$ belong to an interval of~the form $\fragment{p}{p+k}$ for some integer $p$.
        Moreover, note that $Q_i=S_i$ holds for at least $k+1$ fragments $Q_i$;
        these fragments satisfy $Q=\rot^{\delta_i}(S_i)$ and thus contribute $\delta_i$ to $R$.
        We conclude that there is an interval $\fragment{p}{p+k}$ containing $x$ and at least
        $k+1$ elements of~$R$.
        Consequently, we have $x\in I$. By definition of~$J$, this yields $x\bmod |Q| \in
        J \bmod |Q|$.
    \end{claimproof}

    Now, let $Q^\infty\fragmentco{x}{y}$ denote a witness that satisfies
    $\edl{S}{Q}=\ed(S,Q^\infty\fragmentco{x}{y})$.
    By \cref{cl:wt0}, there is a matching fragment
    $Q^\infty\fragmentco{x'}{y'}=Q^\infty\fragmentco{x}{y}$
    starting at $x'\in J$. Thus, we may assume without loss of~generality that $x\in J$.
    As we verify all possible starting positions in $J$ using \verify from
    \cref{lm:verifye}, we correctly compute the starting position $x$ of~a witness occurrence of
    $S$ in $Q^{\infty}$. Further, as we use an \errgen from \cref{lm:misope}, we also
    compute the corresponding ending position correctly.

    As for the running time, we prove the following characterization of~$J$:
    \begin{claim}\label{cl:wt1}
        The interval $J$ satisfies $|J|\le 3k+1$.
    \end{claim}
    \begin{claimproof}
        The claim trivially holds if $|Q|\le 3k+1$, so we assume that $|Q|>3k+1$.
        Recall that the multiset~$R$ in \cref{alg:witness} is the union of~$2k+1$ sets
        $\{j\in \mathbb{Z} : \rot^j(S_i)=Q\}$.
        As the string $Q$ is primitive, $R$ is the union of~at most $2k+1$ infinite
        arithmetic progressions with difference $|Q|$.
        In particular, if $\fragment{p}{p+k}$ and $\fragment{p'}{p'+k}$
        contain at least $k+1$ elements of~$R$ each, then $(\fragment{p}{p+k}\bmod |Q|)
        \cap (\fragment{p'}{p'+k} \bmod |Q|) \ne \emptyset$, and thus $\fragment{p'}{p'+k}
        \bmod |Q| \sub \fragment{p-k}{p+2k}\bmod |Q|$.
        Since $J$ is the union of~such intervals $\fragment{p'}{p'+k}$,
        we have $J \bmod |Q| \sub \fragment{p-k}{p+2k}\bmod |Q|$.
        By definition of~$I$, we conclude that $|I| \le |\fragment{p-k}{p+2k}|=3k+1$.
    \end{claimproof}

    Now, observe that computing the multiset $R$ (represented as the union of~infinite
    arithmetic progressions modulo $|Q|$) takes $\Oh(k)$ time in the \modelname
    model; computing the sets $I$ and $J$ can be done in $\Oh(k \log\log k)$ time
    by sorting $R$ (restricted to $\fragmentco{0}{|Q|}$) and a subsequent cyclic scan over $R$.
    Further, by \cref{cl:wt1}, we call \verify on an interval of~length $\Oh(k)$;
    hence the call to \verify takes $\Oh(k^2)$ time in the \modelname model.
    Finally, as we query the \errgen for up to $k$ errors,
    the last step of~the algorithm takes $\Oh(k^2)$ time in the \modelname model as well.
    Hence in total, \witness runs in $\Oh(k^2)$ time in the \modelname model, completing
    the proof.
\end{proof}

\begin{lemma}[\texttt{FindRelevantFragment($P$, $T$, $k$, $d$, $Q$)}]\label{lem:erelevant}
    Let $P$ denote a pattern of~length $m$, let $T$ denote a text of~length $n$,
    and let $0 \le k\le m$ denote a threshold such that $n<\threehalfs m+k$.
    Further, let $d\ge 2k$ denote a positive integer and let
    $Q$ denote a primitive string that satisfies $|Q|\le m/8d$ and $\edl{P}{Q}\le d$.

    Then, there is an algorithm that computes a fragment $T'=T\fragmentco{\ell}{r}$ and an
    integer range~$I$ such that
    $\edl{T'}{Q}\le 3d$, $|\OccE_k(P,T)|=|\OccE_k(P,T')|$, $|I|\le 6d+1$, and
    $\OccE_k(P,T')\bmod |Q| \subseteq I \bmod |Q|$.
    The algorithm runs in $\Oh(d^2)$ time in the \modelname model,
\end{lemma}
\begin{proof}
    We start with two calls to \texttt{FindAWitness} from \cref{lem:findAPeriod}
    in order to find a fragment $Q^\infty\fragmentco{x}{y}$ such that
    $\ed(P,Q^\infty\fragmentco{x}{y})=\edl{P}{Q}\le d$
    and a fragment $Q^\infty\fragmentco{x'}{y'}$ such that \[
        \ed(T\fragmentco{n-m+k}{m-k},Q^\infty\fragmentco{x'}{y'})
        =\edl{T\fragmentco{n-m+k}{m-k}}{Q} \le \threehalfs d.
    \]
    If the latter fragment does not exist, then we return the empty string $T'=\varepsilon$
    and the empty interval $I=\emptyset$.
    Otherwise, we proceed by computing the rightmost position $r$
    such that $\ed(T\fragmentco{n-m+k}{r},\rot^{-x'}(Q)^*)\le \threehalfs d$
    and the leftmost position $\ell$ such that
    $\eds{T\fragmentco{\ell}{{m-k}}}{\rot^{-y'}(Q)}\le \threehalfs d$.
    That is, we ``extend'' the fragment found in the text as much as possible.
    Afterwards, we return the fragment $T'=T\fragmentco{\ell}{r}$
    and the interval $I=\fragment{n-m+k-\ell+x-x'-3d}{n-m+k-\ell+x-x'+3d}$.
    Consider \cref{alg:erelevant} for implementation details.

    For the correctness, note that the due to the assumption on $Q$,
    the first call to \witness\ is valid and returns a witness $Q^\infty\fragmentco{x}{y}$
    (rather than $\bot$).
    Next, consider a $k$-error occurrence $T\fragmentco{p}{q}$ of~$P$.
    By triangle inequality (\cref{Etria}), we have\[
        \ed(T\fragmentco{p}{q},Q^\infty\fragmentco{x}{y})\le
        k+\ed(P,Q^\infty\fragmentco{x}{y}))\le \threehalfs d.
    \] Due to $|T\fragmentco{p}{q}|\ge m-k$, we have $q\ge m-k$ and $p\le n-m+k$,
    which yields \[
        \ed(T\fragmentco{n-m+k}{m-k},Q^\infty\fragmentco{x''}{y''})\le \threehalfs d
    \] for some integers $x'',y''$ with $x\le x'' \le y'' \le y$.
    Moreover, as in the proof~of~\cref{lem:Eaux}, we have
    $|T\fragmentco{n-m+k}{m-k}|=2(m-k)-n\ge(3d + 1)|Q|$ or $|Q|=1$,
    so the second call to \witness\ is valid.
    Thus, if the call returns $\bot$, then $\OccE_k(P,T)=\emptyset$.

    We henceforth assume that the call returned a witness $Q^\infty\fragmentco{x'}{y'}$.
    Next, we apply \cref{lem:fixed} for $S=T\fragmentco{n-m+k}{m-k}$. This is indeed
    possible because  $|S|\ge(3d + 1)|Q|$ or $|Q|=1$.
    Due to \[
        \edl{T\fragmentco{n-m+k}{q}}{Q}\le
        \ed(T\fragmentco{n-m+k}{q},Q^\infty\fragmentco{x''}{y})\le \threehalfs d,
    \]\cref{lem:fixed} yields
    \[
        \ed(T\fragmentco{n-m+k}{q},\rot^{-x'}(Q)^*)=\edl{T\fragmentco{n-m+k}{q}}{Q}\le
    \threehalfs d.\]
    Consequently, we have $q \le r$, because $r$ is computed correctly using \errgen from
    \cref{lm:misope}.
    Symmetrically, due to \[
        \edl{T\fragmentco{p}{m+k}}{Q}\le
    \ed(T\fragmentco{p}{m+k},Q^\infty\fragmentco{x}{y''})\le \threehalfs d,\]
    \cref{lem:fixed} yields
    \[\eds{T\fragmentco{p}{m-k}}{\rot^{-y'}(Q)}=\edl{T\fragmentco{p}{m+k}}{Q}\le
    \threehalfs d.\]
    Consequently, we have $p \ge \ell$, because $\ell$ is computed correctly using
    \errgenR from \cref{lm:misope}.
    We conclude that $T\fragmentco{p}{q}$ is contained in $T'=T\fragmentco{\ell}{r}$.
    Since $T\fragmentco{p}{q}$ was an arbitrary $k$-error occurrence of~$P$ in $T$,
    this implies $|\OccE_k(P,T')|=|\OccE_k(P,T)|$.

    Now consider the fragment $T\fragmentco{p}{m-k}$ whose prefix $T\fragmentco{p}{m-k}$
    satisfies \[\ed(T\fragmentco{p}{m-k},Q^\infty\fragmentco{x}{y''})\le  \threehalfs d\]
    and whose suffix $T\fragmentco{n-m+k}{m-k}$ satisfies
    \[\ed(T\fragmentco{n-m+k}{m-k},Q^\infty\fragmentco{x'}{y'})\le  \threehalfs d.\]
    We apply \cref{lem:synchr} to $T\fragmentco{p}{m-k}$;
    this is indeed possible because $|T\fragmentco{n-m+k}{m-k}|\ge(3d + 1)|Q|$ or $|Q|=1$.
    \Cref{lem:synchr} now implies $(n-m+k-p+x-x'+3d)\bmod |Q| \le 6d$.
    Consequently, we have \[
        p\bmod |Q| \in \fragment{n-m+k+x-x'-3d}{n-m+k+x-x'+3d}\bmod |Q|.
    \]
    Since $T\fragmentco{p}{q}$ was an arbitrary $k$-error occurrence of~$P$ in $T$,
    this implies \[\OccE_k(P,T')\bmod |Q| \sub I \bmod |Q|.\]
    Moreover, $|I|\le 6d+1$ holds trivially by construction.

    Now consider the fragment $T$,
    whose prefix $T\fragmentco{\ell}{m-k}$ satisfies
    \[\edl{T\fragmentco{\ell}{m-k}}{Q}\le \threehalfs d\]
    and whose suffix $T\fragmentco{n-m+k}{r}$ satisfies
    \[\edl{T\fragmentco{n-m+k}{r}}{Q}\le \threehalfs d.\]
    Again, we apply \cref{lem:synchr} to $T'$;
    this is indeed possible because
    $|T\fragmentco{n-m+k}{m-k}|\ge(3d + 1)|Q|$ or $|Q|=1$.
    \Cref{lem:synchr} now implies $\edl{T'}{Q}\le 3d$, as claimed.

    As for the running time in the \modelname model, the calls to
    \witness\ use $\Oh(d^2)$ time;
    the same is true for the usage of~\errgen\ and \errgenR.
    Thus, the algorithm takes $\Oh(d^2)$ time in the \modelname model.
\end{proof}
\begin{algorithm}[t!]
    \SetKwBlock{Begin}{}{end}
    \SetKwFunction{RFR}{FindRelevantFragment}
    \RFR{$P$, $T$, $k$, $d$, $Q$}\Begin{
        $Q^\infty\fragmentco{x}{y}\gets \witness{$d$, $Q$, $P$}$\;
        $Q' \gets \witness{$\floor{\threehalfs d}$, $Q$, $T\fragmentco{n-m+k}{m-k}$}$\;
        \lIf{$Q'=\bot$}{\Return{$(\eps,\emptyset)$}}
        $Q^\infty\fragmentco{x'}{y'}\gets Q'$\;
        \BlankLine
        \tcp{Extend $Q'$ as much as possible to the right.}
        generator $\mathbf{G} \gets \errgen{$T\fragmentco{n-m+k}{n}$, $\rot^{-x'}(Q)$}$\;
        \lFor{$i \gets 0$ \KwSty{to} $\floor{\threehalfs d}$}{%
            $(\lambda,\_)\gets \nxt{$\mathbf{G}$}$%
        }
        $r \gets n-m+k+\lambda$\;
        \BlankLine
        \tcp{Extend $Q'$ as much as possible to the left.}
        generator $\mathbf{G'} \gets \errgenR{$T\fragmentco{0}{m-k}$, $\rot^{-y'}(Q)$}$\;
        \lFor{$i \gets 0$ \KwSty{to} $\floor{\threehalfs d}$}{%
            $(\lambda',\_)\gets \nxt{$\mathbf{G'}$}$%
        }
        $\ell \gets m-k-\lambda'$\;
        \Return{$(T\fragmentco{\ell}{r},\fragment{n-m+k-\ell+x-x'-3d}{n-m+k-\ell+x-x'+3d})$}\;
    }
    \caption{A \modelname algorithm computing a \emph{relevant} fragment $T'$ of~$T$
    containing all $k$-error  occurrences of~$P$ in $T$,
    and an interval $I$ such that  $\OccE_k(P,T')\bmod |Q| \subseteq I \bmod
|Q|$.}\label{alg:erelevant}
\end{algorithm}

\begin{lemma}[{\tt Locked($S$, $Q$, $d$, $k$)}: Implementation of~\cref{lem:klocked}]\label{lem:alg_locked}
    Let $S$ denote a string, let $Q$ denote a primitive string,
    let $d$ denote a positive integer such that $\edl{S}{Q}\le d$
    and $|S| \ge (2d+1)|Q|$, and let $k$ denote a non-negative integer.

    Then, there is an algorithm that computes
    disjoint locked fragments $L_1,\ldots,L_{\ell} \preceq S$
    such that $L_1$ is a $k$-locked prefix of~$S$, $L_\ell$ is a suffix of~$S$, and
    $\edl{L_i}{Q} > 0$ for $1<i<\ell$. Moreover, we have \[
        \edl{S}{Q}=\sum_{i=1}^{\ell} \edl{L_i}{Q}\quad
        \text{and}\quad\sum_{i=1}^{\ell}|L_i| \le (5|Q|+1)d+2(k+1)|Q|.\]
    The algorithm takes $\Oh(d^2+k)$ time in the \modelname model.
\end{lemma}

\begin{algorithm}[p]
    \SetKwBlock{Begin}{}{end}
    \SetKwFunction{Locked}{Locked}
    \SetKwFunction{push}{push}
    \SetKwFunction{pop}{pop}
    \SetKwFunction{front}{front}
    \SetKwFunction{tops}{top}
    \SetKwRepeat{Do}{do}{while}
    \Locked{$S$, $Q$, $d$, $k$}\Begin{
        $Q^\infty\fragmentco{x}{y} \gets \witness{$d$, $Q$, $S$}$\;
        generator $\mathbf{G} \gets \errgen{$S$, $\rot^{-x}(Q)$}$\;
        \lDo{$\pi < |S|$}{%
            $(\pi,\pi')\gets \nxt{$\mathbf{G}$}$%
        }
        $A \gets \align{$\mathbf{G}$}$\;
        $\ell_Q \gets x$; $\ell_S \gets 0$\;
        $r_Q \gets |Q|\lceil{x/|Q|\rceil}$; $r_S \gets r_Q - \ell_Q$\;
        $\Delta \gets k+1$\;
        queue $F$\;
        \ForEach{$(s,q)\in A\cup (\pi,\pi')$}{\label{ln:locked_for}
            \lIf{$s=\bot$}{$s \gets q+x+r_S - r_Q-1$}
            \lIf{$q=\bot$}{$q \gets s-x+r_Q - r_S-1$}
            \If{$x+q  \ge r_Q$}{
                \push{$F$, $(S\fragmentco{\ell_S}{r_S},\Delta)$}\;
                $\ell_Q \gets |Q|\lfloor{(x+q)/|Q|\rfloor}$; $\ell_S \gets r_S+\ell_Q-r_Q$\;
                $r_Q \gets \ell_Q+|Q|$; $r_S \gets r_S+|Q|$\;
                $\Delta \gets 0$\;
            }
            $r_S \gets r_Q - x + s-q$\;
            \lIf{$(s,q)\ne (\pi,\pi')$}{$\Delta \gets \Delta+1$}
        }
        \push{$F$, $(S\fragmentco{\ell_S}{|S|},\Delta)$}\;
        stack $L$\;
        \While{$F$ is not empty}{
            $(S\fragmentco{\ell}{r},\Delta) \gets \front(F)$;
            \pop{$F$}\;
            \While{\KwSty{true}}{\label{ln:locked_while}
            \If{$\tops(L)=S\fragmentco{\ell'}{r'}$ \KwSty{and} $r'=\ell$}{
                $\ell \gets \ell'$\;
                \pop{$L$}\;
            }
            \ElseIf{$\front(F)=(S\fragmentco{\ell'}{r'},\Delta')$ \KwSty{and} $\ell'=r$}{
                $r \gets r'$\;
                $\Delta \gets \Delta+\Delta'$\;
                \pop{$F$}\;
            }
            \ElseIf{$\Delta > 0$}{
                $\ell \gets \max(0, \ell-|Q|)$\;
                $r \gets \min(|S|,r+|Q|)$\;
                $\Delta \gets \Delta-1$\;
            }
            \Else{
                \push{$L$, $S\fragmentco{\ell}{r}$}\;
                \KwSty{break}\;
            }
        }
    }
    \Return{$L$}\;
}
\caption{Computing locked fragments in a string $S$.}\label{alg:locked}
\end{algorithm}
\begin{proof}
    We implement the process from in the proof of \cref{lem:klocked}.
    Our algorithm is described below; see also \cref{alg:locked} for implementation details.

    First, we construct an optimal alignment between $S$ and a substring of~$Q^\infty$.
    For this, we first use \witness\ of~\cref{lm:witness}
    to obtain positions $x\le y$ such that $\edl{S}{Q}=\ed(S,Q^\infty\fragmentco{x}{y})$.
    Then, we apply a generator \errgen($S$, $\rot^{-x}(Q)$) of~\cref{lm:misope}
    to construct an optimal alignment $A$ between $S$ and $Q^\infty\fragmentco{x}{x+\pi'}$
    for some integer $\pi'\ge 0$ (note that we cannot guarantee $y=x+\pi'$).

    Then, based on the alignment $A$, we construct a decomposition $S=S_0^{(0)}\cdots
    S^{(0)}_{s^{(0)}}$
    such that $S^{(0)}_i$ is aligned against \[
        Q^{(0)}_i :=
        Q^\infty\fragmentco{\max(x,(|Q|-1)\lceil{x/|Q|\rceil})}{\min(|Q|\lceil{x/|Q|\rceil},x+\pi)}
    \] in the decomposition $A$, and a sequence $\Delta^{(0)}_i$ such that we have
     $\Delta^{(0)}_i=\ed(S^{(0)}_i,Q^{(0)}_i)$ for $i>0$ and
        $\Delta^{(0)}_0=\ed(S^{(0)}_0,Q^{(0)}_0)+k+1.$
    Since this sequence might be long, we only generate \emph{interesting} fragments $S^{(0)}_i$
    and store them, along with the values $\Delta^{(0)}_i$, in a queue $F$ in
    left-to-right order.
    (Recall that $S^{(t)}_i$ is interesting if $i=0$, $i=s^{(t)}$, $S_i^{(t)}\ne Q$, or
    $\Delta_i^{(t)}>0$.)

    The process of constructing the interesting fragments $S^{(0)}_i$ is somewhat tedious.
    We maintain a fragment $Q^\infty\fragmentco{\ell_Q}{r_Q}$, interpreted as $Q^{(0)}_i$
    for increasing values of $i$, a fragment $S\fragmentco{\ell_S}{r_S}$, interpreted as a
    candidate for $S^{(0)}_i$,
    and an integer $\Delta$, interpreted as $\Delta^{(0)}_i$.
    They are initialized to $Q^\infty\fragmentco{x}{|Q|\lceil{x/|Q|\rceil}}$,
    $S\fragmentco{0}{Q|\lceil{x/|Q|\rceil}-x}$, and $k+1$, respectively.

    Next, we process pairs $(s,q)$ corresponding to subsequent errors in the alignment $A$.
    The interpretation of the $j$-th pair $(s,q)$ is that $S\fragmentco{0}{s}$ is aligned
    with $Q^\infty\fragmentco{x}{x+q}$ with $j$ errors so that the $j$-th error is a
    substitution of $S\position{s}$ into $Q^\infty\position{x+q}$, and insertion of
    $Q^\infty\position{x+q}$, or a deletion of $S\position{s}$.

    The first step of processing $(s,q)$ is only performed if
    $Q^\infty\fragmentco{x}{x+q}$ is not (yet)
    contained in $Q^\infty\fragmentco{\ell_Q}{r_Q}$.
    If this is not the case, then we push $S\fragmentco{\ell_S}{r_S}$ with budget $\Delta$
    to the queue $F$ of interesting fragments, and we update the maintained data:
    The fragment $Q^\infty\fragmentco{\ell_Q}{r_Q}$ is set to be the fragment of $Q^\infty$
    matching $Q$ and containing $Q^\infty\position{x+q}$;
    between the previous and the current value of $Q^\infty\fragmentco{\ell_Q}{r_Q}$,
    there are zero or more copies of $Q$ aligned in $A$ without error. Hence, we skip the
    same number of copies of $Q$ in $S$
    (these are the uninteresting fragments $S^{(0)}_i$)
    and set $S\fragmentco{\ell_S}{r_S}$ to be the subsequent fragment of length $|Q|$.
    Finally, the budget $\Delta$ is reset to $0$.

    In the second step, we update $r_S$ according to the type of the currently processed error:
    We increment~$r_S$ in case of deletion of $S\position{s}$ and
    we decrement $r_S$  in case of insertion of
    $Q^\infty\position{x+q}$. This way, we guarantee that
    $|S\fragmentoo{s}{r_S}|=|Q^\infty\fragmentoo{x+q}{r_Q}|$,
    and that $A$ aligns $S\fragmentco{\ell_S}{r_S}$ with
    $Q^\infty\fragmentco{\ell_Q}{r_Q}$ provided
    that we have already processed all errors involving $Q^\infty\fragmentco{\ell_Q}{r_Q}$.
    Additionally, we increase $\Delta$ to acknowledge the currently processed error between
    $S\fragmentco{\ell_S}{r_S}$ and $Q^\infty\fragmentco{\ell_Q}{r_Q}$.

    In a similar way, we process $(s,q)=(|S|,\pi')$, interpreting it as extra substitution.
    This time, however, we do not increase $\Delta$ (because this is a not a real error).
    Finally, we push $S\fragmentco{\ell_S}{|S|}=S^{(0)}_{s^{(0)}}$ with budget~$\Delta$ to
    the queue $F$.

    In the second phase of the algorithm, we transform the decomposition
    $S=S_0^{(0)}\cdots S^{(0)}_{s^{(0)}}$ and the sequence $\Delta_0^{(0)}\cdots
    \Delta^{(0)}_{s^{(0)}}$
    using the four types of merge operations described in the proof of \cref{lem:locked}.

    We maintain an invariant that a stack $L$ contains already processed interesting fragments,
    all with budget equal to $0$, in left-to-right order (so that $\tops(L)$
    represents the rightmost one), while $F$ contains fragments that have not been
    processed yet (and may have positive budgets) also in the left-to-right order (so that
    $\front(F)$ represents the leftmost one).
    Additionally, the currently processed fragment $S\fragmentco{\ell}{r}$ is guaranteed
    to be to the right of all fragments in $L$ and to the left of all fragments in $F$.
    The fragments in $L$, the fragment $S\fragmentco{\ell}{r}$, and the fragments in $F$
    form the sequence of all interesting fragments in the current decomposition
    $S=S_0^{(t)}\cdots S^{(t)}_{s^{(t)}}$.

    In each iteration of the main loop, we pop the front fragment $S\fragmentco{\ell}{r}$
    with budget $\Delta$ from the queue~$F$ and exhaustively perform merge operations
    involving it:
    We first try applying a type-\ref{it:type1} merge with the fragment to the left (which
    must be $\tops(L)$).
    If this is not possible, we type applying a type-\ref{it:type1} merge with the
    fragment to the right (which must be $\front(F)$).
    If also this is not possible, then $S\fragmentco{\ell}{r}$ is surrounded by
    uninteresting fragments.
    In this case, we perform a type-\ref{it:type2}, type-\ref{it:type3}, or \ref{it:type4}
    merge provided that $\Delta > 0$. Otherwise, we push $S\fragmentco{\ell}{r}$ to $L$
    and proceed to the next iteration.

    Finally, the algorithm returns the sequence of (locked) fragments represented in the stack $L$.

    The correctness of the algorithm follows from \cref{lem:klocked}; no deep insight is needed
    to prove that our implementation indeed follows the procedure described in the proof
    of \cref{lem:locked} and extended in the proof of \cref{lem:klocked}.

    For the running time, the initial call to \witness and applying the generator
    $\mathbf{G}$    each take  $\Oh(d^2)$ time in the \modelname model.
    As the alignment $A$ is of size $|A|\le d$, the \textbf{for} loop in \cref{ln:locked_for}
    takes $\Oh(d)$ time and generates $\Oh(d)$ interesting locked fragments with total
    budget $\Oh(d+k)$.
    Each iteration of the \textbf{while} loop in \cref{ln:locked_while} decreases
    the number of interesting locked fragments or their total budget, so there are
    at most $\Oh(d + k)$ iterations in total. Overall the algorithm runs in $\Oh(d^2+k)$ time
    in the \modelname model.
\end{proof}

\begin{lemma}[\texttt{SynchedMatches($P$, $T$, $I$, $d$, $d'$, $k$, $Q$)}]\label{lem:sth}
    Let $P$ denote a pattern of~length $m$, let $0 \le k \le m$ denote a threshold,
    and let $T$ denote a text of~length $n \le  \threehalfs m + k$.
    Further, let $I$ denote an integer range
    and let $Q$ denote a primitive string that satisfies $\edl{P}{Q}\le d$ and $\edl{T}{Q}\le d'$.

    There is an algorithm that computes the set $\OccE_k(P,T)\cap
    (I+|Q|\mathbb{Z})$ as $\Oh(|I|d'(d+k))$ arithmetic progressions.
    The algorithm takes $\Oh(kd'(d+k)(k+|I|+d+d'))$ time in the \modelname model.
\end{lemma}
\begin{proof}
    The algorithm resembles the proof~of~\cref{lem:Eaux}\eqref{it:Eprog}.
    Consult \cref{alg:sth} for a visualization of~the algorithm as pseudo-code; in the
    interest of~readability we
    use $\OccE_k(P,T)$ instead of~$\OccE_k(P,T)\cap (I+|Q|\mathbb{Z})$ in the pseudo-code.
    Note that if $|I|>|Q|$ we can replace $I$ with $\fragmentco{0}{Q}$, since in this case
    $I+|Q|\mathbb{Z}=\mathbb{Z}=\fragmentco{0}{Q}+|Q|\mathbb{Z}$.
    Hence, we can assume that $|I|\leq |Q|$.

    We first compute $\mathcal{L}^P:=\Locked{$P$, $Q$, $d$, $k$}$ and $\mathcal{L}^T:=
    \Locked{$T$, $Q$, $d'$, $0$}$ using \cref{alg:locked}.
    We have $\ell^P:=|\mathcal{L}^P| \le \edl{P}{Q}+2\le d+2$ and $\ell^T=|\mathcal{L}^T|
    \le \edl{T}{Q}\le d'+2$.

    Then, for each of~the $\cO(dd')$ pairs of~locked fragments
    $L_i^P=P\fragmentco{\ell}{r}\in \mathcal{L}^P$ and
    $L_j^T=T\fragmentco{\ell'}{r'} \in \mathcal{L}^T$ we (implicitly) mark the positions
    in the interval $\fragmentco{\ell'-r-k}{r'-\ell+k}$.
    We also mark all positions in $\fragment{n-m-k}{n-m+k}$.
    We decompose the set of~marked positions $M$ into $\cO(dd')$ maximal ranges $J \subseteq M$.
    For each such maximal range $J$, for each maximal range $J'\subseteq J\cap (I+|Q|\mathbb{Z})$,
    we call $\verify{$P$, $T$, $k$, $J'$}$ and add its output to $\OccE_k(P,T)\cap
    (I+|Q|\mathbb{Z})$.
    This guarantees that we correctly compute all elements of~the set $\OccE_k(P,T)\cap
    (I+|Q|\mathbb{Z}) \cap M$.

    The decomposition of~$M$ into maximal ranges yields a decomposition
    of~$\fragmentco{0}{n-m+k}\setminus M$ into $\cO(dd')$ maximal ranges.
    For each such maximal range $J$, we rely on the characterization of~\cref{cl:progr_un}
    in order to compute $\OccE_k(P,T)\cap (I+|Q|\mathbb{Z}) \cap J$.
    Recall that for $p,p' \in J$ with $p \equiv p' \pmod{|Q|}$ we have $p \in
    \OccE_k(P,T)$ if and only if $p' \in \OccE_k(P,T)$.
    Hence, it suffices to restrict our attention to the intersection of~the first (at
    most) $|Q|$ positions of~$J$ with $I+|Q|\mathbb{Z}$.
    This intersection consists of at most two intervals of~total size at most $|I|$.
    We call \verify for each of~them, and for each position returned by these \verify
    queries, we add an arithmetic progression to $\OccE_k(P,T)\cap (I+|Q|\mathbb{Z})$.

    We now proceed to analyze the time complexity of~the algorithm in the \modelname model.
    The two calls to \Locked require $\cO(d^2+k+d'^2)$ time in total in the \modelname
    model due to \cref{lem:alg_locked}.
    We then decompose $M$ into maximal ranges, which can be implemented in
    $\Oh(dd'\log\log(dd'))$ time.
    The interval~$R$ of~positions marked due to locked regions $L_i^P$ and $L_j^T$ is of~size
    $|L_i^P|+|L_j^T|+2k-1$;
    the number of~maximal ranges $R' \subseteq R \cap (I+|Q|\mathbb{Z})$ is at most
    $(|R|+2|Q|-2|I|)/|Q|$.
    Consequently, the total number of maximal ranges of size at most $|I|$ that we need
    to consider intervals does not exceed $2k+1$ plus
    \begin{align*}
        \sum_{i=1}^{\ell^P}\sum_{j=1}^{\ell^T}\frac{|L_i^P|+|L_j^T|+2|Q|-2|I|}{|Q|}
         &
         \le \frac{\ell^T
         }{|Q|}\sum_{i=1}^{\ell^P}|L_i^P|+\frac{\ell^P}{|Q|}\sum_{i=1}^{\ell^T}|L_i^T|+
         \frac{2\ell^P \ell^T(|Q|-|I|)}{|Q|} \\
         &=\cO((d'(d|Q|+k|Q|)+dd'|Q|+dd'|Q|)/|Q|) \\
         &=\cO(d'(d+k)).
    \end{align*}
    Each call to \verify in \cref{ln:maxrangever} of~\cref{alg:sth} requires time
    $\cO(k(k+|J'|))$ by~\cref{lm:verifye}.
    By the above analysis, we make $\cO(d'(d+k))$ calls to \verify, each time for an
    interval of~size at most~$|I|$.
    Hence, we can upper bound the overall running time for this step by $\cO(d'(d+k) k(k+|I|))$.
    Finally, the total time required by \verify queries in~\cref{ln:maxrangever2}
    of~\cref{alg:sth} is $\cO(dd'k(k+|I|))$
    as we call \verify $\cO(dd')$ times, each time for an interval of~size $\cO(|I|)$.
    Thus, the overall running time is $\cO(d'(d+k)
    k(k+|I|)+d^2+d'^2+dd'\log\log(dd'))=\cO(kd'(d+k)(k+|I|+d+d'))$.

    The bounds obtained in the time complexity analysis also imply that our representation
    of~$\OccE_k(P,T)\cap (I+|Q|\mathbb{Z})$
    consists of $\cO(|I|d'(d+k))$ arithmetic progressions.
\end{proof}
\begin{algorithm}
    \SetKwBlock{Begin}{}{end}
    \SetKwFunction{SynchedMatches}{SynchedMatches}
    \SynchedMatches{$P$, $T$, $I$, $k$, $d$, $d'$, $Q$}\Begin{
        $\mathcal{L}^P\gets \Locked{$P$, $Q$, $d$, $k$}$\;
        $\mathcal{L}^T\gets \Locked{$T$, $Q$, $d'$, $0$}$\;
        $M \gets \fragment{n-m-k}{n-m+k}$\;
        \ForEach{$P\fragmentco{\ell}{r}\in \mathcal{L}^P$}{
            \ForEach{$T\fragmentco{\ell'}{r'}\in \mathcal{L}^T$}{
                $M\gets M \cup\fragmentoo{\ell'-r-k}{r'-\ell+k}$\;
            }
        }
        $M \gets M \cap \fragmentco{0}{n-m+k}$\;
        \ForEach{maximal range $J\subseteq M$}{
            \ForEach{maximal range $J'\subseteq J\cap (I+|Q|\mathbb{Z})$}{
                $\OccE_k(P, T) \gets \OccE_k(P, T) \cup
                \{ pos \mid (pos, k_{pos}) \in \verify{$P$, $T$, $k$, $J'$}
                \}$\;\label{ln:maxrangever}
            }
        }
        \ForEach{maximal range $\fragmentco{\ell}{r}\subseteq \fragmentco{0}{n-m+k}\setminus M$}{
            $J \gets \fragmentco{\ell}{\min(r,\ell+|Q|)}$\;
            \ForEach{maximal range $J'\subseteq J\cap (I+|Q|\mathbb{Z})$}{
                \ForEach{$(pos, k_{pos}) \in \verify{$P$, $T$, $k$, $J'$}$\label{ln:maxrangever2}
                }{
                    $\OccE_k(P, T)\gets \OccE_k(P, T)\cup ((pos+|Q|\mathbb{Z})\cap
                    \fragmentco{\ell}{r})$\;
                }
            }
        }
        \Return{$\OccE_k(P, T)$}\;
    }
    \caption{Computing $k$-error occurrences in the presence of locked regions in text
    and pattern.}\label{alg:sth}
\end{algorithm}

\begin{lemma}[\texttt{PeriodicMatches($P$, $T$, $k$, $d$, $Q$)}]\label{lm:milpermat}
    Let $P$ denote a pattern of~length $m$ and let $T$ denote a text of~length $n$.
    Further, let $0 \le k \le m$ denote a threshold, let $d \ge 2k$
    denote a positive integer, and let $Q$ denote a primitive string that satisfies
    $|Q| \le m/8d$ and $\edl{P}{Q} \le d$. 

    There is an algorithm that computes the set $\OccE_k(P,T)$,
    using $\Oh(n/m\cdot d^4)$ time in the \modelname model.
\end{lemma}
\begin{proof}
    We consider $\floor{2n/m}$ blocks $T_0, \dots, T_{\floor{2n/m}-1}$ of~$T$, each
    of~length at most $\threehalfs m+k-1$, where the $i$-th block starts at position
    $i\cdot m/2$, that is, \[
        T_i := T\fragmentco{\floor{i\cdot {m}/2}}
        {\min\{n, \floor{(i+3)\cdot {m}/2} + k - 1\}}.
    \]
    Observe that each $k$-error occurrence of~$P$ in $T$ is contained in at least one of~the
    fragments $T_i$:
    Specifically, $T_i$ covers all occurrences starting in $\fragmentco{\floor{i\cdot
    m/2}}{\floor{(i+1)\cdot m/2}}$.
    For each block $T_i$, we call \texttt{FindRelevantFragment($P$, $T_i$, $k$, $d$, $Q$)}
    from~\cref{lem:erelevant}
    and obtain a fragment $T'_i=T\fragmentco{\ell_i}{r_i}$ containing all $k$-error
    occurrences of~$P$ in $T_i$ and an integer range $I_i$.
    \cref{lem:erelevant} guarantees that $\ed(T'_i,Q)\leq 3d$ and $|I_i|\le 6d+1$.
    Next, we call \texttt{SynchedMatches($P$, $T'_i$, $I_i$, $d$, $3d$, $k$, $Q$)} from
    \cref{lem:sth}.
    The output of~the call to \texttt{SynchedMatches} consists of
    $\cO((6d + 1)3d(d+k))=\cO(d^3)$ arithmetic progressions.
    For each obtained arithmetic progression, we first add $\floor{i\cdot m/2}$ to all
    of~its elements, and, if $i<{\floor{2n/m}-1}$, we intersect the resulting arithmetic
    progression with $\fragmentco{\floor{i\cdot m/2}}{\floor{(i+1)\cdot m/2}}$; finally,
    we add the obtained set to $\OccE_k(P,T)$.
    The intersection step guarantees that each
    $k$-error occurrence is accounted for by exactly one block.

    For the correctness, note that by \cref{lem:erelevant}, for each $i$, we have
    $\OccE_k(P,T'_i)\bmod |Q| \subseteq I_i \bmod |Q|$.
    Hence, the call to \texttt{SynchedMatches} indeed computes all occurrences of~$P$ in $T_i$.

    Each call to \texttt{FindRelevantFragment} requires $\cO(d^2)$ time, while each call
    to \texttt{SynchedMatches}
    requires time $\cO(3kd(d+k)(k + (6d+1) + d+3d))=\cO(d^4)$. The claimed overall running
    time follows.
\end{proof}

\subsection{Computing Occurrences in the Non-Periodic Case}

\begin{algorithm}[t]
    \SetKwBlock{Begin}{}{end}
    \SetKwFunction{verify}{Verify}
    \SetKwFunction{brmtch}{BreakMatches}
    \SetKwFunction{exmtch}{ExactMatches}
    \brmtch{$P$, $T$, $\{ B_1 = P\fragmentco{b_1}{b_1 + |B_1|}, \dots, B_{2k} =
        P\fragmentco{b_{2k}}{b_{2k} + |B_{2k}|} \}$, $k$}\Begin{
        multi-set $M \gets \{\}$; $\OccE_k(P, T) \gets \{\}$\;
        \For{$i \gets 1$ \KwSty{to} $2k$}{
            \ForEach{$\tau \in \exmtch{$B_i$, $T$}$}{
                $M \gets M\cup\{ \floor{(\tau - b_i - k) / k}\}$\tcp*{Mark block
                $\floor{(\tau - b_i - k) / k}$ of~$T$}
                $M \gets M\cup\{ \floor{(\tau - b_i) / k}\}$\tcp*{Mark block
                $\floor{(\tau - b_i) / k}$ of~$T$}
                $M \gets M\cup\{ \floor{(\tau - b_i + k) / k}\}$\tcp*{Mark block
                $\floor{(\tau - b_i + k) / k}$ of~$T$}
                $M \gets M\cup\{ \floor{(\tau - b_i + 2k) / k}\}$\tcp*{Mark block
                $\floor{(\tau - b_i + 2k) / k}$ of~$T$}
            }
        }
        sort $M$\;
        \ForEach{$\pi\in \fragment{0}{n-m}$ that appears at least $k$ times in $M$}{
            $\OccE_k(P, T) \gets \OccE_k(P, T) \cup
            \{ pos \mid (pos, k_{pos}) \in \verify{$P$, $T$, $k$, $\fragmentco{\pi \cdot
            k}{(\pi + 1)\cdot k}$} \}$\;
        }
        \Return{$\OccE_k(P, T)$}\;
    }
    \caption{A \modelname model algorithm for \cref{lm:EdC}.}\label{alg:EdC}
\end{algorithm}
\begin{lemma}[{\tt BreakMatches($P$, $T$, $\{B_1,\dots,B_{2k}\}$, $k$)}:
    Implem. of~\cref{lm:EdC}]\label{lm:impEdA}
    Let $k$ denote a threshold and
    let $P$ denote a pattern of~length $m$ having $2k$ disjoint breaks $B_1,\dots,B_{2k}
    \substr P$ each satisfying $\per(B_i) \ge m / \alphav k$.
    Further, let $T$ denote a text of~length $n \le \threehalfs m + k$.

    Then, we can compute the set $\OccE_k(P, T)$ using $\Oh(k^3)$
    time in the \modelname model.
\end{lemma}
\begin{proof}
    We proceed similarly to \cref{lm:imphdA}: Instead of~marking positions, we now mark
    blocks of~length~$k$; in the end, we verify complete blocks at once using \verify
    from \cref{lm:verifye}. Consider \cref{alg:EdC} for the complete algorithm visualized
    as pseudo-code.

    For the correctness, note that we have placed the marks as in the proof~of~\cref{lm:EdC};
    in particular, by \cref{cl:Ec2}, any block $\fragmentco{jk}{(j+1)k}$ that contains any
    position $\pi\in \OccE_k(P,T)$ has at least $k$ marks.
    As we verify each such block using \verify from \cref{lm:verifye},
    we report no false positives, and thus the algorithm is correct.

    We continue with analyzing the running time.
    As every break~$B_i$ has period $\per(B_i)>m/\alphav k$,
    every call to {\tt ExactMatches} uses $\Oh(k)$ time in the \modelname model by \cref{lm:emath};
    thus, all calls to {\tt ExactMatches} in total take $\Oh(k^2)$ time in total.
    Next, by \cref{cl:Ec1}, we place at most $\Oh(k^2)$ marks in $T$, so
    the marking step uses $\Oh(k^2)$ operations in total.
    Further, finding all positions in $T$ with at least $k$ marks can be done via a linear scan
    over the multi-set $M$ of~all marks after sorting $M$, which can be done in time
    $\Oh(k^2 \log \log k)$.
    Finally, as there are at most $\Oh(k^2 / k) = O(k)$ blocks that we verify, and every
    call to {\tt Verify} takes time $\Oh(k^2)$ in the  \modelname model,
    the verifications take $\Oh(k^3)$ time in the \modelname in total.
    Overall, \cref{alg:EdC} thus takes $\Oh(k^3)$ time in the \modelname model.
\end{proof}

\begin{algorithm}[t]
    \SetKwBlock{Begin}{}{end}
    \SetKwFunction{appm}{PeriodicMatches}
    \SetKwFunction{rpmtch}{RepetitiveMatches}
    \rpmtch{$P$, $T$, $\{ (R_1 = P\fragmentco{r_1}{r_1 + |R_1|}, Q_1) \dots, (R_{r} =
        P\fragmentco{r_{r}}{r_{r} + |R_{r}|}, Q_r) \}$, $k$}\Begin{
        multi-set $M \gets \{\}$; $\OccE_k(P, T) \gets \{\}$\;
        \For{$i \gets 1$ \KwSty{to} $r$}{
            set $M_i \gets \{\}$\;
            \ForEach{$\tau \in \appm{$R_i$, $T$, $\floor{\betavh \cdot k/m \cdot |R_i|}$,
                $\ceil{\betav\cdot k/m \cdot |R_i|}$, $Q_i$}$}{
                $M_i \gets M_i\cup\{ (\floor{(\tau - r_i - k)/k}, |R_i|) \}$\tcp*{Place $|R_i|$
                    marks at bl.$\floor{(\tau - r_i - k)/k}$}
                $M_i \gets M_i\cup\{ (\floor{(\tau - r_i)/k}, |R_i|) \}$\tcp*{Place $|R_i|$ marks at
                block $\floor{(\tau - r_i)/k}$}
                $M_i \gets M_i\cup\{ (\floor{(\tau - r_i + k)/k}, |R_i|) \}$\tcp*{Place $|R_i|$
                    marks at bl.$\floor{(\tau - r_i + k)/k}$}
                $M_i \gets M_i\cup\{ (\floor{(\tau - r_i + 2k)/k}, |R_i|) \}$\tcp*{Place $|R_i|$
                    marks at bl.$\floor{(\tau - r_i + 2k)/k}$}
            }
            $M \gets M \cup M_i$\;
        }
        sort $M$ by positions\;
        \ForEach{$\pi\in \fragment{0}{n-m}$ appearing at least
            $\sum_{(\pi, v) \in M} v \ge \sum_{i=1}^r |R_i| - m/\betavh$ times in $M$}{
            $\OccE_k(P, T) \gets \OccE_k(P, T) \cup
            \{ pos \mid (pos, k_{pos}) \in \verify{$P$, $T$, $k$, $\fragmentco{\pi \cdot
            k}{(\pi + 1)\cdot k}$} \}$\;
        }
        \Return{$\OccE_k(P, T)$}\;
    }
    \caption{A \modelname model algorithm for \cref{lm:EdB}.}\label{alg:EdB}
\end{algorithm}
\begin{lemma}[{\tt RepetitiveMatches($P$,$T$,$\{ (R_1, Q_1) \dots, (R_{r},Q_r)\}$,$k$)}:
    Implementation of~\cref{lm:EdB}]\label{lm:impEdB}
    Let $P$ denote a pattern of~length~$m$
    and let $k \le m$ denote a threshold.
    Further, let $T$ denote a string of~length~$n \le \threehalfs m + k$.
    Suppose that $P$ contains disjoint repetitive regions $R_1,\ldots, R_{r}$
    of~total length at least $\sum_{i=1}^r |R_i| \ge \deltavN/\deltavD\cdot m$
    such that each region $R_i$ satisfies $|R_i| \ge m/\betav k$ and has a
    primitive approximate period~$Q_i$
    with $|Q_i| \le m/\alphav k$ and $\hd(R_i,Q_i^*) = \ceil{\betav k/m\cdot |R_i|}$.

    Then, we can compute the set $\OccE_k(P,T)$ using $\Oh(k^4)$ time in the \modelname model.
\end{lemma}
\begin{proof}
    As in the proof~of~\cref{lm:EdB},
    set $m_R := \sum_{i=1}^r |R_i| \ge \deltavN/\deltavD\cdot m$ and
    define for every $1 \le i \le r$ the values
    $k_i := \floor{\betavh \cdot k/m \cdot |R_i|}$ and  $d_i := \ceil{\betav \cdot k/m
    \cdot |R_i|}=|\MIS(R_i, Q_i^*)|$.
    Further, write $R_i = P\fragmentco{r_i}{r_i + |R_i|}$.

    Again, we proceed similarly to the Hamming distance setting (\cref{lm:imphdB}).
    However, instead of~marking positions, we now mark
    blocks of~length $k$; in the end, we then verify complete blocks at once using \verify
    from \cref{lm:verifye}. Note that we need to ensure that we mark a block of~$T$
    only at most once for each repetitive part $R_i$; we do so by first computing a set
    of~all blocks to be marked due to $R_i$ (thereby removing duplicates) and
    then merging the sets computed for every $R_i$ into a multi-set.
    Consider \cref{alg:EdB} for the complete algorithm visualized
    as pseudo-code.

    For the correctness, first note that in every call to \appm from
    \cref{lm:milpermat}, we have \[
        \tbetav k/m \cdot |R_i| \ge d_i =  \ceil{\betav k/m\cdot |R_i|}  = \hd(R_i,
        Q_i^*) \ge 2k_i,\]
    hence $|Q_i| \le m/\alphav k \le |R_i|/8d_i$; thus, we can indeed call \appm in this case.
    Further, note that we have placed the marks as in the proof~of~\cref{lm:EdB};
    in particular, by \cref{cl:Eb2}, any block $\fragment{jk}{(j+1)k}$ that contains
    any position $\pi \in \OccE_k(P,T)$ has at least $m_R - m/\betavh$ marks.
    As we verify every possible candidate using \verify from \cref{lm:verifye},
    we report no false positives, and thus the algorithm is correct.

    For the running time in the \modelname model,
    observe that during the marking step, for every repetitive region~$R_i$
    we call \appm once. In total, all calls to \appm take
    \[\sum_i \Oh(n/|R_i|\cdot d_i^4)=\sum_i \Oh(|R_i|/m \cdot k^4) = \Oh(k^4)\]
    time in the \modelname model.
    Further, for every~$R_i$, we place at most
    $\Oh(|\floor{\OccE_{k_i}(R_i, T)/k}|)$ (weighted) marks,
    which can be bounded by
    $\Oh(|\floor{\OccE_{k_i}(R_i, T)/k}|) = \Oh(n/|R_i| \cdot d_i) = \Oh(k)$  using
    \cref{cor:Eaux}.
    Thus, we place $|M| = \Oh(k^2)$ (weighted) marks in total.
    Hence, the marking step in total takes $\Oh(k^4)$ time in the \modelname model.

    As the multi-set $M$ contains at most $\Oh(k^2)$ (weighted) marks,
    we can sort $M$ (by positions) in time $\Oh(k^2 \log\log k)$;
    afterwards, we can find the elements with total weight at least $m_R - m/\betavh$
    via a linear scan over $M$ in time $\Oh(k^2)$.
    As there are (by \cref{cl:Eb1,cl:Eb2}) at most $O(k)$ blocks with
    at least $m_R - m/\betavh$ marks, we call \verify
    at most $\Oh(k)$ times. As we call verify always on a whole block of~length~$k$ at
    once, each call to \verify takes $\Oh(k^2)$ time in the \modelname model.
    Hence, the verification step in total takes $\Oh(k^3)$ time in the \modelname
    model.

    In total, \cref{alg:EdB} thus takes time $\Oh(k^4)$ in the \modelname model.
\end{proof}

\subsection{A \modelname Model Algorithm for Pattern Matching with Edits}

Finally, we are ready to prove \cref{thm:edalgI}.

\edalgI*
\begin{algorithm}[t]
    \SetKwBlock{Begin}{}{end}
    \SetKwFunction{nonpermatch}{EditOccurrences}
    \nonpermatch{$P$, $T$, $k$}\Begin{
        {\bf (} $B_1,\dots,B_{2k}$ {\bf or} $(R_1, Q_1),\dots,(R_r, Q_r)$ {\bf or}
        $Q$ {\bf )} $\gets \anly{$P$, $k$}$\;
        $\OccE_k(P, T) \gets \{\}$\;
        \If{approximate period $Q$ exists}{
             \Return{\appm{$P$, $T$, $k$, $\betav k$, $Q$}}\;
        }
        \For{$i \gets 0$ \KwSty{to} $\floor{2n/m}-1$}{
            $T_i \gets T\fragmentco{\floor{i\cdot {m}/2}}   {\min\{n,
            \floor{(i+3)\cdot {m}/2}-1 + k\}}$\;
        \If{breaks $B_1,\dots,B_{2k}$ exist}{
             $\OccE_k(P, T_i) \gets \brmtch{$P$, $T_i$, $\{ B_1, \dots, B_{2k}\}$, $k$}$\;
        }\ElseIf{repetitive regions $(R_1, Q_1),\dots,(R_r, Q_r)$ exist}{
             $\OccE_k(P, T_i) \gets \rpmtch{$P$, $T_i$, $\{(R_1, Q_1),\dots,(R_r, Q_r)\}$, $k$}$\;
        }
        \If{$i<\floor{2n/m}-1$}{
             $V \gets \{\ell + \floor{i\cdot m/2} \mid \ell \in \OccE_k(P,T_i)\} \cap
             \fragmentco{\floor{i\cdot m/2}}{\floor{(i+1)\cdot m/2}}$\;
        }\label{ln:int}
        $\OccE_k(P, T)\gets \OccE_k(P, T) \cup V $\;\label{ln:union}
    }
    \Return{$\OccE_k(P, T)$}\;
    }
    \caption{Computing $k$-error occurrences in the \modelname model.}\label{alg:5.1}
\end{algorithm}
\begin{proof}
    We proceed, as in \cref{thm:hdalg}, by separately considering each of~the three
    possible outcomes of \anly$(P, k)$. Consider \cref{alg:5.1} for a visualization
    of~the whole algorithm as pseudo-code.

    If there is an approximate period $Q$ of~$P$
    we call \appm (from \cref{lm:milpermat}).
    Else, for each of~the $\floor{2n/m}$ blocks $T_0, \dots, T_{\floor{2n/m}-1}$, where
    $T_i := T\fragmentco{\floor{i\cdot {m}/2}} {\min\{n, \floor{(i+3)\cdot {m}/2} + k - 1\}}$,
    we call \brmtch (from \cref{lm:impEdA}) or \rpmtch
    (from \cref{lm:impEdB}), depending on the case we are in, and add the computed
    occurrences in $\OccE_k(P, T)$.

	The correctness in the approximately periodic case follows from \cref{lm:milpermat}
    and the fact that we can indeed call \appm since, due to \cref{prp:EIalg}, string $Q$
    satisfies $\edl{P}{Q} \le \betav k$ and
	$|Q|\le m/\alphav k \le m/(8\cdot \betav k)$.
	In the other cases, first observe that each length-$(m + k)$ fragment of
    $T$ is contained in at least one of~the fragments $T_i$ and hence we do not lose any
    occurrences.
    Second, by \cref{prp:EIalg} and due to $|T_i|\le \threehalfs m + k$,
    the parameters in the calls to \brmtch (from \cref{lm:impEdA})  and \rpmtch
    (from \cref{lm:impEdB}) each satisfy the requirements.
    Finally, the intersection step in \cref{ln:int} of~\cref{alg:5.1} guarantees that we
    account for each $k$-error occurrence exactly once.

    For the running time in the \modelname model, we have that the call to \anly
    takes $\Oh(k^2)$ time in the \modelname model,
    the call to \appm takes $\Oh(n/m \cdot k^4)$ time in the \modelname model,
    each call to \brmtch takes $\Oh(k^3)$ time in the \modelname model,
    and each call to \rpmtch takes $\Oh(k^4)$ time in the \modelname model.
    Finally, as the output of~our calls to \brmtch and \rpmtch is of~size $\cO(k^2)$ and is sorted,
    \cref{ln:int,ln:union} require $\Oh(k^2)$ time.
    As there are at most $\Oh(n/m)$ calls to \brmtch and \rpmtch,
    we can bound the total time in the \modelname model by $\Oh(n/m \cdot k^4)$,
    completing the proof.
\end{proof}

\section{Implementing the \modelname Model: Faster Approximate Pattern Matching}\label{sec:model}

In this section we implement the \modelname model in the static, fully compressed and
dynamic settings, thereby lifting \cref{thm:hdalg,thm:edalgI} to these settings.
For each setting, we first show how to implement each of~the following primitive
\modelname operations:
\begin{itemize}
    \item $\extractOpName(S,\ell,r)$: Retrieve a string $S\fragment{\ell}{r}$.
    \item $\lceOp{S}{T}$: Compute the length of~the longest common prefix of~$S$ and $T$.
    \item $\ipmOp{P}{T}$: Assuming that $|T|\le 2|P|$, compute $\OccEx(P,T)$ (represented
        as an arithmetic progression with difference $\per(P)$).
    \item $\accOpName(S,i)$: Retrieve the character $\accOp{S}{i}$.
    \item $\lenOpName(S)$: Compute the length $|S|$ of~the string $S$.
\end{itemize}
\noindent Operation $\lcbOp{S}{T}$ can be implemented analogously to
$\lceOp{S}{T}$ by reversing all the strings in scope.
We then apply our main algorithmic results (\cref{thm:hdalg,thm:edalgI})
in order to obtain efficient algorithms for approximate pattern matching.

\subsection{Implementing the \modelname Model in the Standard Setting}

We start with the implementation of~the \modelname model in the standard setting.
This turns out to be a straightforward application of~known tools for strings.

Recall that in the \modelname model, we are to maintain a collection $\X$ of~strings.
Let us denote the total length of~all strings in $\X$ by $n$.
In the standard setting, a handle to $S=X\fragmentco{\ell}{r}$ is implemented as a pointer to $X\in \X$ (recall that $X$ is stored explicitly) along with the indices $\ell$ and $r$.
Hence, the implementations of~$\extractOpName$, $\accOpName$, and $\lenOpName$ are trivial.

Next, $\lceOp{S}{T}$ queries can by implemented efficiently as follows.
We construct the generalized suffix tree for the elements of~$\X$ in $\cO(n)$
time~\cite{F97} and preprocess it within the same time complexity
so that we can support $\cO(1)$-time lowest common ancestor queries~\cite{Bender2000}.

As for efficiently answering $\ipmOp{P}{T}$ queries, we build the data structure of
Kociumaka et al.~\cite{IPM,thesis}, encapsulated in the following theorem,
for the concatenation of~the elements of~$\X$.
\begin{lemma}[\cite{IPM,thesis}]
    For every string $S$ of~length $n$, there is a data structure of~size $\cO(n)$, which
    can be constructed in $\cO(n)$ time and answers $\ipmOp{P}{T}$ queries in $\cO(1)$
    time for fragments $P$ and $T$ of~$S$.\lipicsEnd
\end{lemma}

We summarize the above discussion in the following theorem.

\begin{theorem}\label{thm:pilis}
    After an $\cO(n)$-time preprocessing of~a collection of~strings of~total length $n$,
    each \modelname operation can be performed in $\Oh(1)$ time.\lipicsEnd
\end{theorem}

This yields the following theorems, which are not new, but as a warm-up,
they are perhaps somewhat instructive.

Combining~\cref{thm:pilis,thm:hdalg}, we obtain an algorithm for pattern matching with
mismatches with the same running time as the algorithm of~Clifford et
al.~\cite{CliffordFPSS16},
which is essentially optimal when $k=\Oh(\sqrt{m})$.\footnote{Strictly speaking, the algorithm in
\cite{CliffordFPSS16} runs in time $\Oh(n\, {\rm polylog}(m) + n/m \cdot k^2\log k)$, so our
algorithm is slightly faster. However, an even better improvement in the logarithmic
factors was already obtained recently in \cite{cgkkp20}.}
\begin{theorem}
    Given a text $T$ of~length $n$, a pattern $P$ of~length $m$ and a threshold $k$,
    we can compute the set $\Occ_k(P, T)$ in time $\Oh(n + n/m \cdot k^2 \log \log
    k)$.\lipicsEnd
\end{theorem}

Similarly, combining~\cref{thm:pilis,thm:edalgI}, we obtain an algorithm for pattern
matching with edits that is, again, not slower than the known algorithm \cite{ColeH98}.
\begin{theorem}
    Given a text $T$ of~length $n$, a pattern $P$ of~length $m$ and a threshold $k$,
    we can compute the set $\OccE_k(P, T)$ in time $\Oh(n + n/m \cdot k^4)$.\lipicsEnd
\end{theorem}

\begin{remark}
    The discussion of~this subsection implies that our algorithms also apply to the
    internal setting.
    That is, a string $S$ of~length $n$ can be preprocessed in $\cO(n)$ time, so that
    given fragments $P$~and~$T$ of~$S$, and a threshold $k$, we can compute
    $\Occ_k(P, T)$ in time $\Oh(|T|/|P| \cdot k^2 \log \log k)$ and $\OccE_k(P, T)$ in time
    $\Oh(|T|/|P| \cdot k^4)$.\lipicsEnd
\end{remark}

\subsection{Implementing the \modelname Model in the Fully Compressed Setting}

Next, we focus on the fully compressed setting, where we want to solve approximate pattern
matching when both the text and the pattern are given as a straight-line programs---that is, in this
setting, we maintain a collection $\X$ of~straight-line programs and show how to implement
the primitive \modelname operations on this collection.
We start with a short exposition on straight-line programs and related concepts.

\paragraph*{Straight-Line Programs}

We denote the set of~non-terminals of~a context-free grammar $\G$ by $N_\G$
and call the elements of~$\mathcal{S}_G = N_\G \cup \Sigma$ symbols.
Then, a \emph{straight line program} (\emph{SLP}) $\G$ is a context-free grammar
that consists of~a set $N_\G=\{ A_1 , \ldots , A_n \}$ of~non-terminals, such that each $A
\in N_\G$ is associated with a unique production rule $A \to f_\G(A)$, where $f_\G(A) \in
\mathcal{S}_\G^*$.
For SLPs given as input, we can assume without loss of~generality that each production
rule is of~the form $A \to BC$ for some symbols $B$ and $C$ (that is, the given SLP is in
Chomsky normal form).

Every symbol $A \in \mathcal{S}_\G$ generates a unique string, which we denote by $\gen(A) \in
\Sigma^*$. The string $\gen(A)$ can be obtained from $A$ by repeatedly replacing each
non-terminal by its production. In addition, $A$ is associated
with its \emph{parse tree} $\Tr(A)$ consisting of~a root labeled with $A$ to which zero or
more subtrees are attached:
\begin{itemize}
    \item If $A$ is a terminal, there are no subtrees.
    \item If $A$ is a non-terminal $A\to B_1 \cdots B_p$, then $\Tr(B_i)$ are attached in
        increasing order of~$i$.
\end{itemize}
\noindent Note that if we traverse the leaves of~$\Tr(A)$ from left to right, spelling out the
corresponding non-terminals, then we obtain $\gen(A)$.

The parse tree $\Tr_\G$ of~$\G$ is the parse tree of~the starting symbol $A_n \in N_\G$;
$\gen(A_n)=S$, where $S$ is the unique string generated by $\G$.
We write $\gen(\G) := S$.
Finally, an SLP can be represented naturally as a directed acyclic graph $H_\G$ of~size
$|\mathcal{S}_\G|$. Consult~\cref{fig:slp} for an example of~an SLP, its parse tree, and the
corresponding acyclic graph.

\begin{figure}
        \begin{center}
        \begin{subfigure}[b]{.22\textwidth}
            \begin{center}
                \begin{tikzpicture}
                    \foreach \x/\r in {5/{$A_4 A_4$}, 
                    4/{$A_1 A_3$}, 3/{$A_1 A_2$}, 2/{\tt b}, 1/{\tt a}}{
                        \node(\x) at (0,\x/1) {$A_{\x} \rightarrow$ \r};
                    }
                \end{tikzpicture}
                \caption{}
            \end{center}
        \end{subfigure}
        \begin{subfigure}[b]{.48\textwidth}
            \begin{center}
                \begin{tikzpicture}
                    \node(60) at (-3,1) {\tt a};
                    \node(61) at (-2,1) {\tt a};
                    \node(71) at (-1,1) {\tt b};
                    \node(63) at (0,1) {\tt a};
                    \node(62) at (1,1) {\tt a};
                    \node(72) at (2,1) {\tt b};
                    \node(14) at (-3,3) {$A_1$};
                    \node(11) at (-2,2) {$A_1$};
                    \node(21) at (-1,2) {$A_2$};
                    \node(13) at (0,3) {$A_1$};
                    \node(12) at (1,2) {$A_1$};
                    \node(22) at (2,2) {$A_2$};
                    \foreach \i in {1,...,3}{
                        \draw (1\i) to (6\i);
                        \if\i3\else
                            \draw (2\i) to (7\i);
                        \fi
                    }
                    \node(31) at (-1.5,3) {$A_3$};
                    \node(32) at (1.5,3) {$A_3$};
                    \foreach \i in {1, 2}{
                        \draw (3\i) to (1\i);
                        \draw (3\i) to (2\i);
                    }
                    \node(4) at (0, 4) {$A_4$};
                    \node(42) at (-1.5, 4) {$A_4$};
                    \draw (42) to (31);
                    \draw (4) to (13);
                    \draw (42) to (14);
                    \draw (14) to (60);
                    \draw (4) to (32);
                    \node(5) at (-.75, 5) {$A_5$};
                    \draw (5) to (4);
                    \draw (5) to (42);
                \end{tikzpicture}
                \caption{}
            \end{center}
        \end{subfigure}
        \begin{subfigure}[b]{.19\textwidth}
            \begin{center}
                \begin{tikzpicture}
                    \foreach \c in {3, ..., 5}{
                        \node(\c) at (0,\c/1) {$A_\c$};
                    }
                    \node(1) at (-.5,2/1) {$A_1$};
                    \node(2) at (.5,2/1) {$A_2$};
                    \node(6) at (-.5,1/1) {\tt a};
                    \node(7) at (.5,1/1) {\tt b};
                    \foreach \c [count=\x from 3] in { {1,2}, {1,3}, {4,4}}{
                        \foreach \d [count=\y from 0] in \c {
                            \if\y0
                                \draw[->] (\x.south)++(-.1,0) to[bend right] (\d.north);
                            \else
                                \draw[->] (\x.south)++(.05,0) to[bend left] (\d.north);
                            \fi
                        }
                    }
                    \draw[->] (1) to (6);
                    \draw[->] (2) to (7);
                \end{tikzpicture}
                \caption{}
            \end{center}
        \end{subfigure}
        \caption{(a) An SLP $\mathcal{G}$ generating {\tt aabaab}. (b) The corresponding parse tree $\mathsf{PT}_\mathcal{G}$.
        (c) The corresponding directed acyclic graph $H_\mathcal{G}$.}\label{fig:slp}
    \end{center}
\end{figure}
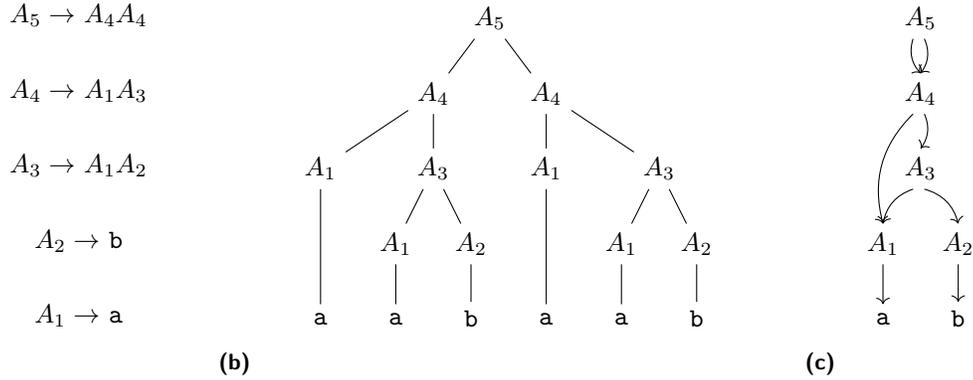

We define the \emph{value} $\val(v)$ of~a node $v$ in $\Tr_\G$ to be the fragment
$S\fragment{a}{b}$ corresponding to the leaves
$S\position{a},\ldots, S\position{b}$ in the subtree of~$v$.
Note that $\val(v)$ is an occurrence of~$\gen(A)$ in $\gen(\G)$, where $A$ is the label of~$v$.
A sequence of~nodes in $\Tr_\G$ is a \emph{chain} if their values are consecutive fragments in $T$.

Given an SLP $\G$ of~size $n$, we can compute $|\gen(\G)|$ in $\cO(n)$ time using dynamic
programming.
We compute the topological order of~$H_\G$ and process the nodes in the reverse order:
For each node corresponding to a non-terminal $A$ with production rule $A \to BC$, we just
need to compute $|\gen(B)|+|\gen(C)|$.

Bille et al.~\cite{BilleLRSSW15} have shown that we can efficiently access any character in $\gen(\G)$.

\begin{theorem}[\cite{BilleLRSSW15}]\label{thm:slp_access}
    An SLP $\G$ of~size $n$, generating a string $S$ of~size $N$, can be preprocessed in
    time $\cO(n \log (N/n))$ so that, for any $i \in \fragmentco{0}{N}$,
    we can access $\gen(\G)\position{i}$ in $\cO(\log N)$ time.\lipicsEnd
\end{theorem}

I in~\cite{I17} presented an efficient data structure for answering longest common prefix queries
for suffixes of~a string given by an SLP, which we encapsulate in the following theorem.
This data structure is based on the recompression technique, which we discuss in the next
subsection in more detail.

\begin{theorem}[\cite{I17}]\label{thm:slp_lce}
    An SLP $\G$ of~size $n$, generating a string $S$ of~size $N$, can be preprocessed in
    time $\cO(n \log (N/n))$ so that for any $i$ and $j$, we can compute
    $\lceOp{S\fragmentco{i}{N}}{S\fragmentco{j}{N}}$ in $\cO(\log N)$ time.\lipicsEnd
\end{theorem}

Finally, we discuss how to ``concatenate'' two SLPs.
Given two SLPs $\G_1$ and $\G_2$, with $\gen(\G_1)=S_1$ and $\gen(\G_2)=S_2$,
we can construct an SLP generating $S_1 S_2$ in $\cO(|\G_1|+|\G_2|)$ time as follows.
We first rename the non-terminals in $N_{\G_2}$ to make sure that they are disjoint from
the non-terminals in $N_{\G_1}$.
Next, let $R_1$ and $R_2$ denote the starting non-terminals of~$\G_1$ and $\G_2$, respectively.
We construct a new SLP $\G$ with $N_\G := N_{\G_1}\cup N_{\G_2} \cup \{R\}$, where
$R$ has production rule $R \to R_1 R_2$.
Note that this procedure can be applied to more than two strings: We first apply a global renaming,
and then repeatedly ``concatenate'' two strings in the collection.

Let us now denote the total size of~all SLPs in $\X$ by $n$, and the total length of~all strings
generated by those SLPs by $N$.
We implement the handle of a fragment $S=X\fragmentco{\ell}{r}$ so that it consists of a pointer to the SLP $\G\in \X$ generating $X$ along with the positions $\ell$ and $r$; this makes $\extractOpName$
and $\lenOpName$ trivial.

The above discussion on computing the length of~the string generated by an SLP
implies that the handles to $\gen(\G)\fragmentco{0}{|\gen(\G)|}$ for all $\G\in \X$ can be constructed in $\Oh(n)$ time.
Moreover, \cref{thm:slp_access} implies that we can preprocess $\X$ in $\cO(n \log(N/n))$ time so that
operation $\accOpName$ requires $\cO(\log N)$ time.

For efficiently answering \lceOpName queries, we rely on~\cref{thm:slp_lce}. We build I's
data structure for an SLP that generates the concatenation of~all elements in the
multi-set $\{\gen(\G) \mid \G \in \X\}$. Thus, after an $\cO(n \log (N/n))$-time
preprocessing, each \lceOpName operation takes $\cO(\log N)$ time.

To implement the \ipmOpName operation efficiently, we rely on the recompression technique due to
Je{\.z}~\cite{talg/Jez15,jacm/Jez16}, which we discuss next.

\paragraph*{Recompression of~Straight-Line Programs}

We start with some additional notation.
A \emph{run-length straight line program} (\emph{RLSLP}) is a straight-line program $\G$
that contains two kinds of~non-terminals:
\begin{itemize}
    \item \emph{Concatenations}: Non-terminals with production rules of
        the form $A \to BC$ (for symbols $B$ and $C$).
    \item \emph{Powers}: Non-terminals with production rules of~the form $A \to B^p$ (for
        a symbol $B$ and an integer $p\ge 2$).
\end{itemize}

\newcommand{\R}{\mathcal{R}}

The key idea of~the \emph{recompression} technique by Je\.{z}~\cite{talg/Jez15,jacm/Jez16}
is the construction of~a particular RLSLP~$\mathcal{R}$ that generates the input string $S$.
The parse tree $\Tr_{\R}$ is of~depth $\Oh(\log N)$ and it can be traversed efficiently based on the underlying directed acyclic graph $H_{\R}$. 
In particular, the name of~the technique stems from the fact that an SLP $\G$ of~size $n$
generating a string $S$ of~length $N$ can be efficiently recompressed to
the RLSLP $\R(\G)$ in-place, that is, without first uncompressing $\G$;
efficiently here means in $\cO(n \log N)$ time.
As observed by I~\cite{I17}, the parse tree~$\Tr_{\R(\G)}$ is \emph{locally consistent}
in a certain sense. To formalize this property, he introduced the \emph{popped sequence}
of~every fragment $S\fragment{a}{b}$, which is a sequence of~symbols labeling a certain
chain of~nodes of $\Tr_{\R(\G)}$, whose values constitute $S\fragment{a}{b}$.


\begin{theorem}[\cite{I17}]\label{thm:recomp}
    If two fragments match, then their popped sequences are equal.
    Moreover, each popped sequence consists of~$\Oh(\log N)$ runs (maximal powers of~a
    single symbol) and can be constructed in $\Oh(\log N)$ time.
    The nodes corresponding to symbols in a run share a single parent.
    Furthermore, the popped sequence consists of~a single symbol only for fragments of
    length $1$. \lipicsEnd
\end{theorem}

Let $F_1^{p_1}\cdots F_t^{p_t}$ denote the run-length encoding of~the popped sequence of~a
substring $U$ of~$S$ and set \[L(U) :=
    \{|\gen(F_1)|, |\gen(F_1^{p_1})|, |\gen(F_1^{p_1}F_2^{p_2})|, \ldots,
        |\gen(F_1^{p_1}\cdots F_{t-1}^{p_{t-1}})|,|\gen(F_1^{p_1}\cdots
    F_{t-1}^{p_{t-1}}F_{t}^{p_{t}-1})|\}.
\]
By \cref{thm:recomp}, the set $L(U)$ can be constructed in $\Oh(\log N)$ time given a
fragment $S\fragment{a}{b}=U$.
Now, the following lemma from~\cite{IDM19} allows us to efficiently implement internal pattern
matching queries.

\begin{lemma}[\cite{IDM19}]\label{lem:rec_sync}
    Let $v$ denote a non-leaf node of~$\Tr_{\R(\G)}$ and let $S\fragment{a}{b}$ denote an occurrence of
    $S$ that is contained in $\val(v)$, but not contained in $\val(u)$ for any child $u$ of~$v$.
    If $S\fragment{a}{c}$ is the longest prefix of~$S\fragment{a}{b}$ contained in $\val(u)$ for
    a child $u$ of~$v$,
    then $|S\fragment{a}{c}| \in L(U)$.
    Symmetrically, if $S\fragment{c'+1}{b}$ is the longest suffix of~$S\fragment{a}{b}$
    contained in $\val(u)$ for a child $u$ of~$v$, then $|S\fragment{a}{c'}| \in L(U)$.\lipicsEnd
\end{lemma}

Finally, we discuss how to implement the \ipmOpName operation.

\begin{lemma}\label{lem:rec_ipm}
    Given $H_{\R(\G)}$, a fragment $T=S\fragmentco{j}{j+\nu}$, and a fragment
    $P=S\fragmentco{i}{i+\mu}$ with $|T|\leq 2|P|$, we can compute
    $\OccEx(P,T)$ in the time required by $\cO(\log^2 N)$ \lceOpName and \lcbOpName
    operations on fragments of~$S$.

    In particular, the \lceOpName and \lcbOpName operations are between pairs consisting in one of
    $\cO(\log N)$ fragments of $P$ and one of $\cO(\log N)$ fragments of $S$.
\end{lemma}
\begin{proof}
    We can assume that $|T|\ge |P|>1$; otherwise it suffices to perform a constant number of
    letter comparisons which can be done in $\cO(\log N)$ time.
    We first compute the popped sequence of~$P$ and $L(P)$ using~\cref{thm:recomp}.
    Let $v \in \Tr_{\R(\G)}$ denote the lowest common ancestor of~the leaves representing
    $S\position{j}$ and $S\position{j+\nu-1}$;
    the node $v$ can be naively computed in $\cO(\log N)$ time by a forward search from
    the root of~the parse tree.
    As $T$ is a fragment of~$\val(v)$, all occurrences of~$P$ in $T$ are also contained in
    $\val(v)$. Our first aim is to compute all occurrences of~$P$ in $T$ that are not
    contained in $\val(u)$ for any child $u$ of~$v$; we then appropriately recurse on the
    children of~$v$ that may contain sought occurrences.

    Let us first analyze the case that the label of~$v$ is a concatenation symbol $A \to
    BC$. Write $v_\ell$ for the left child of~$v$ and $v_r$ for the right child of~$v$.
    Further, let $T_{\ell}=S\fragmentco{j}{t}$ denote the longest prefix of~$T$ that is
    completely contained in $\val(v_\ell)$ and, similarly, let
    $T_{r}=S\fragmentco{t}{j+\nu}$ denote the longest suffix of~$T$ that is completely
    contained in $\val(v_r)$.
    Suppose that there is a fragment $U=T\fragment{a}{b}$ of~$T$ that equals $P$ and
    overlaps with both $\val(v_\ell)$ and $\val(v_r)$.
    The fragment $U$ can then be naturally decomposed into a non-empty suffix $U_\ell$ of
    $T_\ell$ and a non-empty prefix $U_r$ of~$T_r$.
    \Cref{lem:rec_sync} implies that $|U_\ell| \in L(P)$.
    It thus suffices to check for each $q \in L(P)$ whether $P\fragmentco{0}{q}$ is a
    suffix of~$T_\ell$ and $P\fragmentco{q}{m}$ is a prefix of~$T_r$.
    There are $|L(P)|=\Oh(\log N)$ choices for $q$, and for each of~them we can perform
    the check using $\lcbOp{P\fragmentco{0}{q}}{T_\ell}$ and $\lceOp{P\fragmentco{q}{m}}{T_r}$
    operations.

    We now consider the case that the label of~$v$ is a power symbol $A \to B^p$ and
    denote the children of~$v$ in the left-to-right order by $v_1, \ldots, v_p$.
    Let $T\fragmentco{x}{y}$ denote the overlap of~$T$ with $\val(v)$.

    If $T$ overlaps with $\val(v_d)$ for two children $v_d$ of~$v$,
    then we can process $v$ as in the previous case.
    In the case that~$T$ overlaps with
    exactly three children of~$v$, some care is needed
    to avoid double-counting occurrences that overlap with all of~them.
    In particular, let these three children be $v_x, v_{x+1}$, and $v_{x+2}$.
    We consider separately:
    \begin{itemize}
    \item occurrences that overlap with both $\val(v_x)$ and $\val(v_{x+1})$ by setting
    $T_{\ell}$ to be the longest prefix of~$T$ that is
    completely contained in $\val(v_x)$ and $T_{r}$ to be
     the longest suffix of~$T$ that is completely
    contained in $\val(v_{x+1})\val(v_{x+2})$, and
    \item occurrences that overlap with both $\val(v_{x+1})$ and $\val(v_{x+2})$, but not with $\val(v_{x})$, by setting
    $T_{\ell}$ to be $\val(v_{x+1})$ and $T_{r}=S\fragmentco{t}{j+\nu}$ to be
     the longest suffix of~$T$ that is completely
    contained in $\val(v_{x+2})$.
    \end{itemize}

    We can thus assume that $T$ overlaps with $\val(v_d)$ for more than three children
    $v_d$ of~$v$. In that case,  for all $d$, we have $\val(v_d)<\mu$ and
    hence no occurrence of~$P$ in $T\fragmentco{x}{y}$ can be completely contained in
    $\val(v_d)$. We set $T_\ell := \val(v_1)$ and $T_r := \val(v_2) \cdots \val(v_p)$.
    Using $|L(P)|$ many \lceOpName operations
    and $|L(P)|$ many \lcbOpName operations,
    we can compute the set $Y$ of~occurrences of~$P$ in~$S$ (but not necessarily in $T$)
    which can be decomposed into a prefix $U_\ell$ that is a suffix of~$T_\ell$ and a
    suffix~$U_r$ that is a prefix of~$T_r$.

    Then, by the periodicity of~$\val(v)=\gen(B)^p$, the desired set of~occurrences
    is \[
        Z := \{ i + j\cdot |\gen(B)| : i \in Y, j \in \fragment{0}{p-1} \} \cap
    \fragment{x}{y-\mu}.
\]
    Note that $Z$ trivially decomposes into $\cO(\log N)$ arithmetic progressions;
    these arithmetic progressions can be replaced by a single arithmetic progression  with
    difference $\per(P)$ in $\cO(\log N)$ time.

    If the overlap of~the value of~each child of~$v$ with $T$ has length less than
    $\mu$, we terminate the algorithm.
    Otherwise, if $v$ has two children whose values have an overlap of~length $\mu$ with
    $T$, we check whether either of~them is equal to $P$
    using a single \lceOpName operation and terminate the algorithm.
    Finally, at most one of~$v$'s children has length greater than $\mu$. In that
    case, we repeat the above procedure for this child.
    As the depth of~$\Tr_{\R(\G)}$ is $\cO(\log N)$, the overall running time is upper bounded
    by the time required for $\cO(\log^2 N)$ \lceOpName and \lcbOpName operations.

    In the end, we have $\OccEx(P,T)$ represented by at most one arithmetic progression
    and $\cO(\log^2 N)$ single occurrences.
    We postprocess this representation in $\cO(\log^2 N)$ time, in order to represent
    $\OccEx(P,T)$ with a single arithmetic progression with difference $\per (P)$.

    Note that, in each level, the \lceOpName (resp.~\lcbOpName) queries we perform are
    between $\cO(1)$ fragments of $S$ and each $P\fragmentco{0}{q}$ (resp.~$P\fragmentco{q}{m}$) for $q \in L(P)$.
    This observation implies the last claim of the statement of this lemma, concluding its proof.
\end{proof}

\begin{remark}
Essentially the same proof of the above lemma has recently appeared in~\cite{KK20}.\lipicsEnd
\end{remark}

We can get an $\cO(\log^3 N)$-time implementation of \ipmOpName by employing~\cref{lem:rec_ipm}, and answering each \lceOpName or \lcbOpName query in $\cO(\log N)$ time.
However, the structure of the queries allows for a more efficient implementation.

\begin{lemma}
We can preprocess an SLP $\G$ of size $n$, generating a string $S$ of length $N$, in $\cO(n \log N)$ time, so that
$\ipmOp{P}{T}$ queries for fragments $P$ and $T$ of $S$ can be answered in $\cO(\log^2 N \log\log N)$ time.\lipicsEnd
\end{lemma}
\begin{proof}
We first build, in $\cO(n \log(N/n))$ time, the data structures of~\cref{thm:slp_access,thm:slp_lce} for performing \accOpName,
\lceOpName, and \lcbOpName operations in $\cO(\log N)$ time.
Then, we recompress $\G$ to an RLSLP $\R(\G)$ in $\cO(n \log N)$ time.

\cref{lem:rec_ipm} then reduces the task at hand to answering $\cO(\log^2 N)$ \lceOpName and \lcbOpName queries between pairs consisting in one of
$\cO(\log N)$ fragments of $P$ and one of $\cO(\log N)$ fragments of $S$.
Let us focus on efficiently answering all such \lceOpName queries; \lcbOpName queries can be answered analogously.

We sort all fragments in scope using $\cO(\log N \log \log N)$ comparisons, implementing each comparison
in $\cO(\log N)$ time using an \lceOpName operation, followed by two \accOpName operations.
This step thus takes $\cO(\log^2 N \log\log N)$ time.
Then, we construct an array $A$ of size $\cO(\log N)$ such that $A\position{i}$ stores the length of the longest
common prefix of the $i$-th and the $(i+1)$-st elements in our sorted list.
After preprocessing array $A$ in $\cO(\log N)$ time so that range minimum queries over it can be answered in constant time~\cite{Bender2000},
we answer each of the $\cO(\log^2 N)$ \lceOpName queries in $\cO(1)$ time.
\end{proof}

In total, we have thus proved the following result.

\begin{theorem}\label{thm:pilgc}
    Given a collection of~SLPs of~total size $n$, generating strings of~total length $N$,
    each \modelname operation can be performed in $\cO(\log^2 N \log\log N)$ time
    after an $\cO(n \log N)$-time preprocessing.\lipicsEnd
\end{theorem}

In the next subsection (cf.~\cref{rem:rand_pilgc}), we discuss how to perform each \modelname operation in $\cO(\log^2 N)$
after an $\cO(n \log N)$-time preprocessing at the cost of randomization.

\paragraph*{Approximate Pattern Matching in Fully Compressed Strings}

We are now ready to present efficient algorithms for approximate pattern matching in the
fully compressed setting.
We choose to state our results using our deterministic implementation of the \modelname model, that is~\cref{thm:pilgc}.

We are given an SLP $\G_T$ of~size $n$ with $T=\gen(\G_T)$, an SLP $\G_P$ of~size $m$ with
$P=\gen(\G_T)$ and a threshold $k$ and are required to compute
the $k$-mismatch or $k$-error occurrences of~$P$ in $T$.

Set $N:=|T|+|P|$ and $\X := \{\G_T, \G_P\}$.
The overall structure of~our algorithm is as follows:
We first preprocess the collection $\X$ in $\cO((n+m) \log N)$ time according to~\cref{thm:pilgc}.
Next, we traverse~$\G_T$ and compute for every non-terminal~$A$ of~$\G_T$ the approximate
occurrences of~$P$ in $T$ that ``cross'' $A$. Depending on the setting,  we combine
\cref{thm:pilgc} with \cref{thm:hdalg} or \cref{thm:edalgI} to compute the occurrences.
Finally, we combine the computed occurrences using dynamic programming.

Formally, for each non-terminal $A \in N_{\G_T}$, with production rule $A \to BC$, we
wish to compute all approximate occurrences of~$P$ in the string
\[\gen(B)\fragmentco{|\gen(B)|-|P|+1}{|\gen(B)|}\gen(C)\fragmentco{0}{|P|-1},\] which
is indeed a fragment of~$\gen(\G_T)$ and is of~length $2|P|-2$.
These approximate occurrences can be computed in time:
\begin{itemize}
    \item $\cO(k^2 \log^2 N \log\log N)$ in the Hamming distance case by combining~\cref{thm:pilgc,thm:hdalg};
    \item $\cO(k^4 \log^2 N \log\log N)$ in the edit distance case, by combining~\cref{thm:pilgc,thm:edalgI}.
\end{itemize}
Other approximate occurrences in $\gen(A)$ lie entirely in $\gen(B)$ or $\gen(C)$; hence
they are computed when considering $B$ and $C$. (Compare \cite[Theorem~4.1]{bkw19} for
a similar algorithm.)

Now, the number of~approximate occurrences of~$P$ in $T$ (that is $|\Occ_k(P, T)|$ or
$|\OccE_k(P, T)|$) can be computed by dynamic programming that is analogous to the dynamic
programming used to compute the length of~a string generated by an SLP.
Further, all approximate occurrences can be reported in time proportional to their number by
performing a traversal of~$\Tr_\G$, avoiding to explore subtrees that correspond to
fragments of~$T$ that do not contain any approximate occurrences.

We hence obtain the following algorithm for pattern matching with
mismatches in the fully compressed setting.
\gchdalgmain*

Similarly, we obtain the following algorithm for pattern
matching with edits in the fully compressed setting.
\gcedalgmain*

\subsection{Implementing the \modelname Model in the Dynamic Setting}

Lastly, we consider the dynamic setting. In particular, we consider the dynamic
maintenance of a collection of non-empty persistent strings $\X$ that is initially empty
and undergoes updates specified by the following operations:
\begin{itemize}
    \item $\makestring(U)$: Insert a non-empty string $U$ to $\X$.
    \item $\concat(U,V)$: Insert $UV$ to $\X$, for $U,V \in \X$.
    \item $\splitOp(U,i)$: Insert $U\fragmentco{0}{i}$ and $U\fragmentco{i}{|U|}$ in $\X$,
        for $U \in \X$ and $i \in \fragmentco{0}{|U|}$.
\end{itemize}

Let $N$ denote an upper bound on the total length of~all strings in $\X$
throughout the execution of~the algorithm.
Gawrychowski et al.~\cite{ods} presented a data structure that efficiently maintains such
a collection and allows for efficient longest common prefix queries.

\begin{theorem}[\cite{ods}]\label{thm:ods}
    A collection $\X$ of non-empty persistent strings of~total length $N$ can be
    dynamically maintained with update operations $\makestring(U)$, $\concat(U,V)$,
    $\splitOp(U,i)$ requiring time $\cO(\log N +|U|)$, $\cO(\log N)$, and $\cO(\log
    N)$,\footnote{These running times hold w.h.p.}
    respectively, so that $\lceOp{U}{V}$ queries for $U,V \in \X$ can be answered in time
    $\cO(1)$.\lipicsEnd
\end{theorem}

\lceOpName and \lcbOpName operations for arbitrary fragments of elements of $\X$ can be answered in $\cO(\log N)$ time (w.h.p.)
by first performing a constant number of \splitOp{} operations to add the corresponding fragments to the collection
and then asking an \lceOpName query between them.

The lengths of~the strings in $\X$ can be maintained explicitly.
Upon a $\makestring$ operation we naively compute the length of~$U$, while upon a
$\concat$ or a $\splitOp$ operation, we can compute the lengths of~the strings that are
inserted in $\X$ in constant time from the arguments of~the operation.

For each string of the collection $\X$, the data structure of~\cite{ods} maintains
an RLSLP stemming from recompression that is of~depth $\cO(\log N)$ w.h.p.
Given a string $X \in \X$, a pointer to the root of~the parse tree of~$X$ can be retrieved
in $\cO(1)$ time.

Each $\extractOpName(X,\ell,r)$ operation can be performed using at most two $\splitOp$ operations in $\cO(\log N)$ time w.h.p.

We now show that $\accOpName(X,i)$ for $X \in \X$ can be performed efficiently.
Although the parse trees of the strings in $\X$ are not maintained explicitly, given a pointer to some node $v$
in the parse tree of~$X$, we can retrieve in $\cO(1)$ time the endpoints $a,b$ of~the fragment
$\val(v)=X\fragment{a}{b}$, the degree of~$v$, a pointer to the parent of~$v$, and
a pointer to the $j$-th child of~$v$, provided that such a child exists.
Thus, an $\accOpName(X,i)$ operation can be implemented in time proportional to the height
of~the parse tree, that is, in $\cO(\log N)$ time w.h.p.

We are left with showing that \ipmOpName operations can be performed efficiently.
Let us remark that the RLSLPs of~all $X \in \X$ maintained by the data structure underlying~\cref{thm:ods}
are locally consistent with each other, that
is~\cref{thm:recomp} is also true for fragments of~different strings $X_1, X_2 \in \X$.
Thus, \cref{lem:rec_sync,lem:rec_ipm} also hold in this setting.
Combining~\cref{lem:rec_ipm,thm:ods} we get that $\ipmOp{P}{T}$ queries can be answered in
$\cO(\log^2 N)$ time (w.h.p.) by performing $\cO(\log N)$ \splitOp{} operations and $\cO(\log^2 N)$
\lceOpName queries (cf.~the last statement of~\cref{lem:rec_ipm}).
Note that the only other component of the proof of~\cref{lem:rec_ipm} is a forward search from the root of the
relevant parse tree,
which can be efficiently performed given the available pointers.

We summarize the above discussion in the following theorem.

\begin{restatable}{theorem}{thmpildyn}\label{thm:pildyn}
    A collection $\X$ of non-empty persistent strings of~total length $N$ can be
    dynamically maintained with operations $\makestring(U)$, $\concat(U,V)$,
    $\splitOp(U,i)$ requiring time $\cO(\log N +|U|)$,
    $\cO(\log N)$ and $\cO(\log N)$, respectively, so that \modelname operations can
    be performed in time $\Oh(\log^2 N)$.\footnote{All running time bounds hold w.h.p.}\ifx\thmpildynt\undefined\lipicsEnd\fi
\end{restatable}
\def\thmpildynt{1}

\begin{remark}\label{rem:rand_pilgc}
Given an SLP $\G$ of size $n$, generating a string $S$ of size $N$, we can efficiently implement the \modelname operations through
dynamic strings. Let us start with an empty collection $\X$ of dynamic strings.
Using $\cO(n)$ $\makestring(a)$ operations, for $a\in \Sigma$, and $\cO(n)$ $\concat$ operations (one for each non-terminal
of $\G$), we can insert $S$ to $\X$ in $\cO(n \log N)$ time w.h.p.
Then, we can perform each \modelname operation in $\cO(\log^2 N)$ time w.h.p., due to~\cref{thm:pildyn},
thus outperforming~\cref{thm:pilgc} at the cost of randomization.\lipicsEnd
\end{remark}

Combining \cref{thm:pildyn,thm:hdalg,thm:edalgI}, we obtain an algorithm for approximate pattern matching
for dynamic strings.
\dynalgmain*

\clearpage

\bibliographystyle{plainurl}
\bibliography{ms}

\begin{thebibliography}{10}

\bibitem{abbk17}
Amir Abboud, Arturs Backurs, Karl Bringmann, and Marvin K{\"{u}}nnemann.
\newblock Fine-grained complexity of analyzing compressed data: Quantifying
  improvements over decompress-and-solve.
\newblock In Chris Umans, editor, {\em 58th Annual {IEEE} Symposium on
  Foundations of Computer Science, {FOCS} 2017}, pages 192--203. {IEEE}
  Computer Society, 2017.
\newblock \href {https://doi.org/10.1109/FOCS.2017.26}
  {\path{doi:10.1109/FOCS.2017.26}}.

\bibitem{Abrahamson}
Karl~R. Abrahamson.
\newblock Generalized string matching.
\newblock {\em {SIAM} Journal on Computing}, 16(6):1039--1051, 1987.
\newblock \href {https://doi.org/10.1137/0216067} {\path{doi:10.1137/0216067}}.

\bibitem{abr00}
Stephen Alstrup, Gerth~St{\o}lting Brodal, and Theis Rauhe.
\newblock Pattern matching in dynamic texts.
\newblock In David~B. Shmoys, editor, {\em 11th Annual {ACM-SIAM} Symposium on
  Discrete Algorithms, {SODA} 2000}, pages 819--828. SIAM, 2000.
\newblock URL: \url{http://dl.acm.org/citation.cfm?id=338219.338645}.

\bibitem{AmirLP04}
Amihood Amir, Moshe Lewenstein, and Ely Porat.
\newblock Faster algorithms for string matching with $k$ mismatches.
\newblock {\em Journal of Algorithms}, 50(2):257--275, 2004.
\newblock \href {https://doi.org/10.1016/S0196-6774(03)00097-X}
  {\path{doi:10.1016/S0196-6774(03)00097-X}}.

\bibitem{BabenkoGKKS16}
Maxim Babenko, Paweł Gawrychowski, Tomasz Kociumaka, Ignat Kolesnichenko, and
  Tatiana Starikovskaya.
\newblock Computing minimal and maximal suffixes of a substring.
\newblock {\em Theoretical Computer Science}, 638:112--121, 2016.
\newblock \href {https://doi.org/10.1016/j.tcs.2015.08.023}
  {\path{doi:10.1016/j.tcs.2015.08.023}}.

\bibitem{bi18}
Arturs Backurs and Piotr Indyk.
\newblock Edit distance cannot be computed in strongly subquadratic time
  (unless {SETH} is false).
\newblock {\em SIAM Journal on Computing}, 47(3):1087--1097, 2018.
\newblock \href {https://doi.org/10.1137/15M1053128}
  {\path{doi:10.1137/15M1053128}}.

\bibitem{Bender2000}
Michael~A. Bender and Martin Farach{-}Colton.
\newblock The {LCA} problem revisited.
\newblock In {\em {LATIN} 2000: Theoretical Informatics, 4th Latin American
  Symposium, Punta del Este, Uruguay, April 10-14, 2000, Proceedings}, pages
  88--94, 2000.
\newblock \href {https://doi.org/10.1007/10719839_9}
  {\path{doi:10.1007/10719839_9}}.

\bibitem{BilleLRSSW15}
Philip Bille, Gad~M. Landau, Rajeev Raman, Kunihiko Sadakane, Srinivasa~Rao
  Satti, and Oren Weimann.
\newblock {Random Access to Grammar-Compressed Strings and Trees}.
\newblock {\em SIAM Journal on Computing}, 44(3):513--539, 2015.
\newblock \href {https://doi.org/10.1137/130936889}
  {\path{doi:10.1137/130936889}}.

\bibitem{Moore91}
Robert~S. Boyer and J.~Strother Moore.
\newblock {MJRTY}{\textemdash}a fast majority vote algorithm.
\newblock In Robert~S. Boyer, editor, {\em Automated Reasoning: Essays in Honor
  of Woody Bledsoe}, Automated Reasoning Series, pages 105--117. Kluwer
  Academic Publishers, 1991.
\newblock \href {https://doi.org/10.1007/978-94-011-3488-0_5}
  {\path{doi:10.1007/978-94-011-3488-0_5}}.

\bibitem{BG95}
Dany Breslauer and Zvi Galil.
\newblock Finding all periods and initial palindromes of a string in parallel.
\newblock {\em Algorithmica}, 14(4):355--366, 1995.
\newblock \href {https://doi.org/10.1007/BF01294132}
  {\path{doi:10.1007/BF01294132}}.

\bibitem{bkw19}
Karl Bringmann, Marvin K{\"{u}}nnemann, and Philip Wellnitz.
\newblock {Few Matches or Almost Periodicity: Faster Pattern Matching with
  Mismatches in Compressed Texts}.
\newblock In Timothy~M. Chan, editor, {\em 30th Annual {ACM-SIAM} Symposium on
  Discrete Algorithms, {SODA} 2019}, pages 1126--1145. {SIAM}, 2019.
\newblock \href {https://doi.org/10.1137/1.9781611975482.69}
  {\path{doi:10.1137/1.9781611975482.69}}.

\bibitem{BWT}
Michael Burrows and David~J. Wheeler.
\newblock A block-sorting lossless data compression algorithm.
\newblock Technical Report 124, Digital Equipment Corporation, Palo Alto,
  California, 1994.

\bibitem{cgkkp20}
Timothy~M. Chan, Shay Golan, Tomasz Kociumaka, Tsvi Kopelowitz, and Ely Porat.
\newblock Approximating text-to-pattern hamming distances.
\newblock In Julia Chuzhoy, editor, {\em 52nd Annual {ACM} Symposium on Theory
  of Computing, {STOC} 2020}, pages 643--656. ACM, 2020.
\newblock \href {https://doi.org/10.1145/3357713.3384266}
  {\path{doi:10.1145/3357713.3384266}}.

\bibitem{IDM19}
Panagiotis Charalampopoulos, Tomasz Kociumaka, Manal Mohamed, Jakub
  Radoszewski, Wojciech Rytter, and Tomasz Walen.
\newblock Internal dictionary matching.
\newblock In Pinyan Lu and Guochuan Zhang, editors, {\em 30th International
  Symposium on Algorithms and Computation, {ISAAC} 2019, December 8-11, 2019,
  Shanghai University of Finance and Economics, Shanghai, China}, volume 149 of
  {\em LIPIcs}, pages 22:1--22:17. Schloss Dagstuhl--Leibniz-Zentrum f{\"{u}}r
  Informatik, 2019.
\newblock \href {https://doi.org/10.4230/LIPIcs.ISAAC.2019.22}
  {\path{doi:10.4230/LIPIcs.ISAAC.2019.22}}.

\bibitem{CliffordFPSS16}
Rapha{\"{e}}l Clifford, Allyx Fontaine, Ely Porat, Benjamin Sach, and Tatiana
  Starikovskaya.
\newblock The $k$-mismatch problem revisited.
\newblock In Robert Krauthgamer, editor, {\em 27th Annual {ACM-SIAM} Symposium
  on Discrete Algorithms, {SODA} 2016}, pages 2039--2052. {SIAM}, 2016.
\newblock \href {https://doi.org/10.1137/1.9781611974331.ch142}
  {\path{doi:10.1137/1.9781611974331.ch142}}.

\bibitem{CGLS18}
Rapha{\"{e}}l Clifford, Allan Gr{\o}nlund, Kasper~Green Larsen, and Tatiana~A.
  Starikovskaya.
\newblock Upper and lower bounds for dynamic data structures on strings.
\newblock In Rolf Niedermeier and Brigitte Vall{\'{e}}e, editors, {\em 35th
  Symposium on Theoretical Aspects of Computer Science, {STACS} 2018}, pages
  22:1--22:14. Schloss Dagstuhl--Leibniz-Zentrum f{\"{u}}r Informatik, 2018.
\newblock \href {https://doi.org/10.4230/LIPIcs.STACS.2018.22}
  {\path{doi:10.4230/LIPIcs.STACS.2018.22}}.

\bibitem{ColeH98}
Richard Cole and Ramesh Hariharan.
\newblock {Approximate String Matching: A Simpler Faster Algorithm}.
\newblock {\em {SIAM Journal on Computing}}, 31(6):1761--1782, 2002.
\newblock \href {https://doi.org/10.1137/S0097539700370527}
  {\path{doi:10.1137/S0097539700370527}}.

\bibitem{F97}
Martin Farach.
\newblock Optimal suffix tree construction with large alphabets.
\newblock In {\em 38th Annual {IEEE} Symposium on Foundations of Computer
  Science, {FOCS} '97, Miami Beach, Florida, USA, October 19-22, 1997}, pages
  137--143, 1997.
\newblock \href {https://doi.org/10.1109/SFCS.1997.646102}
  {\path{doi:10.1109/SFCS.1997.646102}}.

\bibitem{FG98}
Paolo Ferragina and Roberto Grossi.
\newblock Optimal on-line search and sublinear time update in string matching.
\newblock {\em {SIAM} Journal on Computing}, 27(3):713--736, 1998.
\newblock \href {https://doi.org/10.1137/S0097539795286119}
  {\path{doi:10.1137/S0097539795286119}}.

\bibitem{GG86}
Zvi Galil and Raffaele Giancarlo.
\newblock Improved string matching with $k$ mismatches.
\newblock {\em SIGACT News}, 17(4):52--54, 1986.
\newblock \href {https://doi.org/10.1145/8307.8309}
  {\path{doi:10.1145/8307.8309}}.

\bibitem{DBLP:conf/isaac/GawrychowskiS13}
Pawel Gawrychowski and Damian Straszak.
\newblock Beating $\mathcal{O}(nm)$ in approximate {LZW}-compressed pattern
  matching.
\newblock In Leizhen Cai, Siu{-}Wing Cheng, and Tak~Wah Lam, editors, {\em 24th
  International Symposium on Algorithms and Computation, {ISAAC} 2013}, pages
  78--88. Springer, 2013.
\newblock \href {https://doi.org/10.1007/978-3-642-45030-3\_8}
  {\path{doi:10.1007/978-3-642-45030-3\_8}}.

\bibitem{GawrychowskiU18}
Pawel Gawrychowski and Przemyslaw Uzna\'nski.
\newblock Towards unified approximate pattern matching for {Hamming} and
  ${L}_1$ distance.
\newblock In Ioannis Chatzigiannakis, Christos Kaklamanis, D{\'{a}}niel Marx,
  and Donald Sannella, editors, {\em 45th International Colloquium on Automata,
  Languages, and Programming, {ICALP} 2018}, volume 107 of {\em LIPIcs}, pages
  62:1--62:13. Schloss Dagstuhl--Leibniz-Zentrum für Informatik, 2018.
\newblock \href {https://doi.org/10.4230/LIPIcs.ICALP.2018.62}
  {\path{doi:10.4230/LIPIcs.ICALP.2018.62}}.

\bibitem{ods}
Paweł Gawrychowski, Adam Karczmarz, Tomasz Kociumaka, Jakub Łącki, and Piotr
  Sankowski.
\newblock Optimal dynamic strings.
\newblock In Artur Czumaj, editor, {\em 29th Annual {ACM-SIAM} Symposium on
  Discrete Algorithms, {SODA} 2018}, pages 1509--1528. SIAM, 2018.
\newblock \href {http://arxiv.org/abs/1511.02612} {\path{arXiv:1511.02612}},
  \href {https://doi.org/10.1137/1.9781611975031.99}
  {\path{doi:10.1137/1.9781611975031.99}}.

\bibitem{Gu94}
Ming Gu, Martin Farach, and Richard Beigel.
\newblock An efficient algorithm for dynamic text indexing.
\newblock In Daniel~Dominic Sleator, editor, {\em 5th Annual ACM-SIAM Symposium
  on Discrete Algorithms, {SODA} 1994}, pages 697--704. {ACM/SIAM}, 1994.

\bibitem{HKNS15}
Monika Henzinger, Sebastian Krinninger, Danupon Nanongkai, and Thatchaphol
  Saranurak.
\newblock Unifying and strengthening hardness for dynamic problems via the
  online matrix-vector multiplication conjecture.
\newblock In Ronitt Rubinfeld, editor, {\em 47th Annual {ACM} on Symposium on
  Theory of Computing, {STOC} 2015}, pages 21--30. {ACM}, 2015.
\newblock \href {https://doi.org/10.1145/2746539.2746609}
  {\path{doi:10.1145/2746539.2746609}}.

\bibitem{I17}
Tomohiro I.
\newblock Longest common extensions with recompression.
\newblock In Juha K{\"{a}}rkk{\"{a}}inen, Jakub Radoszewski, and Wojciech
  Rytter, editors, {\em 28th Annual Symposium on Combinatorial Pattern
  Matching, {CPM} 2017}, volume~78 of {\em LIPIcs}, pages 18:1--18:15. Schloss
  Dagstuhl--Leibniz-Zentrum für Informatik, 2017.
\newblock \href {https://doi.org/10.4230/LIPIcs.CPM.2017.18}
  {\path{doi:10.4230/LIPIcs.CPM.2017.18}}.

\bibitem{talg/Jez15}
Artur Jeż.
\newblock Faster fully compressed pattern matching by recompression.
\newblock {\em {ACM} Transactions on Algorithms}, 11(3):20:1--20:43, 2015.
\newblock \href {https://doi.org/10.1145/2631920} {\path{doi:10.1145/2631920}}.

\bibitem{jacm/Jez16}
Artur Jeż.
\newblock Recompression: {A} simple and powerful technique for word equations.
\newblock {\em Journal of the {ACM}}, 63(1):4:1--4:51, 2016.
\newblock \href {https://doi.org/10.1145/2743014} {\path{doi:10.1145/2743014}}.

\bibitem{KK20}
Dominik Kempa and Tomasz Kociumaka.
\newblock Resolution of the {Burrows--Wheeler} transform conjecture.
\newblock In Sandy Irani, editor, {\em 61st Annual {IEEE} Symposium on
  Foundations of Computer Science, {FOCS} 2020}. {IEEE} Computer Society, 2020.

\bibitem{KP18}
Dominik Kempa and Nicola Prezza.
\newblock At the roots of dictionary compression: string attractors.
\newblock In Monika Henzinger, editor, {\em 50th Annual {ACM} Symposium on
  Theory of Computing, {STOC} 2018}, pages 827--840. {ACM}, 2018.
\newblock \href {https://doi.org/10.1145/3188745.3188814}
  {\path{doi:10.1145/3188745.3188814}}.

\bibitem{thesis}
Tomasz Kociumaka.
\newblock {\em Efficient Data Structures for Internal Queries in Texts}.
\newblock PhD thesis, University of Warsaw, October 2018.
\newblock URL: \url{https://www.mimuw.edu.pl/~kociumaka/files/phd.pdf}.

\bibitem{IPM}
Tomasz Kociumaka, Jakub Radoszewski, Wojciech Rytter, and Tomasz Walen.
\newblock Internal pattern matching queries in a text and applications.
\newblock In Piotr Indyk, editor, {\em 26th Annual {ACM-SIAM} Symposium on
  Discrete Algorithms, {SODA} 2015}, pages 532--551. {SIAM}, 2015.
\newblock \href {https://doi.org/10.1137/1.9781611973730.36}
  {\path{doi:10.1137/1.9781611973730.36}}.

\bibitem{Kosaraju}
S.R. Kosaraju.
\newblock Efficient string matching.
\newblock Manuscript, 1987.

\bibitem{LandauV86}
Gad~M. Landau and Uzi Vishkin.
\newblock Efficient string matching with $k$ mismatches.
\newblock {\em Theoretical Computer Science}, 43:239--249, 1986.
\newblock \href {https://doi.org/10.1016/0304-3975(86)90178-7}
  {\path{doi:10.1016/0304-3975(86)90178-7}}.

\bibitem{LandauV89}
Gad~M. Landau and Uzi Vishkin.
\newblock Fast parallel and serial approximate string matching.
\newblock {\em Journal of Algorithms}, 10(2):157--169, 1989.
\newblock \href {https://doi.org/10.1016/0196-6774(89)90010-2}
  {\path{doi:10.1016/0196-6774(89)90010-2}}.

\bibitem{LM00}
N.J. Larsson and A.~Moffat.
\newblock Off-line dictionary-based compression.
\newblock {\em Proceedings of the {IEEE}}, 88(11):1722--1732, 2000.
\newblock \href {https://doi.org/10.1109/5.892708}
  {\path{doi:10.1109/5.892708}}.

\bibitem{l12}
Markus Lohrey.
\newblock Algorithmics on {SLP}-compressed strings: {A} survey.
\newblock {\em Groups Complexity Cryptology}, 4(2):241--299, 2012.

\bibitem{ksu97}
Kurt Mehlhorn, R.~Sundar, and Christian Uhrig.
\newblock {Maintaining Dynamic Sequences under Equality Tests in
  Polylogarithmic Time}.
\newblock {\em Algorithmica}, 17(2):183--198, 1997.
\newblock \href {https://doi.org/10.1007/BF02522825}
  {\path{doi:10.1007/BF02522825}}.

\bibitem{nw97}
Craig~G Nevill-Manning and Ian~H Witten.
\newblock Compression and explanation using hierarchical grammars.
\newblock {\em The Computer Journal}, 40(2 and 3):103--116, 1997.

\bibitem{nii20}
Takaaki Nishimoto, Tomohiro I, Shunsuke Inenaga, Hideo Bannai, and Masayuki
  Takeda.
\newblock Dynamic index and {LZ} factorization in compressed space.
\newblock {\em Discrete Applied Mathematics}, 274:116--129, 2020.
\newblock \href {https://doi.org/10.1016/j.dam.2019.01.014}
  {\path{doi:10.1016/j.dam.2019.01.014}}.

\bibitem{ab10}
Roberto Radicioni and Alberto Bertoni.
\newblock Grammatical compression: compressed equivalence and other problems.
\newblock {\em Discrete Mathematics and Theoretical Computer Science},
  12(4):109, 2010.

\bibitem{r03}
Wojciech Rytter.
\newblock {Application of Lempel-Ziv factorization to the approximation of
  grammar-based compression}.
\newblock {\em Theoretical Computer Science}, 302(1-3):211--222, 2003.
\newblock \href {https://doi.org/10.1016/S0304-3975(02)00777-6}
  {\path{doi:10.1016/S0304-3975(02)00777-6}}.

\bibitem{r04}
Wojciech Rytter.
\newblock Grammar compression, {LZ}-encodings, and string algorithms with
  implicit input.
\newblock In Josep D{\'{\i}}az, Juhani Karhum{\"{a}}ki, Arto Lepist{\"{o}}, and
  Donald Sannella, editors, {\em 31st International Colloquium on Automata,
  Languages, and Programming, {ICALP} 2004}, pages 15--27. Springer, 2004.
\newblock \href {https://doi.org/10.1007/978-3-540-27836-8\_5}
  {\path{doi:10.1007/978-3-540-27836-8\_5}}.

\bibitem{SV96}
S{\"{u}}leyman~Cenk Sahinalp and Uzi Vishkin.
\newblock Efficient approximate and dynamic matching of patterns using a
  labeling paradigm (extended abstract).
\newblock In Martin Tompa, editor, {\em 37th Annual {IEEE} Symposium on
  Foundations of Computer Science, {FOCS} 1996}, pages 320--328. {IEEE}
  Computer Society, 1996.
\newblock \href {https://doi.org/10.1109/SFCS.1996.548491}
  {\path{doi:10.1109/SFCS.1996.548491}}.

\bibitem{sa14}
Hiroshi Sakamoto.
\newblock Grammar compression: Grammatical inference by compression and its
  application to real data.
\newblock In {\em 12th International Conference on Grammatical Inference,
  {ICGI} 2014}, pages 3--20. JMLR.org, 2014.

\bibitem{shishia99}
Yusuxke Shibata, Takuya Kida, Shuichi Fukamachi, Masayuki Takeda, Ayumi
  Shinohara, Takeshi Shinohara, and Setsuo Arikawa.
\newblock Byte pair encoding: A text compression scheme that accelerates
  pattern matching.
\newblock Technical report, Technical Report DOI-TR-161, Department of
  Informatics, Kyushu University, 1999.

\bibitem{st94}
Rajamani Sundar and Robert~E Tarjan.
\newblock {Unique binary-search-tree representations and equality testing of
  sets and sequences}.
\newblock {\em SIAM Journal on Computing}, 23(1):24--44, 1994.

\bibitem{t14}
Alexander Tiskin.
\newblock Threshold approximate matching in grammar-compressed strings.
\newblock In Jan Holub and Jan Žďárek, editors, {\em Prague Stringology
  Conference, {PSC} 2014}, pages 124--138, 2014.

\bibitem{w84}
T.~Welch.
\newblock A technique for high-performance data compression.
\newblock {\em Computer}, 17:8--19, 1984.

\bibitem{lz77}
Jacob Ziv and Abraham Lempel.
\newblock A universal algorithm for sequential data compression.
\newblock {\em IEEE Transactions on Information Theory}, 23(3):337--343, 1977.

\bibitem{lz78}
Jacob Ziv and Abraham Lempel.
\newblock Compression of individual sequences via variable-rate coding.
\newblock {\em {IEEE} Transactions on Information Theory}, 24(5):530--536,
  1978.

\end{thebibliography}


\end{document}